\documentclass[format=acmsmall, screen=true, authorversion=true]{acmart}
\pdfoutput=1

\usepackage{packages}
\usepackage{macros}

\acmJournal{TWEB}
\acmVolume{9}
\acmNumber{4}
\acmArticle{39}
\acmYear{2010}
\acmMonth{3}
\copyrightyear{2009}

\setcopyright{acmlicensed}

\setcopyright{acmcopyright}
\acmJournal{TOCL}
\acmYear{2021} \acmVolume{1} \acmNumber{1} \acmArticle{1} \acmMonth{1} \acmPrice{15.00}\acmDOI{10.1145/3448269}

\received{September 2020}
\received[accepted]{January 2021}

\begin{document}

\title{The Effects of Adding Reachability Predicates in Quantifier-Free Separation Logic}

\author{St{\'e}phane Demri}
\affiliation{%
  \institution{Universit\'e Paris-Saclay, ENS Paris-Saclay, CNRS, LSV, 91190, Gif-sur-Yvette}
  \country{France}
  }
\email{demri@lsv.fr}
\author{Etienne Lozes}
\affiliation{%
  \institution{I3S, Universit{\'e} C{\^o}te d'Azur}
  \country{France}
}
\email{elozes@i3s.unice.fr}
\author{Alessio Mansutti}
\affiliation{%
  \institution{Universit\'e Paris-Saclay, ENS Paris-Saclay, CNRS, LSV, 91190, Gif-sur-Yvette}
  \country{France}
  }
\email{alessio.mansutti@lsv.fr}
\begin{abstract}
The list segment predicate $\ls$ used in separation logic for verifying programs with pointers
is well-suited to express properties on singly-linked lists. 
We study the effects of adding $\ls$ to the full quantifier-free
separation logic with the separating conjunction and implication, which is motivated
by the recent design of new fragments in which all these ingredients are used indifferently and verification tools
start to handle the magic wand connective. 
This is a very natural extension that has not been studied so far. 
We show that the restriction without the separating implication can be solved in polynomial space
by using an appropriate abstraction for memory states whereas the full extension is shown undecidable
by reduction from first-order separation logic. Many variants of  the logic and  fragments are also investigated
from the computational point of view when $\ls$ is added, providing numerous results
about adding reachability predicates to quantifier-free separation logic. 
\end{abstract}
%
%



\maketitle
\section{Introduction}
Separation logic~\cite{Ishtiaq&OHearn01,OHearn&Reynolds&Yang01,Reynolds02} is a well-known assertion logic
for reasoning about programs with dynamic data structures. Since the implementation of Smallfoot
and the evidence that the method is scalable~\cite{Berdine&Calcagno&OHearn05,Yangetal08},
many tools supporting separation
logic as an assertion language have been
developed~\cite{Berdine&Calcagno&OHearn05,Distefano&OHearn&Yang06,Yangetal08,Calcagno&DiStefano11,Calcagnoetal11,Haaseetal13}.
Even though the first tools could handle relatively limited fragments of separation logic, like symbolic heaps~\cite{Cooketal11}, there is a growing
interest and demand to consider extensions with richer expressive power.
We can point out three particular extensions of symbolic heaps (without list predicates) that have been proved decidable.
\begin{itemize}
\item Symbolic heaps with generalised inductive predicates, adding a fixpoint combinator to the language, is a convenient logic for
  specifying data structures that are more advanced than lists or trees. The entailment problem  is known to be decidable
  by means of tree automata techniques for the bounded tree-width fragment~\cite{Iosif&Rogalewicz&Simacek13,Antonopoulosetal14}, whereas satisfiability
  is \exptime-complete~\cite{Brotherstonetal14}. Other related results can be found in~\cite{Leetal17,Katelaan&Matheja&Zuleger19}.
\item List-free symbolic heaps with all classical Boolean connectives $\wedge$ and $\neg$ (and with the separating conjunction $\separate$),
      called herein $\seplogic{\separate}$,
     is a convenient extension when combinations of results of various analysis need to be expressed, or when
  the analysis requires a complementation. This extension already is \pspace-complete~\cite{Calcagno&Yang&OHearn01}.
\item Quantifier-free separation logic with separating implication, a.k.a. magic wand $(\magicwand)$, is a convenient fragment (called herein
      $\seplogic{\separate, \magicwand}$) with decidable
  frame inference and abduction, two problems that play an important role in static analysers and provers built on top of separation logic.
  $\seplogic{\separate, \magicwand}$ can be decided in \pspace thanks to a small model
  property~\cite{Yang01}.
\end{itemize}

A natural question is how to combine these extensions, and which  separation
logic fragment
admitting Boolean connectives, magic wand and generalised recursive predicates can be decided with some adequate restrictions.
As already advocated in~\cite{Brotherston&Villard14,Thakur&Breck&Reps14bis,Schwerhoff&Summers15,Hou&Gore&Tiu15,Muller&Schwerhoff&Summers16},
dealing with
the separating implication $\magicwand$  is a desirable feature for program verification and several semi-automated or automated
verification tools support it in some way, see e.g.~\cite{Thakur&Breck&Reps14bis,Schwerhoff&Summers15,Muller&Schwerhoff&Summers16,Hou&Gore&Tiu15}.

Besides, allowing quantifications is another direction to extend the symbolic heap fragment:
in~\cite{Bozga&Iosif&Perarnau10}, an extension of the symbolic heap fragment with quantification over locations and over
arithmetic variables for list lengths is introduced and several fragments are shown decidable (the whole extension is
undecidable). Such an extension combines shape and arithmetic specifications (see also~\cite{David&Kroening&Lewis15}
for a theory of singly-linked lists with length combining such features) and the decidability results are obtained
by using so-called symbolic shape graphs that are finite representations
of sets of heaps. In the current paper, we consider only shape analysis (since herein, the heaps are restricted to a single record field)
but the separating implication is admitted.

\paragraph{Our contribution.}
In this paper, we address the question of combining the magic wand and inductive predicates in the extremely limited case where
the only inductive predicate is the gentle list segment predicate $\ls$.
The starting point of this work is this puzzling question:
what is the complexity/de\-ci\-da\-bi\-li\-ty status of quantifier-free separation logic $\seplogic{\separate,\magicwand}$ enriched with
the list segment predicate $\ls$ (herein called  $\seplogic{\separate,\magicwand,\ls}$)?
More precisely, we
study the decidability/complexity status of extensions of quantifier-free separation
logic $\seplogic{\separate, \magicwand}$ by adding one of the reachability predicates among $\ls$ (precise predicate as usual in separation logic),
$\reach$ (existence of a path, possibly empty) and $\reachplus$ (existence of a non-empty path).
At this point, it is worth noting  that in the presence of the separating conjunction $\separate$, $\ls$ and $\reach$ are
interdefinable, and $\reachplus$ can easily define   $\ls$  and $\reach$. Consequently, the complexity upper bounds will be stated
with $\reachplus$ and the complexity lower bounds or undecidability results are sharper with $\ls$ or $\reach$.

First, we establish that the satisfiability problem for the
quantifier-free separation logic $\seplogic{\separate, \magicwand, \ls}$
is undecidable. Our proof is by reduction
from the undecidability of first-order separation
logic $\foseplogic{\magicwand}$~\cite{Brochenin&Demri&Lozes12,DemriDeters16}, using an encoding of the
variables as heap cells (see Theorem~\ref{theorem-main-undecidability}).
As a consequence, we also establish that $\seplogic{\separate, \magicwand, \ls}$ is not finitely axiomatisable. Moreover, our reduction requires a rather
limited expressive power of the list segment predicate, and we can strengthen
our undecidability results to some fragments of $\seplogic{\separate, \magicwand, \ls}$. For instance, surprisingly,
the extension of $\seplogic{\separate, \magicwand}$ with the atomic formulae
of the form $\reach(\avariable,\avariablebis) = 2$
and $\reach(\avariable,\avariablebis) = 3$ (existence of a path between $\avariable$ and $\avariablebis$ of respective length 2 or 3)
is already undecidable, whereas the satisfiability problem for $\seplogic{\separate, \magicwand,\reach(\avariable,\avariablebis) = 2}$
is known to be in \pspace~\cite{Demrietal17}.

Second, we show that the satisfiability problem for
$\seplogic{\separate, \reachplus}$  is \pspace-complete, extending
the well-known result on $\seplogic{\separate}$. The \pspace upper bound relies
on a small heap property based on the techniques of test formulae,
see e.g.~\cite{Lozes04,Lozes04bis,Brochenin&Demri&Lozes09,Demrietal17,Mansutti18,Echenim&Iosif&Peltier20},
and the \pspace-hardness of $\seplogic{\separate}$ is inherited
from~\cite{Calcagno&Yang&OHearn01}. The \pspace upper bound can be extended to
the fragment of $\seplogic{\separate, \magicwand, \reachplus}$
made of Boolean combinations of formulae from  $\seplogic{\separate, \reachplus} \cup \seplogic{\separate, \magicwand}$
(see the developments
in Section~\ref{section-decidability-star}).
As a by-product of our proof technique,
we obtain that the satisfiability problem for Boolean combinations of pure formulae and spatial formulae from the
symbolic heap fragment, $\mathsf{Bool(SHF)}$, is \np-complete via a proof different from the one in~\cite{Piskac&Wies&Zufferey13}.
Figure~\ref{figure-recap} presents a summary of the main results of the paper. An unlabelled arrow between two logics means
that there is a many-one reduction between the satisfiability problem of the first logic and the problem for  the second one (sometimes,
the reduction is the identity in the case of syntactic fragments).

\begin{figure}
\scalebox{0.9}{
\begin{tikzpicture}
\node (undec-sl2) at (0,0)
  {$\seplogic{\separate,\sepimp,\reach(\avariable,\avariablebis) = 2,\reach(\avariable,\avariablebis) = 3}$\\
  \emph{undecidable}\\
  (Corollary~\ref{corollary-undecidability-limited-reach})
  };
  \node (undec-sl) [right = 2cm of undec-sl2]
    {$\seplogic{\separate,\magicwand,\ls}$\\
    \emph{undecidable}\\
    (Theorem~\ref{theorem-main-undecidability})
    };
  \node (slr) [below = of undec-sl2]
    {$\seplogic{\separate,\reachplus}$\\
    \emph{\pspace-complete}\\
    (Theorem~\ref{theorem-pspace})
    };
  \node (bool-shf) [below = of slr]
    {$\mathsf{Bool(SHF)}$\\
    \emph{\np-complete}\\
    (Corollary~\ref{corollary-bool-shf} and \cite{Piskac&Wies&Zufferey13})
    };
  \node (undec-fosl) [left = 2.5cm of undec-sl2]
    {$\foseplogic{\magicwand}$\\
    \emph{undecidable}\\
    \cite{Brochenin&Demri&Lozes12,DemriDeters16}
    };
  \node (psl) [left = of slr]
    {$\seplogic{\separate}$\\
    \emph{\pspace-complete}\\
    \cite{Calcagno&Yang&OHearn01}
    };
  \node (shf) [left = of bool-shf]
    {$\mathsf{SHF}$\\
    \emph{\ptime}\\
    \cite{Cooketal11,Haaseetal13}
    };

  \draw[pto] (undec-fosl) --node[above]{} (undec-sl2);
  \draw[pto] (undec-sl2) --node[above]{} (undec-sl);
  \draw[pto] (psl) --node[above]{} (slr);
  \draw[pto] (shf) --node[above]{} (bool-shf);
  \draw[pto] (bool-shf) --node[right]{} (slr);
  \draw[pto] (slr) --node[below right]{} (undec-sl);
\end{tikzpicture}
}
\caption{Main contributions.}
\label{figure-recap}
\end{figure}

This paper is an extended and completed version of~\cite{Demri&Lozes&Mansutti18}.
\section{Preliminaries}
\label{section-preliminaries}

\subsection{Separation logic with the list segment predicate}

Let $\PVAR = \set{\avariable, \avariablebis, \ldots}$ be a countably infinite set of \defstyle{program variables}
and $\LOC = \set{\alocation_0,\alocation_1, \alocation_2, \ldots}$ be a countable infinite set of \defstyle{locations}.
A \defstyle{memory state} is a pair $\pair{\astore}{\aheap}$ such that
$\astore: \PVAR \rightarrow \LOC$ is a variable valuation (known as the \defstyle{store})
and $\aheap: \LOC \to_{\fin} \LOC$ is a partial function with finite domain, known as the \defstyle{heap}.
We write $\domain{\aheap}$ to denote its domain and $\range{\aheap}$ to denote its range.
Given a heap $\aheap$ with $\domain{\aheap} = \set{\alocation_1, \ldots, \alocation_n}$,
we also write $\set{\alocation_1 \mapsto \aheap(\alocation_1), \ldots,\alocation_n \mapsto \aheap(\alocation_n)}$
to denote $\aheap$. Each $\alocation_i \mapsto \aheap(\alocation_i)$ is understood as a \defstyle{memory cell} of $\aheap$.

Let $\aheap_1$ and $\aheap_2$ be two heaps.
$\aheap_1$ and $\aheap_2$ are said to be \defstyle{disjoint}, written $\aheap_1 {\perp} \aheap_2$,
whenever their domains are disjoint, i.e.\ $\domain{\aheap_1} \cap \domain{\aheap_2} = \emptyset$.
When $\aheap_1 {\perp} \aheap_2$, we write $\aheap_1 + \aheap_2$ to denote the heap
corresponding to the disjoint union of the graphs of $\aheap_1$ and $\aheap_2$, hence $\domain{\aheap_1 + \aheap_2} = \domain{\aheap_1} \uplus \domain{\aheap_2}$.
If the domains of $\aheap_1$ and $\aheap_2$  are not disjoint, then the union $\aheap_1 + \aheap_2$ is not defined.
We say that $\aheap_1$ is a \emph{subheap} of $\aheap_2$, written $\aheap_1 \sqsubseteq \aheap_2$, if $\domain{\aheap_1} \subseteq \domain{\aheap_2}$ and for all
$\alocation \in \domain{\aheap_1}$, we have $\aheap_1(\alocation) = \aheap_2(\alocation)$.
For instance, if $\aheap_1 {\perp} \aheap_2$ then $\aheap_1 \sqsubseteq \aheap_1 + \aheap_2$.
Given a heap $\aheap$, we write $\aheap^i$ for the (partial) function obtained from $i$ functional composition(s) of $\aheap$.
By definition, $\aheap^0$ is the identity function on $\LOC$,
$\aheap^1 \egdef \aheap$ and for all $\beta \geq 2$ and $\alocation \in \LOC$, we have
$\aheap^{\beta}(\alocation) \egdef \aheap(\aheap^{\beta-1}(\alocation))$, assuming that $\aheap^{\beta-1}(\alocation)$ is defined and belongs to $\domain{\aheap}$, otherwise
$\aheap^{\beta}(\alocation)$ is undefined.

The formulae $\aformula$ of the separation logic $\seplogic{\separate, \magicwand, \ls}$ and its atomic formulae
$\aatomicformula$ are built from  the grammars below (where $\avariable, \avariablebis \in \PVAR$ and the connectives $\Rightarrow$, $\Leftrightarrow$ and $\vee$ are defined as usually).
$$
\aatomicformula ::=
                    \emptyconstant \ \mid \
                    \avariable = \avariablebis \ \mid \
                    \avariable \hpto \avariablebis \ \mid \
                    \ls(\avariable, \avariablebis)
                    $$
$$\aformula ::= \aatomicformula \ \mid \  \neg \aformula \ \mid \ \aformula \wedge \aformula
                 \ \mid \ \aformula \separate \aformula  \ \mid \ \aformula \magicwand \aformula
$$
Models of $\seplogic{\separate, \magicwand, \ls}$ are
memory states and the satisfaction relation $\models$ is defined as follows:
\begin{center}
\begin{tabular}[t]{lll}
  $\pair{\astore}{\aheap} \models \emp$ & $\iff$ & $\domain{\aheap} = \emptyset$.\\
  $\pair{\astore}{\aheap} \models \avariable = \avariablebis$ & $\iff$ & $\astore(\avariable) = \astore(\avariablebis)$.\\
  $\pair{\astore}{\aheap} \models \avariable \hpto \avariablebis$ & $\iff$  & $\astore(\avariable) \in \domain{\aheap}$
  and $\aheap(\astore(\avariable)) = \astore(\avariablebis)$.\\
  $\pair{\astore}{\aheap} \models \ls(\avariable,\avariablebis)$ & $\iff$ & either
  ($\domain{\aheap} = \emptyset$ and  $\astore(\avariable) = \astore(\avariablebis)$) or  \\
  & &  $\aheap = \set{\alocation_0 \mapsto \alocation_1, \alocation_1 \mapsto \alocation_2, \ldots,
         \alocation_{n-1}\mapsto \alocation_n}$ for some  $n \geq 1$,  \\
& &  $\alocation_0 = \astore(\avariable)$, $\alocation_n = \astore(\avariablebis)$   and  for all  $i \neq j \in \interval{0}{n}$, $\alocation_i \neq \alocation_j$.\\
  $\pair{\astore}{\aheap} \models \lnot \aformula$ & $\iff$ & $\pair{\astore}{\aheap} \not \models \aformula$.\\
  $\pair{\astore}{\aheap} \models \aformula_1 \land \aformula_2$  & $\iff$ & $\pair{\astore}{\aheap} \models \aformula_1$
  and $\pair{\astore}{\aheap}  \models \aformula_2$.\\
  $\pair{\astore}{\aheap} \models \aformula_1 \separate \aformula_2$ & $\iff$ & there are $\aheap_1$ and $\aheap_2$  such that
  $\aheap_1 \bot \aheap_2$, $\aheap_1 + \aheap_2 = \aheap$,  \\
  & &  $\pair{\astore}{\aheap_1}  \models \aformula_1$ and  $\pair{\astore}{\aheap_2}  \models \aformula_2$.\\
  $\pair{\astore}{\aheap} \models \aformula_1 \magicwand \aformula_2$ & $\iff$ & for all $\aheap_1$
  such that  $\aheap_1 \bot \aheap$ and  $\pair{\astore}{\aheap_1} \models \aformula_1$, we have
 $\pair{\astore}{\aheap + \aheap_1} \models \aformula_2$.
\end{tabular}
\end{center}
The semantics for $\separate$ (separating conjunction), $\magicwand$ (separating implication), $\hpto$ (points to), $\ls$ and for all other ingredients is the usual
one in separation logic, where $\ls(\avariable,\avariablebis)$ is the {\em precise list segment predicate} stating that $\aheap^i(\astore(\avariable)) = \astore(\avariablebis)$ for some $i \in \Nat$, but this property does not hold in any strict subheap of $\aheap$
(which excludes the presence of cycles).

In the sequel, we use the following
abbreviations: $\size \geq 0 \egdef \top$  and for all $\beta \geq 0$,
\begin{itemize}
\item $\size \geq \beta +1  \egdef (\size \geq \beta) \separate \neg \emptyconstant$,
\item $\size \leq \beta \egdef \neg (\size \geq \beta +1)$ and,
\item ${\size = \beta} \egdef (\size \leq \beta) \wedge (\size \geq \beta)$.
\end{itemize}
It is easy to see that $\pair{\astore}{\aheap} \models \size \geq \beta$ iff $\card{\domain{\aheap}} \geq \beta$,
where $\card{\aset}$ denotes the cardinality of a (finite) set $\aset$.
We introduce the \defstyle{septraction connective} $\septraction$,
defined as $\aformula_1 \septraction \aformula_2 \egdef \neg
(\aformula_1 \magicwand \neg \aformula_2)$. The connective  $\septraction$ can be viewed as a form of dual operator for the separating implication.
So,  $\pair{\astore}{\aheap} \models \aformula_1 \septraction \aformula_2$
 iff there is some heap $\aheap_1$ disjoint from $\aheap$
such that $\pair{\astore}{\aheap_1} \models \aformula_1$ and
$\pair{\astore}{\aheap + \aheap_1} \models \aformula_2$.
We introduce also the following standard abbreviations:
\begin{center}
  \hfill
  $\alloc{\avariable} \egdef (\avariable \hpto \avariable)\, \magicwand \perp$
  \hfill
  $\avariable \mapsto \avariablebis \egdef (\avariable \hpto \avariablebis) \wedge \size = 1.$
  \hfill\,
\end{center}
It holds that $\pair{\astore}{\aheap} \models \alloc{\avariable}$ iff $\astore(\avariable) \in \domain{\aheap}$, whereas $\pair{\astore}{\aheap} \models \avariable \mapsto \avariablebis$ iff $\domain{\aheap} = \{\astore(\avariable)\}$ and
$\aheap(\astore(\avariable)) = \astore(\avariablebis)$.
Without loss of generality, we assume that $\LOC = \Nat$ (see also Section~\ref{section-generalised-memory-states}).
We write $\seplogic{\separate, \magicwand}$ to denote the restriction of $\seplogic{\separate, \magicwand, \ls}$
without $\ls$. Similarly, $\seplogic{\separate}$
denotes the restriction of $\seplogic{\separate, \magicwand}$ without $\magicwand$ and $\seplogic{\magicwand}$ denotes its restriction without $\separate$.
Given two formulae $\aformula, \aformula'$ (possibly from different logical languages), we write
$\aformula \equiv \aformula'$ whenever for all  memory states $\pair{\astore}{\aheap}$, we have
$\pair{\astore}{\aheap} \models \aformula$ iff $\pair{\astore}{\aheap} \models \aformula'$.
When $\aformula \equiv \aformula'$, the formulae $\aformula$ and $\aformula'$ are said to be \defstyle{equivalent}.

\subsection{Variants with other reachability predicates}

We  use two additional reachability predicates $\reach(\avariable,\avariablebis)$ and $\reachplus(\avariable,\avariablebis)$.
We write  $\seplogic{\separate, \magicwand, \reach}$ (resp. $\seplogic{\separate, \magicwand, \reachplus}$)
to denote the variant of $\seplogic{\separate, \magicwand, \ls}$ in which $\ls$ is replaced by $\reach$ (resp. by $\reachplus$).
The satisfaction relation $\models$ is extended as follows:
\begin{itemize}
\item  $\pair{\astore}{\aheap} \models \reach(\avariable,\avariablebis)$ holds when
  there is  $i \geq 0$  such that $\aheap^i(\astore(\avariable)) = \astore(\avariablebis)$,
\item $\pair{\astore}{\aheap} \models \reachplus(\avariable,\avariablebis)$  holds when
  there is  $i \geq 1$  such that $\aheap^i(\astore(\avariable)) = \astore(\avariablebis)$.
\end{itemize}
When the heap $\aheap$ is understood as a directed graph with a finite 
relation,
$\reach(\avariable,\avariablebis)$ corresponds to the standard reachability predicate (differently from $\ls$, which also imposes constraints on strict subheaps). $\reachplus(\avariable,\avariablebis)$ corresponds instead to the reachability predicate
in at least one step.
For instance, $\reach(\avariable, \avariable)$ always holds, whereas
$\ls(\avariable,\avariable)$ holds only on the empty heap and
$\reachplus(\avariable,\avariable)$ holds on heaps such that there is $i \geq 1$ with $\aheap^i(\astore(\avariable)) =
\astore(\avariable)$, i.e. there is a non-empty path (cycle) from $\astore(\avariable)$ to itself.

As $\ls(\avariable, \avariablebis) \equiv \reach(\avariable, \avariablebis) \wedge \neg (\neg \emptyconstant \separate
\reach(\avariable, \avariablebis))$ and $\reach(\avariable, \avariablebis) \equiv \top \separate \ls(\avariable, \avariablebis)$,
the logics $\seplogic{\separate, \magicwand, \reach}$ and $\seplogic{\separate, \magicwand,
\ls}$ have identical decidability status.
Besides, these two logics can be seen as fragments of $\seplogic{\separate, \magicwand, \reachplus}$, thanks to the equivalence below:
$$\reach(\avariable, \avariablebis) \equiv \avariable = \avariablebis \vee  \reachplus(\avariable, \avariablebis).$$
Notice that this analysis can be carried out as soon as $\separate$, $\neg$, $\wedge$ and $\emptyconstant$
are parts of the logic (none of the equivalences above uses the separating implication~$\magicwand$). 
More specifically, $\seplogic{\separate,  \reach}$ and $\seplogic{\separate, \ls}$ have the same decidability status, and can be viewed as fragments of
$\seplogic{\separate,  \reachplus}$. 
It is therefore stronger
to establish decidability or complexity upper bounds with $\reachplus$
and to show undecidability or complexity lower bounds with $\ls$ or $\reach$.
Herein, we provide the optimal results.

\subsection{Decision problems}

Let $\alogic$ be a logic defined above.
As usual, the \defstyle{satisfiability problem for $\alogic$} takes as an input a formula $\aformula$ from $\alogic$ and
asks whether there is a memory state $\pair{\astore}{\aheap}$ satisfying it, i.e.~$\pair{\astore}{\aheap} \models \aformula$.
The \defstyle{validity problem for $\alogic$} asks instead whether $\aformula$ is satisfied by every memory state.
If $\alogic$ is not closed under negation, then it is also worth considering the \defstyle{entailment problem}
that takes as inputs two formulae~$\aformula$ and $\aformula'$, and asks whether for all the memory states
 $\pair{\astore}{\aheap}$, we have $\pair{\astore}{\aheap} \models \aformula$ implies
$\pair{\astore}{\aheap} \models \aformula'$ (written $\aformula \models \aformula'$).

The \defstyle{model-checking problem for $\alogic$} takes as an input a formula $\aformula$ from $\alogic$ and
a finite representation of a memory state $\pair{\astore}{\aheap}$, 
and asks whether $\pair{\astore}{\aheap} \models \aformula$.
Notice that, given a formula $\aformula$, it is easy to find a finite representation of $\pair{\astore}{\aheap}$: it is sufficient to restrict $\astore$ to the variables occurring in $\aformula$ and to encode $\aheap$ as a finite 
directed graph such that each vertex has at most one outgoing vertex.
Unless otherwise specified, the \defstyle{size} of a formula $\aformula$ is understood as the size of its syntax tree.
Moreover, the size of the finite representation of a memory state
$\pair{\astore}{\aheap}$ is defined naturally considering a reasonably succinct encoding of
$\astore$ (only interpreting the program variables in $\aformula$) and the graph representation of $\aheap$.

Below, we recall a few complexity results about well-known strict fragments
of $\seplogic{\separate, \magicwand, \ls}$.
\begin{proposition} \

\begin{enumerate}[label=\normalfont{\textbf{(\Roman*)}}]
\itemsep 0 cm
\item The satisfiability problem is \pspace-complete for both $\seplogic{\separate}$ and $\seplogic{\separate, \magicwand}$~\cite{Calcagno&Yang&OHearn01}.

\item The satisfiability and entailment problems for the symbolic heap fragment
      are in \ptime~\cite{Cooketal11}.

\item  The satisfiability problem for the fragment of\, $\seplogic{\separate,\ls}$ restricted to
formulae
obtained by Boolean combinations of formulae from the symbolic heap fragment is
\np-complete~\cite{Cooketal11,Piskac&Wies&Zufferey13}.

\end{enumerate}
\end{proposition}
We refer the reader to~\cite{Cooketal11} for a complete description of the symbolic heap fragment, or to
Section~\ref{section:complexity-upper-bounds} for its definition.
The main purpose of this paper is to study the decidability/complexity status of $\seplogic{\separate, \magicwand, \ls}$,
as well as for its
fragments and variants.
\section{Undecidability of the Satisfiability Problem for $\seplogic{\separate, \magicwand, \ls}$}
\label{section-undecidability}

In this section, we show that $\seplogic{\separate, \magicwand,\ls}$
has an undecidable satisfiability problem
even though it does not admit \emph{explicitly} first-order quantification.
We stress the word ``explicitly'' as we show that $\seplogic{\separate, \magicwand,\ls}$ can simulate
the first-order quantification from $\foseplogic{\magicwand}$: the first-order extension of $\seplogic{\magicwand}$, studied in~\cite{Brochenin&Demri&Lozes12,DemriDeters16}.

In the versions of $\foseplogic{\magicwand}$  defined in~\cite{Brochenin&Demri&Lozes12,DemriDeters16},
one distinguishes program variables from quantified variables (used with the first-order quantification $\forall$).
This distinction made by the two sets of variables is not necessary herein and for the sake of simplicity, we adopt
a version of   $\foseplogic{\magicwand}$ with a unique type of variables.
The formulae $\aformula$ of $\foseplogic{\magicwand}$ are built from the grammars below:
$$
\aatomicformula ::= \avariable = \avariablebis \ \mid \
                    \avariable \hpto \avariablebis$$
$$\aformula ::= \aatomicformula \ \mid \  \neg \aformula \ \mid \ \aformula \land \aformula
               \ \mid \ \aformula \magicwand \aformula
               \ \mid \ \forall \avariable \ \aformula,$$
where $\avariable, \avariablebis \in \PVAR$.
Models of the logic $\foseplogic{\magicwand}$ are
memory states and the satisfaction relation $\models$ is defined as for $\seplogic{\magicwand}$
with the additional clause:
\[
\pair{\astore}{\aheap} \models \forall \avariable \ \aformula \iff
  \text{ for all } \alocation \in \LOC, \text{ we have } \pair{\astore[\avariable \leftarrow \alocation]}{\aheap} \models \aformula,
\]
where $\astore[\avariable \leftarrow \alocation]$ is defined from the store $\astore$ by only changing
that $\avariable$ takes the value $\alocation$.
Note that $\emptyconstant$ can be easily defined by $\neg \exists \ \avariable \ (\avariable \hpto \avariable \magicwand \perp)$.
Without any loss of generality, we can assume that the satisfiability (resp. validity) problem for
$\foseplogic{\magicwand}$ is defined by taking as inputs closed formulae (i.e. without free occurrences of the variables).
\begin{proposition}~\cite{Brochenin&Demri&Lozes12,DemriDeters16}
The satisfiability problem for $\foseplogic{\magicwand}$ is undecidable
and the set of valid formulae for $\foseplogic{\magicwand}$ is not recursively enumerable.
\end{proposition}
We recall that the undecidability proof of $\foseplogic{\magicwand}$ makes extensive use of the  operator $\magicwand$, whereas a similar result can be achieved without $\magicwand$ if we interpret the logic on heaps having at least two record fields
(i.e. $\aheap$ is of the form $\LOC \to_{\fin} \LOC^k$ with $k \geq 2$)~\cite{Calcagno&Yang&OHearn01}.
In a nutshell, we establish the undecidability of $\seplogic{\separate, \magicwand, \ls}$ by a reduction from
the satisfiability problem for  $\foseplogic{\magicwand}$. The reduction is nicely decomposed in two intermediate steps:
(1) the undecidability of $\seplogic{\separate, \magicwand}$ extended with a few atomic predicates, to be
defined soon, and (2) a \emph{tour de force} resulting in the encoding of these atomic predicates in $\seplogic{\separate, \magicwand, \ls}$.
In particular, Section~\ref{section-encoding} explains how the additional predicates can be used to encode stores as subheaps, and how this helps to mimic first-order quantification.
Section~\ref{section-the-translation} provides the formal presentation of the translation as well as its correctness.
Section~\ref{section-expressing-predicates} establishes how the additional predicates can be indeed expressed in
$\seplogic{\separate, \magicwand, \ls}$.
Finally, Section~\ref{section-undecidability-results-and-nfax} provides the concluding results
about undecidability and refines the results of the previous sections.

\subsection{Generalised memory states}
\label{section-generalised-memory-states}

For technical convenience, in order to reduce (the satisfiability problem of) 
$\foseplogic{\magicwand}$ to $\seplogic{\separate,\magicwand,\ls}$
we consider a slight alternative
for the semantics of these two logics,
which does not modify the notion of satisfiability/validity and such that
the set of formulae and
the definition of the satisfaction relation
$\models$ remain unchanged.
So far, the memory states are pairs of the form $\pair{\astore}{\aheap}$ with
${\astore: \PVAR \rightarrow \LOC}$  and $\aheap: \LOC \to_{\fin} \LOC$  for a {\em fixed} countably infinite set of locations $\LOC$, e.g.~$\Nat$.
Alternatively, the models for  $\foseplogic{\magicwand}$ and $\seplogic{\separate,\magicwand,\ls}$ can be defined
as triples $\triple{\LOC_1}{\astore_1}{\aheap_1}$ such that $\LOC_1$ is a countably infinite set,
$\astore_1: \PVAR \rightarrow \LOC_1$  
and ${\aheap_1: \LOC_1 \to_{\fin} \LOC_1}$. 
Most of the time, a generalised memory state $\triple{\LOC_1}{\astore_1}{\aheap_1}$ shall be written
$\pair{\astore_1}{\aheap_1}$ when no confusion is possible.
Given a bijection ${\amap: \LOC_1 \rightarrow \LOC_2}$ and a heap $\aheap_1: \LOC_1 \to_{\fin} \LOC_1$ 
that can be represented by 
${\set{\alocation_1 \mapsto \aheap_1(\alocation_1), \ldots,\alocation_n \mapsto \aheap_1(\alocation_n)}}$,
we write $\amap(\aheap_1)$ to denote the heap  $\aheap_2: \LOC_2 \to_{\fin} \LOC_2$ with
$\aheap_2 = {\set{\amap(\alocation_1) \mapsto \amap(\aheap_1(\alocation_1)), \ldots, \amap(\alocation_n) \mapsto \amap(\aheap_1(\alocation_n))}}$.

\begin{definition}\label{def:partial_memory_state_iso}
Let $\triple{\LOC_1}{\astore_1}{\aheap_1}$ and  $\triple{\LOC_2}{\astore_2}{\aheap_2}$ be generalised
memory sta\-tes and $\avarset \subseteq \PVAR$. An \defstyle{$\aset$-isomorphism}
from  $\triple{\LOC_1}{\astore_1}{\aheap_1}$ to  $\triple{\LOC_2}{\astore_2}{\aheap_2}$ is a bijection
$\amap: \LOC_1 \rightarrow \LOC_2$ such that
$\aheap_2 = \amap(\aheap_1)$ and for all $\avariable \in \avarset$,  $\amap(\astore_1(\avariable)) = \astore_2(\avariable)$.
\end{definition}

Note that if $\amap$ is an $\aset$-isomorphism from $\triple{\LOC_1}{\astore_1}{\aheap_1}$ to
$\triple{\LOC_2}{\astore_2}{\aheap_2}$, then $\amap^{-1}$ is also an $\aset$-isomorphism from
$\triple{\LOC_2}{\astore_2}{\aheap_2}$ to  $\triple{\LOC_1}{\astore_1}{\aheap_1}$.
Two generalised memory states $\triple{\LOC_1}{\astore_1}{\aheap_1}$ and  $\triple{\LOC_2}{\astore_2}{\aheap_2}$
are said to be \defstyle{isomorphic with respect to $\avarset$}, written $\triple{\LOC_1}{\astore_1}{\aheap_1} \approx_{\avarset}
\triple{\LOC_2}{\astore_2}{\aheap_2}$, if and only if there exists an $\aset$-isomorphism
between them.

It is easy to check that $\approx_{\avarset}$ is an equivalence relation.
A folklore result states that isomorphic memory states satisfy the same formulae,
which implies that considering generalised memory states over (standard) memory states
does not
change the notion of satisfiability and validity.
Below, we state the precise result
we need in the sequel, the proof being by a standard  induction on the formula structure.

\begin{lemma}
\label{lemma-equivalence-isomorphism}
Let $\triple{\LOC_1}{\astore_1}{\aheap_1}$, $\triple{\LOC_2}{\astore_2}{\aheap_2}$ be generalised memory states
and $\avarset \subseteq \PVAR$ be a finite set of variables such that
$\triple{\LOC_1}{\astore_1}{\aheap_1} \approx_{\avarset}  \triple{\LOC_2}{\astore_2}{\aheap_2}$. Given $\aformula$ in $\seplogic{\separate,\magicwand,\ls}$ or $\foseplogic{\magicwand}$, with free variables from $\avarset$, 
$\triple{\LOC_1}{\astore_1}{\aheap_1} \models \aformula$ iff $\triple{\LOC_2}{\astore_2}{\aheap_2} \models \aformula$.
\end{lemma}

\begin{proof} The proof is by induction on the tree structure of $\aformula$.
Let $\avarset$ be a set of variables that includes the free variables from $\aformula$.
To be more concise, this is done on formulae from $\seplogic{\forall,\separate,\magicwand,\ls}$.
Let $\amap: \LOC_1 \to \LOC_2$ be a bijection defined as in Definition~\ref{def:partial_memory_state_iso}.
Since $\approx_{\aset}$ is an equivalence relation, 
showing one direction suffices to prove the result.
We start by considering the base case with $\ls(\avariable,\avariablebis)$, 
the cases for  $\avariable = \avariablebis$ and $\avariable \hpto \avariablebis$ being omitted as this poses no
difficulty. 
\begin{description}
\itemsep 0 cm
  \item[base case: $\ls(\avariable,\avariablebis)$.] If $\triple{\LOC_1}{\astore_1}{\aheap_1} \models \ls(\avariable,\avariablebis)$, then either
  ($\domain{\aheap_1} = \emptyset$ and  $\astore_1(\avariable) = \astore_1(\avariablebis)$)
  or $\aheap_1 = \set{\alocation_0 \mapsto \alocation_1, \alocation_1 \mapsto \alocation_2, \ldots, \alocation_{n-1}\mapsto \alocation_n}$
  with  $n \geq 1$, $\alocation_0 = \astore_1(\avariable)$, $\alocation_n = \astore_1(\avariablebis)$ and for all $i \neq j \in \interval{0}{n}$, $\alocation_i \neq \alocation_j$.
  Notice that since for all $\avariable \in \avarset$, $\amap(\astore_1(\avariable)) = \astore_2(\avariable)$ (Definition~\ref{def:partial_memory_state_iso}), we have $\astore_1(\avariable) = \astore_1(\avariablebis)$ if and only if $\astore_2(\avariable) = \astore_2(\avariablebis)$.
  Moreover, again from the definition of $\avarset$-isomorphism, 
  $\domain{\aheap_1}$ and $\domain{\aheap_2}$ have the same cardinality.
  Thus, if $\domain{\aheap_1} = \emptyset$ and $\astore_1(\avariable) = \astore_1(\avariablebis)$, we conclude that $\domain{\aheap_2} = \emptyset$ and $\astore_2(\avariable) = \astore_2(\avariablebis)$.
  Otherwise, let us consider the case where $\aheap_1 = \set{\alocation_0 \mapsto \alocation_1, \alocation_1 \mapsto \alocation_2, \ldots, \alocation_{n-1}\mapsto \alocation_n}$
  with  $n \geq 1$, $\alocation_0 = \astore_1(\avariable)$, $\alocation_n = \astore_1(\avariablebis)$ and for all $i \neq j \in \interval{0}{n}$ $\alocation_i \neq \alocation_j$.
  Notice that $\card{\domain{\aheap_1}} = n$.
  By definition of $\avarset$-isomorphism, ${\astore_2(\avariable) = \amap(\astore_1(\avariable)) = \amap(\alocation_0)}$, $\astore_2(\avariablebis) = \amap(\astore_1(\avariablebis)) = \amap(\alocation_n)$ and 
  given~$i \in \interval{0}{n-1}$, 
  ${\aheap_2(\amap(\alocation_i)) = \amap(\aheap_1(\alocation_i)) = \amap(\alocation_{i+1})}$.
So, $\set{\astore_2(\avariable) \mapsto \amap(\alocation_1), \amap(\alocation_1) \mapsto \amap(\alocation_2), \ldots, \amap(\alocation_{n-1})\mapsto \astore_2(\avariablebis)} \subseteq \aheap_2$.
Moreover, from the fact that for all $i \neq j \in \interval{0}{n}$, $\alocation_i \neq \alocation_j$, we conclude that for all ${i \neq j \in \interval{0}{n}}$, $\amap(\alocation_i) \neq \amap(\alocation_j)$.
From $\card{\domain{\aheap_1}} = n$ we derive $\card{\domain{\aheap_2}} = n$,
and therefore 
$\aheap_2 = \set{\astore_2(\avariable) \mapsto \amap(\alocation_1), \amap(\alocation_1) \mapsto \amap(\alocation_2), \ldots, \amap(\alocation_{n-1})\mapsto \astore_2(\avariablebis)}$, where $\astore_2(\avariable),\amap(\alocation_1),\dots,\amap(\alocation_{n-1}),\astore_2(\avariablebis)$ are $n+1$ distinct locations.
We conclude that $\triple{\LOC_2}{\astore_2}{\aheap_2}$ satisfies $\ls(\avariable,\avariablebis)$.
\end{description}
Concerning the cases for the induction step, we omit the obvious cases
when the outermost connective is the conjunction or the negation. 
\cut{
For negation and conjunction, the proof simply applies the induction hypothesis:
\begin{description}
  \setcounter{enumi}{3}
  \item[induction step: $\land$.] If $\triple{\LOC_1}{\astore_1}{\aheap_1} \models \aformula_1 \land \aformula_2$,
        then $\triple{\LOC_1}{\astore_1}{\aheap_1}$ satisfies $\aformula_1$ and $\aformula_2$.
        By the induction hypothesis, $\triple{\LOC_2}{\astore_2}{\aheap_2}$ also satisfies $\aformula_1$ and $\aformula_2$.
        Equivalently, $\triple{\LOC_2}{\astore_2}{\aheap_2} \models \aformula_1 \land \aformula_2$.
  \item[induction step: $\lnot$.] If $\triple{\LOC_1}{\astore_1}{\aheap_1} \models \lnot \aformula_1$, then $\triple{\LOC_1}{\astore_1}{\aheap_1}$ do not satisfy $\aformula_1$.
        By the induction hypothesis, $\aformula_1$ cannot therefore be satisfied by $\triple{\LOC_2}{\astore_2}{\aheap_2}$.
\end{description}
}
We conclude the proof by considering
formulae of the form $\aformula_1 \sepcnj \aformula_2$, $\aformula_1 \magicwand \aformula_2$ and $\forall \avariable \ \aformula_1$.
\begin{description}
\itemsep 0 cm
  \setcounter{enumi}{5}
  \item[induction step: case with $\sepcnj$.] If $\triple{\LOC_1}{\astore_1}{\aheap_1} \models \aformula_1 \sepcnj \aformula_2$, then there are
          two heaps $\aheap_1'$ and $\aheap_1''$ such that $\aheap_1 = \aheap_1' + \aheap_1''$,
          $\triple{\LOC_1}{\astore_1}{\aheap_1'} \models \aformula_1$ and $\triple{\LOC_1}{\astore_1}{\aheap_1''} \models \aformula_2$.
          Let $\aheap_2' = \amap(\aheap_1')$ and $\aheap_2'' = \amap(\aheap_1'')$ the images of $\aheap_1'$ and $\aheap_1''$ via $\amap$.
          Since $\amap$ is an $\aset$-isomorphism between $\aheap_1$ and $\aheap_2$ it holds that:
  \[
    \aheap_2 = \amap(\aheap_1) = \amap(\aheap_1' + \aheap_1'') = \amap(\aheap_1') + \amap(\aheap_1'') = \aheap_2' + \aheap_2''
  \]
  and moreover $\aheap_2' \bot \aheap_2''$.
  Lastly, $\triple{\LOC_1}{\astore_1}{\aheap_1'} \approx_\aset \triple{\LOC_2}{\astore_2}{\aheap_2'}$, since
  $\amap$ is an $\aset$-isomorphism also between these structures. The same holds for $\aheap_1''$ and $\aheap_2''$.
  Therefore, by the induction hypothesis, we get
  $\triple{\LOC_2}{\astore_2}{\aheap_2'} \models \aformula_1$ and $\triple{\LOC_2}{\astore_2}{\aheap_2''} \models \aformula_2$.
  We conclude that $\triple{\LOC_2}{\astore_2}{\aheap_2} \models \aformula_1 \sepcnj \aformula_2$.
  \item[induction step: case with $\magicwand$.] If $\triple{\LOC_1}{\astore_1}{\aheap_1} \models \aformula_1 \magicwand \aformula_2$, then for every heap $\aheap_1'$, if $\aheap_1' \bot \aheap_1$ and
  $\triple{\LOC_1}{\astore_1}{\aheap_1'} \models \aformula_1$ then
  $\triple{\LOC_1}{\astore_1}{\aheap_1 + \aheap_1'} \models \aformula_2$. We show that
  $\triple{\LOC_2}{\astore_2}{\aheap_2} \models \aformula_1 \magicwand \aformula_2$, which is true whenever for all $\aheap_2'$,
  if $\aheap_2' \bot \aheap_2$ and
  $\triple{\LOC_2}{\astore_2}{\aheap_2'} \models \aformula_1$ then
  ${\triple{\LOC_2}{\astore_2}{\aheap_2 + \aheap_2'} \models \aformula_2}$.
  Consider a heap~$\aheap_2'$ such that $\aheap_2' \bot \aheap_2$ and
  $\triple{\LOC_2}{\astore_2}{\aheap_2'} \models \aformula_1$.
  By recalling that $\amap$ is a bijection from $\LOC_1$ to $\LOC_2$,
  we construct the heap $\aheap_1' = \amap^{-1}(\aheap_2')$. 
  Directly from the fact that $\amap$ is a $\avarset$-isomorphism between 
  $\triple{\LOC_1}{\astore_1}{\aheap_1}$ and $\triple{\LOC_2}{\astore_2}{\aheap_2}$,
  it is easy to see that the following two properties hold:
  \begin{enumerate}
    \item $\triple{\LOC_1}{\astore_1}{\aheap_1'} \approx_\aset \triple{\LOC_2}{\astore_2}{\aheap_2'}$ and $\amap$ is an $\aset$-isomorphism between the two structures.
    \item $\aheap_1' \bot \aheap_1$, from (1) together with $\triple{\LOC_1}{\astore_1}{\aheap_1'} \approx_\aset \triple{\LOC_2}{\astore_2}{\aheap_2'}$, $\aheap_2' \bot \aheap_2$ and $\amap^{-1}(\aheap_2') = \aheap_1'$.
  \end{enumerate}

  By the induction hypothesis, from (1) we conclude that
  $\triple{\LOC_1}{\astore_1}{\aheap_1'} \models \aformula_1$.
  Then, from (2) and the initial hypothesis $\triple{\LOC_1}{\astore_1}{\aheap_1} \models \aformula_1 \magicwand \aformula_2$, we obtain that $\triple{\LOC_1}{\astore_1}{\aheap_1 + \aheap_1'}$ satisfies~$\aformula_2$.
  By definition $\amap(\aheap_1 + \aheap_1') = \amap(\aheap_1) + \amap(\amap^{-1}(\aheap_2')) = \aheap_2 + \aheap_2'$ and therefore by the induction hypothesis, we get $\triple{\LOC_2}{\astore_2}{\aheap_2 + \aheap_2'} \models \aformula_2$.
  Thus, $\triple{\LOC_2}{\astore_2}{\aheap_2} \models \aformula_1 \magicwand \aformula_2$.

  \item[induction step: case with $\forall$.] If $\triple{\LOC_1}{\astore_1}{\aheap_1} \models \forall \avariable \ \aformula_1$,
        then ${\triple{\LOC_1}{\astore_1[\avariable \gets \alocation]}{\aheap_1} \models \aformula_1}$ holds for all $\alocation \in \LOC_1$.
  We prove that $\triple{\LOC_2}{\astore_2}{\aheap_2} \models \forall \avariable \ \aformula_1$, which is true whenever
  for every $\alocation' \in \LOC_2$, $\triple{\LOC_2}{\astore_2[\avariable \gets \alocation']}{\aheap_2} \models \aformula_1$.
  Notice that $\avariable \not \in \aset$ since $\aset$ contains only the free variables in $\forall \avariable \ \aformula_1$.
   Let $\alocation'$ be a location in $\LOC_2$.
  We have $\amap^{-1}(\alocation') \in \LOC_1$ and
  ${\triple{\LOC_1}{\astore_1[\avariable \gets \amap^{-1}(\alocation')]}{\aheap_1} \models \aformula_1}$.
  From $\triple{\LOC_1}{\astore_1}{\aheap_1} \approx_\aset \triple{\LOC_2}{\astore_2}{\aheap_2}$ and $\avariable \not \in \avarset$,
  we derive $\amap(\aheap_1) = \aheap_2$ and 
  for every $\avariablebis \in \avarset$, $\amap(\astore_1[\avariable \gets \amap^{-1}(\alocation')](\avariablebis)) = \amap(\astore_1(\avariablebis)) = \astore_2(\avariablebis) =  
  \astore_2[\avariable \gets \alocation'](\avariablebis)$.
  Thus, $\triple{\LOC_1}{\astore_1[\avariable \gets \amap^{-1}(\alocation')]}{\aheap_1} \approx_{\aset \cup \{\avariable\}} \triple{\LOC_2}{\astore_2[\avariable \gets \alocation']}{\aheap_2}$.
  Using the induction hypothesis, we derive $\triple{\LOC_2}{\astore_2[\avariable \gets \alocation']}{\aheap_2} \models \aformula_1$.
  Consequently, $\triple{\LOC_2}{\astore_2}{\aheap_2} \models \forall \avariable \ \aformula_1$. \qedhere
\end{description}
\end{proof}

As a direct consequence, satisfiability in  $\seplogic{\separate,\magicwand,\ls}$  as defined in Section~\ref{section-preliminaries},
is equivalent to satisfiability with generalised memory states, the same holds for $\foseplogic{\magicwand}$.
Indeed, suppose that $\aformula$ is satisfiable in the memory state $\pair{\astore}{\aheap}$, then
 $\aformula$ is satisfiable with the generalised semantics by considering the model $\triple{\Nat}{\astore}{\aheap}$.
Similarly, suppose that $\aformula$ is satisfiable in the generalised memory state $\triple{\LOC_1}{\astore_1}{\aheap_1}$.
As $\LOC_1$ is countably infinite, there is a bijection $\amap: \LOC_1 \rightarrow \Nat$. Let $\pair{\astore}{\aheap}$
such that $\amap(\aheap_1) = \aheap$ and for every $\avariable \in \PVAR$, $\astore(\avariable) \egdef \amap(\astore_1(\avariable))$.
Let $\avarset$ be the set of free variables occurring in $\aformula$. By construction of
$\pair{\astore}{\aheap}$, we have
$\triple{\LOC_1}{\astore_1}{\aheap_1} \approx_{\avarset}  \triple{\Nat}{\astore}{\aheap}$.
By Lemma~\ref{lemma-equivalence-isomorphism}, we conclude $\triple{\Nat}{\astore}{\aheap} \models \aformula$,
which is equivalent to $\pair{\astore}{\aheap} \models \aformula$.

\subsection{Encoding quantified variables as cells in the heap}
\label{section-encoding}

In this section, we introduce three additional atomic predicates that allow us to encode the first-order quantification of~$\foseplogic{\separate,\magicwand}$ as memory updates of~$\seplogic{\separate,\magicwand,\ls}$:
$$
\allocback{\avariable},  \qquad n(\avariable) = n(\avariablebis), \qquad  n(\avariable) \hpto n(\avariablebis).$$
The characterisation of these predicates in terms of formulae in~$\seplogic{\separate,\magicwand,\ls}$ is given in~Section~\ref{section-expressing-predicates}. For the time being, we simply assume that they can be defined in the logic, and work following their semantics, given below:
\begin{center}
      \begin{tabular}[t]{lll}
        $\pair{\astore}{\aheap} \models \allocback{\avariable}$ & $\iff$ & $\astore(\avariable) \in \range{\aheap}$.\\
        $\pair{\astore}{\aheap} \models n(\avariable)=n(\avariablebis)$ & $\iff$ & $\{\astore(\avariable),\astore(\avariablebis)\} \subseteq \domain{\aheap}$ and $\aheap(\astore(\avariable))=\aheap(\astore(\avariablebis))$.\\
        $\pair{\astore}{\aheap} \models n(\avariable) \hpto n(\avariablebis)$ & $\iff$ & $\{\astore(\avariable),\astore(\avariablebis)\} \subseteq \domain{\aheap}$ and ${\aheap^2(\astore(\avariable))=\aheap(\astore(\avariablebis))}$.
      \end{tabular}
\end{center}
In other words, $\allocback{\avariable}$ holds in $\pair{\astore}{\aheap}$ whenever $\astore(\avariable)$ has a predecessor in $\aheap$.
The satisfaction of $n(\avariable)=n(\avariablebis)$ corresponds to the existence of the following pattern in the memory state
$\pair{\astore}{\aheap}$:
\begin{center}
\begin{tikzpicture}[baseline=-2pt,->,>=stealth',shorten >=1pt, node distance=1cm, thick,auto,bend angle=30]
        \tikzstyle{every state}=[fill=white,draw=black,text=black,minimum width=0.1cm]

        \node (i) at (0,0) {$\astore(\avariable)$};
        \node[aloc] (j) [right=1cm of i] {};
        \node (k) [right=1cm of j] {$\astore(\avariablebis)$};
        \draw[pto] (i) edge node {} (j);
        \draw[pto] (k) edge node {} (j);
      \end{tikzpicture}
\end{center}
whereas the satisfaction of $n(\avariable) \hpto n(\avariablebis)$ corresponds to the existence of the following pattern:
\begin{center}
\begin{tikzpicture}[scale=0.6,->,>=stealth',shorten >=1pt, node distance=1cm, thick,auto,bend angle=30]
        \tikzstyle{every state}=[fill=white,draw=black,text=black,minimum width=0.1cm]

       \node[aloc] (i) {};
       \node (j) [left = of i] {$\astore(\avariable)$};
       \node[aloc] (w) [below = 0.5cm of i] {};
       \node (k) [left =  of w] {$\astore(\avariablebis)$};

       \draw[pto] (j) edge node {} (i);
       \draw[pto] (i) edge node {} (w);
       \draw[pto] (k) edge node {} (w);
    \end{tikzpicture}
\end{center}
As a rule of thumb, `$n(\avariable)$' in  $n(\avariable)=n(\avariablebis)$ or in $n(\avariable)\!\hpto\!n(\avariablebis)$
refers to the ``next'' value of $\avariable$ understood as $\aheap(\astore(\avariable))$ -- if it exists.
Let us first intuitively explain how the two last predicates will help encoding $\seplogic{\forall , \magicwand}$.
By definition, the satisfaction of the quantified formula $\forall \avariable \ \aformulabis$
from $\foseplogic{\magicwand}$ requires the satisfaction of the formula $\aformulabis$ for all the values $\alocation$
in $\LOC$ assigned to $\avariable$. The principle of the encoding is to use a value $\astore(\avariable)$
not in the domain or range of the heap to mimic the store by modifying how $\astore(\avariable)$ is allocated:
typically `$\avariable$ takes the value $\alocation$'' is encoded by the heap $\set{\astore(\avariable) \mapsto
\alocation}$ where $\astore(\avariable)$ is not in the range and domain of the heap.
Figure~\ref{figure:undecidability:rough-idea:store-into-heaps} intuitively depicts this encoding, highlighting how the store of the memory state 
in~Figure~\ref{figure:undecidability:rough-idea:a-memory-state}
can be simulated directly by the heap.
\begin{figure}
\vspace{3pt}
{\hphantom{.}\hfill
\begin{minipage}{.45\textwidth}
  \centering
  \begin{tikzpicture}[baseline]
    \node[aloc] (aa) {};
    \node[aloc] (ab) [above = 0.6cm of aa] {};
    \node[aloc] (ac) [above = 0.6cm  of ab] {};
    \node[aloc, label={[label distance=-0.05cm]180:{$\avariable_2$}}] (ad) [above = 0.6cm  of ac] {};
    \node[aloc, label={[label distance=-0.05cm]180:{$\avariable_1$}}] (ae) [above = 0.6cm  of ad] {};
    \node[aloc] (af) [above = 0.6cm  of ae] {};

    \node[aloc] (ba) [right = 0.7cm of aa] {};
    \node[aloc] (bb) [above = 0.6cm of ba] {};
    \node[aloc, label={[label distance=-0.05cm]-135:{$\avariable_3$}}] (bc) [above = 0.6cm  of bb] {};
    \node[aloc] (bd) [above = 0.6cm  of bc] {};
    \node[aloc] (be) [above = 0.6cm  of bd] {};
    \node[aloc] (bf) [above = 0.6cm  of be] {};

    \node[aloc] (ca) [right = 0.7cm of ba] {};
    \node[aloc] (cb) [above = 0.6cm of ca] {};
    \node[aloc] (cc) [above = 0.6cm  of cb] {};
    \node[aloc] (cd) [above = 0.6cm  of cc] {};
    \node[aloc] (ce) [above = 0.6cm  of cd] {};
    \node[aloc] (cf) [above = 0.6cm  of ce] {};

    \draw[pto] (cc) -- (cd);
    \draw[pto] (cd) -- (ce);
    \draw[pto] (ce) -- (bf);
    \draw[pto] (bf) to [bend right=50] (ae);
    \draw[pto] (ae) -- (be);
    \draw[pto] (be) -- (bf);

    \draw[pto] (bc) to [bend left=35] (ad);
    \draw[pto] (ad) to [bend left=35] (bc);
    \draw[pto] (cb) -- (bc);
    \draw[pto] (ba) -- (cb);
    \draw[pto] (ab) -- (ba);

          \begin{pgfonlayer}{bg}
            \draw[white,very thick,dashed,rounded corners=5pt,fill=white] ($(bf.north west)+(0,0.3)$)  rectangle ($(ba.south east)+(0.5,-0.3)$);
          \end{pgfonlayer}

  \end{tikzpicture}
  \captionof{figure}[]{A memory state.}\label{figure:undecidability:rough-idea:a-memory-state}
\end{minipage}%
\hfill%
\begin{minipage}{.45\textwidth}
  \centering
  \begin{tikzpicture}[baseline]
    \node[aloc] (aa) {};
    \node[aloc] (ab) [above = 0.6cm of aa] {};
    \node[aloc] (ac) [above = 0.6cm  of ab] {};
    \node[aloc] (ad) [above = 0.6cm  of ac] {};
    \node[aloc] (ae) [above = 0.6cm  of ad] {};
    \node[aloc] (af) [above = 0.6cm  of ae] {};

    \node[aloc] (ba) [right = 0.7cm of aa] {};
    \node[aloc] (bb) [above = 0.6cm of ba] {};
    \node[aloc] (bc) [above = 0.6cm  of bb] {};
    \node[aloc] (bd) [above = 0.6cm  of bc] {};
    \node[aloc] (be) [above = 0.6cm  of bd] {};
    \node[aloc] (bf) [above = 0.6cm  of be] {};

    \node[aloc] (ca) [right = 0.7cm of ba] {};
    \node[aloc] (cb) [above = 0.6cm of ca] {};
    \node[aloc] (cc) [above = 0.6cm  of cb] {};
    \node[aloc] (cd) [above = 0.6cm  of cc] {};
    \node[aloc] (ce) [above = 0.6cm  of cd] {};
    \node[aloc] (cf) [above = 0.6cm  of ce] {};

    \node[aloc] (sa) [left = 1cm of aa] {};
    \node[aloc, label={[label distance=-0.05cm]180:{$\avariable_3$}}] (sb) [above = 0.6cm of sa] {};
    \node[aloc] (sc) [above = 0.6cm of sb] {};
    \node[aloc, label={[label distance=-0.05cm]180:{$\avariable_2$}}] (sd) [above = 0.6cm of sc] {};
    \node[aloc, label={[label distance=-0.05cm]180:{$\avariable_1$}}] (se) [above = 0.6cm of sd] {};
    \node[aloc] (sf) [above = 0.6cm of se] {};

    \draw[pto] (cc) -- (cd);
    \draw[pto] (cd) -- (ce);
    \draw[pto] (ce) -- (bf);
    \draw[pto] (bf) to [bend right=50] (ae);
    \draw[pto] (ae) -- (be);
    \draw[pto] (be) -- (bf);

    \draw[pto] (bc) to [bend left=35] (ad);
    \draw[pto] (ad) to [bend left=35] (bc);
    \draw[pto] (cb) -- (bc);
    \draw[pto] (ba) -- (cb);
    \draw[pto] (ab) -- (ba);

    \draw[bordeaupto] (se) -- (ae);
    \draw[bordeaupto] (sb) -- (bc);
    \draw[bordeaupto] (sd) -- (ad);

    \begin{pgfonlayer}{bg}
      \draw[bordeau,very thick,dashed,rounded corners=5pt,fill=bordeau!4] ($(sf.north west)+(-1.7,0.3)$)  rectangle ($(sa.south east)+(0.5,-0.3)$);
      \node[rotate=90,align=center,bordeau] (text) [above right = -1.55cm and -0.8cm of sc] {
        locations reserved to\\[-2pt] simulate the store
      };
    \end{pgfonlayer}

  \end{tikzpicture}
  \captionof{figure}[]{Simulating the store with the heap.}\label{figure:undecidability:rough-idea:store-into-heaps}
\end{minipage}%
\hfill%
\hfill\,}
\end{figure}
As showed in this picture,
the principle of the encoding is to use a set $L$ of locations  initially not in the domain or range of the heap to mimic the store by modifying how they are allocated. 
In this way, a variable will be interpreted by a location in the heap and,
instead of checking whether $\avariable \hpto \avariablebis$ (or $\avariable = \avariablebis$) holds, we will check if $n(\avariable) \hpto n(\avariablebis)$ (or $n(\avariable) = n(\avariablebis)$) holds, where $\avariable$ and $\avariablebis$  correspond, after the translation,
to the locations in $L$ that mimic the store for those variables.

Let $\aset$ be the finite set of variables needed for the translation.
To properly encode the store, each location in $L$ only mimics exactly one variable, i.e. there is a bijection between
$\aset$ and $L$, and cannot be reached by any location.
Afterwards, the universal quantification for the quantifier $\forall$ is simulated by means of the separating implication $\magicwand$:
the formula $\forall \avariable \ \aformulabis$ will be encoded by
a formula of the form
$$(\alloc{\avariable} \wedge \size = 1) \magicwand (\OK(\aset) \Rightarrow
\atranslation(\aformulabis)),$$
where $\OK(\aset)$ (formally defined below)  checks whether the locations in $L$ still satisfy the auxiliary conditions just described, 
i.e.~locations in $L$ are not reached by other locations,
whereas
$\atranslation(\aformulabis)$ is the translation of $\aformulabis$.
So, as hinted earlier, the above formula universally quantifies over the addition of
all subheaps of the form $\set{\astore(\avariable) \mapsto
\alocation}$. The locations in $L$ play a special role and this shall be specified with the formula
 $\OK(\aset)$.

The formula $\aformulabis_1 \magicwand \aformulabis_2$ cannot simply be translated into
$\atranslation(\aformulabis_1) \magicwand (\Safe(\aset) \Rightarrow \atranslation(\aformulabis_2))$ because the evaluation of $\atranslation(\aformulabis_1)$
in a disjoint heap may need the values of free variables occurring in $\aformulabis_1$ but our encoding  of the variable valuations via the heap
does not allow to preserve these values through disjoint heaps. In order to solve this problem, for each variable $\avariable$ in the formula, $X$ will contain an auxiliary variable $\overline{\avariable}$, or alternatively we define on $X$ an involution $\overline{(.)}$. If the translated formula has $q$
variables then the set $X$ of variables needed for the translation will have cardinality $2q$. Roughly speaking, in the translation
of a formula whose outermost connective is the magic wand,
the locations corresponding to variables of the form $\overline{\avariable}$ will be allocated
on the left side of the magic wand, and checked to be equal to their non-bar versions on the right side of the magic wand.

Let us formalise the intuition of using part of the memory to mimic the store
depicted in 
Figure~\ref{figure:undecidability:rough-idea:store-into-heaps},
by defining a suitable encoding between generalised memory states.

\begin{definition}\label{def:rhd}
      Let $\aset \subseteq \asetbis$ be finite sets of program variables.
      Let
      $(\locations_1,\astore_1,\aheap_1)$ and $(\locations_2,\astore_2,\aheap_2)$ be two (generalised) memory states.
      We say that $(\locations_1,\astore_1,\aheap_1)$ is encoded by $(\locations_2,\astore_2,\aheap_2)$ with respect to $\aset$ and $\asetbis$,
      written ${(\locations_1,\astore_1,\aheap_1) \rhd^{\aset}_{\asetbis}(\locations_2,\astore_2,\aheap_2)}$, if the following conditions hold:
      \begin{enumerate}
      \item $\locations_1=\locations_2\setminus\set{\astore_2(\avariable)\mid\avariable \in \asetbis}$,
      \item for all $\avariable \neq \avariablebis \in \asetbis$, $\astore_2(\avariable)\neq\astore_2(\avariablebis)$,
      \item $\aheap_2=\aheap_1 + \set{\astore_2(\avariable)\mapsto\astore_1(\avariable)\mid\avariable\in\aset}$.
      \end{enumerate}
    
    \end{definition}
    
    Notice that the heap $\aheap_2$ is equal to the heap $\aheap_1$ augmented with
    the heap $\set{\astore_2(\avariable)\mapsto\astore_1(\avariable)\mid\avariable\in\aset}$. 
      From the encoding, we can retrieve the initial heap by removing the memory cells  corresponding to variables in $\avarset$.
      By way of example, notice how the memory state in Figure~\ref{figure:undecidability:rough-idea:a-memory-state} satisfies the formula $\avariable_3 \hpto \avariable_2$, whereas its encoding 
      in Figure~\ref{figure:undecidability:rough-idea:store-into-heaps}
      satisfies the formula ${n(\avariable_3) \hpto n(\avariable_2)}$.
      Furthermore, notice that the memory state in Figure~\ref{figure:undecidability:rough-idea:store-into-heaps} satisfies the formula
            $$
            \Safe(\set{\avariable_1,\avariable_2,\avariable_3}) = \avariable_1 \neq \avariable_2 \wedge
            \avariable_2 \neq \avariable_3 \wedge
            \avariable_1 \neq \avariable_3 \wedge
            \neg \mathtt{alloc}^{-1}(\avariable_1) \wedge
            \neg \mathtt{alloc}^{-1}(\avariable_2) \wedge
            \neg \mathtt{alloc}^{-1}(\avariable_3).
            $$
      This formula guarantees that the memory states are  encodings of  other memory states (in particular, the one in Figure~\ref{figure:undecidability:rough-idea:a-memory-state}).
      In general, we write $\Safe(\asetbis)$ for the following formula:
      $$
            \Safe(\asetbis) \egdef (\bigwedge_{\mathclap{{\rm distinct} \ \avariable,\avariablebis \in \asetbis}} \avariable \neq \avariablebis) \wedge
            (\bigwedge_{\mathclap{\avariable \in \asetbis}}  \neg \mathtt{alloc}^{-1}(\avariable))
            .
      $$    

\subsection{The translation}
\label{section-the-translation}
We are now ready to define the translation of a first-order formula in quantifier-free separation logic extended with the three predicates introduced at the
beginning of the section.
Let $\aformula$ be a closed formula of $\seplogic{\forall , \magicwand}$ with quantified variables $\aset = \{\avariable_1,\dots,\avariable_q\}$. Without loss of generality,
we can assume that~$\aformula$ is \defstyle{well-quantified}, i.e.~distinct quantifications involve distinct variables.
We consider a set $\asetbis = \{\avariable_1,\dots,\avariable_{2q}\}$ and $\overline{(.)}$ to be the (unique) involution on $\asetbis$ such that for all $i \in \interval{1}{q}$, $\overline{\avariable_i} \egdef \avariable_{i+q}$.
Notice that $\overline{\avariable_{i+q}} = \avariable_i$.
We extend the involution to arbitrary subsets of $\asetbis$, so that in particular $\overline{\aset} = \{\avariable_{q+1},\dots,\avariable_{2q}\}$ and $\overline{\overline{\aset}} = \aset$.

The translation function $\atranslation(\aformulabis, \asetbis)$ has two arguments: the formula~$\aformulabis$ in $\foseplogic{\magicwand}$ to be recursively translated,
and the total set of variables potentially appearing in the target formula, 
i.e.~$\asetbis$. The translation function is given below, 
assuming that the variables in $\aformulabis$ are  either included in $\aset$ or included $\overline{\aset}$.
\begin{align*}
\atranslation(\avariable_i = \avariable_j, \asetbis) & \egdef n(\avariable_i) = n(\avariable_j)  \\
\atranslation(\avariable_i \hpto \avariable_j, \asetbis) & \egdef n(\avariable_i) \hpto n(\avariable_j)  \\
\atranslation(\neg \aformulabis, \asetbis) & \egdef \neg \atranslation(\aformulabis, \asetbis)  \\
\atranslation(\aformulabis_1 \wedge \aformulabis_2, \asetbis) & \egdef \atranslation(\aformulabis_1, \asetbis) \wedge  \atranslation(\aformulabis_2, \asetbis) \\
\atranslation(\forall \avariable_i \ \aformulabis, \asetbis) & \egdef (\alloc{\avariable_i} \wedge \size = 1) \magicwand
(\Safe(\asetbis) \Rightarrow
\atranslation(\aformulabis, \asetbis))
\end{align*}
Lastly, the translation $\atranslation(\aformulabis_1 \magicwand \aformulabis_2, \asetbis)$  is defined as
$$
\Big(\big(\bigwedge_{\mathclap{\avariableter \in \asetter}} \alloc{\overline{\avariableter}}) \wedge
(\bigwedge_{\mathclap{\avariableter \in \asetbis \setminus \asetter}} \neg \alloc{\overline{\avariableter}})
\wedge \Safe(\asetbis) \wedge \atranslation(\aformulabis_1[\avariable \leftarrow \overline{\avariable} \mid \avariable \in \asetbis],\asetbis)\big)
\magicwand
$$
$$ \big(((\bigwedge_{\mathclap{\avariableter \in \asetter}} n(\avariableter) = n(\overline{\avariableter})) \wedge \Safe(\asetbis))
   \Rightarrow
((\bigwedge_{\mathclap{\avariableter \in \asetter}} \alloc{\overline{\avariableter}} \wedge \size = \card{\asetter}) \separate
\atranslation(\aformulabis_2, \asetbis))\big)\Big),
$$
where $\asetter$ is the set of free variables in $\aformulabis_1$, and 
$\aformulabis_1[\avariable \leftarrow \overline{\avariable} \mid \avariable \in \asetbis]$ denotes
the formula obtained from $\aformulabis_1$  by replacing simultaneously
all the variables $\avariable \in \asetbis$
by $\overline{\avariable}$.

\begin{example} Assume that $q =2$, and therefore $\asetbis = \set{\avariable_1, \avariable_2, \avariable_3, \avariable_4}$. 
We have $\overline{\avariable_1} = \avariable_3$ and $\overline{\avariable_2} = \avariable_4$. Thus, the translation
$\atranslation(\avariable_1 \hpto \avariable_2 \magicwand \avariable_1 \hpto \avariable_2, \asetbis)$ is defined as the formula below:
$$
\Big((\alloc{\avariable_3} \wedge \alloc{\avariable_4} \wedge \neg \alloc{\avariable_1} \wedge \neg \alloc{\avariable_2} \wedge \Safe(\asetbis)
\wedge n(\avariable_3) \hpto n(\avariable_4))\ \magicwand
$$
$$
\big((n(\avariable_1) = n(\avariable_3) \wedge n(\avariable_2) = n(\avariable_4) \wedge \Safe(\asetbis)) \Rightarrow
((\alloc{\avariable_3} \wedge \alloc{\avariable_4} \wedge \size = 2) \separate  n(\avariable_1) \hpto n(\avariable_2))\big)
\Big).
$$
\end{example}

Let us informally analyse the translation $\atranslation(\aformulabis_1 \magicwand \aformulabis_2, \asetbis)$.
It is a formula of the form $\aformulabis_1' \magicwand \aformulabis_2'$ where $\aformulabis'_1$ expresses the four constraints below (which follow the four conjuncts of $\aformulabis_1'$, from left to right):
\begin{enumerate}
\itemsep 0 cm
\item All the variables $\overline{\avariableter}$ with $\avariableter \in \asetter$ are allocated.
\item None of the variables $\overline{\avariable}$ with $ \avariable \in \asetbis \setminus \asetter$ is allocated. 
\item The environment around the variables $\avariable \in \asetbis$ is safe. In particular, none of the variables in $\asetbis$ corresponds to a location that is pointed by another location.
\item Finally, one needs to guarantee that the translated formula $\atranslation(\aformulabis_1,\asetbis)$ holds true when each variable
$\avariable \in \asetbis$ is replaced by its copy $\overline{\avariable}$, 
i.e.~that $\atranslation(\aformulabis_1[\avariable \gets \overline{\avariable} \mid \avariable \in \asetbis],\asetbis)$ holds.
\end{enumerate}
Essentially, the formula $\aformulabis_1'$ considers heaps satisfying $\atranslation(\aformulabis_1,\asetbis)$, but with respect to the copies $\overline{\asetter}$ of the free variables $\asetter$ in $\aformulabis_1$.
On these heaps, $\aformulabis'_2$ checks whether the following two conditions hold (antecedent of $\aformulabis'_2$):
\begin{enumerate}
\item[(A)] The values for $\avariableter$ and $\overline{\avariableter}$ are equal for all $\avariableter \in \asetter$.
\item[(B)] As in (3), the environment around the variables $\avariable \in \asetbis$ is safe.
\end{enumerate}
When these two conditions hold, $\aformulabis_2'$ removes the 
assignments corresponding to the variables in $\overline{\asetter}$,
see subformula $\bigwedge_{\avariableter \in \asetter} \alloc{\overline{\avariableter}} \wedge \size = \card{\asetter}$, and then check whether the formula $\atranslation(\aformulabis_2,\asetbis)$ holds.
Notice that freeing the variables in $\overline{\asetter}$ allows us to reuse them whenever a magic wand appears in $\atranslation(\aformulabis_2,\asetbis)$.

Here is the main result of this section, that is the correctness of the translation $\atranslation{(\aformulabis,\asetbis)}$.

\begin{lemma}
\label{lemma-correctness-translation}
Let $\asetbis = \set{\avariable_1, \ldots, \avariable_{2q}}$ and
$E$ be either $\{\avariable_1, \ldots, \avariable_q\}$ or $\overline{\{\avariable_1, \ldots, \avariable_q\}} = \{\avariable_{q+1},\dots,\avariable_{2q}\}$.
Let $\aset \subseteq E$.
Let $\aformulabis$
be a well-quantified formula in $\foseplogic{\magicwand}$ with free variables among $\aset$ and bound variables among $E \setminus \aset$.
If 
 $\triple{\LOC_1}{\astore_1}{\aheap_1} \rhd_\asetbis^{\aset} \ \triple{\LOC_2}{\astore_2}{\aheap_2}$ then
$\pair{\astore_1}{\aheap_1}  \models \aformulabis$ iff $\pair{\astore_2}{\aheap_2} \models \atranslation(\aformulabis,\asetbis)$.
\end{lemma}

The statement of Lemma~\ref{lemma-correctness-translation}, as well as its proof, refers to the involution $\overline{(.)}$ such that $\overline{\avariable_i} = \avariable_{i+q}$.

\begin{proof}
  The proof is by induction on the structure of $\aformulabis$. We start by proving the two base cases
with the atomic formulae $\avariable_i = \avariable_j$ and $\avariable_i \hpto \avariable_j$.
\begin{enumerate}
  \item Let $\aformulabis$ be $\avariable_i = \avariable_j$. We have $\{\avariable_i,\avariable_j\} \subseteq \aset \subseteq \asetbis$
  and $\atranslation(\aformulabis,\asetbis) = n(\avariable_i) = n(\avariable_j)$ by definition. The following equivalences hold.
  \begin{center}
    \begin{tabular}{rclr}
      $\pair{\astore_1}{\aheap_1} \models \avariable_i = \avariable_j$
      & $\iff$ & 
      $\astore_1(\avariable_i) = \astore_1(\avariable_j),$\\
      & $\iff$ & 
      $\aheap_2(\astore_2(\avariable_i)) = \aheap_2(\astore_2(\avariable_j)),$
        &\qquad(Definition~\ref{def:rhd})\\
      & $\iff$ & $\pair{\astore_2}{\aheap_2} \models n(\avariable_i) = n(\avariable_j)$.
    \end{tabular}
  \end{center}
  \item Let $\aformulabis$ be
   $\avariable_i \hpto \avariable_j$.  We have $\{\avariable_i,\avariable_j\} \subseteq \aset \subseteq \asetbis$ and
 $\atranslation(\aformulabis,\asetbis) = n(\avariable_i) \hpto n(\avariable_j)$ by definition. The following equivalences hold.
 \begin{center}
  \begin{tabular}{rclr}
    $\pair{\astore_1}{\aheap_1} \models \avariable_i \hpto \avariable_j$
    & $\iff$ & 
    $\aheap_1(\astore_1(\avariable_i)) = \astore_1(\avariable_j),$\\
    & $\iff$ & 
    $\aheap_2(\aheap_2(\astore_2(\avariable_i))) = \aheap_2(\astore_2(\avariable_j)),$
      &\qquad(Definition~\ref{def:rhd})\\
    & $\iff$ & $\pair{\astore_2}{\aheap_2} \models n(\avariable_i) \hpto n(\avariable_j)$.
  \end{tabular}
\end{center}
  \end{enumerate}
  We now consider the induction step and we distinguish the different cases depending on the outermost connective.
  To be precise, we assume the induction hypothesis to be the following:

  \vspace{3pt}
  \begin{minipage}{0.95\linewidth}
  Let $N \in \Nat$.
  Let $E$ be either $\{\avariable_1, \ldots, \avariable_q\}$ or $\overline{\{\avariable_1, \ldots, \avariable_q\}} = \{\avariable_{q+1},\dots,\avariable_{2q}\}$.
  Let $\aset \subseteq E$.
  Let $\aformula$
  be a well-quantified formula in $\foseplogic{\magicwand}$ with free variables among $\aset$, bound variables among $E \setminus \aset$, and size at most $N$.
  If 
   $\triple{\LOC_1}{\astore_1}{\aheap_1} \rhd_\asetbis^{\aset} \ \triple{\LOC_2}{\astore_2}{\aheap_2}$ then
  $\pair{\astore_1}{\aheap_1}  \models \aformula$ iff $\pair{\astore_2}{\aheap_2} \models \atranslation(\aformula,\asetbis)$.
  \end{minipage}
  \vspace{3pt}
  
  \noindent
  In the induction step, we take a formula $\aformulabis$ of size $N+1$. We omit the obvious cases for the
Boolean connectives. 
  We now prove the result for $\forall \avariable_i \ \aformula$ and $\aformula_1 \magicwand \aformula_2$.
  \begin{enumerate}
    \setcounter{enumi}{2}
    \item Let $\aformulabis = \forall \ \avariable_i \ \aformula$ with $\avariable_i \in E \setminus \aset$ and therefore by definition
    \[
      \atranslation(\forall \ \avariable_i \ \aformula,\asetbis) = (\alloc{\avariable_i} \wedge \size = 1) \magicwand (\Safe(\asetbis) \Rightarrow \atranslation(\aformula, \asetbis)).
    \]
    By definition, $\pair{\astore_1}{\aheap_1} \models \forall \avariable_i \ \aformula$ if and only if for all locations $\alocation \in \LOC_1$,
    $\pair{\astore_1[\avariable_i \gets \alocation]}{\aheap_1} \models \aformula$.

    Let us now consider the generalised memory state $\triple{\LOC_2}{\astore_2}{\aheap_2 + \{\astore_2(\avariable_i) \pto \alocation\}}$.
    Since $\avariable_i \not\in \aset$, the heap $\aheap_2 + \{\astore_2(\avariable_i) \pto \alocation\}$ is well-defined.
    From $\triple{\LOC_1}{\astore_1}{\aheap_1} \rhd_\asetbis^{\aset} \ \triple{\LOC_2}{\astore_2}{\aheap_2}$ and by~Definition~\ref{def:rhd}, we obtain
    \[
      \triple{\LOC_1}{\astore_1[\avariable_i \gets \alocation]}{\aheap_1} \rhd_\asetbis^{\aset \cup \{\avariable_i\}} \ \triple{\LOC_2}{\astore_2}{\aheap_2+ \{\astore_2(\avariable_i) \pto \alocation\}}.
    \]
    We apply the induction hypothesis, and conclude that:
    \begin{center}
      \begin{tabular}{rl}
      & for all $\alocation \in \LOC_1$, $\pair{\astore_1[\avariable_i \gets \alocation]}{\aheap_1} \models \aformula$,\\
      iff &
      for all $\alocation \in \LOC_1$, $\pair{\astore_2}{\aheap_2 + \{\astore_2(\avariable_i) \pto \alocation\}} \models \atranslation(\aformula,\asetbis)$.
      \end{tabular}
    \end{center}
    Moreover, due to the fulfilment of the conditions in $\rhd_\asetbis^{\aset \cup \{\avariable_i\}}$, any 
    memory state $\pair{\astore_2}{\aheap_2 + \{\astore_2(\avariable_i) \pto \alocation'\}}$ satisfies
    $\Safe(\asetbis)$ if and only if $\alocation' \in \LOC_1$.
    Therefore, 
    \begin{center}
      \begin{tabular}{rl}
      & for all $\alocation \in \LOC_1$, $\pair{\astore_2}{\aheap_2 + \{\astore_2(\avariable_i) \pto \alocation\}} \models \atranslation(\aformula,\asetbis)$,\\
      iff &
      for all $\alocation \in \LOC_2$, $\pair{\astore_2}{\aheap_2 + \{\astore_2(\avariable_i) \pto \alocation\}} \models \Safe(\asetbis) \Rightarrow \atranslation(\aformula,\asetbis)$.
      \end{tabular}
    \end{center}
    Indeed, if $\alocation \in \LOC_1$ then $\Safe(\asetbis)$ and $\atranslation(\aformula,\asetbis)$ are both satisfied (for $\Safe(\asetbis)$, remember that all locations in $\set{\astore_2(\avariablebis) \ \mid \ \avariablebis \in \aset} \cap \domain{\aheap_2}$ already point to elements of $\LOC_1$, due to the fulfilment of the conditions in $\rhd_\asetbis^{\aset}$),
    otherwise whenever $\alocation \in \LOC_2 \setminus \LOC_1$ the premise of the implication  does not hold
    and therefore the formula is trivially satisfied. Observe that 
    a memory state $\pair{\astore_2}{\aheap'}$ satisfies $\alloc{\avariable_i} \wedge \size = 1$ if and only if $\aheap'$ is of the form $\{\astore_2(\avariable_i) \pto \alocation''\}$, for some $\alocation'' \in \LOC_2$.
    Then, we conclude that 
    \begin{center}
      \begin{tabular}{rl}
      & 
      for all $\alocation \in \LOC_2$, $\pair{\astore_2}{\aheap_2 + \{\astore_2(\avariable_i) \pto \alocation\}} \models \Safe(\asetbis) \Rightarrow \atranslation(\aformula,\asetbis)$,\\
      iff & 
      for all $\aheap'$, if $(\aheap' \bot \aheap_2 \text{ and }\pair{\astore_2}{\aheap'} \models \alloc{\avariable_i} \wedge \size = 1)$\\
      & \quad then $\pair{\astore_2}{\aheap_2 + \aheap'} \models \Safe(\aset) \Rightarrow \atranslation(\aformula,\asetbis)$ \\
      iff & 
      $\pair{\astore_2}{\aheap_2} \models (\alloc{\avariable_i} \wedge \size = 1) \magicwand 
      (\Safe(\asetbis) \Rightarrow \atranslation(\aformula,\asetbis))$.
      \end{tabular}
    \end{center}
  \item Let $\aformulabis = \aformula_1 \magicwand \aformula_2$.
  Let $\asetter$ be the set of free variables in $\aformula_1$.
  Below, we assume that 
  for every $\avariable \in \aset$, $\astore_1(\avariable) = \astore_1(\overline{\avariable})$. Notice that this assumption is without loss of generality, 
  as $\aformulabis$ is written with free variables among $\aset$, where $\aset \cap \overline{\aset} = \emptyset$, and the memory state 
  $\pair{\astore_1[\overline{\avariable} \gets \astore_1(\avariable) \mid \avariable \in \aset]}{\aheap_1}$ is $\aset$-isomorphic to $\pair{\astore_1}{\aheap_1}$, 
  where $\astore_1[\overline{\avariable} \gets \astore_1(\avariable) \mid \avariable \in \aset]$ is the store obtained from $\astore_1$ by assigning $\astore_1(\avariable)$ to every variable $\overline{\avariable} \in \overline{\aset}$. 
  By Lemma~\ref{lemma-equivalence-isomorphism}, 
  $\pair{\astore_1}{\aheap_1} \models \aformulabis$ if and only if
  $\pair{\astore_1[\overline{\avariable} \gets \astore_1(\avariable) \mid \avariable \in \aset]}{\aheap_1} \models \aformulabis$.

  ($\Rightarrow$):
  If $\pair{\astore_1}{\aheap_1} \models \aformula_1 \magicwand \aformula_2$, then by definition,
  for all heaps $\aheap_1'$, if $\aheap_1' \bot \aheap_1$ and $\pair{\astore_1}{\aheap_1'} \models \aformula_1$, then 
  we have $\pair{\astore_1}{\aheap_1+\aheap_1'} \models \aformula_2$.
  We show that $\pair{\astore_2}{\aheap_2} \models \atranslation(\aformula_1 \magicwand \aformula_2, \asetbis)$.
  By definition, 
  $\atranslation(\aformula_1 \magicwand \aformula_2, \asetbis)$,
  holds whenever for all heaps $\aheap_2'$, if 
  \begin{enumerate}[align=left]
    \item[(A1)] $\aheap_2' \bot \aheap_2$,
    \item[(A2)]  $\pair{\astore_2}{\aheap_2'} \, \models \, (\displaystyle\,\bigwedge_{\mathclap{\avariableter \in Z}}\, \alloc{\overline{\avariableter}}) \wedge
    (\displaystyle\,\bigwedge_{\mathclap{\avariableter \in \asetbis \setminus Z}}\, \neg \alloc{\overline{\avariableter}})
    \wedge \Safe(\asetbis) \wedge \atranslation(\aformula_1[\avariable \gets \overline{\avariable} \mid \avariable \in \asetbis],\asetbis)$,
    \item[(A3)] $\pair{\astore_2}{\aheap_2 + \aheap_2'} \, \models \, (\displaystyle\,\bigwedge_{\mathclap{\avariableter \in Z}}\, n(\avariableter) = n(\overline{\avariableter})) \wedge \Safe(\asetbis)$,
  \end{enumerate}
  then 
  \begin{enumerate}[align=left]
   \item[(C1)]
    $\pair{\astore_2}{\aheap_2 + \aheap_2'} \models (\displaystyle\,\bigwedge_{\mathclap{\avariableter \in \asetter}}\, \alloc{\overline{\avariableter}} \wedge \size = \card{\asetter}) \separate
    \atranslation(\aformula_2, \asetbis).$
  \end{enumerate}

  Let $\aheap_2'$ be some heap that satisfies the premises (A1)--(A3) 
  of the implication.
  From (A3) together with
   $\triple{\LOC_1}{\astore_1}{\aheap_1} \rhd_\asetbis^{\aset} \ \triple{\LOC_2}{\astore_2}{\aheap_2}$, which implies that
  for all $\avariable \in \aset
   \supseteq \asetter$, we have
  $\aheap_2(\astore_2(\avariable)) = \astore_1(\avariable)$, we conclude that  for all $\avariableter \in \asetter$,
  we have $\aheap_2'(\astore_2(\overline{\avariableter})) = \astore_1(\avariableter)$. Moreover, (A2) 
  entails that $\aheap_2'$ can be written as $\aheap_{\asetter} + \aheap'_1$
  where $\aheap_{\asetter} \egdef \{ \astore_2(\overline{\avariableter}) \pto \astore_1(\avariableter) \mid \avariableter
  \in \asetter \}$
  ($\aheap'_1$ is then defined as the unique subheap satisfying $\aheap_2' = \aheap_{\asetter} + \aheap'_1$,
  once $\aheap_{\asetter}$ is defined).
  Observe that the domain of $\aheap_{\asetter}$ is the set of locations in $\LOC_2$ that corresponds to variables in $\overline{\asetter}$, and its codomain is included in $\LOC_1$. For $\aheap'_1$ instead, both its domain and codomain are included in $\LOC_1$.
  By Definition~\ref{def:rhd} and the assumption that for all $\avariable \in \aset$, $\astore_1(\avariable) = \astore_1(\overline{\avariable})$, it holds that
  \[
  \triple{\LOC_1}{\astore_1}{\aheap_1'} \rhd_{\asetbis}^{\overline{\asetter}}
  \triple{\LOC_2}{\astore_2}{\aheap_2'}.
  \]
  Since $\asetter$ is the set of free variables in $\aformula_1$,
  $\overline{\asetter}$ is the set of free variables in~${\aformula_1[\avariable \gets \overline{\avariable} \mid \avariable \in \asetbis]}$.
  From
  $\pair{\astore_2}{\aheap_2'} \models \atranslation(\aformula_1[\avariable \gets \overline{\avariable} \mid \avariable \in \asetbis],\asetbis)$, 
  we apply the induction hypothesis and conclude that
  $\pair{\astore_1}{\aheap_1'} \models \aformula_1[\avariable \gets \overline{\avariable} \mid \avariable \in \asetbis]$.
  Again thanks to the assumption that for all $\avariable \in \aset$, $\astore_1(\avariable) = \astore_1(\overline{\avariable})$, 
  this allows us to conclude that $\pair{\astore_1}{\aheap_1'} \models \aformula_1$. 
  Indeed, this holds directly from the well-known axiom of separation logic $\avariable = \avariablebis \land \aformula \implies \aformula[\avariable \gets \avariablebis]$ (see e.g.~\cite{DemriLM20}).
  Since $\aheap_1' \sqsubseteq \aheap_2'$,
  $\aheap_1 \sqsubseteq \aheap_2$ and $\aheap_2' \bot \aheap_2$ (A1), the property $\aheap_1' \bot \aheap_1$ holds and 
  therefore, by the assumption
  $\pair{\astore_1}{\aheap_1} \models \aformula_1 \magicwand \aformula_2$,
  we get $\pair{\astore_1}{\aheap_1 + \aheap_1'} \models \aformula_2$.
  As such, since $\triple{\LOC_1}{\astore_1}{\aheap_1 + \aheap_1'} \rhd_\asetbis^{\aset} \triple{\LOC_2}{\astore_2}{\aheap_2 + \aheap_1'}$
  holds because $\aheap_1'$ has domain and codomain included in $\LOC_1$ and is a subset of $\aheap_2'$ which is disjoint from $\aheap_2$,
  we can use the induction hypothesis and obtain that $\pair{\astore_2}{\aheap_2 + \aheap_1'} \models \atranslation(\aformula_2,\asetbis)$.
  Together with $\aheap_2' = \aheap_1' + \aheap_{\asetter}$
  and
  $\pair{\astore_2}{\aheap_{\asetter}} \models \bigwedge_{\avariableter \in \asetter} \alloc{\overline{\avariableter}} \wedge
  \size = \card{\asetter}$, 
  this allows us to derive (C1).

  ($\Leftarrow$):
  For the other direction, suppose that $\pair{\astore_2}{\aheap_2} \models \atranslation(\aformula_1 \magicwand \aformula_2, \asetbis)$.
This means that for all heaps $\aheap_2'$ if (A1)-(A3) holds, then (C1) holds. 
\cut{  
This means that for all heaps $\aheap_2'$ if 
  \begin{enumerate}[align=left]
    \item[(A1)] $\aheap_2' \bot \aheap_2$,
    \item[(A2)]  $\pair{\astore_2}{\aheap_2'} \, \models \, (\bigwedge_{{\avariableter \in Z}} \alloc{\overline{\avariableter}}) \wedge
    (\bigwedge_{{\avariableter \in \asetbis \setminus Z}} \neg \alloc{\overline{\avariableter}})
    \wedge \Safe(\asetbis) \wedge \atranslation(\aformula_1[\avariable \gets \overline{\avariable} \mid \avariable \in \asetbis],\asetbis)$,
    \item[(A3)] $\pair{\astore_2}{\aheap_2 + \aheap_2'} \, \models \, (\bigwedge_{\avariableter \in Z} n(\avariableter) = n(\overline{\avariableter})) \wedge \Safe(\asetbis)$,
  \end{enumerate}
  then 
  \begin{enumerate}[align=left]
   \item[(C1)]
    $\pair{\astore_2}{\aheap_2 + \aheap_2'} \models (\bigwedge_{{\avariableter \in \asetter}} \alloc{\overline{\avariableter}} \wedge \size = \card{\asetter}) \separate
    \atranslation(\aformula_2, \asetbis).$
  \end{enumerate}
}
  Let us  prove that
  $\pair{\astore_1}{\aheap_1} \models \aformula_1 \magicwand \aformula_2$, which by definition holds whenever for all heaps
  $\aheap_1'$, if $\aheap_1' \bot \aheap_1$ and $\pair{\astore_1}{\aheap_1'} \models \aformula_1$ then
  $\pair{\astore_1}{\aheap_1+\aheap_1'} \models \aformula_2$.
  Let $\aheap_1'$ be a heap disjoint from $\aheap_1$ satisfying $\pair{\astore_1}{\aheap_1'} \models \aformula_1$.
  Let $\aheap_2' \egdef \aheap_1' + \{\astore_2(\overline{\avariableter}) \pto \astore_1(\avariableter) \mid
  \avariableter \in \asetter\}$.
  By definition of $\rhd^{\aset}_\asetbis$
  together with the assumption that for all $\avariable \in \aset$, 
  $\astore_1(\avariable) = \astore_1(\overline{\avariable})$, we have 
  \[
    \triple{\LOC_1}{\astore_1}{\aheap_1'} \rhd^{\overline{\asetter}}_\asetbis \triple{\LOC_2}{\astore_2}{\aheap_2'}.
  \]
  By the induction hypothesis, we get
  $\pair{\astore_2}{\aheap_2'} \models \atranslation(\aformula_1[\avariable \gets \overline{\avariable} \mid \avariable \in \asetbis],\asetbis)$,
  whereas from
  the definition of $\rhd^{\overline{\asetter}}_\asetbis$, it holds that $\pair{\astore_2}{\aheap_2'}$ satisfies
  \[
  (\bigwedge_{\mathclap{\avariableter \in \asetter}} \alloc{\overline{\avariableter}}) \wedge
  (\bigwedge_{\mathclap{\avariableter \in \aset \setminus \asetter}} \neg \alloc{\overline{\avariableter}})
  \wedge \Safe(\asetbis).
  \]
  Thus, (A2) holds.
  Besides this, $\aheap_2$ can be written as
  $\aheap_1 + \{ \astore_2(\avariablebis) \pto \astore_1(\avariablebis) \mid \avariablebis \in \aset\}$ 
  (condition (3) in~Definition~\ref{def:rhd}).
  Consequently, (A1) holds and the heap $\aheap_2 + \aheap_2'$ can be written as
  \[
    \aheap_1 + \aheap_1' + \{ \astore_2(\avariable) \pto \astore_1(\avariable) \mid \avariable \in \aset\} +
   \{\astore_2(\overline{\avariableter}) \pto \astore_1(\avariableter) \mid \avariableter \in \asetter\}.
  \]
  Therefore, it holds that for all $\avariableter \in \asetter$,
  $(\aheap_2 + \aheap_2')(\astore_2(\avariableter)) = (\aheap_2 + \aheap_2')(\astore_2(\overline{\avariableter}))$ or,
  equivalently, $\pair{\astore_2}{\aheap_2 + \aheap_2'} \models \bigwedge_{\avariableter \in \asetter} n(\avariableter) = n(\overline{\avariableter})$. Indeed, $\asetter \subseteq \aset$, $\avariableter \in \asetter$,
and, $\astore_2(\avariableter)$ and $\astore_2(\overline{\avariableter})$ point to the same location. 
  Furthermore, $\pair{\astore_2}{\aheap_2 + \aheap_2'}$ also satisfies $\Safe(\asetbis)$. (A3) holds and thus so does (C1), i.e.
  \[
  \pair{\astore_2}{\aheap_2 + \aheap_2'} \models (\bigwedge_{\avariableter \in \asetter}
   \alloc{\overline{\avariableter}} \wedge \size = \card{\asetter}) \separate
  \atranslation(\aformula_2, \asetbis).
  \]
  Again, $\aheap_2 + \aheap_2'$ can be written as
  $\aheap_1 + \aheap_1' + \{ \astore_2(\avariable) \pto \astore_1(\avariable) \mid \avariable \in \aset\} +
    \{\astore_2(\overline{\avariableter}) \pto \astore_1(\avariableter) \mid \avariableter \in \asetter\}$,
  and the formula $\bigwedge_{\avariableter \in \asetter} \alloc{\overline{\avariableter}} \wedge \size = \card{\asetter}$ can  be only  satisfied
  by  the memory state with the heap $\{\astore_2(\overline{\avariableter}) \pto \astore_1(\avariableter) \mid \avariableter \in
  \asetter\}$. Consequently, $\aheap_2 + \aheap_1' = \aheap_1 + \aheap_1' + \{ \astore_2(\avariable) \pto \astore_1(\avariable) \mid \avariablebis \in \aset\}$
  is such that $\pair{\astore_2}{\aheap_2 + \aheap_1'} \models \atranslation(\aformula_2, \asetbis)$.
  Lastly, since $\triple{\LOC_1}{\astore_1}{\aheap_1 + \aheap_1'} \rhd^{\aset}_\asetbis \triple{\LOC_2}{\astore_2}{\aheap_2 + \aheap_1'}$,
  we can use the induction hypothesis to conclude that $\pair{\astore}{\aheap_1 + \aheap_1'} \models \aformula_2$.
  \qedhere
  \end{enumerate}
\end{proof}

Let $\aformula$ be a formula of $\foseplogic{\magicwand}$ written with variables among 
$\set{\avariable_1,\dots,\avariable_q}$.
We define the translation $\mathcal{T}_{\rm SAT}(\aformula)$ into $\seplogic{\separate, \magicwand, \ls}$
augmented with the new predicates where
$\atranslation(\aformula, \asetbis)$ is defined recursively as before ( $\asetbis = \set{\avariable_1, \ldots, \avariable_{2q}}$).
$$
\mathcal{T}_{\rm SAT}(\aformula) \egdef
(\bigwedge_{\mathclap{i \in \interval{1}{2q}}} \neg \alloc{\avariable_i}) \wedge \Safe(\asetbis) \wedge \atranslation(\aformula, \asetbis).
$$
The first two conjuncts specify initial conditions, namely each variable $\avariable$ in $\asetbis$
is interpreted by a location that is unallocated, it is not in the  heap range and it is distinct from the interpretation of
all other variables; in other words, the value for $\avariable$ is isolated.
Similarly, let $\mathcal{T}_{\rm VAL}(\aformula)$ be the formula in $\seplogic{\separate, \magicwand, \ls}$
defined by
$$
\mathcal{T}_{\rm VAL}(\aformula) \egdef
((\bigwedge_{\mathclap{i \in \interval{1}{2q}}} \neg \alloc{\avariable_i}) \wedge \Safe(\asetbis)) \Rightarrow
\atranslation(\aformula, \asetbis).
$$
As a consequence of Lemma~\ref{lemma-correctness-translation},
$\aformula$ and $\mathcal{T}_{\rm SAT}(\aformula)$ are shown equisatisfiable, whereas $\aformula$ and $\mathcal{T}_{\rm VAL}(\aformula)$ are shown equivalid.
\begin{corollary}
\label{corollary-satisfiability-validity}
Let $\aformula$ be a closed formula in $\foseplogic{\magicwand}$ using quantified variables
among $\set{\avariable_1, \ldots, \avariable_q}$.
\begin{center}
  \hfill
  \textbf{(I)}\ \,$\aformula$ and $\mathcal{T}_{\rm SAT}(\aformula)$ are equisatisfiable.
  \hfill 
  \textbf{(II)}\ \,$\aformula$ and $\mathcal{T}_{\rm VAL}(\aformula)$ are equivalid.
  \hfill\,
\end{center}
\end{corollary}

\begin{proof} (I) First, suppose that $\aformula$ is satisfiable, i.e. there is a memory state $\pair{\astore}{\aheap}$
such that $\pair{\astore}{\aheap} \models \aformula$. It is then easy to define a generalised memory state
$\triple{\LOC \setminus \set{1, \ldots, 2q}}{\astore'}{\aheap'}$ isomorphic
with respect to the empty set to $\pair{\astore}{\aheap}$ (i.e., using the equivalence relation $\approx_{\emptyset}$)
and satisfying $\aformula$ by Lemma~\ref{lemma-equivalence-isomorphism}.
Typically, to define $\pair{\astore'}{\aheap'}$ from $\pair{\astore}{\aheap}$, it is sufficient to shift all the values with an offset
equal to $2q+1$.
Now, let $\triple{\LOC}{\astore''}{\aheap''}$ be the generalised memory state such that $\aheap'' = \aheap'$,
and for all $i \in \interval{1}{2q}$, we have $\astore''(\avariable_i) = i$. One can check that
$\triple{\LOC}{\astore''}{\aheap''} \models
(\bigwedge_{i \in \interval{1}{2q}} \neg \alloc{\avariable_i}) \wedge \Safe(\asetbis)$
with $\asetbis = \set{\avariable_1, \ldots, \avariable_{2q}}$. Moreover,
$\triple{\LOC \setminus \set{1, \ldots, 2q}}{\astore'}{\aheap'} \rhd_\asetbis^{\emptyset} \triple{\LOC}{\astore''}{\aheap''}$.
By Lemma~\ref{lemma-correctness-translation}, we have $\pair{\astore''}{\aheap''} \models  \atranslation(\aformula, \asetbis)$.
So, $\pair{\astore''}{\aheap''} \models \mathcal{T}_{\rm SAT}(\aformula)$.

Conversely, suppose that  $\pair{\astore}{\aheap} \models \mathcal{T}_{\rm SAT}(\aformula)$.
Let $\triple{\LOC'}{\astore'}{\aheap'}$ be the generalised memory state defined as follows:
\begin{enumerate}
\itemsep 0 cm
\item $\LOC' = \LOC \setminus \set{\astore(\avariable) \ \mid \ \avariable \in \asetbis}$,
\item the heap $\aheap'$ is equal to the restriction of $\aheap$ to locations in $\LOC'$.
\end{enumerate}
Since $\pair{\astore}{\aheap} \models \Safe(\asetbis)$ and by construction of
$\triple{\LOC'}{\astore'}{\aheap'}$,
we have
$$
\triple{\LOC'}{\astore'}{\aheap'} \rhd_\asetbis^{\emptyset} \triple{\LOC}{\astore}{\aheap}.
$$
As $\pair{\astore}{\aheap} \models  \atranslation(\aformula, \asetbis)$, by Lemma~\ref{lemma-correctness-translation},
we get $\triple{\LOC'}{\astore'}{\aheap'} \models \aformula$. Now, note that $\LOC' \subseteq \LOC$ and therefore
we also have $\triple{\LOC}{\astore'}{\aheap'} \models \aformula$, where $\triple{\LOC}{\astore'}{\aheap'}$ is understood
as a standard memory state $\pair{\astore'}{\aheap'}$. \\
\noindent
(II) Similar to (I). 
\end{proof}

\subsection{Expressing the auxiliary atomic predicates}
\label{section-expressing-predicates}

To complete the reduction, we explain how to express the formulae $\allocback{\avariable}$,
$n(\avariable) =  n(\avariablebis)$ and $n(\avariable) \hpto  n(\avariablebis)$ within
$\seplogic{\separate, \magicwand, \ls}$.
Let us introduce a few macros that shall be helpful.
\begin{itemize}

\item Given $\aformula$ in $\seplogic{\separate, \magicwand, \reachplus}$ and $\gamma \geq 0$, we write
$[\aformula]_{\gamma}$ to denote the formula ${(\size = \gamma \land \aformula) \separate \true}$. It is easy to show that
for any memory state $\pair{\astore}{\aheap}$, we have $\pair{\astore}{\aheap} \models [\aformula]_{\gamma}$
iff there is  $\aheap' \sqsubseteq \aheap$ such that $\card{\domain{\aheap'}} = \gamma$ and
$\pair{\astore}{\aheap'} \models \aformula$.
Such formulae $[\aformula]_{\gamma}$ can be used to ensure that the minimum path between two locations interpreted by
program variables is of length $\gamma$.

\item We write $\reach(\avariable,\avariablebis) = \gamma$ to denote the formula
$[\ls(\avariable,\avariablebis)]_{\gamma}$, which is satisfied in any memory state $\pair{\astore}{\aheap}$ where
$\aheap^\gamma(\astore(\avariable)) = \astore(\avariablebis)$ and $\gamma$ is minimal (no cycles allowed).
Lastly, we write $\reach(\avariable,\avariablebis) \leq \gamma$ to denote the formula $\bigvee_{0 \leq \gamma' \leq \gamma} \reach(\avariable,\avariablebis) = \gamma'$.
\end{itemize}
In order to define the existence of a predecessor (i.e. $\allocback{\avariable}$) in  $\seplogic{\separate, \magicwand, \ls}$, we need to take advantage of
an auxiliary variable $\avariablebis$ whose value is different from the one for $\avariable$.
Let $\mathtt{alloc}_{\avariablebis}^{-1}(\avariable)$ be the formula
\[
\avariable \hpto \avariable \lor \avariablebis \hpto \avariable \lor (\true \separate ((\alloc{\avariablebis} \land \lnot
(\avariablebis \hpto \avariable)
\land \size = 1) \lollipop \reach(\avariablebis,\avariable) = 2))
\]

\begin{lemma}
\label{lemma-inverse-alloc}
Let $\avariable,\avariablebis \in \PVAR$.
\begin{enumerate}[label=\normalfont{\textbf{(\Roman*)}}]
\item For all memory states $\pair{\astore}{\aheap}$, if $\astore(\avariable) \neq \astore(\avariablebis)$,
then $\pair{\astore}{\aheap} \models \mathtt{alloc}_{\avariablebis}^{-1}(\avariable)$ iff $\astore(\avariable) \in \range{\aheap}$.
\item In the translation $\atranslation(\aformula, \asetbis)$,
$\mathtt{alloc}^{-1}(\avariable)$ (with $\avariable \in \asetbis$) can be replaced with $\mathtt{alloc}_{\overline{\avariable}}^{-1}(\avariable)$.
\end{enumerate}
\end{lemma}
As stated in Lemma~\ref{lemma-inverse-alloc}(II), we can exploit the fact that in the translation of a formula with variables in
$\set{\avariable_1,\dots,\avariable_q}$, we use $2q$ variables that correspond to $2q$ distinguished locations in the heap in order
to retain the soundness of the translation while using $\allocbacktwo{\overline{\avariable}}{\avariable}$ as $\allocback{\avariable}$.

\begin{proof}
  For (I), ($\Rightarrow$):
  Suppose that $\pair{\astore}{\aheap} \models \allocbackaux{\avariable}{\avariablebis}$.
  If $\pair{\astore}{\aheap}$ satisfies either $\avariable \hpto \avariable$ or $\avariablebis \hpto \avariable$, then obviously
  $\astore(\avariable) \in \range{\aheap}$.
  Otherwise, $\pair{\astore}{\aheap}$ must satisfy the third conjunct of $\allocbackaux{\avariable}{\avariablebis}$:
  \begin{center}
    $\pair{\astore}{\aheap} \models \true \separate ((\alloc{\avariablebis} \land \sizeeq{1}) \septraction \reach(\avariablebis,\avariable) = 2)$.
  \end{center}
  \noindent In this case,
   ${\pair{\astore}{\aheap'} \models (\alloc{\avariablebis} \land \sizeeq{1}) \septraction \reach(\avariablebis,\avariable) = 2}$
   holds for some heap $\aheap' \sqsubseteq
   \aheap$.
  From the semantics of the septraction operator, there is a heap $\aheap''$ disjoint from $\aheap'$ and such that
  \begin{center}
    \hfill
    (1) $\astore(\avariablebis) \in \domain{\aheap''}$,
    \hfill
    (2) $\card{\aheap''} = 1$,
    \hfill
    (3) $\pair{\astore}{\aheap' + \aheap''} \models \reach(\avariablebis,\avariable) = 2$.
    \hfill\,
  \end{center}
  \noindent From (1) and $\aheap'' \perp \aheap'$, we conclude that $\astore(\avariablebis) \not \in \domain{\aheap'}$.
  Now,
  (3) implies that there is a location~$\alocation$ such that
  $\set{\astore(\avariablebis) \pto \alocation \pto \astore(\avariable)} \sqsubseteq \aheap' + \aheap''$, and moreover
  $\alocation$ must be distinct from both $\astore(\avariable)$ and $\astore(\avariablebis)$ (which are also assumed to be distinct).
  From~(1) and (2), it must be that $\aheap'' = \{\astore(\avariablebis) \pto{\alocation}\}$ and therefore
  $\{{\alocation}\pto{\astore(\avariable)}\} \sqsubseteq \aheap'$.
  From $\aheap' \sqsubseteq \aheap$ we then conclude that $\astore(\avariable) \in \range{\aheap}$.

  ($\Leftarrow$):
  Suppose there is a location $\alocation \in \domain{\aheap}$ such that~$\aheap(\alocation) = \astore(\avariable)$ (i.e.~$\astore(\avariable) \in \range{\aheap}$).
  First, suppose that~$\alocation = \astore(\avariable)$ or~$\alocation = \astore(\avariablebis)$.
  In this case, we directly derive that
  $\pair{\astore}{\aheap} \models \avariable \hpto \avariable \lor \avariablebis \hpto \avariable$,
  which in turn shows that $\pair{\astore}{\aheap} \models \allocbackaux{\avariable}{\avariablebis}$.
  Otherwise, consider the case where~$\alocation \neq \astore(\avariable)$ and~$\alocation \neq \astore(\avariablebis)$.
  Let $\aheap' \sqsubseteq \aheap$ be the heap~$\set{{\alocation} \pto {\astore(\avariable)}}$.
  As $\alocation \neq \astore(\avariablebis)$, the location $\astore(\avariablebis)$ is not a memory cell of $\aheap'$.
  Let us consider the heap $\aheap'' = \set{{\astore(\avariablebis)} \pto {\alocation}}$, so that $\pair{\astore}{\aheap''} \models \alloc{\avariablebis} \land \sizeeq{1}$.
  The heaps $\aheap'$ and $\aheap''$ are disjoint, and from their definition we have $\aheap' + \aheap'' = \set{\astore(\avariablebis) \pto \alocation \pto \astore(\avariable)}$.
  As~$\astore(\avariable)$,~$\alocation$ and~$\astore(\avariablebis)$ are all distinct locations,~${\pair{\astore}{\aheap' + \aheap''} \models \reach(\avariablebis,\avariable) = 2}$ holds, which in turn allows us to conclude
  that~$\pair{\astore}{\aheap'} \models (\alloc{\avariablebis} \land \sizeeq{1}) \septraction \reach(\avariablebis,\avariable) = 2$.
  As~$\aheap' \sqsubseteq \aheap$, this implies that
  $\pair{\astore}{\aheap} \models \true \separate ((\alloc{\avariablebis} \land \sizeeq{1}) \septraction \reach(\avariablebis,\avariable) = 2)$,
  which implies $\pair{\astore}{\aheap} \models \allocbackaux{\avariable}{\avariablebis}$.

  For (II), since by definition, $\avariable$ and $\overline{\avariable}$ are interpreted by different locations, they satisfy
  the additional hypothesis $\astore(\avariable) \not= \astore(\avariablebis)$ (where $\avariablebis = \overline{\avariable}$). Therefore,  we can  use one of these two variables to check if the other is in $\range{\aheap}$. As such, instead of $\mathtt{alloc}^{-1}(\avariable)$ we can use $\mathtt{alloc}_{\overline{\avariable}}^{-1}(\avariable)$ in the translation.
\end{proof}

Moreover, $\allocbacktwo{\avariablebis}{\avariable}$ allows us to express in $\seplogic{\separate, \magicwand, \ls}$ whether 
a location corresponding to a program variable reaches itself in exactly two steps (we use this property in the definition of 
$n(\avariable) \hpto n(\avariablebis)$).
We write $\avariable \hpto_\avariablebis^2 \avariable$ to denote the formula
$$\lnot (\avariable \hpto \avariable) \land ({\avariable \hpto \avariablebis} \Rightarrow \avariablebis \hpto \avariable) \land
[\alloc{\avariable} \land \allocbacktwo{\avariablebis}{\avariable} \land (\top \magicwand \lnot \reach(\avariable,\avariablebis) = 2)]_2.
$$

\begin{lemma}
\label{lemma-two-steps}
For any memory state $\pair{\astore}{\aheap}$ such that $\astore(\avariable) \neq \astore(\avariablebis)$,
we have $\pair{\astore}{\aheap} \models \avariable \hpto_\avariablebis^2 \avariable$ if and only if
$\aheap^2(\astore(\avariable)) = \astore(\avariable)$ and $\aheap(\astore(\avariable)) \neq \astore(\avariable)$.
\end{lemma}

\begin{proof} So, let $\pair{\astore}{\aheap}$ be a memory state such that $\astore(\avariable)$ is distinct from
$\astore(\avariablebis)$.

First, we assume that $\aheap^2(\astore(\avariable)) = \astore(\avariable)$ and 
$\aheap(\astore(\avariable)) \neq \astore(\avariable)$. So, there is a location $\alocation$ distinct from
$\astore(\avariable)$ such that $\aheap(\astore(\avariable)) = \alocation$ and $\aheap(\alocation) = \astore(\avariable)$.
Obviously, $\pair{\astore}{\aheap} \models \lnot (\avariable \hpto \avariable)$ as  
$\alocation$ is distinct from $\astore(\avariable)$. Below, we distinguish two cases.
\begin{description}
\item[Case 1: $\alocation \neq \astore(\avariablebis)$.] So, 
$\pair{\astore}{\aheap} \models  \neg \avariable \hpto \avariablebis$
and therefore $\pair{\astore}{\aheap} \models {\avariable \hpto \avariablebis} \Rightarrow \avariablebis \hpto \avariable$.
Let $\aheap'$ be the subheap of $\aheap$ with $\domain{\aheap'} = \set{\astore(\avariable), \alocation}$. 
We have $\pair{\astore}{\aheap'} \models \alloc{\avariable}$ and $\pair{\astore}{\aheap'} \models \allocbacktwo{\avariablebis}{\avariable}$
by Lemma~\ref{lemma-inverse-alloc}(I). Moreover, for all heaps $\aheap''$ such that $\aheap' \sqsubseteq \aheap''$, we
have  $\pair{\astore}{\aheap''} \not \models \reach(\avariable,\avariablebis) = 2$, as $\alocation \in \domain{\aheap'}$ and 
$\aheap' \sqsubseteq \aheap''$. 

\item[Case 2: $\alocation = \astore(\avariablebis)$.] 
So, 
$\pair{\astore}{\aheap}  \models \avariable \hpto \avariablebis \wedge \avariablebis \hpto \avariable$
and therefore $\pair{\astore}{\aheap} \models {\avariable \hpto \avariablebis} \Rightarrow \avariablebis \hpto \avariable$.
Let $\aheap'$ be the subheap of $\aheap$ with $\domain{\aheap'} = \set{\astore(\avariable), \alocation}$. 
We have $\pair{\astore}{\aheap} \models \alloc{\avariable}$ and $\pair{\astore}{\aheap} \models \allocbacktwo{\avariablebis}{\avariable}$
by Lemma~\ref{lemma-inverse-alloc}(I). Moreover, for all heaps $\aheap''$ such that $\aheap' \sqsubseteq \aheap''$, we
have  $\pair{\astore}{\aheap''} \not \models \reach(\avariable,\avariablebis) = 2$ as
$\pair{\astore}{\aheap''} \models \avariable \hpto \avariablebis$. 
\end{description}

So, in both cases,  $\pair{\astore}{\aheap}  \models  
[\alloc{\avariable} \land \allocbacktwo{\avariablebis}{\avariable} \land (\top \magicwand \lnot \reach(\avariable,\avariablebis) = 2)]_2$
and therefore $\pair{\astore}{\aheap}  \models \avariable \hpto_\avariablebis^2 \avariable$. 

Conversely, assume that $\pair{\astore}{\aheap}  \models \avariable \hpto_\avariablebis^2 \avariable$ and let us show that $\astore(\avariable)$
can reach itself in two steps but not in one step.

\begin{description}
\item[Case 1: $\pair{\astore}{\aheap}  \models \avariable \hpto \avariablebis \wedge \avariablebis \hpto \avariable$.]
Let $\alocation = \astore(\avariablebis)$. So, $\aheap(\astore(\avariable)) = \alocation$, $\aheap(\alocation) = \astore(\avariable)$,
and $\alocation \neq \astore(\avariable)$. Hence, we are done. 

\item[Case 2: $\pair{\astore}{\aheap}  \models \neg \avariable \hpto \avariablebis$.]
First, note that the case $ \avariable \hpto \avariablebis \wedge \neg (\avariablebis \hpto \avariable)$ 
is rule out because the memory state satisfies 
${\avariable \hpto \avariablebis} \Rightarrow \avariablebis \hpto \avariable$. 
As $\pair{\astore}{\aheap}  \models  
[\alloc{\avariable} \land \allocbacktwo{\avariablebis}{\avariable} \land (\top \magicwand \lnot \reach(\avariable,\avariablebis) = 2)]_2$,
we conclude that $\pair{\astore}{\aheap} \models \alloc{\avariable}$ and therefore $\astore(\avariable) \in \domain{\aheap}$. 
Similarly, we can conclude that $\pair{\astore}{\aheap} \models \allocbacktwo{\avariablebis}{\avariable}$. 
Let $\alocation = \aheap(\astore(\avariable))$ and by assumption $\alocation$ is distinct from 
$\astore(\avariablebis)$ and $\alocation$ is different from $\astore(\avariable)$, whence $\aheap(\astore(\avariable)) \neq \astore(\avariable)$. 
{\em Ad absurdum}, suppose that either $\alocation \not \in \domain{\aheap}$ or $\aheap(\alocation) \neq
\astore(\avariable)$. So, there is a location $\alocation'$ distinct from $\alocation$ such that 
$\aheap(\alocation') = \astore(\avariable)$ (indeed, we have $\pair{\astore}{\aheap} \models \allocbacktwo{\avariablebis}{\avariable}$
and then use Lemma~\ref{lemma-inverse-alloc}(I)).
As $\pair{\astore}{\aheap}  \models  
[\alloc{\avariable} \land \allocbacktwo{\avariablebis}{\avariable} \land (\top \magicwand \lnot \reach(\avariable,\avariablebis) = 2)]_2$,
the only heap $\aheap'$ with two memory cells satisfying 
$\pair{\astore}{\aheap'}  \models  
\alloc{\avariable} \land \allocbacktwo{\avariablebis}{\avariable}$ is the one with
$\domain{\aheap'} = \set{\astore(\avariable), \alocation'}$ and therefore $\alocation \not \in 
\domain{\aheap'}$, which is important at this point.
Let $\aheap''$ be any heap disjoint from $\aheap'$ such that $\aheap''(\alocation) = \astore(\avariablebis)$. 
We know such a heap exists as $\alocation  \not \in \domain{\aheap'}$. 
As $\astore(\avariable)$, $\alocation$ and $\astore(\avariablebis)$ are pairwise distinct, we get that 
$\pair{\astore}{\aheap''} \models \reach(\avariable,\avariablebis) = 2$ and therefore,
$\pair{\astore}{\aheap'} \not \models (\top \magicwand \lnot \reach(\avariable,\avariablebis) = 2)$, which leads to a contradiction.
Consequently, $\alocation \in \domain{\aheap}$ and $\aheap(\alocation) = \astore(\avariable)$, so $\aheap^2(\astore(\avariable)) = \astore(\avariable)$.
\qedhere
\end{description}
\end{proof}

The predicate $n(\avariable) = n(\avariablebis)$ can be defined in $\seplogic{\separate, \magicwand, \ls}$ as
the formula $\mathsf{ext}(n(\avariable) = n(\avariablebis))$ ('ext' stands for 'extension')
\begin{gather*}
(\avariable \neq \avariablebis \Rightarrow [\alloc{\avariable} \land \alloc{\avariablebis} \land ((\avariable \hpto \avariablebis \land \avariablebis \hpto \avariablebis)
\lor
 (\avariablebis \hpto \avariable \land \avariable \hpto \avariable) \lor\\
((\bigwedge_{\mathclap{\avariableter,\avariableter' \in \{\avariable,\avariablebis\}}} \lnot (\avariableter \hpto \avariableter')) \land (\true \sepimp \lnot(\reach(\avariable,\avariablebis)=2 \land \reach(\avariablebis,\avariable)=2))))]_2) \land \alloc{\avariable}
\end{gather*}
\begin{lemma}
\label{lemma-next-equality}
Let $\avariable,\avariablebis \in \PVAR$. For all memory states $\pair{\astore}{\aheap}$,
we have $\pair{\astore}{\aheap} \models \mathsf{ext}(n(\avariable) = n(\avariablebis))$ iff $\aheap(\astore(\avariable)) = \aheap(\astore(\avariablebis))$.
\end{lemma}

\begin{proof}
\cut{
We recall the definition of $n(\avariable) = n(\avariablebis)$ in $\seplogic{\separate, \magicwand, \ls}$:
\begin{gather*}
(\avariable \neq \avariablebis \Rightarrow [\alloc{\avariable} \land \alloc{\avariablebis} \land ((\avariable \hpto \avariablebis \land \avariablebis \hpto \avariablebis)
\lor
 (\avariablebis \hpto \avariable \land \avariable \hpto \avariable) \lor\\
((\bigwedge_{\mathclap{\avariableter,\avariableter' \in \{\avariable,\avariablebis\}}} \lnot \avariableter \hpto \avariableter') \land (\true \sepimp \lnot(\reach(\avariable,\avariablebis)=2 \land \reach(\avariablebis,\avariable)=2))))]_2) \land \alloc{\avariable}
\end{gather*}
}
First, suppose $\aheap(\astore(\avariable)) = \aheap(\astore(\avariablebis))$. Then obviously $(\astore,\aheap) \models \alloc{\avariable} \land \alloc{\avariablebis}$. Suppose $\astore(\avariable) \neq \astore(\avariablebis)$.
We need to show that $\pair{\astore}{\aheap}$ satisfies the formula
\begin{gather*}
[\alloc{\avariable} \land \alloc{\avariablebis} \land ((\avariable \hpto \avariablebis \land \avariablebis \hpto \avariablebis)
\lor
 (\avariablebis \hpto \avariable \land \avariable \hpto \avariable) \lor\\
((\bigwedge_{\mathclap{\avariableter,\avariableter' \in \{\avariable,\avariablebis\}}} \lnot \avariableter \hpto \avariableter') \land (\true \sepimp \lnot(\reach(\avariable,\avariablebis)=2 \land \reach(\avariablebis,\avariable)=2))))]_2
\end{gather*}
Let $\aheap' \sqsubseteq \aheap$ be the two-memory-cells heap such that $\aheap'(\astore(\avariable)) = \aheap'(\astore(\avariablebis))$. In particular,
$\domain{\aheap'} = \set{\astore(\avariable),\astore(\avariablebis)}$
and therefore $\pair{\astore}{\aheap'} \models \alloc{\avariable} \land \alloc{\avariablebis}$.
Moreover, $\aheap'$ is represented by one of the following memory states.

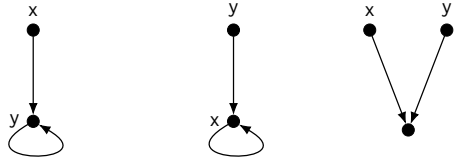
\begin{center}
  \hfill
  \begin{tikzpicture}[baseline]
    \node[dot,label=above:$\avariable$] (l1) at (0,0) {};
    \node[dot,label=left:$\avariablebis$] (l2) [below of=l1] {};

    \draw[pto] (l1) -- (l2);
    \draw[pto] (l2) edge [loop below, in=-30, out=-150,looseness=20] node {} (l2);
  \end{tikzpicture}
  \hfill
  \begin{tikzpicture}[baseline]
    \node[dot,label=above:$\avariablebis$] (l1) at (0,0) {};
    \node[dot,label=left:$\avariable$] (l2) [below of=l1] {};

    \draw[pto] (l1) -- (l2);
    \draw[pto] (l2) edge [loop below, in=-30, out=-150,looseness=20] node {} (l2);
  \end{tikzpicture}
  \hfill
  \begin{tikzpicture}[baseline]
    \node[dot,label=above:$\avariable$] (l1) at (0,0) {};
    \node[dot] (l2) [below right = 1.5cm and 0.5cm of l1] {};
    \node[dot,label=above:$\avariablebis$] (l3) [above right = 1.5cm and 0.5cm of l2] {};

    \draw[pto] (l1) -- (l2);
    \draw[pto] (l3) -- (l2);
  \end{tikzpicture}
  \hfill\,
\end{center}
Each memory state above satisfies one of the three disjuncts from the third conjunct of the above formula.
The first memory state satisfies
$\avariable \hpto \avariablebis \land \avariablebis \hpto \avariablebis$, the second one satisfies
${\avariablebis \hpto \avariable \land \avariable \hpto \avariable}$ and the third one satisfies
$$(\bigwedge_{\mathclap{\avariableter,\avariableter' \in \{\avariable,\avariablebis\}}} \lnot \avariableter \hpto \avariableter') \land (\true \sepimp \lnot(\reach(\avariable,\avariablebis)=2 \land \reach(\avariablebis,\avariable)=2)).$$

The first two cases are trivial. The last one represents a memory state with a location $\alocation$ such that
 $\alocation = \aheap'(\astore(\avariable)) = \aheap'(\astore(\avariablebis))$ and $\astore(\avariable) \neq \alocation \neq \astore(\avariablebis)$.
As such, this memory state trivially satisfies $\bigwedge_{\avariableter,\avariableter' \in \{\avariable,\avariablebis\}} \lnot \avariableter \hpto \avariableter'$.
Now, consider $\aheap''$ disjoint from $\aheap'$ and  $\pair{\astore}{\aheap' + \aheap''} \models \reach(\avariable,\avariablebis)=2$.
In particular, it must hold that $\aheap''(\alocation) = \astore(\avariablebis)$. As such, $(\aheap'+\aheap'')^2(\astore(\avariablebis)) = \astore(\avariablebis)$ and so $(\astore, \aheap' + \aheap'') \not\models \reach(\avariablebis,\avariable)=2$.
It follows that the last memory state of the picture satisfies
${(\bigwedge_{\avariableter,\avariableter' \in \{\avariable,\avariablebis\}} \lnot \avariableter \hpto \avariableter')} \land (\true \sepimp \lnot(\reach(\avariable,\avariablebis)=2 \land \reach(\avariablebis,\avariable)=2))$.
Thus, $\aheap(\astore(\avariable)) = \aheap(\astore(\avariablebis))$ implies $\pair{\astore}{\aheap} \models \mathsf{ext}(n(\avariable) = n(\avariablebis))$.

Conversely, suppose $\pair{\astore}{\aheap} \models 
                    \mathsf{ext}(n(\avariable) = n(\avariablebis))$. 
Then $\pair{\astore}{\aheap} \models \alloc{\avariable}$. If $\astore(\avariable) = \astore(\avariablebis)$, $\aheap(\astore(\avariable)) = \aheap(\astore(\avariablebis))$ follows trivially.
Instead, if $\astore(\avariable) \neq \astore(\avariablebis)$, there must exist a two-memory-cells heap $\aheap' \sqsubseteq \aheap$ such that $\pair{\astore}{\aheap'}$ satisfies
$\alloc{\avariable} \land \alloc{\avariablebis} \land ((\avariable \hpto \avariablebis \land \avariablebis \hpto \avariablebis)
\lor
 (\avariablebis \hpto \avariable \land \avariable \hpto \avariable) \lor
((\bigwedge_{\avariableter,\avariableter' \in \{\avariable,\avariablebis\}} \lnot \avariableter \hpto \avariableter') \land (\true \sepimp \lnot(\reach(\avariable,\avariablebis)=2 \land \reach(\avariablebis,\avariable)=2))))$. As such, $\astore(\avariable)$ and $\astore(\avariablebis)$ are both in $\domain{\aheap'}$.
Trivially, if
$\pair{\astore}{\aheap'} \models (\avariable \hpto \avariablebis \land \avariablebis \hpto \avariablebis)
\lor
 (\avariablebis \hpto \avariable \land \avariable \hpto \avariable)$ then $\aheap'(\astore(\avariable)) = \aheap'(\astore(\avariablebis))$.
 The same holds true if
 $\pair{\astore}{\aheap'} \models (\bigwedge_{\avariableter,\avariableter' \in \{\avariable,\avariablebis\}} \lnot \avariableter \hpto \avariableter') \land (\true \sepimp \lnot(\reach(\avariable,\avariablebis)=2 \land \reach(\avariablebis,\avariable)=2))$.
 Indeed, $\pair{\astore}{\aheap'} \models \bigwedge_{\avariableter,\avariableter' \in \{\avariable,\avariablebis\}} \lnot \avariableter \hpto \avariableter'$
 leaves open only two possible memory states, that are represented below.
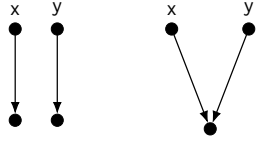
\begin{center}
  \hfill
  \begin{tikzpicture}[baseline]
    \node[dot,label=above:$\avariable$] (l1) at (0,0) {};
    \node[dot] (l2) [below of=l1] {};
    \node[dot,label=above:$\avariablebis$] (l3) [right=0.5cm of l1] {};
    \node[dot] (l4) [below of=l3] {};

    \draw[pto] (l1) -- (l2);
    \draw[pto] (l3) -- (l4);
  \end{tikzpicture}
  \hfill
  \begin{tikzpicture}[baseline]
    \node[dot,label=above:$\avariable$] (l1) at (0,0) {};
    \node[dot] (l2) [below right = 1.5cm and 0.5cm of l1] {};
    \node[dot,label=above:$\avariablebis$] (l3) [above right = 1.5cm and 0.5cm of l2] {};

    \draw[pto] (l1) -- (l2);
    \draw[pto] (l3) -- (l2);
  \end{tikzpicture}
  \hfill\,
\end{center}
The last part of the formula, $\true \sepimp \lnot(\reach(\avariable,\avariablebis)=2 \land \reach(\avariablebis,\avariable)=2)$
allows to  differentiate  these two cases by excluding the left heap.
Indeed, that very formula holds  on $\pair{\astore}{\aheap'}$ if and only if there is no heap $\aheap''$ such that  $\aheap'' \bot \aheap'$ and
$\pair{\astore}{\aheap' + \aheap''} \models \reach(\avariable,\avariablebis)=2 \land \reach(\avariablebis,\avariable)=2$.
This property holds on the right heap, as already discussed in the first part of the proof, but not on the left one, as can be shown by defining $\aheap''$ as $\{\aheap'(\astore(\avariable)) \pto \astore(\avariablebis), \aheap'(\astore(\avariablebis)) \pto \astore(\avariable)\}$.
\end{proof}

Similarly to $\allocback{\avariable}$, we can show that $n(\avariable) \hpto n(\avariablebis)$ is definable in $\seplogic{\separate, \magicwand, \ls}$ by using one additional variable $\avariableter$ whose value is different from both $\avariable$ and $\avariablebis$. Let $\aformula_{\hpto}(\avariable,\avariablebis, \avariableter)$ be $(\mathsf{ext}(n(\avariable) = n(\avariablebis)) \land \aformula^{=}_{\hpto}(\avariable,\avariablebis,\avariableter)) \lor (
\neg \mathsf{ext}(n(\avariable) = n(\avariablebis)) \land \aformula^{\neq}_{\hpto}(\avariable,\avariablebis))$ where $\aformula^{=}_{\hpto}(\avariable,\avariablebis,\avariableter)$ is defined as
\begin{align*}
  \aformula^{=}_{\hpto}(\avariable,\avariablebis,\avariableter)\ \egdef &\ (\avariable \hpto \avariable \land \avariablebis \hpto \avariable) \lor
	(\avariablebis \hpto \avariablebis \land \avariable \hpto \avariablebis) \lor
	(\avariable \hpto \avariableter \land \avariableter \hpto \avariableter)\\
  & \lor [\alloc{\avariable}\land \lnot\allocbacktwo{\avariableter}{\avariable} \land (\top \magicwand \lnot \reach(\avariable,\avariableter)\leq 3)]_2
\end{align*}
whereas $\aformula^{\neq}_{\hpto}(\avariable,\avariablebis)$ is defined as
\begin{align*}
  \aformula^{\neq}_{\hpto}(\avariable,\avariablebis)\ \egdef &\ (\avariable \hpto \avariablebis \land \alloc{\avariablebis}) \lor
(\avariablebis \hpto \avariablebis \land \reach(\avariable,\avariablebis) = 2) \lor
(\avariablebis \hpto \avariable \land \avariable \hpto^2_\avariablebis \avariable) \,\lor\\
& [\alloc{\avariable} \land \alloc{\avariablebis} \land \textstyle (\bigwedge_{\avariableter,\avariableter' \in \{\avariable,\avariablebis\}} \lnot \avariableter
\hpto \avariableter') \land \lnot\reach(\avariable,\avariablebis) \leq 3\\ 
& \land ((\size=1 \land \allocbacktwo{\avariable}{\avariablebis}) \lollipop (\reach(\avariable,\avariablebis)=3 \land \avariablebis \hpto^2_\avariable \avariablebis))
 ]_3
\end{align*}

\begin{lemma}
 \label{lemma-three-in-one}
 Let $\avariable, \avariablebis, \avariableter \in \PVAR$.
\begin{enumerate}[label=\normalfont{\textbf{(\Roman*)}}]
\item For all memory states $(\astore,\aheap)$ such that ${\astore(\avariable) \neq \astore(\avariableter)}$ and $\astore(\avariablebis) \neq \astore(\avariableter)$,
we have $(\astore,\aheap) \models \aformula_{\hpto}(\avariable,\avariablebis, \avariableter)$ iff $\{\astore(\avariable),\astore(\avariablebis)\} \subseteq \domain{\aheap}$
and $\aheap(\aheap(\astore(\avariable))) = \aheap(\astore(\avariablebis))$.
\item In the translation $\atranslation(\aformula,\aset)$, $n(\avariable) \hpto  n(\avariablebis)$ can be replaced by $\aformula_{\hpto}(\avariable,\avariablebis,\overline{\avariable})$.
\end{enumerate}
\end{lemma}
\begin{proof} (I)
\cut{
We recall the definition of $\aformula_{\hpto}(\avariable,\avariablebis, \avariableter)$: $(n(\avariable) = n(\avariablebis) \land \aformula^{=}_{\hpto}(\avariable,\avariablebis,\avariableter)) \lor (\lnot n(\avariable) = n(\avariablebis) \land \aformula^{\neq}_{\hpto}(\avariable,\avariablebis))$ where $\aformula^{=}_{\hpto}(\avariable,\avariablebis,\avariableter)$ is defined as
\begin{gather*}
(\avariable \hpto \avariable \land \avariablebis \hpto \avariable) \lor
	(\avariablebis \hpto \avariablebis \land \avariable \hpto \avariablebis) \lor
	(\avariable \hpto \avariableter \land \avariableter \hpto \avariableter)\\
\lor\ [\alloc{\avariable}\land \lnot\allocbacktwo{\avariableter}{\avariable} \land (\top \magicwand \lnot \reach(\avariable,\avariableter)\leq 3)]_2
\end{gather*}
whereas $\aformula^{\neq}_{\hpto}(\avariable,\avariablebis)$ is defined as
\begin{gather*}
(\avariable \hpto \avariablebis \land \alloc{\avariablebis}) \lor
(\avariablebis \hpto \avariablebis \land \reach(\avariable,\avariablebis) = 2) \lor
(\avariablebis \hpto \avariable \land \avariable \hpto^2_\avariablebis \avariable) \lor\\
[\alloc{\avariable} \land \alloc{\avariablebis} \land \textstyle (\bigwedge_{\avariableter,\avariableter' \in \{\avariable,\avariablebis\}} \lnot \avariableter
\hpto \avariableter') \land \lnot\reach(\avariable,\avariablebis) \leq 3\\ \land ((\size=1 \land \allocbacktwo{\avariable}{\avariablebis}) \lollipop (\reach(\avariable,\avariablebis)=3 \land \avariablebis \hpto^2_\avariable \avariablebis))
 ]_3
\end{gather*}
}
First, suppose $\{\astore(\avariable),\astore(\avariablebis)\} \subseteq \domain{\aheap}$, $\aheap(\aheap(\astore(\avariable))) = \aheap(\astore(\avariablebis))$, $\astore(\avariable) \neq \astore(\avariableter)$ and $\astore(\avariablebis) \neq \astore(\avariableter)$.
Let us distinguish two cases.

\begin{description}
\itemsep=0pt
\item[Case 1:]
$\aheap(\astore(\avariable)) = \aheap(\astore(\avariablebis))$. Equivalently, $\mathsf{ext}(n(\avariable) = n(\avariablebis))$ 
holds from Lemma~\ref{lemma-next-equality} and we must prove that $\pair{\astore}{\aheap}$ satisfies $\aformula^{=}_{\hpto}(\avariable,\avariablebis,\avariableter)$. 
Consider the subheap $\aheap' \sqsubseteq \aheap$ such that $\domain{\aheap'} = \{\astore(\avariable), \astore(\avariablebis), \aheap(\astore(\avariable)) \}$.
There are only a bounded number of subheaps of this kind (up to isomorphism with respect to $\{\avariable,\avariablebis,\avariableter\}$).
Remember that $\aheap(\astore(\avariable)) = \aheap(\astore(\avariablebis))$, $\aheap(\aheap(\astore(\avariable))) = \aheap(\astore(\avariablebis))$, $\astore(\avariable) \neq \astore(\avariableter)$
and $\astore(\avariablebis) \neq \astore(\avariableter)$. 
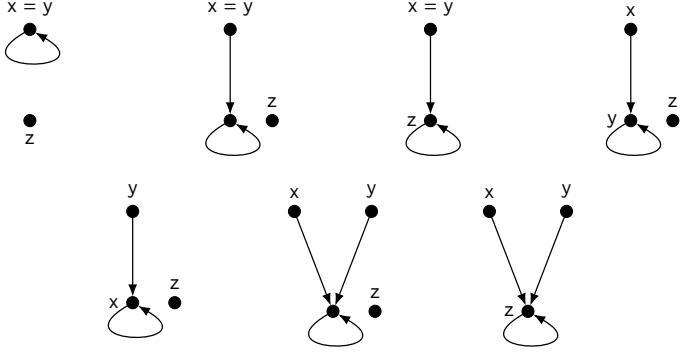
 \begin{center}
    \begin{tikzpicture}[baseline]
      \node[dot,label=above:${\avariable=\avariablebis}$] (l1) at (0,0) {};
      \node[dot,label=below:$\avariableter$] (l3) [below of=l1] {};
      \draw[pto] (l1) edge [loop below, in=-30, out=-150,looseness=20] (l1);
    \end{tikzpicture}
    \hfill
		\begin{tikzpicture}[baseline]
      \node[dot,label=above:${\avariable=\avariablebis}$] (l1) at (0,0) {};
      \node[dot] (l2) [below of=l1] {};
      \node[dot,label=above:$\avariableter$] (l3) [right=0.5cm of l2] {};

      \draw[pto] (l1) -- (l2);
      \draw[pto] (l2) edge [loop below, in=-30, out=-150,looseness=20] (l2);
    \end{tikzpicture}
    \hfill
    \begin{tikzpicture}[baseline]
      \node[dot,label=above:${\avariable=\avariablebis}$] (l1) at (0,0) {};
      \node[dot,label=left:$\avariableter$] (l2) [below of=l1] {};

      \draw[pto] (l1) -- (l2);
      \draw[pto] (l2) edge [loop below, in=-30, out=-150,looseness=20] (l2);
    \end{tikzpicture}
    \hfill
    \begin{tikzpicture}[baseline]
      \node[dot,label=above:$\avariable$] (l1) at (0,0) {};
      \node[dot,label=left:$\avariablebis$] (l2) [below of=l1] {};
      \node[dot,label=above:$\avariableter$] (l3) [right=0.5cm of l2] {};

      \draw[pto] (l1) -- (l2);

      \draw[pto] (l2) edge [loop below, in=-30, out=-150,looseness=20] (l2);
    \end{tikzpicture}
\end{center}
\vspace{0.6cm}
\begin{center}
		\hfill
    \begin{tikzpicture}[baseline]
      \node[dot,label=above:$\avariablebis$] (l1) at (0,0) {};
      \node[dot,label=left:$\avariable$] (l2) [below of=l1] {};
      \node[dot,label=above:$\avariableter$] (l3) [right=0.5cm of l2] {};

      \draw[pto] (l1) -- (l2);
      \draw[pto] (l2) edge [loop below, in=-30, out=-150,looseness=20] (l2);
    \end{tikzpicture}
    \hfill
    \begin{tikzpicture}[baseline]
      \node[dot,label=above:$\avariable$] (l1) at (0,0) {};
      \node[dot] (l2) [below right = 1.5cm and 0.5cm of l1] {};
      \node[dot,label=above:$\avariablebis$] (l3) [above right = 1.5cm and 0.5cm of l2] {};
      \node[dot,label=above:$\avariableter$] (l4) [right=0.5cm of l2] {};

      \draw[pto] (l1) -- (l2);
      \draw[pto] (l3) -- (l2);
      \draw[pto] (l2) edge [loop below, in=-30, out=-150,looseness=20] (l2);
    \end{tikzpicture}
		\hfill
    \begin{tikzpicture}[baseline]
      \node[dot,label=above:$\avariable$] (l1) at (0,0) {};
      \node[dot,label=left:$\avariableter$] (l2) [below right = 1.5cm and 0.5cm of l1] {};
      \node[dot,label=above:$\avariablebis$] (l3) [above right = 1.5cm and 0.5cm of l2] {};

      \draw[pto] (l1) -- (l2);
      \draw[pto] (l3) -- (l2);
      \draw[pto] (l2) edge [loop below, in=-30, out=-150,looseness=20] (l2);
    \end{tikzpicture}
		\hfill\,
\end{center}
We need to check that for each case, the memory state satisfies one of the disjuncts of the formula $\aformula^{=}_{\hpto}(\avariable,\avariablebis,\avariableter)$.
Indeed, it is easy to see that the first and fifth memory states (enumerated from the top left to the bottom right) satisfy $\avariable \hpto \avariable \land \avariablebis \hpto \avariable$, the third and seventh memory states satisfy $\avariable \hpto \avariableter \land \avariableter \hpto \avariableter$ and the fourth memory state satisfies $\avariablebis \hpto \avariablebis \land \avariable \hpto \avariablebis$.
Lastly, the second and sixth memory states satisfy
$ [\alloc{\avariable}\land \lnot\allocbacktwo{\avariableter}{\avariable} \land (\top \magicwand \lnot {\reach(\avariable,\avariableter)\leq 3)]_{2}}$.
For these last two cases, consider the two-memory-cells heap $\aheap'' = \{\astore(\avariable) \pto \aheap'(\astore(\avariable)), \aheap'(\astore(\avariable)) \pto \aheap'(\astore(\avariable))\} \sqsubseteq \aheap'$. Trivially, $\pair{\astore}{\aheap''} \models \alloc{\avariable}\land \lnot\allocbacktwo{\avariableter}{\avariable}$.
Moreover, since $\aheap''(\aheap'(\astore(\avariable))) = \aheap'(\astore(\avariable))$, for any heap $\aheap'''$ such that $\aheap''' \bot \aheap''$, $\aheap''(\aheap'(\astore(\avariable)))$ cannot be in $\domain{\aheap'''}$ and therefore $\pair{\astore}{\aheap'' + \aheap'''}$ does not satisfy $\reach(\avariable,\avariableter)\leq 3$.
It follows that $\pair{\astore}{\aheap''}$ satisfies $(\top \magicwand \lnot \reach(\avariable,\avariableter)\leq 3)$.
%
%
%

\item[Case 2:]
$\aheap(\astore(\avariable)) \neq \aheap(\astore(\avariablebis))$. Equivalently, from Lemma~\ref{lemma-next-equality}, 
$\mathsf{ext}(n(\avariable) = n(\avariablebis))$ does not hold and we must prove that $\pair{\astore}{\aheap} \models \aformula^{\neq}_{\hpto}(\avariable,\avariablebis)$.

Moreover, $\astore(\avariable) \neq \astore(\avariablebis)$.
Consider the subheap $\aheap' \sqsubseteq \aheap$ such that $\domain{\aheap'} = \{\astore(\avariable), \astore(\avariablebis), \aheap(\astore(\avariable)) \}$.
There are only a bounded number of subheaps of this kind (up to isomorphism with respect to $\{\avariable,\avariablebis\}$).
Remember that $\aheap(\aheap(\astore(\avariable))) = \aheap(\astore(\avariablebis))$ and $\aheap(\astore(\avariable)) \neq \aheap(\astore(\avariablebis))$.
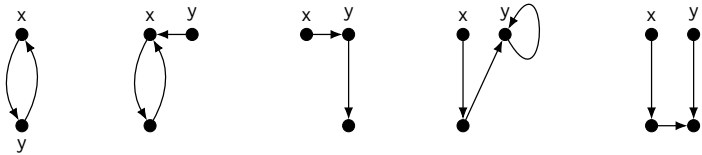
\begin{center}
\begin{tikzpicture}[baseline]
  \node[dot,label=above:$\avariable$] (l1) at (0,0) {};
  \node[dot,label=below:$\avariablebis$] (l2) [below of=l1] {};

  \draw[pto] (l1) edge[bend right] (l2);
  \draw[pto] (l2) edge[bend right] (l1);
\end{tikzpicture}
\hfill
\begin{tikzpicture}[baseline]
  \node[dot,label=above:$\avariable$] (l1) at (0,0) {};
  \node[dot] (l2) [below of=l1] {};
\node[dot,label=above:$\avariablebis$] (l3) [right=0.5cm of l1] {};

  \draw[pto] (l1) edge[bend right] (l2);
  \draw[pto] (l2) edge[bend right] (l1);
  \draw[pto] (l3) -- (l1);
\end{tikzpicture}
\hfill
\begin{tikzpicture}[baseline]
  \node[dot,label=above:$\avariable$] (l1) at (0,0) {};
  \node[dot,label=above:$\avariablebis$] (l3) [right=0.5cm of l1] {};
  \node[dot] (l4) [below of=l3] {};

  \draw[pto] (l1) -- (l3);
  \draw[pto] (l3) -- (l4);
\end{tikzpicture}
\hfill
\begin{tikzpicture}[baseline]
  \node[dot,label=above:$\avariable$] (l1) at (0,0) {};
  \node[dot] (l2) [below of=l1] {};
  \node[dot,label=above:$\avariablebis$] (l3) [right=0.5cm of l1] {};

  \draw[pto] (l1) -- (l2);
  \draw[pto] (l3) edge [loop right, in=60, out=-60,looseness=20]  (l3);
  \draw[pto] (l2) -- (l3);
\end{tikzpicture}
\hfill
\begin{tikzpicture}[baseline]
  \node[dot,label=above:$\avariable$] (l1) at (0,0) {};
  \node[dot] (l2) [below of=l1] {};
  \node[dot,label=above:$\avariablebis$] (l3) [right=0.5cm of l1] {};
  \node[dot] (l4) [below of=l3] {};

  \draw[pto] (l1) -- (l2);
  \draw[pto] (l3) -- (l4);
  \draw[pto] (l2) -- (l4);
\end{tikzpicture}
\end{center}

Let us check that for each of the five cases, one of the disjuncts of $\aformula^{\neq}_{\hpto}(\avariable,\avariablebis)$ holds. Trivially, the first and third memory states satisfy the first disjunct ${\avariable \hpto \avariablebis} \land \alloc{\avariablebis}$, the second one satisfies $\avariablebis \hpto \avariable \land \avariable \hpto^2_\avariablebis \avariable$, whereas the fourth one satisfies $\avariablebis \hpto \avariablebis \land \reach(\avariable,\avariablebis) = 2$. Lastly, the fifth one satisfies the disjunct
\begin{gather*}
[\alloc{\avariable} \land \alloc{\avariablebis} \land \textstyle (\bigwedge_{\avariableter,\avariableter' \in \{\avariable,\avariablebis\}} \lnot \avariableter
\hpto \avariableter') \land \lnot\reach(\avariable,\avariablebis) \leq 3\\ \land ((\size=1 \land \allocbacktwo{\avariable}{\avariablebis}) \lollipop (\reach(\avariable,\avariablebis)=3 \land \avariablebis \hpto^2_\avariable \avariablebis))
 ]_3
\end{gather*}
For the fifth memory state, trivially $\pair{\astore}{\aheap'} \models \alloc{\avariable} \land \alloc{\avariablebis} \land \textstyle {(\bigwedge_{\avariableter,\avariableter' \in \{\avariable,\avariablebis\}} \lnot \avariableter
\hpto \avariableter')} \land \lnot\reach(\avariable,\avariablebis) \leq 3$.
Let $\aheap'' = \{\aheap'(\aheap'(\astore(\avariable))) \pto \astore(\avariablebis)\}$. As $\aheap'(\aheap'(\astore(\avariable))) \not\in \domain{\aheap'}$, it holds that $\aheap'' \bot \aheap'$.
Moreover, $\pair{\astore}{\aheap''}$ satisfies $\size=1 \land \allocbacktwo{\avariable}{\avariablebis}$.
Consider now $\pair{\astore}{\aheap' + \aheap''}$. Since $\aheap''(\aheap'(\aheap'(\astore(\avariable)))) = \astore(\avariablebis)$, this memory state satisfies $\reach(\avariable,\avariablebis)=3$, whereas $\avariablebis \hpto^2_\avariable \avariablebis$ is satisfied from the hypothesis
$\astore(\avariable) \neq \astore(\avariablebis)$ and $\aheap(\aheap(\astore(\avariable))) = \aheap(\astore(\avariablebis))$.
Indeed, $\aheap''(\aheap'(\astore(\avariablebis))) = \aheap''(\aheap'(\aheap'(\astore(\avariable)))) = \astore(\avariablebis)$.
We derive $\pair{\astore}{\aheap'} \models ((\size=1 \land \allocbacktwo{\avariable}{\avariablebis}) \lollipop ({\reach(\avariable,\avariablebis)=3} \land {\avariablebis \hpto^2_\avariable \avariablebis}))$.
\end{description}

Conversely, let us suppose that  $\pair{\astore}{\aheap} \models \aformula_{\hpto}(\avariable,\avariablebis, \avariableter)$,
$\astore(\avariable) \neq \astore(\avariableter)$ and  $\astore(\avariablebis) \neq \astore(\avariableter)$. Two cases
are considered.
\begin{description}
\itemsep 0 cm
\item[Case 1:]
$\pair{\astore}{\aheap} \models \mathsf{ext}(n(\avariable) = n(\avariablebis)) \land \aformula^{=}_{\hpto}(\avariable,\avariablebis,\avariableter)$. If one of the three first
disjuncts of $\aformula^{=}_{\hpto}(\avariable,\avariablebis,\avariableter)$ holds, we have immediately $\{\astore(\avariable),\astore(\avariablebis)\} \subseteq \domain{\aheap}$ and
$\aheap(\aheap(\astore(\avariable)))$ $= \aheap(\astore(\avariablebis))$. Indeed, let us consider the three possible situations.

\begin{enumerate}

\item  In the case $\pair{\astore}{\aheap} \models \mathsf{ext}(n(\avariable) = n(\avariablebis)) \wedge \avariable \hpto \avariable \wedge \avariablebis \hpto \avariable$,
       we get $\aheap(\aheap(\astore(\avariable))) = \aheap(\astore(\avariablebis)) = \astore(\avariable)$ and therefore $\pair{\astore}{\aheap} \models n(\avariable) \hpto n(\avariablebis)$. 

\item In the case $\pair{\astore}{\aheap} \models \mathsf{ext}(n(\avariable) = n(\avariablebis)) \wedge \avariablebis \hpto \avariablebis \wedge \avariable \hpto \avariablebis$,
       we get $\aheap(\aheap(\astore(\avariable))) = \aheap(\astore(\avariablebis)) = \astore(\avariablebis)$  and therefore $\pair{\astore}{\aheap} \models n(\avariable) \hpto n(\avariablebis)$. 

\item In the case $\pair{\astore}{\aheap} \models \mathsf{ext}(n(\avariable) = n(\avariablebis)) \wedge \avariable \hpto \avariableter \wedge \avariableter \hpto \avariableter$,
       we get $\aheap(\aheap(\astore(\avariable))) = \aheap(\astore(\avariablebis)) = \astore(\avariableter)$  and therefore $\pair{\astore}{\aheap} \models n(\avariable) \hpto n(\avariablebis)$. 

\end{enumerate}

The remaining case is when
none of the three first disjuncts holds and therefore $\pair{\astore}{\aheap} \models [\alloc{\avariable}\land \lnot\allocbacktwo{\avariableter}{\avariable} \land (\top \magicwand \lnot \reach(\avariable,\avariableter)\leq 3)]_2$.
Let $\aheap' \sqsubseteq \aheap$ be some heap with $\card{\domain{\aheap'}} = 2$  and
$\pair{\astore}{\aheap'} \models \alloc{\avariable} \land \lnot\allocbacktwo{\avariableter}{\avariable} \land
(\top \magicwand \lnot \reach(\avariable,\avariableter)\leq 3)$. Obviously, $\astore(\avariable) \in \domain{\aheap'}$, say $\aheap'(\astore(\avariable))
= \alocation$ and $\alocation \in \domain{\aheap'}$ (otherwise $\pair{\astore}{\aheap' + \set{\alocation \mapsto \astore(\avariableter)}}
\models \reach(\avariable,\avariableter) =  2$, which leads to a contradiction). We can assume that $\alocation$ is distinct from $\astore(\avariable)$
as the first disjunct of $\aformula^{=}_{\hpto}(\avariable,\avariablebis,\avariableter)$ does not hold.
Similarly, $\pair{\astore}{\aheap'} \models \neg (\avariable \hpto \avariableter) \wedge \neg ({\reach(\avariable,\avariableter) =  2})$,
otherwise we reach a contradiction with $\pair{\astore}{\aheap'} \models \top \magicwand \lnot \reach(\avariable,\avariableter)\leq 3$.
To sump up, the heap $\aheap'$ satisfies the following properties:
\begin{itemize}
\item $\card{\domain{\aheap'}} = 2$, $\astore(\avariable) \in \domain{\aheap'}$ and $\aheap'(\astore(\avariable)) = \alocation$ with $\alocation \neq \astore(\avariable)$ and $\alocation \in \domain{\aheap'}$. 
\item $\pair{\astore}{\aheap'} \models \lnot\allocbacktwo{\avariableter}{\avariable} \land
(\top \magicwand \lnot \reach(\avariable,\avariableter)\leq 3) \land \neg (\avariable \hpto \avariableter) \wedge \neg (\reach(\avariable,\avariableter) =  2)$.
\end{itemize}
If $\aheap(\aheap(\astore(\avariable))) \neq\aheap(\astore(\avariable))$
or $\aheap(\aheap(\astore(\avariable)))  = \astore(\avariable)$, then we would have $\pair{\astore}{\aheap'} \models 
\allocbacktwo{\avariableter}{\avariable} \vee (\top \septraction  (\reach(\avariable,\avariableter)\leq 3))$, which is precisely excluded from above.
Consequently, we conclude that necessarily $\aheap(\aheap(\astore(\avariable))) = \aheap(\astore(\avariable)) \neq \astore(\avariable)$.
\cut{
\begin{center}
\begin{tikzpicture}[baseline]
  \node[dot,label=above:$\avariable$] (l1) at (0,0) {};
  \node[dot,label=above:$\avariablebis$] (l2) [right=0.5cm of l1] {};
  \node[dot,label=above:$\avariableter$] (l5) [left=0.5cm of l1] {};
  \node[dot] (l3) [below of=l1] {};
  \node[dot] (l4) [below of=l3] {};

  \draw[pto] (l1) edge (l3);
  \draw[pto] (l3) edge (l4);
\end{tikzpicture}
\hfill
\begin{tikzpicture}[baseline]
  \node[dot,label=above:$\avariable$] (l1) at (0,0) {};
  \node[dot,label=above:$\avariablebis$] (l2) [right=0.5cm of l1] {};
  \node[dot] (l3) [below of=l1] {};
  \node[dot,label=above:$\avariableter$] (l5) [left=0.5cm of l1] {};

  \draw[pto] (l1) edge (l3);
  \draw[pto] (l3) edge [loop left, in=60, out=-60,looseness=20]  (l3);
\end{tikzpicture}
\hfill
\begin{tikzpicture}[baseline]
  \node[dot,label=above:$\avariable$] (l1) at (0,0) {};
  \node[dot,label=above:$\avariablebis$] (l2) [right=0.5cm of l1] {};
  \node[dot] (l3) [below of=l1] {};
  \node[dot] (l4) [below of=l2] {};
  \node[dot,label=above:$\avariableter$] (l5) [left=0.5cm of l1] {};

  \draw[pto] (l1) edge (l2);
  \draw[pto] (l3) edge (l4);
\end{tikzpicture}
\hfill
\begin{tikzpicture}[baseline]
  \node[dot,label=above:$\avariable$] (l1) at (0,0) {};
  \node[dot,label=above:$\avariablebis$] (l2) [right=0.5cm of l1] {};
  \node[dot] (l3) [below of=l1] {};
  \node[dot,label=above:$\avariableter$] (l4) [below of=l2] {};

  \draw[pto] (l1) edge (l2);
  \draw[pto] (l3) edge (l4);

\end{tikzpicture}
\hfill
\begin{tikzpicture}[baseline]
  \node[dot,label=above:$\avariable$] (l1) at (0,0) {};
  \node[dot,label=above:$\avariablebis$] (l2) [right=0.5cm of l1] {};
  \node[dot] (l4) [below of=l2] {};
  \node[dot,label=above:$\avariableter$] (l3) [below of=l1] {};

  \draw[pto] (l1) edge (l2);
  \draw[pto] (l3) edge (l4);
\end{tikzpicture}
\hfill
\begin{tikzpicture}[baseline]
  \node[dot,label=above:${\avariable=\avariablebis}$] (l1) at (0,0) {};
  \node[dot] (l2) [below of=l1] {};  
  \draw[pto] (l1) edge (l2);
  \draw[pto] (l2) edge [loop left, in=60, out=-60,looseness=20]  (l2);
\end{tikzpicture}
\end{center}
Only the second and sixth memory states satisfy 
$\top \magicwand \lnot \reach(\avariable,\avariableter)\leq 3$. For the four other memory states,
it is always possible to add memory cells so that there is a path of length less than three between $\astore(\avariable)$
and $\astore(\avariableter)$. Consequently, $\aheap(\aheap(\astore(\avariable))) = \aheap(\astore(\avariable))$.
}
Moreover, by Lemma~\ref{lemma-next-equality}, 
$\pair{\astore}{\aheap} \models \mathsf{ext}(n(\avariable) = n(\avariablebis))$ implies that $\aheap(\astore(\avariable)) = \aheap(\astore(\avariablebis))$. $\aheap(\aheap(\astore(\avariable))) = \aheap(\astore(\avariablebis))$ and $\{\astore(\avariable),\astore(\avariablebis)\} \subseteq \domain{\aheap}$ follows.

\item[Case 2:] $\pair{\astore}{\aheap} \models \lnot \mathsf{ext}(n(\avariable) = n(\avariablebis)) \land \aformula^{\neq}_{\hpto}(\avariable,\avariablebis)$. This implies $\astore(\avariable) \neq \astore(\avariablebis)$.
If one of the three first disjuncts of $\aformula^{\neq}_{\hpto}(\avariable,\avariablebis)$ holds,  we have immediately $\{\astore(\avariable),\astore(\avariablebis)\} \subseteq \domain{\aheap}$ and
$\aheap(\aheap(\astore(\avariable)))$ $= \aheap(\astore(\avariablebis))$. The remaining case is when
none of the three first disjuncts holds and $\pair{\astore}{\aheap}$ satisfies the fourth one.
So, there is a heap $\aheap' \sqsubseteq \aheap$ with $\card{\domain{\aheap'}} = 3$  such that $\pair{\astore}{\aheap'}$ satisfies
\begin{gather*}
\alloc{\avariable} \land \alloc{\avariablebis} \land \textstyle (\bigwedge_{\avariableter,\avariableter' \in \{\avariable,\avariablebis\}} \lnot \avariableter
\hpto \avariableter') \land \lnot\reach(\avariable,\avariablebis) \leq 3\\ \land ((\size=1 \land \allocbacktwo{\avariable}{\avariablebis}) \lollipop (\reach(\avariable,\avariablebis)=3 \land \avariablebis \hpto^2_\avariable \avariablebis)).
\end{gather*}
Since $\pair{\astore}{\aheap} \models \lnot \mathsf{ext}(n(\avariable) = n(\avariablebis))$, the only memory states
with three memory cells
satisfying $\alloc{\avariable} \land \alloc{\avariablebis} \land \textstyle (\bigwedge_{\avariableter,\avariableter' \in \{\avariable,\avariablebis\}} \lnot \avariableter
\hpto \avariableter')$ are the following (up to isomorphism with respect to $\{\avariable,\avariablebis\}$):
\begin{center}
\setlength{\tabcolsep}{25pt}
\begin{longtable}{c c c c}
  \begin{tikzpicture}[baseline]
    \node[dot,label=above:$\avariable$] (l1) at (0,0) {};
    \node[dot] (l2) [below of=l1] {};
    \node[dot,label=above:$\avariablebis$] (l3) [right=0.5cm of l1] {};
    \node[dot] (l4) [below of=l3] {};
    \node[dot] (l5) [right=0.5cm of l3] {};
    \node[dot] (l6) [below of=l5] {};

    \draw[pto] (l1) -- (l2);
    \draw[pto] (l3) -- (l4);
    \draw[pto] (l5) -- (l6);
  \end{tikzpicture}
&
  \begin{tikzpicture}[baseline]
    \node[dot,label=above:$\avariable$] (l1) at (0,0) {};
    \node[dot] (l2) [below of=l1] {};
    \node[dot,label=above:$\avariablebis$] (l3) [right=0.5cm of l1] {};
    \node[dot] (l4) [below of=l3] {};
    \node[dot] (l5) [right=0.5cm of l3] {};

    \draw[pto] (l1) -- (l2);
    \draw[pto] (l3) -- (l4);
    \draw[pto] (l5) edge [loop below, in=-30, out=-150,looseness=20] node {} (l5);
  \end{tikzpicture}
&
  \begin{tikzpicture}[baseline]
    \node[dot,label=above:$\avariable$] (l1) at (0,0) {};
    \node[dot] (l2) [below of=l1] {};
    \node[dot,label=above:$\avariablebis$] (l3) [right=0.5cm of l1] {};
    \node[dot] (l4) [below of=l3] {};
    \node[dot] (l5) [right=0.5cm of l3] {};

    \draw[pto] (l1) -- (l2);
    \draw[pto] (l3) -- (l4);
    \draw[pto] (l5) -- (l3);
  \end{tikzpicture}
&
  \begin{tikzpicture}[baseline]
    \node[dot,label=above:$\avariable$] (l1) at (0,0) {};
    \node[dot] (l2) [below of=l1] {};
    \node[dot,label=above:$\avariablebis$] (l3) [right=0.5cm of l1] {};
    \node[dot] (l4) [below of=l3] {};
    \node[dot] (l5) [right=0.5cm of l3] {};

    \draw[pto] (l1) -- (l2);
    \draw[pto] (l3) -- (l4);
    \draw[pto] (l5) -- (l4);
  \end{tikzpicture}
\\ \\
  \begin{tikzpicture}[baseline]
    \node[dot,label=above:$\avariable$] (l1) at (0,0) {};
    \node[dot] (l2) [below of=l1] {};
    \node[dot,label=above:$\avariablebis$] (l3) [right=0.5cm of l1] {};
    \node[dot] (l4) [below of=l3] {};
    \node[dot] (l5) [left=0.5cm of l1] {};

    \draw[pto] (l1) -- (l2);
    \draw[pto] (l3) -- (l4);
    \draw[pto] (l5) -- (l1);
  \end{tikzpicture}
&
  \begin{tikzpicture}[baseline]
    \node[dot,label=above:$\avariable$] (l1) at (0,0) {};
    \node[dot] (l2) [below of=l1] {};
    \node[dot,label=above:$\avariablebis$] (l3) [right=0.5cm of l1] {};
    \node[dot] (l4) [below of=l3] {};
    \node[dot] (l5) [left=0.5cm of l1] {};

    \draw[pto] (l1) -- (l2);
    \draw[pto] (l3) -- (l4);
    \draw[pto] (l5) -- (l2);
  \end{tikzpicture}
&
  \begin{tikzpicture}[baseline]
    \node[dot,label=above:$\avariable$] (l1) at (0,0) {};
    \node[dot] (l2) [below of=l1] {};
    \node[dot,label=above:$\avariablebis$] (l3) [right=0.5cm of l1] {};
    \node[dot] (l4) [below of=l3] {};

    \draw[pto] (l1) -- (l2);
    \draw[pto] (l3) -- (l4);
    \draw[pto] (l4) edge [loop below, in=-30, out=-150,looseness=20] node {} (l4);
  \end{tikzpicture}
&
  \begin{tikzpicture}[baseline]
    \node[dot,label=above:$\avariable$] (l1) at (0,0) {};
    \node[dot] (l2) [below of=l1] {};
    \node[dot,label=above:$\avariablebis$] (l3) [right=0.5cm of l1] {};
    \node[dot] (l4) [below of=l3] {};
    \node[dot] (l5) [right=0.5 of l4] {};

    \draw[pto] (l1) -- (l2);
    \draw[pto] (l3) -- (l4);
    \draw[pto] (l4) -- (l5);
  \end{tikzpicture}
\\ \\
  \begin{tikzpicture}[baseline]
    \node[dot,label=above:$\avariable$] (l1) at (0,0) {};
    \node[dot] (l2) [below of=l1] {};
    \node[dot,label=above:$\avariablebis$] (l3) [right=0.5cm of l1] {};
    \node[dot] (l4) [below of=l3] {};

    \draw[pto] (l1) -- (l2);
    \draw[pto] (l3) edge [bend right] (l4);
    \draw[pto] (l4) edge [bend right] (l3);
  \end{tikzpicture}
&
  \begin{tikzpicture}[baseline]
    \node[dot,label=above:$\avariable$] (l1) at (0,0) {};
    \node[dot] (l2) [below of=l1] {};
    \node[dot,label=above:$\avariablebis$] (l3) [right=0.5cm of l1] {};
    \node[dot] (l4) [below of=l3] {};

    \draw[pto] (l1) -- (l2);
    \draw[pto] (l3) -- (l4);
    \draw[pto] (l4) -- (l1);
  \end{tikzpicture}
&
  \begin{tikzpicture}[baseline]
    \node[dot,label=above:$\avariable$] (l1) at (0,0) {};
    \node[dot] (l2) [below of=l1] {};
    \node[dot,label=above:$\avariablebis$] (l3) [right=0.5cm of l1] {};
    \node[dot] (l4) [below of=l3] {};

    \draw[pto] (l1) -- (l2);
    \draw[pto] (l3) -- (l4);
    \draw[pto] (l4) -- (l2);
  \end{tikzpicture}
&
  \begin{tikzpicture}[baseline]
    \node[dot,label=above:$\avariable$] (l1) at (0,0) {};
    \node[dot] (l2) [below of=l1] {};
    \node[dot,label=above:$\avariablebis$] (l3) [right=0.5cm of l1] {};
    \node[dot] (l4) [below of=l3] {};

    \draw[pto] (l1) -- (l2);
    \draw[pto] (l3) -- (l4);
    \draw[pto] (l2) edge [loop below, in=-30, out=-150,looseness=20] node {} (l2);
  \end{tikzpicture}
\\ \\
  \begin{tikzpicture}[baseline]
    \node[dot,label=above:$\avariable$] (l1) at (0,0) {};
    \node[dot] (l2) [below of=l1] {};
    \node[dot,label=above:$\avariablebis$] (l3) [right=0.5cm of l1] {};
    \node[dot] (l4) [below of=l3] {};
    \node[dot] (l5) [left=0.5 of l2] {};

    \draw[pto] (l1) -- (l2);
    \draw[pto] (l3) -- (l4);
    \draw[pto] (l2) -- (l5);
  \end{tikzpicture}
&
  \begin{tikzpicture}[baseline]
    \node[dot,label=above:$\avariable$] (l1) at (0,0) {};
    \node[dot] (l2) [below of=l1] {};
    \node[dot,label=above:$\avariablebis$] (l3) [right=0.5cm of l1] {};
    \node[dot] (l4) [below of=l3] {};

    \draw[pto] (l1) edge [bend right] (l2);
    \draw[pto] (l3) -- (l4);
    \draw[pto] (l2) edge [bend right] (l1);
  \end{tikzpicture}
&
  \begin{tikzpicture}[baseline]
    \node[dot,label=above:$\avariable$] (l1) at (0,0) {};
    \node[dot] (l2) [below of=l1] {};
    \node[dot,label=above:$\avariablebis$] (l3) [right=0.5cm of l1] {};
    \node[dot] (l4) [below of=l3] {};

    \draw[pto] (l1) -- (l2);
    \draw[pto] (l3) -- (l4);
    \draw[pto] (l2) -- (l3);
  \end{tikzpicture}
&
  \begin{tikzpicture}[baseline]
    \node[dot,label=above:$\avariable$] (l1) at (0,0) {};
    \node[dot] (l2) [below of=l1] {};
    \node[dot,label=above:$\avariablebis$] (l3) [right=0.5cm of l1] {};
    \node[dot] (l4) [below of=l3] {};

    \draw[pto] (l1) -- (l2);
    \draw[pto] (l3) -- (l4);
    \draw[pto] (l2) -- (l4);
  \end{tikzpicture}
\end{longtable}
\end{center}

The formula $(\size=1 \land \allocbacktwo{\avariable}{\avariablebis}) \lollipop (\reach(\avariable,\avariablebis)=3 \land \avariablebis \hpto^2_\avariable \avariablebis)$
can only be satisfied on heaps $\aheap'$ where there exists a way of adding a one-memory-cell heap $\aheap''$
with $\astore(\avariablebis) \in \range{\aheap''}$
and $(\aheap' + \aheap'')^3 (\astore(\avariable)) = \astore(\avariablebis)$.
This rules out all the memory states of the figure but the first and last one of the last row:
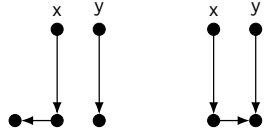
\begin{center}
\hfill
\begin{tikzpicture}[baseline]
  \node[dot,label=above:$\avariable$] (l1) at (0,0) {};
  \node[dot] (l2) [below of=l1] {};
  \node[dot,label=above:$\avariablebis$] (l3) [right=0.5cm of l1] {};
  \node[dot] (l4) [below of=l3] {};
  \node[dot] (l5) [left=0.5 of l2] {};

  \draw[pto] (l1) -- (l2);
  \draw[pto] (l3) -- (l4);
  \draw[pto] (l2) -- (l5);
\end{tikzpicture}
\hfill
\begin{tikzpicture}[baseline]
  \node[dot,label=above:$\avariable$] (l1) at (0,0) {};
  \node[dot] (l2) [below of=l1] {};
  \node[dot,label=above:$\avariablebis$] (l3) [right=0.5cm of l1] {};
  \node[dot] (l4) [below of=l3] {};

  \draw[pto] (l1) -- (l2);
  \draw[pto] (l3) -- (l4);
  \draw[pto] (l2) -- (l4);
\end{tikzpicture}
\hfill\,
\end{center}
Typically, $\aheap''$ can only take the value $\{ \aheap'(\aheap'(\astore(\avariable))) \pto \astore(\avariablebis)\}$.
However, only the second memory state is able to verify the condition $\pair{\astore}{\aheap' + \aheap''} \models \avariablebis \hpto^2_\avariable \avariablebis$. Then $\aheap'(\astore(\avariablebis)) = \aheap'(\aheap'(\astore(\avariable)))$ and we conclude that $\{\astore(\avariable),\astore(\avariablebis)\} \subseteq \domain{\aheap}$ and
$\aheap(\aheap(\astore(\avariable)))$ $= \aheap(\astore(\avariablebis))$.
\end{description}

(II) It is easy to show that $n(\avariable) \hpto n(\overline{\avariable})$ will never occur in the translation. Indeed, by always translating formulae with variables in $\{\avariable_1,\dots,\avariable_q\}$ we can only have $n(\avariable_i) \hpto n(\avariable_j)$ or ${n(\overline{\avariable_i}) \hpto n(\overline{\avariable_j})}$ (the latter case due to the renaming in the translation of the left side of the magic wand).
As such, $\aformula_{\hpto}(\avariable,\avariablebis, \overline{\avariable})$ can be used to check whenever $n(\avariable) \hpto n(\avariablebis)$.
\end{proof}

As for $\allocbacktwo{\avariablebis}{\avariable}$, the properties of the translation imply the equivalence between $n(\avariable) \hpto  n(\avariablebis)$ and $\aformula_{\hpto}(\avariable,\avariablebis,\overline{\avariable})$ (as stated in Lemma~\ref{lemma-three-in-one}(II)).
By looking at the formulae herein defined, the predicate $\reach$ only appears bounded, i.e. in the form of $\reach(\avariable,\avariablebis) = 2$ and $\reach(\avariable,\avariablebis)=3$. The three new predicates  can therefore be defined in $\seplogic{\separate, \magicwand}$ enriched
with $\reach(\avariable,\avariablebis) = 2$ and $\reach(\avariable,\avariablebis) = 3$.
Observe that  $\reach(\avariable,\avariablebis) = 0$ and $\reach(\avariable,\avariablebis) = 1$ are already definable in
the logic $\seplogic{\separate, \magicwand}$.

\subsection{Undecidability results and non-finite axiomatization}
\label{section-undecidability-results-and-nfax}
It is time to collect the fruits of all our efforts and to conclude this part about undecidability.
As a direct consequence of Corollary~\ref{corollary-satisfiability-validity} and the undecidability of
$\foseplogic{\magicwand}$, here is one of the main results of the paper.
\begin{theorem}
\label{theorem-main-undecidability}
The satisfiability problem for $\seplogic{\separate, \magicwand, \ls}$ is undecidable.
\end{theorem}
As a by-product, we get the following (negative) result.

\begin{theorem}
The set of valid formulae for  $\seplogic{\separate, \magicwand, \ls}$ is not
recursively enumerable.
\end{theorem}

Indeed, suppose that the set of valid formulae for $\seplogic{\separate, \magicwand, \ls}$
were recursively enumerable, then one can enumerate the valid formulae of the form $\mathcal{T}_{\rm VAL}(\aformula)$
as it is decidable in \ptime whether $\aformulabis$ in $\seplogic{\separate, \magicwand, \ls}$
is syntactically equal to  $\mathcal{T}_{\rm VAL}(\aformula)$ for some $\foseplogic{\magicwand}$ formula $\aformula$.
This leads to a contradiction since this would allow the enumeration of valid formulae in $\foseplogic{\magicwand}$.
Note also that the set of satisfiable formulae for $\seplogic{\separate, \magicwand, \ls}$  is not
recursively enumerable.
Indeed, suppose $\aformula$ is a formula built over $V = \set{\avariable_1, \ldots, \avariable_n}$.
Given a partition $\aset$ of $V$, we write $\aformulabis_{\aset}$ to denote conjunctions of equalities and
inequalities that state that two variables (resp. not) in the same element of the partition are equal (resp. are different).
One can show that $\aformula$ is valid iff $\emp \wedge \aformulabis_{\aset} \wedge (\top \magicwand \aformula)$
is satisfiable, for all partitions $\aset$. As the set of valid formulae is not r.e., the same applies to
the set of satisfiable formulae.

Below, the essential ingredients to establish the undecidability of $\seplogic{\separate, \magicwand, \ls}$
allow us to propose several refinements. For instance, the key property rests on the fact
that the following properties $n(\avariable) = n(\avariablebis)$, $n(\avariable) \hpto n(\avariablebis)$
and $\allocback{\avariable}$ are expressible in the logic.


\begin{corollary}
\label{corollary-undecidability-next}
$\seplogic{\separate, \magicwand}$ augmented with
formulae
of the form $n(\avariable) = n(\avariablebis)$, ${n(\avariable) \hpto n(\avariablebis)}$ and
$\allocback{\avariable}$
admits an undecidable satisfiability problem.
\end{corollary}

Corollary~\ref{corollary-undecidability-next} can be refined a bit further by noting in which context
the reachability predicates
$\ls$ and $\reach$ are used in the formulae. This allows us to get the result below.

\begin{corollary}
\label{corollary-undecidability-limited-reach}
$\seplogic{\separate, \magicwand}$ augmented with
built-in formulae
of the form $\reach(\avariable,\avariablebis) = 2$ and $\reach(\avariable,\avariablebis) = 3$
admits an undecidable satisfiability problem.
\end{corollary}

It is the addition of $\reach(\avariable,\avariablebis) = 3$ that is crucial for undecidability since
the satisfiability problem for $\seplogic{\separate, \magicwand,\reach(\avariable,\avariablebis) = 2}$
is in \pspace~\cite{Demrietal17}.

Following a similar analysis, let SL1($\forall,\separate,\magicwand$) (resp. SL2($\forall,\separate,\magicwand$)) be the
restriction of $\foseplogic{\separate, \magicwand}$ (i.e.  $\foseplogic{\magicwand}$ plus $\separate$)
to formulae of the form
$
\exists \avariable_1 \ \cdots \  \exists \avariable_q \ \aformula
$,
where $q \geq 1$, the variables in $\aformula$ are among $\set{\avariable_1, \ldots, \avariable_{q+1}}$
(resp. are among $\set{\avariable_1, \ldots, \avariable_{q+2}}$)
and the only quantified variable in $\aformula$ is $\avariable_{q+1}$ (resp. and the only quantified variables in $\aformula$ are $\avariable_{q+1}$ and $\avariable_{q+2}$).
Note that SL2($\forall,\separate,\magicwand$) should not be confused with $\foseplogic{\magicwand}$ restricted to two quantified variables.
The satisfiability problem for SL1($\forall,\separate,\magicwand$) is \pspace-complete~\cite{Demrietal17}.
Observe that SL1($\forall,\separate,\magicwand$) can easily express $n(\avariable) = n(\avariablebis)$
and $\allocback{\avariable}$.
The distance  between the decidability for SL1($\forall,\separate,\magicwand$)
and the undecidability for $\seplogic{\separate, \magicwand, \ls}$, is best witnessed by Corollary~\ref{corollary-final-undecidability}(I) below,
which solves an open problem~\cite[Section 6]{Demrietal17}.
Below, we state several undecidability results for the record, but we do not claim that all these variants
happen to be interesting in practice.
\begin{corollary}
\label{corollary-final-undecidability}
The satisfiability problem of the following logics is undecidable:
\begin{enumerate}[label=\normalfont{\textbf{(\Roman*)}}]
\item SL1($\forall,\separate,\magicwand$) augmented with
$n(\avariable) \hpto n(\avariablebis)$,
\item SL1($\forall,\separate,\magicwand$) augmented with $\ls$,
\item SL2($\forall,\separate,\magicwand$) ,
\item $\seplogic{\separate, \magicwand, \ls}$ restricted to four program
variables,
\item $\seplogic{\magicwand,n(\avariable) = n(\avariablebis), n(\avariable) \hpto n(\avariablebis),\allocback{\avariable}}$.
\end{enumerate}
\end{corollary}

\begin{proof}  (I) Consequence of Corollary~\ref{corollary-undecidability-next} by observing that
$\emptyconstant$, 
$n(\avariable) = n(\avariablebis)$ and $\allocback{\avariable}$ can be expressed with one quantified variable.

(II) Consequence of (I), as $n(\avariable) \hpto n(\avariablebis)$ can be expressed with $\ls$.

(III) Consequence of (I), as $n(\avariable) \hpto n(\avariablebis)$ can be expressed with two quantified variables.

(IV) It is shown in~\cite{DemriDeters16} that $\foseplogic{\magicwand}$ restricted to two quantified variables is undecidable.
The translation provided in Section~\ref{section-the-translation} assumes that distinct quantifications involve distinct variables.
In order to translate $\foseplogic{\magicwand}$ restricted to two quantified variables, it is necessary to give up that assumption
and to update the definition of $\atranslation$. Actually, only the clause for formulae of the form
$\forall \avariable_i \ \aformulabis$ requires a change ($i \in \set{1,2}$ and the formulae to be translated
contains at most two variables). Here is the new value for $\atranslation(\forall \avariable_i \ \aformulabis, \aset)$:
$$
((\alloc{\avariable_i} \wedge \size =1) \vee \emptyconstant) \separate
(\neg \alloc{\avariable_i} \wedge
(\alloc{\avariable_i} \wedge \size = 1) \magicwand
(\Safe(\aset) \Rightarrow
\atranslation(\aformulabis, \aset))).
$$
The proof of Lemma~\ref{lemma-correctness-translation} can be updated accordingly.

(V) In the translation $\atranslation$ to  $\seplogic{\separate, \magicwand, \ls}$ formulae, the separating connective $\separate$
appears only in the definition of $\atranslation(\aformulabis_1 \magicwand \aformulabis_2, \aset)$, assuming that
$n(\avariable) = n(\avariablebis)$, $n(\avariable) \hpto n(\avariablebis)$ and $\allocback{\avariable}$ are built-in predicates.
It is easy to check that if we remove
$$``(\bigwedge_{\mathclap{\avariableter \in \asetter}} \alloc{\overline{\avariableter}} \wedge \size = \card{\asetter}) \separate"$$
from $\atranslation(\aformulabis_1 \magicwand \aformulabis_2, \aset)$ (i.e., we disallow the recycling of variables),
the proof of Lemma~\ref{lemma-correctness-translation} can be updated accordingly.
\end{proof}

\section{Decision Procedures in \pspace for $\seplogic{\separate, \reachplus}$ and  Variants}
\label{section-decidability-star}
We show that the satisfiability problem for $\seplogic{\separate, \reachplus}$ can be solved in
polynomial space.
Refining the arguments used in our proof,
we also show that it is possible to push further the \pspace upper bound to the formulae expressible as a
Boolean combination of formulae from the two fragments $\seplogic{\separate, \reachplus}$ and $\seplogic{\separate, \sepimp}$.
The latter logic is shown to be \pspace-complete in~\cite{Yang01,Lozes04bis}.
Moreover, as shown in Section~\ref{section-preliminaries},
$\seplogic{\separate, \ls}$
can be understood as a fragment of $\seplogic{\separate, \reachplus}$ and therefore the \pspace upper bound for $\seplogic{\separate, \reachplus}$ established
here also holds for $\seplogic{\separate, \ls}$.

Our proof relies on a small heap property: a formula $\aformula$ is satisfiable
if and only if it admits a model with
a polynomial amount of memory cells.
The \pspace upper
bound then follows by establishing that the model-checking problem for $\seplogic{\separate, \reachplus}$
is in \pspace too.
To establish the small heap property,
an equivalence relation on memory states with finite index is designed,
following the standard approach
in~\cite{Yang01,Calcagno&Gardner&Hague05} and using test formulae as
in~\cite{Lozes04,Lozes04bis,Brochenin&Demri&Lozes09,Demrietal17,Echenim&Iosif&Peltier20}.
\subsection{Introduction to test formulae}
In order to show the \pspace upper bound for $\seplogic{\separate, \reachplus}$, we would like to adapt the technique introduced in~\cite{Lozes04bis} for $\seplogic{\separate, \magicwand}$, which relies on defining a set of \emph{test formulae} capturing the essential properties expressible in $\seplogic{\separate, \magicwand}$, as depicted by the following result.

\begin{proposition} \label{proposition-prop-sl} \cite{Lozes04bis,DD-esslli15}
Any formula $\aformula$ in $\seplogic{\separate, \magicwand}$ built over  variables
in $\avariable_1$, \ldots,$\avariable_q$ is logically equivalent to a Boolean combination
of test formulae, i.e. formulae among $\size \geq \beta$, $\alloc{\avariable_i}$,
$\avariable_i \pointsl{} \avariable_j$ and $\avariable_i = \avariable_j$,
with $\beta \geq 0$ and $i,j \in \interval{1}{q}$.
\end{proposition}
The formulae of the form $\size \geq \beta$ and $\alloc{\avariable_i}$ are introduced in
Section~\ref{section-preliminaries} and $\alloc{\avariable_i}$ holds when $\astore(\avariable_i)$
belongs to the heap domain and $\size \geq \beta$ holds when the heap has at least $\beta$ memory cells.
By way of example, $(\neg \emptyconstant \separate ((\avariable_1 \pointsl{} \avariable_1) \magicwand \perp))$
is equivalent to $\size \geq 2 \wedge \alloc{\avariable_1}$.  As a corollary of its proof,
Proposition~\ref{proposition-prop-sl} can be refined so that in the formulae $\size \geq \beta$, we can enforce that
$\beta \leq 2 \times \length{\aformula}$ (rough upper
bound)
where $\length{\aformula}$ is the size of $\aformula$ (seen as a tree).
Consequently,  for any satisfiable formula $\aformula$ in
$\seplogic{\separate, \magicwand}$, there is a memory state with less than $2 \times \length{\aformula}$
memory cells that satisfies it.
Similar results will be shown for $\seplogic{\separate, \reachplus}$
and for some of its extensions.

In order to define a set of test formulae that captures the expressive power of $\seplogic{\separate, \reachplus}$,
we need to study which basic properties on memory states can be expressed by formulae in $\seplogic{\separate, \reachplus}$.
For example, consider the memory states from Figure~\ref{figure-fourms}.
\begin{figure}
\begin{center}
    \newlength{\R}\setlength{\R}{1.2cm}
     \begin{tikzpicture}[baseline]

        \node[dot,label=above right:$\avariable_i$] (i) at ( 60:\R) {};
        \node[dot] (km) at (120:\R) {};
        \node[dot,label=left:$\avariable_k$] (k) at (180:\R) {};
        \node[dot] (jm) at (240:\R) {};
        \node[dot,label=below right:$\avariable_j$] (j) at (300:\R) {};
        \node[dot] (im) at (360:\R) {};

       \draw[pto] (i) edge [bend left = 20] node {} (im);
        \draw[pto] (im) edge [bend left = 20] node {} (j);
       \draw[pto] (j) edge [bend left = 20] node {} (jm);
      \draw[pto] (jm) edge [bend left = 20] node {} (k);
       \draw[pto] (k) edge [bend left = 20] node {} (km);
        \draw[pto] (km) edge [bend left = 20] node {} (i);
     \end{tikzpicture}
  \hfill
     \begin{tikzpicture}[baseline]

        \node[dot,label=above right:$\avariable_i$] (i) at ( 60:\R) {};
        \node[dot] (km) at (120:\R) {};
        \node[dot,label=left:$\avariable_k$] (k) at (180:\R) {};
        \node[dot] (jm) at (240:\R) {};
        \node[dot,label=below right:$\avariable_j$] (j) at (300:\R) {};
        \node[dot] (im) at (360:\R) {};

       \draw[pto] (im) edge [bend right = 20] node {} (i);
        \draw[pto] (j) edge [bend right = 20] node {} (im);
       \draw[pto] (jm) edge [bend right = 20] node {} (j);
      \draw[pto] (k) edge [bend right = 20] node {} (jm);
       \draw[pto] (km) edge [bend right = 20] node {} (k);
        \draw[pto] (i) edge [bend right = 20] node {} (km);
     \end{tikzpicture}
    \hfill
     \begin{tikzpicture}[baseline]

        \node[dot,label=above right:$\avariable_j$] (j) at ( 60:\R) {};
        \node[dot,label=above left:$\avariable_i$] (i) at (120:\R) {};
        \node[dot,label=left:$\alocation$] (m1) [below of=i] {};
        \node[dot,label=right:$\alocation'$] (m2) [below of=j] {};
        \node[dot, label=below:$\avariable_k$] (k) at (270:\R) {};

        \draw[pto] (i) edge node {} (m1);
        \draw[pto] (j) edge node {} (m2);
        \draw[pto] (m1) edge [bend left = 20] node {} (m2);
        \draw[pto] (m2) edge [bend left = 20] node {} (k);
        \draw[pto] (k) edge [bend left = 20] node {} (m1);
     \end{tikzpicture}
  \hfill
     \begin{tikzpicture}[baseline]

        \node[dot,label=above right:$\avariable_j$] (j) at ( 60:\R) {};
        \node[dot,label=above left:$\avariable_i$] (i) at (120:\R) {};
        \node[dot,label=left:$\alocation$] (m1) [below of=i] {};
        \node[dot,label=right:$\alocation'$] (m2) [below of=j] {};
        \node[dot, label=below:$\avariable_k$] (k) at (270:\R) {};

        \draw[pto] (i) edge node {} (m1);
        \draw[pto] (j) edge node {} (m2);
        \draw[pto] (m1) edge [bend right = 20] node {} (k);
        \draw[pto] (k) edge [bend right = 20] node {} (m2);
        \draw[pto] (m2) edge [bend right = 20] node {} (m1);
     \end{tikzpicture}
\end{center}
\caption{Memory states $\pair{\astore_1}{\aheap_1}$, \ldots, $\pair{\astore_4}{\aheap_4}$ (from left to right)}
\label{figure-fourms}
\end{figure}
The memory states $\pair{\astore_1}{\aheap_1}$ and $\pair{\astore_2}{\aheap_2}$ can be distinguished by the formula
$$\true \sepcnj (\reach^+(\avariable_i,\avariable_j) \land \reach^+(\avariable_j,\avariable_k) \land \lnot \reach^+(\avariable_k,\avariable_i)).$$
Indeed,~$\pair{\astore_1}{\aheap_1}$ satisfies this formula (a witness is obtained by removing the only edge that leads to $\astore_1(\avariable_i)$),
whereas $\pair{\astore_2}{\aheap_2}$ does not, as every subheap of $\aheap_2$
that retains the path from $\astore(\avariable_i)$ to $\astore(\avariable_j)$
and
the one from $\astore(\avariable_j)$ to $\astore(\avariable_k)$ necessarily has a path from $\astore(\avariable_k)$ to $\astore(\avariable_i)$. This suggests that $\seplogic{\separate, \reachplus}$ can express whether, for example, any path from $\astore(\avariable_i)$ to $\astore(\avariable_j)$ also contains $\astore(\avariable_k)$.
We will introduce the test formula $\sees_q(\avariable_i,\avariable_j) \geq \beta$ to capture this property.

Similarly, the memory states $\pair{\astore_3}{\aheap_3}$ and $\pair{\astore_4}{\aheap_4}$ can be distinguished by the formula
\[(\size = 1) \sepcnj \big(\reach^+(\avariable_j,\avariable_k) \land \lnot \reach^+(\avariable_i,\avariable_k)
\land \lnot \reach^+(\avariable_k,\avariable_k)\big).\]
The memory state $\pair{\astore_3}{\aheap_3}$ satisfies this formula by separating $\{ \alocation \mapsto \alocation'\}$ from the rest of the heap,
whereas the formula is not satisfied by $\pair{\astore_4}{\aheap_4}$.
Indeed, there is no way to break the loop from $\astore(\avariable_k)$ to itself by removing just one location from the heap while retaining the path from
$\astore(\avariable_j)$ to $\astore(\avariable_k)$ and loosing the path from $\astore(\avariable_i)$ to $\astore(\avariable_k)$. This suggests
that the two locations $\alocation$ and $\alocation'$ are particularly interesting since they are reachable from several locations corresponding
to program variables. Therefore by separating them from the rest of the heap, several paths are lost. In order to capture this, we introduce
the notion of \defstyle{meet-points}.

Informally, given a memory state $\pair{\astore}{\aheap}$,
a meet-point between $\astore({\avariable})$ and $\astore({\avariablebis})$ leading to $\astore({\avariableter})$  is a location
$\alocation$ such that $\alocation$  reaches $\astore({\avariableter})$, both locations
$\astore(\avariable)$ and $\astore({\avariablebis})$ reach $\alocation$,
and there is no location $\alocation'$ satisfying these properties and reachable from $\astore({\avariable})$ in strictly
less steps.
Let us formalise this notion.
Let $\Terms{q}$ be the set $\set{\avariable_1, \ldots, \avariable_q} \cup \set{m_q(\avariable_i,\avariable_j) \ \mid \ i,j \in \interval{1}{q}}$
understood as the set of \defstyle{terms} that are either variables or expressions denoting a meet-point.
We write $\sem{\avariable_i}^q_{\astore,\aheap}$ to denote $\astore(\avariable_i)$
and $\sem{m_q(\avariable_i,\avariable_j)}^q_{\astore,\aheap}$ to denote
(if it exists)
the first location reachable from
$\astore(\avariable_i)$ that is also reachable from $\astore(\avariable_j)$.
Moreover we require that this location can reach another location corresponding to a program variable.
Formally, $\sem{m_q(\avariable_i,\avariable_j)}^q_{\astore,\aheap}$ is defined as the unique location $\alocation$ (if it exists) such that
\begin{itemize}
\itemsep 0 cm
\item
there are $L_1, L_2\geq 0$ such that  $\aheap^{L_1}(\astore(\avariable_i)) = \aheap^{L_2}(\astore(\avariable_j)) = \alocation$, and
\item for all $L_1' < L_1$ and for all $L_2'\geq 0$, $\aheap^{L_1'}\big(\astore(\avariable_i)\big) \not= \aheap^{L_2'}\big(\astore(\avariable_j)\big)$, and
\item there exist  $k \in \interval{1}{q}$ and $L\geq 0$ such that  $\aheap^L(\alocation) = \astore(\avariable_k)$.
\end{itemize}
One can easily show that the notion $\sem{m_q(\avariable_i,\avariable_j)}^q_{\astore,\aheap}$ is well-defined.
\begin{lemma}
\label{lemma-meet-unique}
  Let $i,j \in \interval{1}{q}$. At most one location satisfies the conditions of $\sem{m_q(\avariable_i,\avariable_j)}^q_{\astore,\aheap}$.
\end{lemma}

\begin{proof}
Ad absurdum, assume that there are two different locations $\alocation_1$ and $\alocation_2$ that satisfy the conditions of $\sem{m_q(\avariable_i,\avariable_j)}^q_{\astore,\aheap}$.
In particular, there are $L_1,L_2,L_3,L_4 \geq 0$ such that 
\begin{itemize}
  \item[(A)] $\aheap^{L_1}(\astore(\avariable_i)) = \aheap^{L_2}(\astore(\avariable_j)) = \alocation_1$ and $\aheap^{L_3}(\astore(\avariable_i)) = \aheap^{L_4}(\astore(\avariable_j)) = \alocation_2$,
  \item[(B)] for all $L_1' < L_1$ and for all $L_2'\geq 0$, $\aheap^{L_1'}\big(\astore(\avariable_i)\big) \not= \aheap^{L_2'}\big(\astore(\avariable_j)\big)$.
  \item[(C)] for all $L_3' < L_3$ and for all $L_4'\geq 0$, $\aheap^{L_3'}\big(\astore(\avariable_i)\big) \not= \aheap^{L_4'}\big(\astore(\avariable_j)\big)$.
\end{itemize}
Notice that (A) corresponds to the first condition of $\sem{m_q(\avariable_i,\avariable_j)}^q_{\astore,\aheap}$, with respect to $\alocation_1$ and $\alocation_2$, whereas (B) and (C) correspond to the second condition for $\alocation_1$ and $\alocation_2$, respectively.
Let $L_{\rm min} = \min \set{n \ \mid \ \aheap^n(\astore(\avariable_i)) \in \set{\alocation_1,\alocation_2}}$. 
Let us assume that $\aheap^{L_{\rm min}}(\astore(\avariable_i)) = \alocation_1$ and so, as $\alocation_1 \neq \alocation_2$, for all $L \geq 0$, $\aheap^{L}(\astore(\avariable_i)) = 
\alocation_2$ implies $L_{\rm min} < L$.
From (A) this implies that $L_{\rm min} < L_1'$ and $\aheap^{L_{\rm min}}(\astore(\avariable_i)) = \heap^{L_2}{\astore(\aheap_j)}$, 
which however contradicts (C).
The case where $\aheap^{L_{\rm min}}(\astore(\avariable_i)) = \alocation_2$ is analogous, the contradiction being reached with (B).
\end{proof}

The notion of meet-point is quite natural for studying fragments of separation logic with $\ls$. For instance, a similar notion
of a {\em cut point}, although satisfying
different conditions, is considered for the fragments studied in~\cite{Bozga&Iosif&Perarnau10}. 
Figure~\ref{figure-taxonomy} provides a taxonomy of meet-points, where arrows labelled by `$+$' represent paths of non-zero length
and zig-zag arrows any path (possibly of length zero).
Symmetrical cases, obtained by swapping $\avariable_i$ and $\avariable_j$, are omitted.
\begin{figure}
\begin{center}
  \begin{tabular*}{\textwidth}{c @{\extracolsep{\fill}}  cccc}
    &
     \begin{tikzpicture}[baseline]
       \node[dot,label=above:$\avariable_i$] (i) at (0,0) {};
       \node[dot,label=left:{$m_q$($\avariable_i$,$\avariable_j$)\\$m_q$($\avariable_j$,$\avariable_i$)}] (m) [below right = 1.5cm and 0.5cm of i] {};
       \node[dot,label=above:$\avariable_j$] (j) [above right=1.5cm and 0.5cm of m] {};
       \node[dot,label=below:{$\avariable_k$\\[8pt]$\avariable_k$ not inside a loop,\\but can reach one}] (k) [below of=m] {};

       \draw[reach] (i) -- (m);
       \draw[reach] (j) -- (m);
       \draw[reach] (m) -- (k);
     \end{tikzpicture}
     &
     \begin{tikzpicture}[baseline]
       \node[dot,label=above:$\avariable_i$] (i) at (0,0) {};
       \node[dot,label=left:{$m_q$($\avariable_i$,$\avariable_j$)\\$m_q$($\avariable_j$,$\avariable_i$)}] (m) [below right = 1.5cm and 0.5cm of i] {};
       \node[dot,label=above:$\avariable_j$] (j) [above right=1.5cm and 0.5cm of m] {};
       \node[dot] (mid) [below=0.8cm of m] {};
       \node[dot,label=below:{$\avariable_k$}] (k) [below=1.65cm of m] {};

       \node (zz) [right = 1.2cm of m] {};

       \draw[reach] (i) -- (m);
       \draw[reach] (j) -- (m);
       \draw[reach] (m) -- (mid);
       \draw[reach] (mid) to [out=-35,in=35] (k);
       \draw[pto] (k) edge [out=155,in=-155] node [left] {$+$} (mid);
     \end{tikzpicture}
     &
     \begin{tikzpicture}[baseline]
       \node[dot,label=above:$\avariable_i$] (i) at (0,0) {};
       \node[dot,label=left:{$m_q$($\avariable_i$,$\avariable_j$)}] (m1) [below right = 1.5cm and 0.4cm of i] {};
       \node[dot,label=right:{$m_q$($\avariable_j$,$\avariable_i$)}] (m2) [below right = 1.7cm and 0.4cm of m1] {};
       \node[dot,label=above:$\avariable_j$] (j) [above right=1.5cm and 0.5cm of m2] {};
       \node[dot,label=left:{$\avariable_k$}] (k) [below right= 2.7cm and 0.1 of i] {};

       \draw[reach] (i) -- (m1);
       \draw[pto] (m1) edge node [above right] {$+$} (m2);
       \draw[reach] (j) -- (m2);
       \draw[reach] (m2) to [out=180,in=-75] (k);
       \draw[pto] (k) edge [bend left=30] node [left] {$+$} (m1);

     \end{tikzpicture}
     &
     \\
     &
     (i)
     &
     (ii)
     &
     (iii)
     &
  \end{tabular*}
\end{center}
\caption{A taxonomy of meet-points}
\label{figure-taxonomy}
\end{figure}
Following Figure~\ref{figure-taxonomy},
the taxonomy with three distinct shapes depends on the following conditions
on $\sem{m_q(\avariable_i,\avariable_j)}^q_{\astore,\aheap}$  and $\astore(\avariable_k)$:
\begin{itemize}
  \item $\astore(\avariable_k)$ is not inside a loop iff $\aheap$ witnesses the shape (i) of the taxonomy.
  \item $\astore(\avariable_k)$ is inside a loop and $\sem{m_q(\avariable_i,\avariable_j)}^q_{\astore,\aheap} = \sem{m_q(\avariable_j,\avariable_i)}^q_{\astore,\aheap}$ iff $\aheap$ witnesses the shape (ii).
  \item $\astore(\avariable_k)$ is inside a loop and $\sem{m_q(\avariable_i,\avariable_j)}^q_{\astore,\aheap}
\neq \sem{m_q(\avariable_j,\avariable_i)}^q_{\astore,\aheap}$ iff $\aheap$ witnesses the shape (iii).
\end{itemize}
Notice that, in the case (i), $\astore(\avariable_k)$ can a priori be in~$\domain{\aheap}$, which also means that it can reach a loop, even though it is not inside one.
The asymmetrical definition of meet-points is captured by the case (iii).
Consider the memory states from Figure~\ref{figure-fourms},
$\pair{\astore_3}{\aheap_3}$ and $\pair{\astore_4}{\aheap_4}$ can be
seen as an instance of this case of the taxonomy and, as such, it holds that
$\sem{m_q(\avariable_i,\avariable_j)}^q_{\astore_3,\aheap_3} = \alocation$ and $\sem{m_q(\avariable_j,\avariable_i)}^q_{\astore_3,\aheap_3} = \alocation'$.

\cut{
Instead, the last two pictures capture the case when both $\avariable_i$ and $\avariable_j$ reach a loop that contains a third variable
$\avariable_k$. Moreover, the first location of the loop reached by $\avariable_i$ is different from the first location
of the loop reached by $\avariable_j$.
These two locations satisfy the conditions of meet-points: the first location corresponds to $\sem{m_q(\avariable_i,\avariable_j)}^q_{\astore,\aheap}$
whereas the second one corresponds to $\sem{m_q(\avariable_j,\avariable_i)}^q_{\astore,\aheap}$.
This asymmetrical definition for meet-points is needed since
$\seplogic{\separate,\reachplus}$
is expressive enough to capture the differences between the structures in the picture.%
}

We identify as special locations
the $\astore(\avariable_i)$'s and the meet-points of the form $\sem{m_q(\avariable_i,\avariable_j)}^q_{\astore,\aheap}$
when it exists ($i,j \in \interval{1}{q}$).  We call such locations, \defstyle{labelled} locations and denote by
$\Labels{q}{\astore,\aheap}$ the set of labelled locations.
The following lemma states an important property of labelled locations: taking subheaps can only reduce their set.

\begin{lemma}
\label{lemma-labels}
Let $\pair{\astore}{\aheap}$ be a memory state, $\aheap' \sqsubseteq \aheap$ and $q \geq 1$. We have $\Labels{q}{\astore, \aheap'} \subseteq \Labels{q}{\astore,\aheap}$.
\end{lemma}

\begin{proof}
Let $\alocation \in \Labels{q}{\astore,\aheap'}$. We want to prove that $\alocation \in \Labels{q}{\astore,\aheap}$.
The only interesting case is when
$\alocation \in \set{\sem{m_q(\avariable_i,\avariable_j)}^q_{\astore,\aheap'}  \ \mid \
\sem{m_q(\avariable_i,\avariable_j)}^q_{\astore,\aheap'} \ {\rm is \ defined}, i,j \in \interval{1}{q}} \setminus
\set{\astore(\avariable_i) \ \mid \ i \in \interval{1}{q}}$
as $\pair{\astore}{\aheap}$ and $\pair{\astore}{\aheap'}$ share the same store and therefore the result is trivial
for locations that correspond to the interpretation of program variables.
So, suppose that $\alocation = \sem{m_q(\avariable_i,\avariable_j)}^q_{\astore,\aheap'}$ for some $i,j$.
We just need to prove that
  there exist $i^{\star},j^{\star} \in \interval{1}{q}$ such that
$\sem{m_q(\avariable_{i^{\star}},\avariable_{j^{\star}})}^q_{\astore,\aheap} = \alocation$.

By definition,
  $\sem{m_q(\avariable_i,\avariable_j)}^q_{\astore,\aheap'} = \alocation$ if and only if there exist $k \in \interval{1}{q}$ and $L$ such that ${\aheap'}^L(\alocation) = \astore(\avariable_k)$ and there are
  $L_1, L_2$ such that ${\aheap'}^{L_1}(\astore(\avariable_i)) = {\aheap'}^{L_2}(\astore(\avariable_j)) = \alocation$
  and for all $L_1' < L_1$ and $L_2'$ it holds ${\aheap'}^{L_1'}(\astore(\avariable_i)) \not= {\aheap'}^{L_2'}(\astore(\avariable_j))$.
  As $\aheap' \sqsubseteq \aheap$, each path of $\aheap'$ is also a path in $\aheap$. In particular, ${\aheap}^L(\alocation) = \astore(\avariable_k)$.
As such, we only need to prove that that there exist $i^{\star},j^{\star} \in \interval{1}{q}$
  such that there are
  $L_1, L_2$ such that ${\aheap}^{L_1}(\astore(\avariable_{i^{\star}})) = {\aheap}^{L_2}(\astore(\avariable_{j^{\star}})) = \alocation$
  and for all $L_1' < L_1$ and $L_2'$ it holds ${\aheap}^{L_1'}(\astore(\avariable_{i^{\star}})) \not= {\aheap}^{L_2'}(\astore(\avariable_{j^{\star}}))$.
  We consider the taxonomy of $m_q(\avariable_i,\avariable_j)$ shown in Figure~\ref{figure-taxonomy}.
  Based on the classification, $\sem{m_q(\avariable_i,\avariable_j)}^q_{\astore,\aheap'} = \alocation$
  entails that the heap $\aheap'$ witnesses one of the situations among (i)--(iii).

  If the heap $\aheap'$ belongs to (ii) or (iii),  the locations $\astore(\avariable_i)$ and $\astore(\avariable_j)$ eventually reach a location~$\tilde{\alocation}$ belonging to a loop, i.e.~${\aheap'}^L(\tilde{\alocation}) = \tilde{\alocation}$ for some $L \geq 1$,
  and since $\aheap' \sqsubseteq \aheap$, 
  the same holds true in $\aheap$. 
  Hence, in $\aheap$, the location $\alocation$ is still
  the first location reachable from $\avariable_i$ that is also reachable from $\avariable_j$ and therefore $i^{\star} = i$ and $j^{\star} = j$.

  Alternatively, if $\aheap'$ belongs to (i), the heap $\aheap$ may belong to any of the three situations (i)--(iii) as far
  as $\alocation$ is concerned.
  It is  easy to see that in $\aheap$, either $\alocation = \sem{m_q(\avariable_i,\avariable_j)}^q_{\astore,\aheap}$ or
  $\alocation = \sem{m_q(\avariable_j,\avariable_i)}^q_{\astore,\aheap}$, which guarantees that $\alocation$ is a labelled location
  in $\aheap$. Let us justify this claim.
  As (i) holds, 
  either (a) $\astore(\avariable_k)$ can reach a location that is not in the domain of the heap $\aheap'$
  or (b) $\astore(\avariable_k)$ does not reach such a location, that is, there is a cycle  after 
  $\avariable_k$. If (b) holds, then  a reasoning analogous to the above cases for (ii) and (iii) works.
  Otherwise ((a) holds),  
in $\aheap'$ there is a location $\alocation'$ reachable from $\astore(\avariable_k)$ that is not 
in $\domain{\aheap'}$. We now consider how $\aheap$ extends $\aheap'$.
  \begin{itemize}
  \item If $\aheap$ belongs to (i), then $\alocation = \sem{m_q(\avariable_i,\avariable_j)}^q_{\astore,\aheap} =
        \sem{m_q(\avariable_i,\avariable_j)}^q_{\astore,\aheap'}$. The same holds true if $\aheap$
        extends $\aheap'$ by introducing a loop so that $\aheap$ belongs to (ii).
  \item If $\aheap$ extends $\aheap'$ by adding a path from $\alocation'$ to a location in the path between $\astore(\avariable_j)$ to $\alocation$ ($\alocation$ excluded), then $\alocation = \sem{m_q(\avariable_i,\avariable_j)}^q_{\astore,\aheap} =
 \sem{m_q(\avariable_i,\avariable_j)}^q_{\astore,\aheap'}$  and $\aheap$ belongs to (iii).
  \item Lastly, if $\aheap$ extends $\aheap'$ by adding a path from $\alocation'$ to a location in the path between $\astore(\avariable_i)$ to $\alocation$ ($\alocation$ excluded), then
$\aheap$ belongs also to (iii). Unlike the previous case however, we have
$\alocation = \sem{m_q(\avariable_j,\avariable_i)}^q_{\astore,\aheap} = \sem{m_q(\avariable_i,\avariable_j)}^q_{\astore,\aheap'}$,
so $\alocation$ is again in $\Labels{q}{\astore,\aheap}$. \qedhere
  \end{itemize}
\end{proof}

As shown in Proposition~\ref{proposition-prop-sl}, the test formulae explicit what are the essential features that are expressible in $\seplogic{\separate,\magicwand}$. For $\seplogic{\separate, \reachplus}$, these
test formulae primarily speak about relationships between labelled locations and specific shapes of the heap.
In order to highlight these relationships, below we introduce a class of abstract structures, 
called \defstyle{support graphs}. The semantics for test formulae shall be provided directly
on such support graphs.
\begin{definition}
\label{definition-support-graph}
Let $q \geq 1$ and $\pair{\astore}{\aheap}$ be a memory state.
Its \defstyle{support graph}
$\supportgraph{\astore,\aheap}$
is defined as the tuple
$(\gverts,\gedges,\galloc,\glabels,\gbtw,\grem)$
such that:
\begin{description}
\item[(V)] $\gverts \egdef \Labels{q}{\astore,\aheap}$ is the set of labelled locations.
\item[(E)] $\gedges$ is the (functional) binary relation on $\gverts$
      such that $\pair{\alocation}{\alocation'} \in \gedges$ $\equivdef$
      there is $L \geq 1$ such that $\aheap^{L}(\alocation) = \alocation'$ and for all $0 < L' < L$,
      it holds that $\aheap^{L'}(\alocation) \not \in \gverts$. Informally, $\gedges(\alocation,\alocation')$ if and only if there is a non-empty path from $\alocation$ to $\alocation'$ whose intermediate locations are not labelled.
  \item[(A)] $\galloc \egdef \gverts \cap \domain{\aheap}$.
  \item[(T)] $\glabels: \gverts \to \powerset{\Terms{q}}$ is a function such that for any labelled location $\alocation \in \gverts$, is the (non-empty)
  set of terms corresponding to $\alocation$. Formally,
  $\glabels(\alocation) \egdef \set{\aterm \in \Terms{q} \ \mid \ \sem{\aterm}^q_{\astore,\aheap} = \alocation}$.
  \item[(I)] $\gbtw: \gedges \to \powerset{\locations}$ is a function such that for any $\pair{\alocation}{\alocation'} \in \gedges$, is the set of intermediate
  locations of the shortest path beginning in $\alocation$ and ending with $\alocation'$. Formally,
  \begin{gather*}
  \gbtw(\alocation,\alocation') \egdef
  \left\{
    \alocation'' \in \locations \
    \middle|
  \begin{aligned}
    \text{there are } L,L' \geq 1, \ \text{such that } \aheap^{L}(\alocation) = \alocation'' \text{ and } \aheap^{L'}(\alocation'') = \alocation'\\
    \text{ and for all } L'' \in \interval{1}{L},\ \aheap^{L''}(\alocation) \not\in V
  \end{aligned}
  \right\}
  \end{gather*}
  Despite the definition of $\gedges$, the condition
  ``for all  $L'' \in \interval{1}{L},\ \aheap^{L''}(\alocation) \not\in V$'' in this definition is needed
  in order to handle the case when $\alocation$ and $\alocation'$ belongs to the same cycle.
  Moreover, notice that every location $\alocation'' \in \gbtw(\alocation,\alocation')$ is such that $\alocation'' \not\in V$ (hence, $\alocation''$ is not a labelled location).
  \item[(R)] $\grem \egdef \domain{\aheap} \setminus (\galloc \cup \bigcup_{(\alocation,\alocation') \in \gedges} \gbtw(\alocation,\alocation'))$, i.e. the set of
  locations in the domain that are not labelled locations   and that do not belong to paths between labelled locations.
\end{description}
\end{definition}
One can think of the support graph $\supportgraph{\astore,\aheap}$ as a selection of properties from $\pair{\astore}{\aheap}$.
Indeed, from its definition, it is straightforward to see that $\{\galloc,\grem\} \cup \{ \gbtw(\alocation,\alocation') \mid
(\alocation,\alocation') \in \gedges \}$ is a partition of $\domain{\aheap}$. Moreover,
for all $i,j \in \interval{1}{q}$ and $\alocation \in \locations$,
$\alocation = \astore(\avariable_i) = \astore(\avariable_j)$ iff $\{\avariable_i,\avariable_j\} \subseteq \glabels(\alocation)$.

We are now ready to introduce the test formulae.
Given $q,\alpha \geq 1$, we write $\Test(q, \alpha)$ to denote the following set of atomic \defstyle{test formulae}:
\begin{mathpar}
  \aterm = \aterm'

  \alloc{\aterm}

  \aterm \hpto \aterm'

  \sees_q(\aterm,\aterm')\geq \beta + 1

  \sizeothers_q \geq \beta,
\end{mathpar}
where $\aterm, \aterm' \in \Terms{q}$ and $\beta \in \interval{1}{\alpha}$.
It is worth noting that the  $\alloc{\aterm}$'s are not needed
for the logic $\seplogic{\separate,\reachplus}$ but it is required for extensions (see Theorem~\ref{theorem-boolean-combination}).
Let $\pair{\astore}{\aheap}$ be a memory state.
Below, we present the formal semantics for the test formulae, provided with the elements from the
support graph $\supportgraph{\astore,\aheap} = (\gverts,\gedges,\galloc,\glabels,\gbtw,\grem)$ of $\pair{\astore}{\aheap}$:
\begin{itemize}
\item $\pair{\astore}{\aheap} \models \aterm = \aterm'$ $\equivdef$ there is $\alocation \in \gverts$ such that $\{\aterm,\aterm'\} \subseteq \glabels(\alocation)$.
\item $\pair{\astore}{\aheap} \models \alloc{\aterm}$ $\equivdef$ there is $\alocation \in \galloc$ such that $\aterm \in \glabels(\alocation)$.
\item $\pair{\astore}{\aheap} \models \aterm \hpto \aterm'$ $\equivdef$ $\aterm \in \glabels(\alocation)$, $\aterm' \in \glabels(\alocation')$, $\card{\gbtw(\alocation,\alocation')} = 0$ for some $(\alocation,\alocation') \in \gedges$.
\item $\pair{\astore}{\aheap} \models \sees_q(\aterm,\aterm') \geq \beta\!+\!1$ $\equivdef$ there  is
$\pair{\alocation}{\alocation'}\!\in\!\gedges$ such that
$\aterm\!\in\!\glabels(\alocation)$,  $\aterm'\!\in\!\glabels(\alocation')$ and
$\card{\gbtw(\alocation,\alocation')} \geq \beta$.
\item $\pair{\astore}{\aheap} \models \sizeothers_q \geq \beta$ $\equivdef$ $\card{\grem} \geq \beta$.
\end{itemize}
The semantics is deliberately defined with the help of support graphs
but of course, the clauses could be formulated  with $\astore$ and $\aheap$ only.
For instance, the formula $\sees_q(m_q(\avariable_i,\avariable_j),\avariable_k) \geq 4$ is satisfied whenever $\sem{m_q(\avariable_i,\avariable_j)}^q_{\astore,\aheap}$
is defined and there is a path of length at least $4$ beginning in this location, ending in $\astore(\avariable_k)$ and whose intermediate locations are not labelled.
Similarly, $\sizeothers_q \geq \beta$ holds true whenever the number of allocated locations that are neither between two locations interpreted
by terms nor the interpretation of a term, is at least $\beta$.
Notice that there is no need for test formulae of the form $\sees_q(v,v') \geq 1$ since it is equivalent to
$v \hpto v' \lor \sees_q(v,v') \geq 2$. In the sequel, occurrences of $\sees_q(v,v') \geq 1$ are understood as abbreviations.
One can check whether $\sem{m_q(\avariable_i,\avariable_j)}^q_{\astore,\aheap}$ is defined thanks to the satisfaction
of the formula
$m_q(\avariable_i,\avariable_j) = m_q(\avariable_i,\avariable_j)$.
The satisfaction of $\sees_q(v,v')\geq \beta + 1$ entails the exclusion of labelled locations in the witness path,
which is reminiscent to  atomic formulae of the form $T \step{h \backslash T''} T'$ in the logic GRASS~\cite{Piskac&Wies&Zufferey13}.
Indeed, by way of example, $\avariable_1 \step{\aheap \backslash \avariable_3} \avariable_2$ states that there is a path in the graph
of the heap $\aheap$ between $\astore(\avariable_1)$ and $\astore(\avariable_2)$ without passing via $\astore(\avariable_3)$.
The satisfaction of $\sees_q(v,v')\geq \beta + 1$ requires to exclude labelled locations strictly between the interpretation of
the terms $v$ and $v'$.
Our test formulae
are quite expressive since they capture the atomic formulae from $\seplogic{\separate, \reachplus}$ and the test formulae
for $\seplogic{\separate,\magicwand}$ (introduced in Proposition~\ref{proposition-prop-sl}).

\begin{lemma}
\label{lemma-atomic-formulae-test}
Let $\alpha, q \geq 1$, $i,j \in \interval{1}{q}$. Any atomic formula among 
$\reachplus(\avariable_i, \avariable_j)$, $\emp$ and $\size \geq \beta$ (with $\beta \leq \alpha$),
is equivalent to a Boolean combination of test formulae from  $\Test(q, \alpha)$.
\end{lemma}

\begin{proof}
Let $q,\alpha \geq 1$ and consider the family of test formulae $\Test(q,\alpha)$. Let $\pair{\astore}{\aheap}$ be a memory state and
$
\supportgraph{\astore,\aheap} = (\gverts,\gedges,\galloc,\glabels,\gbtw,\grem)
$
be its support graph.

\begin{itemize}
\item
We show that $\emp$ is equivalent to the Boolean combination of test formulae ${\lnot \sizeothers_q \geq 1} \land \bigwedge_{i \in [1,q]} \lnot \alloc{\avariable_i}$.
Suppose $\pair{\astore}{\aheap} \models \emp$.
Then trivially, for all $i\in[1,q]$ $\pair{\astore}{\aheap} \models \lnot \alloc{\avariable_i}$ and $\pair{\astore}{\aheap} \models {\lnot \sizeothers_q \geq 1}$.

Conversely, suppose $\pair{\astore}{\aheap} \models {\lnot \sizeothers_q \geq 1} \land \bigwedge_{i \in [1,q]} \lnot \alloc{\avariable_i}$.
As for all $i \in [1,q]$ $\avariable_i$ does not correspond to a location in $\domain{\aheap}$, it holds that $\pair{\astore}{\aheap} \models \bigwedge_{\aterm,\aterm' \in \Terms{q}} \lnot \sees_q(\aterm,\aterm') \geq 1$.
Therefore, directly from its definition, $\grem = \domain{\aheap}$. The emptiness of $\domain{\aheap}$ then follows from $\pair{\astore}{\aheap} \models {\lnot \sizeothers_q \geq 1}$.

\item 
The formula $\reachplus(\avariable_i, \avariable_j)$ can be shown equivalent to
$$
\aformula \egdef \bigvee_{\mathclap{\substack{\aterm_1, \ldots, \aterm_n \in \Terms{q}, \\ 
{\rm pairwise \ distinct} \ \aterm_1, \ldots, \aterm_{n-1}, \\ \avariable_i = \aterm_1, \avariable_j = \aterm_n}}}
\ \ \ \ \ \ \ \ \ \ \ \bigwedge_{\delta \in \interval{1}{n-1}} \sees_q(\aterm_{\delta},\aterm_{\delta+1}) \geq 1.
$$
In the expression above, $n$ is arbitrary in the disjunction as soon as the other constraints
are satisfied (whence, $n \leq \card{ \Terms{q}}$). 
First, suppose $\pair{\astore}{\aheap} \models \reachplus(\avariable_i,\avariable_j)$. Then, there exists $L \geq 1$ such that $\aheap^L(\astore(\avariable_i)) = \astore(\avariable_j)$. Let $\alocation_1 = \astore(\avariable_i)$, $\alocation_n = \astore(\avariable_j)$ and let $\alocation_2,\dots,\alocation_{n-1} \in \Labels{q}{\astore,\aheap}$ be all labelled locations in the (minimal) path witnessing $\pair{\astore}{\aheap} \models \reachplus(\avariable_i,\avariable_j)$,
such that $\aheap^{L_\delta}(\alocation_\delta) = \alocation_{\delta+1}$, where $L_\delta \geq 1$, and there are no labelled locations between $\alocation_\delta$ and $\alocation_{\delta+1}$, $1 \leq \delta \leq n-1 \leq \card{\Labels{q}{\astore,\aheap}}$. The path from $\astore(\avariable_i)$ to $\astore(\avariable_j)$ is therefore uniquely split into $n-1$ 
subpaths starting and ending with a labelled location and without labelled locations in between.
Let $\aterm_1 = \avariable_i$, $\aterm_n = \avariable_j$ and $\aterm_2,\dots,\aterm_{n-1} \in \Terms{q}$ be such that
$\sem{\aterm_\delta}^q_{\astore,\aheap} = \alocation_\delta$ for $\delta \in \interval{2}{n-1}$.
From the definition of $\alocation_2,\dots,\alocation_{n-1}$, we conclude that
$\pair{\astore}{\aheap}$ satisfies $\bigwedge_{1 \leq \delta \leq n-1} \sees_q(\aterm_{\delta},\aterm_{\delta+1}) \geq 1$.
The minimality of the path is needed to guarantee that all the $v_i$'s are distinct. 

\cut{
Conversely, suppose
\[
\pair{\astore}{\aheap} \models \bigvee_{\mathclap{\substack{\aterm_1, \ldots, \aterm_n \in \Terms{q}, \\
{\rm pairwise \ distinct} \ \aterm_1, \ldots, \aterm_{n-1}, \\ \avariable_i = \aterm_1, \avariable_j = \aterm_n}}}
\ \ \ \ \ \ \ \ \ \ \  \bigwedge_{\delta \in \interval{1}{n-1}} \sees_q(\aterm_{\delta},\aterm_{\delta+1}) \geq 1.
\]
}
Conversely, suppose $\pair{\astore}{\aheap} \models \aformula$. 
Then, there are $\aterm_1,\dots,\aterm_n \in \Terms{q}$ such that $\aterm_1 = \avariable_i$, $\aterm_n = \avariable_j$,
$\aterm_1$, \ldots, $\aterm_{n-1}$ are pairwise distinct 
 and for all ${\delta \in [1,n-1]}$ $\pair{\astore}{\aheap} \models \sees_q(\aterm_{\delta},\aterm_{\delta+1}) \geq 1$.
From the semantics of $\sees$, this implies that there are $L_1,\dots,L_{n-1} \geq 1$ such that $\aheap^{L_\delta}(\sem{\aterm_\delta}^q_{\astore,\aheap}) = \sem{\aterm_{\delta
+1}}^q_{\astore,\aheap}$ for all $\delta \in \interval{1}{n-1}$.
Therefore $\aheap(\astore(\avariable_i))^{\sum_\delta L_\delta} = \astore(\avariable_j)$, where $\sum_\delta L_\delta \geq 1$ and $\reachplus(\avariable_i, \avariable_j)$ is satisfied.

\item
We define ${\size \geq \beta}$, where $\beta \leq \alpha$.
First of all, we introduce the formula $\sizeV_q(\aterm) \geq \beta$,
with $\aterm \in \Terms{q}$, below:
\begin{mathpar}
  \sizeV_q(\aterm) \geq 0 \egdef \true

  \sizeV_q(\aterm) \geq 1 \egdef \alloc{v}

  \sizeV_q(\aterm) \geq \beta + 1 \egdef \bigvee_{\mathclap{\aterm' \in \Terms{q}}} \sees_q(\aterm,\aterm') \geq \beta + 1\qquad\text{ for } \beta \in \interval{1}{\alpha}
\end{mathpar}
$\sizeV_q(\aterm) \geq 0$ is always true,
$\sizeV_q(\aterm) \geq 1$ holds in a memory state $\pair{\astore}{\aheap}$ if and only if
$\pair{\astore}{\aheap} \models \alloc{v}$, whereas $\sizeV_q(\aterm) \geq \beta + 1$ holds if and only if
there exists $\aterm' \in \Terms{q}$ such that  $\pair{\astore}{\aheap} \models \sees_q(\aterm,\aterm') \geq \beta + 1$.
As such, $\sizeV_q(\aterm) \geq \beta+1$ holds whenever for all $L \in \interval{0}{\beta}$, $\aheap^L(\sem{\aterm}^q_{\astore,\aheap}) \in \domain{\aheap}$, 
$L \geq 1$ implies $\aheap^L(\sem{\aterm}^q_{\astore,\aheap}) \not\in \Labels{q}{\astore,\aheap}$, and there is $L' > \beta$ such that $\aheap^{L'}(\sem{\aterm}^q_{\astore,\aheap}) = \sem{\aterm'}^q_{\astore,\aheap}$, for some~$\aterm' \in \Terms{q}$. 
From the definition of $\sees$ and meet-points, it follows that the locations considered for the satisfaction of $\sizeV_q(\aterm) \geq \beta$ do not play any role in the satisfaction of $\sizeV_q(\aterm') \geq \beta'$, where $\sem{\aterm}^q_{\astore,\aheap} \neq \sem{\aterm'}^q_{\astore,\aheap}$.
Therefore, it is easy to prove that if $\pair{\astore}{\aheap} \models {\sizeV_q(\aterm) \geq \beta} \land \sizeV_q(\aterm') \geq \beta' \land \aterm \neq \aterm'$ then $\card{\domain{\aheap}} \geq \beta + \beta'$.
Similarly, the locations considered for the satisfaction of this formula are not in 
$\grem$
and therefore $\pair{\astore}{\aheap} \models {\sizeV_q(\aterm) \geq \beta} \land \sizeothers_q \geq \beta'$ implies $\card{\domain{\aheap}} \geq \beta + \beta'$. We can then use this formula to define $\size \geq \beta$ as follows, where we write $\sizeothers_q \geq 0$ instead of $\avariable_1 = \avariable_1$ (any tautological Boolean combination
of test formulae would be fine).
\[
  \bigvee_{\mathclap{\substack{V \subseteq \Terms{q}\\
  \beta \leq \beta_R + \sum_{\aterm \in V} \beta_{\aterm}\\ \quad\beta_R \in \interval{0}{\alpha}\ \forall \aterm \in V \beta_{\aterm} \in \interval{0}{\alpha+1}}}}
 \Big(\sizeothers_q \geq \beta_R \land \textstyle \bigwedge_{\aterm \in V}( \sizeV_q(\aterm) \geq \beta_{\aterm} \land \bigwedge_{v'\in V\setminus\{v\}} v \not= v')\Big)
\]
Suppose that this formula is satisfied by $\pair{\astore}{\aheap}$. Then there exists a subset of terms $V$ such that for all $\aterm,\aterm' \in V$, if $\aterm \neq \aterm'$ then $\sem{\aterm}^q_{\astore,\aheap} \neq \sem{\aterm'}^q_{\astore,\aheap}$ also follows (from the last conjunct of the formula). From the property just stated about $\sizeV_q(\aterm) \geq \beta$, it must therefore hold that $\card{\domain{\aheap}} \geq \beta_R + \sum_{\aterm \in V} \beta_\aterm \geq \beta$.

Conversely, suppose $\card{\domain{\aheap}} \geq \beta$.
We can define a partition $\aheap_R + \sum_{\alocation \in \Labels{q}{\astore,\aheap}} \aheap_\alocation = \aheap$ 
such that each subheap $\aheap_\alocation$ of the partition  contains exactly the locations of $\domain{\aheap}$ considered for the satisfaction of the
test formulae $\sizeV_q(\aterm) \geq \beta'$, $\beta' \in [0,\alpha+1]$,
for a specific labelled location $\alocation = \sem{\aterm}^q_{\astore,\aheap}$,
whereas $\aheap_R$ contains all the locations considered for the satisfaction of $\sizeothers_{q} \geq \beta'$, $\beta' \in [0,\alpha]$.
Consider a representative $\aterm \in \Terms{q}$ for each subheap $\aheap_\alocation$, where $\sem{\aterm}^q_{\astore,\aheap} = \alocation$. Let $V$ be the set of these representatives and let
$\beta_R = \min(\alpha,\card{\domain{h_R}})$, $\beta_\aterm = \min(\alpha+1,\card{\domain{\aheap_{\sem{\aterm}^q_{\astore,\aheap}}}})$.
Since $\beta \leq \alpha$ and $\card{\domain{\aheap}} = \card{\domain{\aheap_R}} + \sum_{\aterm \in V}  \geq \beta$, it follows that $\beta_R + \sum_{\aterm \in V} \beta_v \geq \beta$.
Moreover, for each $\aterm \in V$, $\beta_\aterm \in [0,\alpha+1]$, whereas $\beta_R \in [0,\alpha]$.
From the definition of $V$, it immediately holds that $\pair{\astore}{\aheap} \models \bigwedge_{\aterm\neq\aterm' \in V} \aterm \neq \aterm'$.
Lastly, from their definition, it holds that $\pair{\astore}{\aheap} \models \sizeothers_q \geq \beta_R$ and for all $\aterm \in V$ $\pair{\astore}{\aheap} \models \sizeV_q(\aterm) \geq \beta_{\aterm}$.
We conclude that the formula defining $\size \geq \beta$ is satisfied.\qedhere
\end{itemize}
\end{proof}

\subsection{Expressive power and the small heap property}

Now, we show that the sets of test formulae $\Test(q,\alpha)$ are sufficient to capture the expressive power of
$\seplogic{\separate,\reachplus}$ (as shown below, Theorem~\ref{theorem-quantifier-elimination}) and deduce the small
heap property of this logic (Theorem~\ref{theorem-finite-heap-property}).
We introduce an indistinguishability relation $\approx_\alpha^q$
between memory states based on test formulae, see analogous relations
in~\cite{Lozes04bis,David&Kroening&Lewis15,Demrietal17}.

\begin{definition}
Let  $q,\alpha \geq 1$ and, $\pair{\astore}{\aheap}$ and $\pair{\astore'}{\aheap'}$ be memory states.
 $\pair{\astore}{\aheap} \approx_\alpha^q \pair{\astore'}{\aheap'}$
$\equivdef$ $\pair{\astore}{\aheap} \models \aformulabis$ iff $\pair{\astore'}{\aheap'} \models \aformulabis$, for all $\aformulabis \in \Test(q,\alpha)$.
\end{definition}

Forthcoming Theorem~\ref{theo:test_sl_equiv} states that
if $\pair{\astore}{\aheap} \approx_\alpha^q \pair{\astore'}{\aheap'}$, then the  two memory states cannot be distinguished
by formulae whose syntactic resources are bounded in some way by $q$ and $\alpha$
(details will follow).
The following technical lemma lifts the relationship $\approx^q_\alpha$ to an equivalence between support graphs, consolidating this idea of indistinguishable memory states.

\begin{lemma}\label{lemma-test-graph}
  Let $q,\alpha \geq 1$, and $\pair{\astore}{\aheap}$, $\pair{\astore'}{\aheap'}$ be two memory states with support graphs respectively
$\supportgraph{\astore,\aheap} = (\gverts,\gedges,\galloc,\glabels,\gbtw,\grem)$
 and
$\supportgraph{\astore',\aheap'} = (\gverts',\gedges',\galloc',\glabels',\gbtw',\grem')$.
We have $\pair{\astore}{\aheap} \approx^q_\alpha \pair{\astore'}{\aheap'}$ iff there is a map $\amap:  \gverts \to \gverts'$
such that
  \begin{enumerate}[label=\normalfont{\textbf{(A\arabic*)}}]
  \item\label{a1} $\amap$ is a graph isomorphism between $\pair{\gverts}{\gedges}$ and $\pair{\gverts'}{\gedges'}$;
  \item\label{a2} for all $\alocation \in \gverts$, we have $\alocation \in \galloc$ iff  $\amap(\alocation) \in \galloc'$;
  \item\label{a3} for all $\alocation \in \gverts$, we have $\glabels(\alocation) = \glabels'(\amap(\alocation))$;
  \item\label{a4} for all  $\pair{\alocation}{\alocation'} \in \gedges$, we have
       $\min(\alpha,\card{\gbtw(\alocation,\alocation')}) = \min(\alpha,\card{\gbtw'(\amap(\alocation),\amap(\alocation'))})$;
   \item\label{a5}  $\min(\alpha, \card{\grem}) = \min(\alpha, \card{\grem'})$.
 \end{enumerate}
\end{lemma}

\begin{proof}
Suppose $\pair{\astore}{\aheap} \approx^q_\alpha \pair{\astore'}{\aheap'}$.
Let $\amap: \gverts \to \gverts'$ be the map such that for all locations $\alocation \in \gverts$,
we have $\amap(\alocation) \egdef \alocation'$ if and only if there is $\aterm \in \Terms{q}$ such that
$\sem{\aterm}^q_{\astore,\aheap} = \alocation$ and  $\sem{\aterm}^q_{\astore',\aheap'} = \alocation'$.
Let us show that $\amap$ is well-defined. To do so, assume that there are $\aterm,\aterm'$ such that
$\sem{\aterm}^q_{\astore,\aheap} = \sem{\aterm'}^q_{\astore,\aheap} = \alocation$,
$\sem{\aterm}^q_{\astore',\aheap'} = \alocation'$
and $\sem{\aterm'}^q_{\astore',\aheap'} = \alocation''$.
Since $\pair{\astore}{\aheap} \approx^q_\alpha \pair{\astore'}{\aheap'}$, we have that
$\pair{\astore}{\aheap} \models \aterm = \aterm'$ iff $\pair{\astore'}{\aheap'} \models \aterm = \aterm'$.
Therefore,
($\sem{\aterm}^q_{\astore,\aheap}$, $\sem{\aterm'}^q_{\astore,\aheap}$ are defined and
$\sem{\aterm}^q_{\astore,\aheap} = \sem{\aterm'}^q_{\astore,\aheap}$) iff ($\sem{\aterm}^q_{\astore',\aheap'}$, $\sem{\aterm'}^q_{\astore',\aheap'}$ are defined
and
$\sem{\aterm}^q_{\astore',\aheap'} = \sem{\aterm'}^q_{\astore',\aheap'}$). Consequently, $\alocation' = \alocation''$ and $\amap$ is well-defined.
Actually, the equivalence above induced by $\pair{\astore}{\aheap} \approx^q_\alpha \pair{\astore'}{\aheap'}$ allows us to show in a similar way
that $\amap$ is a bijection from $\gverts$ to $\gverts'$ and that the condition~\ref{a3} holds true. Indeed, the statements below are equivalent:
\begin{itemize}

\item $\set{\aterm,\aterm'} \subseteq \glabels(\alocation)$,
\item $\pair{\astore}{\aheap} \models \aterm = \aterm'$ and $\sem{\aterm}^q_{\astore,\aheap} = \sem{\aterm'}^q_{\astore,\aheap} = \alocation$
      (by definition of $\models$ and $\supportgraph{\astore,\aheap}$),
\item $\pair{\astore'}{\aheap'} \models \aterm = \aterm'$ and $\sem{\aterm}^q_{\astore',\aheap'} = \sem{\aterm'}^q_{\astore',\aheap'} = \alocation'$
      for some $\alocation'$
      (by $\pair{\astore}{\aheap} \approx^q_\alpha \pair{\astore'}{\aheap'}$),
\item $\pair{\astore'}{\aheap'} \models \aterm = \aterm'$ and $\sem{\aterm}^q_{\astore',\aheap'} = \sem{\aterm'}^q_{\astore',\aheap'} = \amap(\alocation)$
      (by definition of $\amap$),
\item $\set{\aterm,\aterm'} \subseteq \glabels'(\amap(\alocation))$
      (by definition of $\supportgraph{\astore',\aheap'}$).
\end{itemize}
Consequently,  the condition~\ref{a3} holds true. A similar reasoning allows us to establish~\ref{a2}
and it is omitted below. 
\cut{
Indeed,
the statements below are equivalent:
\begin{itemize}

\item $\alocation \in \galloc$,
\item $\pair{\astore}{\aheap} \models \alloc{\aterm}$  for some $\aterm$ such that $\sem{\aterm}^q_{\astore,\aheap} = \alocation$
      (by definition of $\models$ and $\supportgraph{\astore,\aheap}$),
\item $\pair{\astore'}{\aheap'} \models \alloc{\aterm}$ for some $\aterm$ such that $\sem{\aterm}^q_{\astore',\aheap'} = \alocation'$
      for some $\alocation'$
      (by $\pair{\astore}{\aheap} \approx^q_\alpha \pair{\astore'}{\aheap'}$),
\item $\pair{\astore'}{\aheap'} \models \alloc{\aterm}$ for some $\aterm$ such that $\sem{\aterm}^q_{\astore',\aheap'} = \amap(\alocation)$
      (by definition of $\amap$),
\item $\amap(\alocation) \in \galloc'$
      (by definition of $\supportgraph{\astore',\aheap'}$).
\end{itemize}
}
In order to conclude the first part of the proof, first we show~\ref{a5} and then we focus on~\ref{a1} and~\ref{a4}.
Let us first establish~\ref{a5}. As $\pair{\astore}{\aheap} \approx^q_\alpha \pair{\astore'}{\aheap'}$, we have
($\dag$) for all $\beta \in \interval{1}{\alpha}$, $\pair{\astore}{\aheap} \models \sizeothers_q \geq \beta$
iff $\pair{\astore'}{\aheap'} \models \sizeothers_q \geq \beta$ and ($\dag$) is equivalent to the statements below:
\begin{itemize}
\item  for all $\beta \in \interval{1}{\alpha}$, $\card{\grem} \geq \beta$ iff
      $\card{\grem'} \geq \beta$ (by definition of $\models$),
\item $\min(\alpha, \card{\grem}) = \min(\alpha, \card{\grem'})$ (by a simple arithmetical reasoning).
\end{itemize}
Consequently,  the condition~\ref{a5} holds true. Now, let us show~\ref{a1} and~\ref{a4}.
First, the statements below are equivalent:
\begin{itemize}
\item $\pair{\alocation}{\alocation'} \in \gedges$ and $\gbtw(\alocation,\alocation') = \emptyset$,
\item $\pair{\astore}{\aheap} \models \aterm \hpto \aterm'$, $\sem{\aterm}^q_{\astore,\aheap} = \alocation$ and $\sem{\aterm'}^q_{\astore,\aheap} = \alocation'$
      for some $\aterm, \aterm' \in \Terms{q}$
      (by definition of $\models$),
\item  $\pair{\astore'}{\aheap'} \models \aterm \hpto \aterm'$, $\sem{\aterm}^q_{\astore',\aheap'} = \alocation''$ and $\sem{\aterm'}^q_{\astore',\aheap'} = 
      \alocation'''$
       for some $\aterm, \aterm' \in \Terms{q}$, for some locations $\alocation''$, $\alocation'''$,
      (by $\pair{\astore}{\aheap} \approx^q_\alpha \pair{\astore'}{\aheap'}$),
\item  $\pair{\astore'}{\aheap'} \models \aterm \hpto \aterm'$, $\sem{\aterm}^q_{\astore',\aheap'} = \amap(\alocation)$ and $\sem{\aterm'}^q_{\astore',\aheap'} =
       \amap(\alocation')$ for some $\aterm, \aterm' \in \Terms{q}$
      (by definition of $\amap$),
\item $\pair{\amap(\alocation)}{\amap(\alocation')} \in \gedges'$ and $\gbtw'(\amap(\alocation),\amap(\alocation')) = \emptyset$
       (by definition of $\models$ and $\supportgraph{\astore',\aheap'}$).
\end{itemize}
Similarly, we have the following equivalences, where $\beta \in \interval{1}{\alpha}$:
\begin{itemize}

\item $\pair{\alocation}{\alocation'} \in \gedges$ and $\card{\gbtw(\alocation,\alocation')} \geq \beta$,
\item $\pair{\astore}{\aheap} \models \sees_q(\aterm,\aterm') \geq \beta + 1$,
       $\sem{\aterm}^q_{\astore,\aheap} = \alocation$, $\sem{\aterm'}^q_{\astore,\aheap} = \alocation'$
      for some $\aterm, \aterm' \in \Terms{q}$
      (by definition of $\models$),
\item  $\pair{\astore'}{\aheap'} \models \sees_q(\aterm,\aterm') \geq \beta + 1$,
       $\sem{\aterm}^q_{\astore',\aheap'} = \alocation''$ and $\sem{\aterm'}^q_{\astore',\aheap'} = \alocation'''$
       for some $\aterm, \aterm' \in \Terms{q}$, for some locations $\alocation''$, $\alocation'''$,
      (by $\pair{\astore}{\aheap} \approx^q_\alpha  \pair{\astore'}{\aheap'}$),
\item  $\pair{\astore'}{\aheap'} \models \sees_q(\aterm,\aterm')\!\geq\!\beta\!+\!1$,
        $\sem{\aterm}^q_{\astore',\aheap'}\!= \amap(\alocation)$, $\sem{\aterm'}^q_{\astore',\aheap'}\!=
       \amap(\alocation')$ for some $\aterm, \aterm'\!\in\!\Terms{q}$,
      (by $\amap$),
\item $\pair{\amap(\alocation)}{\amap(\alocation')} \in \gedges'$ and $\card{\gbtw'(\amap(\alocation),\amap(\alocation'))} \geq \beta$
       (by definition of $\models$ and $\supportgraph{\astore',\aheap'}$).
\end{itemize}
Consequently, we get~\ref{a1} and~\ref{a4}.
We omit below the proof of the other direction as it is similar to the first direction.
\cut{
Let us show the other direction. Suppose $\amap$ is a map satisfying \ref{a1}--\ref{a5}.
\begin{itemize}
\item Let us consider the test formula $\aterm = \aterm'$. The statements below are equivalent:
\begin{itemize}
\item $\pair{\astore}{\aheap} \models \aterm = \aterm'$,
\item $\sem{\aterm}^q_{\astore,\aheap}, \sem{\aterm'}^q_{\astore,\aheap} \in V$ and
 $\set{\aterm, \aterm'} \subseteq \glabels(\sem{\aterm}^q_{\astore,\aheap})$ (by definition of $\models$ and $\supportgraph{\astore,\aheap}$),
\item $\amap(\sem{\aterm}^q_{\astore,\aheap}), \amap(\sem{\aterm'}^q_{\astore,\aheap}) \in V'$ and
 $\set{\aterm, \aterm'} \subseteq \glabels(\amap(\sem{\aterm}^q_{\astore,\aheap}))$ (by~\ref{a1}~and~\ref{a3}),
\item $\sem{\aterm}^q_{\astore',\aheap'}, \sem{\aterm'}^q_{\astore',\aheap'} \in V'$ and
 $\set{\aterm, \aterm'} \subseteq \glabels(\sem{\aterm}^q_{\astore',\aheap'})$ (by~\ref{a1}~and~\ref{a3}),
\item $\pair{\astore'}{\aheap'} \models \aterm = \aterm '$  (by definition of $\models$ and $\supportgraph{\astore',\aheap'}$).
\end{itemize}
\item Let us consider the test formula $\alloc{\aterm}$. The statements below are equivalent:
\begin{itemize}
\item $\pair{\astore}{\aheap} \models \alloc{\aterm}$,
\item $\sem{\aterm}^q_{\astore,\aheap} \in \galloc$ and $\aterm \in \glabels(\sem{\aterm}^q_{\astore,\aheap})$
      (by definition of $\models$ and $\supportgraph{\astore,\aheap}$),
\item $\amap(\sem{\aterm}^q_{\astore,\aheap}) \in \galloc'$ and $\aterm \in \glabels'(\amap(\sem{\aterm}^q_{\astore,\aheap}))$
      (by~\ref{a2} and~\ref{a3}),
\item $\sem{\aterm}^q_{\astore',\aheap'} \in \galloc'$ and $\aterm \in \glabels'(\sem{\aterm}^q_{\astore',\aheap'})$
      (by definition of $\amap$),
\item $\pair{\astore'}{\aheap'} \models \alloc{\aterm}$ (by definition of $\models$ and $\supportgraph{\astore',\aheap'}$).
\end{itemize}
\item Let us consider the test formula $\aterm \hpto \aterm'$. The statements below are equivalent:
\begin{itemize}
\item $\pair{\astore}{\aheap} \models \aterm \hpto \aterm'$,
\item $\alocation = \sem{\aterm}^q_{\astore,\aheap}$, $\alocation' =
\sem{\aterm'}^q_{\astore,\aheap}
 \in V$ and $\pair{ \alocation}{ \alocation'} \in \gedges$ and
        $\card{\gbtw( \alocation, \alocation')} = 0$,
      (by definition of $\models$ and $\supportgraph{\astore,\aheap}$),
\item  $\amap(\alocation) = \amap(\sem{\aterm}^q_{\astore,\aheap})$, $\amap(\alocation') =
\amap(\sem{\aterm'}^q_{\astore,\aheap})
 \in V'$ and $\pair{\amap(\alocation)}{\amap(\alocation')} \in \gedges'$ and \\
        $\card{\gbtw'(\amap(\alocation), \amap(\alocation'))} = 0$,
      (by~\ref{a1},~\ref{a4}),
\item $\amap(\alocation) = \sem{\aterm}^q_{\astore',\aheap'}$, $\amap(\alocation') =
\sem{\aterm'}^q_{\astore',\aheap'}
 \in V'$ and $\pair{\amap(\alocation)}{\amap(\alocation')} \in \gedges'$ and \\
        $\card{\gbtw'(\amap(\alocation), \amap(\alocation'))} = 0$,
      (by~\ref{a3}),
\item  $\pair{\astore'}{\aheap'} \models \aterm \hpto \aterm'$,
 (by definition of $\models$ and $\supportgraph{\astore',\aheap'}$).
\end{itemize}
\item Let us consider the test formula $\sees_q(\aterm,\aterm') \geq \beta +1$. The statements below are equivalent:
\begin{itemize}
\item $\pair{\astore}{\aheap} \models \sees_q(\aterm,\aterm') \geq \beta +1$,
\item $\alocation = \sem{\aterm}^q_{\astore,\aheap}$,
 $\alocation' = \sem{\aterm'}^q_{\astore,\aheap} \in V$  and
 $\pair{ \alocation}{ \alocation'} \in \gedges$  and $\card{\gbtw( \alocation, \alocation')} \geq \beta$
      (by definition of $\models$ and $\supportgraph{\astore,\aheap}$),
\item  $\amap(\alocation) = \amap(\sem{\aterm}^q_{\astore,\aheap})$,
 $\amap(\alocation') = \amap(\sem{\aterm'}^q_{\astore,\aheap}) \in V$  and
 $\pair{\amap(\alocation)}{\amap(\alocation')} \in \gedges'$  and \\ $\card{\gbtw'(\amap(\alocation),\amap(\alocation'))} \geq \beta$
(by~\ref{a1},~\ref{a4}),
\item $\amap(\alocation) = \sem{\aterm}^q_{\astore',\aheap'}$, $\amap(\alocation') =
\sem{\aterm'}^q_{\astore',\aheap'}
 \in V'$ and $\pair{\amap(\alocation)}{\amap(\alocation')} \in \gedges'$ and \\
        $\card{\gbtw'(\amap(\alocation), \amap(\alocation'))} \geq \beta$,
      (by~\ref{a3}),
\item  $\pair{\astore'}{\aheap'} \models \sees_q(\aterm,\aterm') \geq \beta +1$,
 (by definition of $\models$ and $\supportgraph{\astore',\aheap'}$).
\end{itemize}
\item Let us consider the test formula $\sizeothers_q \geq \beta$. The statements below are equivalent:
\begin{itemize}
\item $\pair{\astore}{\aheap} \models \sizeothers_q \geq \beta$,
\item $\card{\grem} \geq \beta$
      (by definition of $\models$ and $\supportgraph{\astore,\aheap}$),
\item $\card{\grem'} \geq \beta$
      (by~\ref{a5}),
\item  $\pair{\astore'}{\aheap'} \models \sizeothers_q \geq \beta$,
 (by definition of $\models$ and $\supportgraph{\astore',\aheap'}$).\qedhere
\end{itemize}
\end{itemize}
}
\end{proof}

Now, we  state the key intermediate result of the section that can be viewed as a distributivity lemma.
The expressive  power of the test formulae allows us to mimic the separation between two equivalent
memory states with respect to the relation $\approx^q_\alpha$. Separating conjunction
can therefore be eliminated from the logic in favour of test formulae, which is essential in
the proof of Theorem~\ref{theo:test_sl_equiv}.

\begin{lemma}\label{lemma-star}
 Let $q,\alpha,\alpha_1,\alpha_2 \geq 1$ with $\alpha = \alpha_1 + \alpha_2$
and $\pair{\astore}{\aheap}$, $\pair{\astore'}{\aheap'}$ be  memory states such that $\pair{\astore}{\aheap} \approx^q_\alpha
\pair{\astore'}{\aheap'}$.
 For all heaps $\aheap_1$, $\aheap_2$ such that $\aheap = \aheap_1 + \aheap_2$, there are heaps
$\aheap'_1$, $\aheap'_2$ such that $\aheap' = \aheap'_1 + \aheap'_2$,
  $\pair{\astore}{\aheap_1}  \approx^q_{\alpha_1} \pair{\astore'}{\aheap'_1}$ and
$\pair{\astore}{\aheap_2}  \approx^q_{\alpha_2} \pair{\astore'}{\aheap'_2}$.
\end{lemma}

The proof of Lemma~\ref{lemma-star} is rather long as it first constructs the
subheaps $\aheap'_1$ and $\aheap_2'$, and then check that
$\pair{\astore}{\aheap_i}  \approx^q_{\alpha_i} \pair{\astore'}{\aheap'_i}$ holds (for both $i \in \{1,2\}$)
by verifying the conditions~\ref{a1}--\ref{a5}.
Moreover, the verification of some of the conditions requires first to establish
additional preliminary properties. Though the principle of the proof structure
is quite simple, checking carefully each property is quite lengthy. We believe that having
all the material in a single proof is helpful to emphasize  the proof structure,
or alternatively to skip the proof in a first reading of the document (see also Appendix~\ref{appendix-star}).

\begin{proof}
  Let $q$, $\alpha$, $\alpha_1$, $\alpha_2$, $\pair{\astore}{\aheap}$, $\pair{\astore'}{\aheap'}$, $\aheap_1$ and $\aheap_2$ be
  defined as in the statement. Let
  \begin{itemize}
  \item $\supportgraph{\astore,\aheap} = (\gverts,\gedges,\galloc,\glabels,\gbtw,\grem)$ and
  \item $\supportgraph{\astore',\aheap'} = (\gverts',\gedges',\galloc',\glabels',\gbtw',\grem')$
  \end{itemize}
  be the support graphs
of $\pair{\astore}{\aheap}$ and $\pair{\astore'}{\aheap'}$ respectively, with respect to $q$.
  As $\pair{\astore}{\aheap} \approx^q_\alpha
  \pair{\astore'}{\aheap'}$, let $\amap: \gverts \to \gverts'$ be a map satisfying \ref{a1}--\ref{a5} from
   Lemma~\ref{lemma-test-graph}.
  Below, for the sake of conciseness, let $k \in \{1,2\}$ and let
  $\supportgraph{\astore,\aheap_k} = (\gverts_k,\gedges_k,\galloc_k,\glabels_k,\gbtw_k,\grem_k)$
  be the support graph of $\pair{\astore}{\aheap_k}$.

  The proof is rather long and can be summed up with the following steps.
  \begin{enumerate}
  \item First, we define a strategy to split $\aheap'$ into $\aheap_1'$ and $\aheap_2'$ by closely following the way that $\aheap$ is split into $\aheap_1$ and $\aheap_2$. To do this, we look at the support graphs. For instance, suppose that $\grem$ is split into two sets $R_1 \subseteq \domain{\aheap_1}$ and $R_2 \subseteq \domain{\aheap_2}$.
  By definition, it is quite easy to see that the sets $R_1$ and $R_2$ must be subsets of $\grem_1$ and $\grem_2$, respectively. 
Then, following Lemma~\ref{lemma-test-graph}, to obtain
  $\pair{\astore}{\aheap_k} \approx^q_{\alpha_k} \pair{\astore'}{\aheap_k'}$,
  we are required to split $\grem'$ into
  $R_1' \subseteq \domain{\aheap_1'}$ and
  $R_2' \subseteq \domain{\aheap_2'}$ so that $\min(\alpha_k, \card{R_k}) = \min(\alpha_k, \card{R_k'})$. Indeed, otherwise the equisatisfaction of
  the test formulae of the form $\sizeothers_q \geq \beta$ is not ensured (details on this are formalised later).
  \item After defining $\aheap_1'$ and $\aheap_2'$, we show that $\pair{\astore}{\aheap_k} \approx^q_{\alpha_k} \pair{\astore'}{\aheap_k'}$.
  To do so, again, we  follow Lemma~\ref{lemma-test-graph} and we show that we can find suitable bijections from labelled locations of $\aheap_k$ to the ones of $\aheap_k'$ satisfying \ref{a1}--\ref{a5}.
  \end{enumerate}
  According to the summary above, let us first define explicitly $\aheap_1'$ and $\aheap_2'$ via an iterative process that consists in adding locations
  to $\domain{\aheap'_1}$ or to $\domain{\aheap'_2}$. Whenever we enforce that $\alocation \in \domain{\aheap_k'}$,
  implicitly we have $\aheap'_k(\alocation) \egdef \aheap'(\alocation)$ as $\aheap'_k$ is intended to be a subheap of $\aheap'$.
  \begin{description}
    \item[\descLab{(CA)}{c1-body}] For all $\alocation \in \galloc'$, $\alocation \in \domain{\aheap_k'} \equivdef \amap^{-1}(\alocation) \in \domain{\aheap_k}$.
    This step of the construction, as well as its usefulness, should be self-explanatory.
    For example, if $\pair{\astore}{\aheap_k} \models \alloc{\avariable_i}$ then, by relying on \ref{a3}, this step allows us to conclude that $\pair{\astore'}{\aheap_k'} \models \alloc{\avariable_i}$ (independently on how the definition of $\aheap_k'$ will be completed in the next steps of the construction).
    \item[\descLab{(CR)}{c3-body}] The heaps $\aheap_1'$ and $\aheap_2'$ are further populated depending on $\grem$. Let $R_k = \grem \cap \domain{\aheap_k}$.
    By definition, we have $R_1 \uplus R_2 = \grem$.
    Below, we partition $\grem'$ into two sets $R_1'$ and $R_2'$ so that by definition $R_k' \subseteq \domain{\aheap_k'}$.
    The strategy for defining the partition is split into three cases:
    \begin{center}
    \begin{enumerate}[align=left]
    \item[\descLab{(CR.C1)}{c31-body}] If $\card{R_1} < \alpha_1$ then $R_1'$ is a set of $\card{R_1}$ locations from $\grem'$ and $R_2' \egdef \grem' \setminus R_1'$.
    \item[\descLab{(CR.C2)}{c32-body}] Otherwise, if $\card{R_2} < \alpha_2$ then $R_2'$ is a set of 
    $\card{R_2}$ locations of $\grem'$ and $R_1' \egdef \grem' \setminus R_2'$.
    \item[\descLab{(CR.C3)}{c33-body}] Otherwise we have  $\card{R_1} \geq \alpha_1$ and $\card{R_2} \geq \alpha_2$. Then, $R_1'$ is a set of $\alpha_1$ locations
    from $\grem'$ and
    $R_2' \egdef \grem' \setminus R_1'$.
    \end{enumerate}
    \end{center}
    It is easy to show that the construction satisfies the following property (where $k \in \{1,2\}$):
    \begin{equation}
      \tag*{\textbf{(CR.P1)}}
      \label{crp2-body} \min(\alpha_k,\card{R_k}) = \min(\alpha_k,\card{R_k'})
    \end{equation}

    The property \ref{crp2-body} directly follows from the
     property~\ref{a5} satisfied by $\amap$. The proof (that can be applied also for the next step of the construction, see~\ref{cip2-body}), works as follows.

  First, suppose that the sets of remaining locations in the heap domain are small, i.e.
     \[
     \min(\alpha,\card{\grem}) = \min(\alpha,\card{\grem'}) < \alpha_1 + \alpha_2.
     \]
     So, $\card{\grem} = \card{\grem'}$ and therefore $\card{R_1} + \card{R_2} = \card{R_1'} + \card{R_2'}$.
    By definition, $\card{R_1} = \card{R_1'}$ and $\card{R_2} = \card{R_2'}$ trivially hold for the cases~\ref{c31-body}
    and~\ref{c32-body}, whereas the case~\ref{c33-body}
     ($\card{R_1} \geq \alpha_1$ and $\card{R_2} \geq \alpha_2$) can never be applied since
    $\card{R_1} + \card{R_2} < \alpha_1 + \alpha_2$. We conclude that $\min(\alpha_k,\card{R_k}) = \min(\alpha_k,\card{R_k'})$.

  Second, suppose instead
  \[\min(\alpha,\card{\grem}) = \min(\alpha,\card{\grem'}) = \alpha_1 + \alpha_2.\]
  If the first case~\ref{c31-body} applies, i.e.\ $\card{R_1} < \alpha_1$, then $\card{R_2} \geq \alpha_2$ and by definition $\card{R_1'} = \card{R_1}$.
  Then, $\card{R_2'} \geq \alpha_2$ trivially follows from $\card{R_1'} + \card{R_2'} \geq \alpha_1 + \alpha_2$.
  Symmetrically, $\min(\alpha_k,\card{R_k}) = \min(\alpha_k,\card{R_k'})$ holds when the second case~\ref{c32-body}  applies ($\card{R_2} < \alpha_2$).
     Lastly, suppose $\card{R_1} \geq \alpha_1$ and $\card{R_2} \geq \alpha_2$. Then, the third case~\ref{c33-body} applies and by definition $\card{R_1'} = \alpha_1$. Again, we conclude that $\min(\alpha_k,\card{R_k}) = \min(\alpha_k,\card{R_k'})$ since
    $\card{R_2'} \geq \alpha_2$ trivially follows from $\card{R_1'} + \card{R_2'} \geq \alpha_1 + \alpha_2$.

    \item[\descLab{(CI)}{c2-body}] Lastly, the heaps $\aheap_1'$ and $\aheap_2'$ are further populated with respect to the memory cells in $\gbtw(\alocation,\alocation')$.
    For all $\pair{\alocation}{\alocation'} \in \gedges$, let
     $L_k \egdef \gbtw(\alocation,\alocation') \cap \domain{\aheap_k}$. We have
    $L_1 \uplus L_2 = \gbtw(\alocation,\alocation')$.
    Below, we partition $\gbtw'(\amap(\alocation),\amap(\alocation'))$ into $L_1'$ and $L_2'$ so that by definition
    $L_k' \subseteq \domain{\aheap_k'}$:
    \begin{description}
    \item[\descLab{(CI.C1)}{c21-body}] If $\card{L_1} < \alpha_1$ then $L_1'$ is a set of $\card{L_1}$ locations from $\gbtw'(\amap(\alocation),\amap(\alocation'))$, whereas $L_2' \egdef
     \gbtw'(\amap(\alocation),\amap(\alocation')) \setminus L_1'$.
    \item[\descLab{(CI.C2)}{c22-body}] Else, if $\card{L_2} < \alpha_2$ then $L_2'$ is a set of $\card{L_2}$ locations from $\gbtw'(\amap(\alocation),\amap(\alocation'))$,
    whereas  $L_1' \egdef \gbtw'(\amap(\alocation),\amap(\alocation')) \setminus L_2'$.
    \item[\descLab{(CI.C3)}{c23-body}] Otherwise, we have  $\card{L_1} \geq \alpha_1$ and $\card{L_2} \geq \alpha_2$. Then $L_1'$ is a set of $\alpha_1$ locations
    from $\gbtw'(\amap(\alocation),\amap(\alocation'))$ and
    $L_2' \egdef \gbtw'(\amap(\alocation),\amap(\alocation')) \setminus L_1'$.
    \end{description}
    It is easy to show that the construction satisfies the following properties (where $k \in \{1,2\}$):
    \begin{equation}
      \tag*{\textbf{(CI.P1)}}\label{cip3-body}
    \text{if } L_k' = \emptyset \text{ then } \gbtw'(\amap(\alocation),\amap(\alocation')) \subseteq \domain{\aheap_{3-k}'}
    \end{equation}
    \begin{equation}
      \tag*{\textbf{(CI.P2)}}\label{cip2-body}
    \min(\alpha_k,\card{L_k}) = \min(\alpha_k,\card{L_k'})
    \end{equation}
    (Proof of~\ref{cip3-body}) The first property trivially holds from the cases~\ref{c21-body} and~\ref{c22-body} of the construction.
    Notice that given $k \in \{1,2\}$, $3-k$ corresponds to the index in $\{1,2\}$ that is different from $k$.

    (Proof of~\ref{cip2-body}) The second property directly follows from~\ref{a4} (which is satisfied by $\amap$) and is proved as done for~\ref{crp2-body}.


  \end{description}

This ends the construction of $\aheap'_1$ and $\aheap'_2$ as
any location in $\domain{\aheap'}$ has been assigned to one of the two heaps.
Indeed, $\{\galloc',\grem'\} \cup \{ \gbtw'(\alocation,\alocation') \mid
(\alocation,\alocation') \in \gedges' \}$ is a partition of $\domain{\aheap'}$.
As  $\pair{\astore}{\aheap} \approx^q_\alpha
\pair{\astore'}{\aheap'}$,
the support graphs  $\supportgraph{\astore,\aheap}$ and  $\supportgraph{\astore',\aheap'}$ witness
the existence of a map $\amap:  \gverts \to \gverts'$ satisfying~\ref{a1}--\ref{a5} and therefore
there is an underlying isomorphism between these structures satisfying quantitative properties up to
the value $\alpha$. The construction above can be understood as a way to split $\aheap'$
into $\aheap_1'$ and $\aheap_2'$ mimicking the splitting of $\aheap$ into $\aheap_1$ and $\aheap_2$.
It remains to show below that this is done in a way that guarantees that
 $\pair{\astore}{\aheap_k} \approx^q_{\alpha_k}
\pair{\astore'}{\aheap'_k}$ ($k \in \set{1,2}$).

In the following, we denote the support graphs of $\pair{\astore}{\aheap_k}$ and $\pair{\astore'}{\aheap_k'}$ respectively as
\begin{itemize}
\item $\supportgraph{\astore,\aheap_k} = (\gverts_k,\gedges_k,\galloc_k,\glabels_k,\gbtw_k,\grem_k)$
and,
\item  $\supportgraph{\astore',\aheap_k'} = (\gverts_k',\gedges_k',\galloc_k',\glabels_k',\gbtw_k',\grem_k')$.
\end{itemize}

  First, let us formalise an essential property of the construction of $\aheap_1'$ and $\aheap_2'$.
  \begin{enumerate}[align=left]
  \item[\descLab{(Paths)}{pathproperty-body}] Let $k \in \{1,2\}$ and let $\alocation,\alocation' \in \gverts$ be two labelled locations w.r.t. $\pair{\astore}{\aheap}$.
  $\aheap_k$ witnesses a non-empty path from $\alocation$ to $\alocation'$ if and only if $\aheap_k'$ witnesses a non-empty path from $\amap(\alocation)$ to $\amap(\alocation')$.
  \vspace{2pt}
  \item[(Proof of \ref{pathproperty-body})]
  The proof mainly relies on the properties \ref{cip3-body} and \ref{cip2-body} of the construction.
  Recall that $\amap : \gverts \to \gverts'$ is a bijection satisfying~\ref{a1}--\ref{a5} w.r.t.\ $\pair{\astore}{\aheap}$ and $\pair{\astore'}{\aheap'}$.

  ($\Rightarrow$) Let $\alocation,\alocation' \in \gverts$ be such that
  $\aheap_k$ witnesses a non-empty path from $\alocation$ to $\alocation'$.
  Since $\aheap_k \sqsubseteq \aheap$, then $\aheap$ also witnesses a non-empty path from $\alocation$ to $\alocation'$.
  In particular, by Definition~\ref{definition-support-graph}, this path corresponds to a path
  in the support graph $\supportgraph{\astore,\aheap}$:
  \begin{center}
       \begin{tikzpicture}[baseline, node distance=1.8cm]
         \node[dot,label=above:{$\alocation$}] (zero) at (0,0) {};
         \node[dot,label=above:{$\alocation_1$}] (one) [right of=zero] {};
         \node[dot,label=above:{$\alocation_2$}] (two) [right of=one] {};
         \node[dot,label=above:{$\alocation_{n-1}$}] (n1) [right of=two] {};
         \node[dot,label=above:{$\alocation_{n}$}] (n) [right of=n1] {};
         \node[dot,label=above right:{$\alocation'$}] (np) [right of=n] {};

          \draw[pto] (zero) edge node [above] {$\gedges$} (one);
          \draw[pto] (one) edge node [above] {$\gedges$} (two);
          \draw[pto] (n1) edge node [above] {$\gedges$} (n);
          \draw[pto] (n) edge node [above] {$\gedges$} (np);
          \draw[pto] (one) edge node [above] {$\gedges$} (two);

          \path (two) edge [draw=none] node {$\dots$} (n1);

       \end{tikzpicture}
  \end{center}
  Let us define $\alocation_0 \egdef \alocation$  and $\alocation_{n+1} \egdef \alocation'$.
  In particular, following the picture above, the support graph $\supportgraph{\astore,\aheap}$
  witnesses a path $\{\pair{\alocation_0}{\alocation_1},\dots,\pair{\alocation_n}{\alocation_{n+1}}\} \subseteq \gedges$ from $\alocation_0$ to $\alocation_{n+1}$.
  Since $\amap$ is a graph isomorphism from $\pair{\gverts}{\gedges}$ to $\pair{\gverts'}{\gedges'}$ (by~\ref{a1}), $\supportgraph{\astore',\aheap'}$ witnesses a similar structure, as depicted below:
  \begin{center}
   \scalebox{1}{
       \begin{tikzpicture}[baseline, node distance=1.8cm]
         \node[dot,label=above left:{$\alocation$=$\alocation_0$}] (zero) at (0,0) {};
         \node[dot,label=above:{$\alocation_1$}] (one) [right of=zero] {};
         \node[dot,label=above:{$\alocation_2$}] (two) [right of=one] {};
         \node[dot,label=above:{$\alocation_{n-1}$}] (n1) [right of=two] {};
         \node[dot,label=above:{$\alocation_{n}$}] (n) [right of=n1] {};
         \node[dot,label=above right:{$\alocation'$=$\alocation_{n+1}$}] (np) [right of=n] {};
         \node[dot,label=below:{$\amap(\alocation)$}] (Fzero) [below of=zero] {};
         \node[dot,label=below:{$\amap(\alocation_1)$}] (Fone) [right of=Fzero] {};
         \node[dot,label=below:{$\amap(\alocation_2)$}] (Ftwo) [right of=Fone] {};
         \node[dot,label=below:{$\amap(\alocation_{n-1})$}] (Fn1) [right of=Ftwo] {};
         \node[dot,label=below:{$\amap(\alocation_{n})$}] (Fn) [right of=Fn1] {};
         \node[dot,label=below:{$\amap(\alocation')$}] (Fnp) [right of=Fn] {};

          \draw[pto] (zero) edge node [above] {$\gedges$} (one);
          \draw[pto] (one) edge node [above] {$\gedges$} (two);
          \draw[pto] (n1) edge node [above] {$\gedges$} (n);
          \draw[pto] (n) edge node [above] {$\gedges$} (np);
          \draw[pto] (one) edge node [above] {$\gedges$} (two);
          \draw[pto] (Fzero) edge node [below] {$\gedges'$} (Fone);
          \draw[pto] (Fone) edge node [below] {$\gedges'$} (Ftwo);
          \draw[pto] (Fn1) edge node [below] {$\gedges'$} (Fn);
          \draw[pto] (Fn) edge node [below] {$\gedges'$} (Fnp);
          \draw[pto] (Fone) edge node [below] {$\gedges'$} (Ftwo);

          \path (two) edge [draw=none] node {$\dots$} (n1);
          \path (Ftwo) edge [draw=none] node {$\dots$} (Fn1);

          \path[dotted,->] (zero) edge node [left] {$\amap$} (Fzero);
          \path[dotted,->] (one) edge node [left] {$\amap$} (Fone);
          \path[dotted,->] (np) edge node [left] {$\amap$} (Fnp);
          \path[dotted,->] (n) edge node [left] {$\amap$} (Fn);
          \path[dotted,->] (n1) edge node [left] {$\amap$} (Fn1);
          \path[dotted,->] (two) edge node [left] {$\amap$} (Ftwo);
       \end{tikzpicture}
   }
  \end{center}
  Let us consider $i \in \interval{0}{n}$. Since the path belongs to $\aheap_k$, it must hold that
  $\alocation_i \in \domain{\aheap_k}$ and $\gbtw(\alocation_i,\alocation_{i+1}) \subseteq \domain{\aheap_k}$. We show that then $\amap(\alocation_i) \in \domain{\aheap_k'}$ and $\gbtw'(\amap(\alocation_i),\amap(\alocation_{i+1})) \subseteq \domain{\aheap_k'}$,
  which entails that $\aheap_k'$ witnesses a path from $\amap(\alocation)$ to $\amap(\alocation')$, concluding the proof.
  \begin{itemize}[align = left]
    \item From $\pair{\alocation_i}{\alocation_{i+1}} \in \gedges$, by Definition~\ref{definition-support-graph} we conclude that $\alocation_i \in \gverts$.
    Since $\alocation_i \in \domain{\aheap}$, again by Definition~\ref{definition-support-graph},
    we have $\alocation_i \in \galloc$ and therefore by~\ref{a2}, $\amap(\alocation_i) \in \galloc'$.
    By~\ref{c1-body} together with the fact that $\alocation_i \in \domain{\aheap_k}$, we then conclude that $\amap(\alocation_i) \in \domain{\aheap_k'}$.
    \item If $\gbtw(\alocation_i,\alocation_{i+1})$ is empty then by~\ref{a4} $\gbtw'(\amap(\alocation_i),\amap(\alocation_{i+1}))$ is also empty and the inclusion w.r.t. $\domain{\aheap_k'}$ trivially holds.
    Suppose now $\gbtw(\alocation_i,\alocation_{i+1})$ is non-empty.
    From $\gbtw(\alocation_i,\alocation_{i+1}) \subseteq \domain{\aheap_k}$ and the fact that $\aheap_k$ and $\aheap_{3-k}$ are disjoint we conclude that $\gbtw(\alocation_i,\alocation_{i+1}) \cap \domain{\aheap_{3-k}} = \emptyset$. Hence, by \ref{cip2-body} we conclude that
    $\gbtw'(\amap(\alocation_i),\amap(\alocation_{i+1})) \cap \domain{\aheap_{3-k}'} = \emptyset$
    (notice that in \ref{cip2-body}, $L_{3-k}$ corresponds to $\gbtw(\alocation_i,\alocation_{i+1}) \cap \domain{\aheap_{3-k}}$ whereas $L_{3-k}'$ corresponds to $\gbtw'(\amap(\alocation_i),\amap(\alocation_{i+1})) \cap \domain{\aheap_{3-k}'}$).
    Since $\gbtw'(\amap(\alocation_i),\amap(\alocation_{i+1})) \cap \domain{\aheap_{3-k}'} = \emptyset$,
    by  \ref{cip3-body} we then conclude that
    $\gbtw'(\amap(\alocation_i),\amap(\alocation_{i+1})) \subseteq \domain{\aheap_{k}'}$.
  \end{itemize}

  ($\Leftarrow$) The right-to-left direction is analogous (thanks to the fact that $\amap^{-1}$ is a graph isomorphism from $\pair{\gverts'}{\gedges'}$ to $\pair{\gverts}{\gedges}$). \\
  \end{enumerate}

  Here is the last step of the proof. 
  Given $k \in \{1,2\}$, let $\amap_k$ be the restriction of $\amap$ to $\gverts_k$ and
  $\gverts_k'$.
  We prove that $\amap_k$ satisfies \ref{a1}--\ref{a5} w.r.t. the memory states $\pair{\astore}{\aheap_k}$ and $\pair{\astore'}{\aheap_k'}$ and $\alpha_k$.
  Thanks to Lemma~\ref{lemma-test-graph}, this implies $\pair{\astore}{\aheap_1} \approx^q_{\alpha_1} \pair{\astore'}{\aheap_1'}$ and
  $\pair{\astore}{\aheap_2} \approx^q_{\alpha_2} \pair{\astore}{\aheap_2'}$, ending the proof.
  Because of lack of space, below we prove~\ref{a3} and the proofs for~\ref{a2}, \ref{a1}, \ref{a4} and \ref{a5} can be found
  in Appendix~\ref{appendix-star}. \\

  \begin{enumerate}[label=\textbf{(P\arabic*)}, align = left]
  \item[\textbf{$\amap_k$ satisfies \descLab{(A3)}{mapksatA3-body}}:]
    We prove that for every $\alocation \in \gverts$, the set of terms corresponding to $\alocation$ in $\pair{\astore}{\aheap_k}$ is equal to the set of terms corresponding to $\amap(\alocation)$ in $\pair{\astore'}{\aheap_k'}$. Formally:
    \begin{center}
      \begin{minipage}{0.9\linewidth}
        for every $\alocation \in \gverts$,
        \begin{enumerate}[label=(\alph*)]
            \item $\alocation \in \gverts_k$ iff $\amap(\alocation) \in \gverts_k'$;
            \item if $\alocation \in \gverts_k$, $\glabels_k(\alocation) = \glabels_k'(\amap(\alocation))$.
        \end{enumerate}
      \end{minipage}
    \end{center}
    Notice that (a) implies that $\amap_k$ (which we recall being the restriction of $\amap$ to $\gverts_k$ and $\gverts_k'$) is well-defined and it is a bijection from the labelled locations of $\pair{\astore}{\aheap_k}$ (i.e. $\gverts_k$) to the labelled locations of $\pair{\astore'}{\aheap_k'}$ (i.e. $\gverts_k'$).
This is due to the fact that
   $\amap$ is a bijection from $\gverts$ to $\gverts'$ and by Lemma~\ref{lemma-labels}, 
we have $\gverts_k \subseteq \gverts$ and $\gverts_k' \subseteq \gverts'$.
    Again by $\gverts_k \subseteq \gverts$ (Lemma~\ref{lemma-labels}), (b) is then equivalent
    to \ref{a3}.

    We prove (a) and (b) together, by showing that
    for every $\alocation \in \gverts$ the set of terms corresponding to $\alocation$ in $\pair{\astore}{\aheap_k}$ is equivalent to the set of terms corresponding to $\amap(\alocation)$
    in $\pair{\astore'}{\aheap_k'}$.
    We first show the result for program variables, and then for meet-points.

    (Program variables) Let $\alocation \in \gverts$ and  $i \in \interval{1}{q}$.
    It holds that $\avariable_i \in \glabels_k(\alocation)$
     if and only if $\astore(\avariable_i) = \alocation$, or equivalently
     $\avariable_i \in \glabels(\alocation)$ which, by \ref{a3} in Lemma~\ref{lemma-test-graph}, holds whenever
     $\avariable_i \in \glabels'(\amap(\alocation))$. The latter is equivalent to $\astore'(\avariable_i) =
      \amap(\alocation)$, or equivalently $\avariable_i \in \glabels_k'(\amap(\alocation))$.

    (Meet-points)
    In order to conclude the proof, we show that for all $i,j \in \interval{1}{q}$ and $\alocation \in \gverts$,
    \begin{center}
     $m_q(\avariable_i,\avariable_j) \in \glabels_k(\alocation)$
    if and only if $m_q(\avariable_i,\avariable_j) \in \glabels_k'(\amap(\alocation))$.
    \end{center}
    If $\sem{m_q(\avariable_i,\avariable_j)}^q_{\astore,\aheap}$ is undefined, then so is $\sem{m_q(\avariable_j,\avariable_i)}^q_{\astore,\aheap}$ (by def.)
    and by ${\pair{\astore}{\aheap} \approx^q_\alpha \pair{\astore'}{\aheap'}}$, so are $\sem{m_q(\avariable_i,\avariable_j)}^q_{\astore',\aheap'}$ and $\sem{m_q(\avariable_j,\avariable_i)}^q_{\astore',\aheap'}$.
    By Lemma~\ref{lemma-labels}, if $\sem{m_q(\avariable_i,\avariable_j)}^q_{\astore,\aheap}$ is undefined, then so are
    $\sem{m_q(\avariable_i,\avariable_j)}^q_{\astore,\aheap_k}$, $\sem{m_q(\avariable_j,\avariable_i)}^q_{\astore,\aheap_k}$,
    $\sem{m_q(\avariable_i,\avariable_j)}^q_{\astore',\aheap'_k}$ and $\sem{m_q(\avariable_j,\avariable_i)}^q_{\astore',\aheap'_k}$.
    Otherwise, if $\sem{m_q(\avariable_i,\avariable_j)}^q_{\astore,\aheap} = \alocation$ then by definition $\sem{m_q(\avariable_j,\avariable_i)}^q_{\astore,\aheap} = \alocation'$ for some $\alocation' \in \gverts$ and
    moreover $\sem{m_q(\avariable_i,\avariable_j)}^q_{\astore',\aheap'} = \amap(\alocation)$ and
    $\sem{m_q(\avariable_j,\avariable_i)}^q_{\astore',\aheap'} = \amap(\alocation')$ by $\pair{\astore}{\aheap} \approx^q_\alpha \pair{\astore'}{\aheap'}$.
    Let $\avariableter^i$ (resp. $\avariableter^j$) be the program variable in $\{\avariable_1,\dots,\avariable_q\}$ such that $\astore(\avariableter^i)$ is the first location corresponding to a program variable that is reachable from $\sem{m_q(\avariable_i,\avariable_j)}^q_{\astore,\aheap}$ (resp. $\sem{m_q(\avariable_j,\avariable_i)}^q_{\astore,\aheap}$), itself included.
    The characterisation of such a program variable $\avariableter^i$
    can be captured by the formula
    $\firstvar(m_q(\avariable_i,\avariable_j),\avariableter^i)$ defined as the following Boolean combination of test formulae:
    $$
    m_q(\avariable_i,\avariable_j) = \avariableter^i \vee \ \ \ \
    \bigvee_{\mathclap{\substack{\aterm_1, \ldots, \aterm_n \in \Terms{q}, n>1 \\
    {\rm pairwise \ distinct} \ \aterm_1, \ldots, \aterm_{n-1}, \\ \aterm_1 = m_q(\avariable_i,\avariable_j), \aterm_n = \avariableter^i}}}
    \textstyle\bigwedge_{\delta \in [1,n-1]} \sees_q(\aterm_{\delta},\aterm_{\delta+1}) \geq 1 \land \bigwedge_{\substack{m < n\hphantom{m.}\\ k \in [1,q]}} \avariable_k \neq \aterm_m$$
    Indeed, this formula is satisfied only by memory states where there is a (possibly empty) path from the location corresponding to $m_q(\avariable_i,\avariable_j)$ to the location $\overline{\alocation}$ corresponding to $\avariableter^i$ so that each labelled location in the path, apart from $\overline{\alocation}$, does not correspond to the interpretation of  such program variables.
    Now,  we  recall the taxonomy of meet-points, where the rightmost case from Figure~\ref{figure-taxonomy} is split into three cases, highlighting additional cases depending on $\avariableter^i$ and $\avariableter^j$.
    \begin{center}
       \begin{tabular*}{\linewidth}{c @{\extracolsep{\fill}}  cccc}
         \scalebox{0.68}{
           \begin{tikzpicture}[baseline]
             \node[dot,label=above:$\avariable_i$] (i) at (0,0) {};
             \node[dot,label=left:{$m_q$($\avariable_i$,$\avariable_j$)\\$m_q$($\avariable_j$,$\avariable_i$)}] (m) [below right = 1.5cm and 0.5cm of i] {};
             \node[dot,label=above:$\avariable_j$] (j) [above right=1.5cm and 0.5cm of m] {};
             \node[dot,label=below:{$\avariableter^i=\avariableter^j$\\$\avariableter^i$ not inside a loop}] (k) [below of=m] {};

             \draw[reach] (i) -- (m);
             \draw[reach] (j) -- (m);
             \draw[reach] (m) -- (k);
           \end{tikzpicture}
         }
       &
        \scalebox{0.68}{
           \begin{tikzpicture}[baseline]
             \node[dot,label=above:$\avariable_i$] (i) at (0,0) {};
             \node[dot,label=left:{$m_q$($\avariable_i$,$\avariable_j$)\\$m_q$($\avariable_j$,$\avariable_i$)}] (m) [below right = 1.5cm and 0.5cm of i] {};
             \node[dot,label=above:$\avariable_j$] (j) [above right=1.5cm and 0.5cm of m] {};
             \node[dot] (mid) [below=0.8cm of m] {};
             \node[dot,label=below:{$\avariableter^i=\avariableter^j$}] (k) [below=1.65cm of m] {};

             \draw[reach] (i) -- (m);
             \draw[reach] (j) -- (m);
             \draw[reach] (m) -- (mid);
             \draw[reach] (mid) to [out=-35,in=35] (k);
             \draw[pto] (k) edge [out=155,in=-155] node [left] {$+$} (mid);
           \end{tikzpicture}
        }
       &
        \scalebox{0.68}{
           \begin{tikzpicture}[baseline]
             \node[dot,label=above:$\avariable_i$] (i) at (0,0) {};
             \node[dot,label=left:{$m_q$($\avariable_i$,$\avariable_j$)}] (m1) [below right = 1.5cm and 0.4cm of i] {};
             \node[dot,label=below:{$m_q$($\avariable_j$,$\avariable_i$)}] (m2) [below right = 1.7cm and 0.4cm of m1] {};
             \node[dot,label=above:$\avariable_j$] (j) [above right=1.5cm and 0.5cm of m2] {};
             \node[dot,label=left:{$\avariableter^i=\avariableter^j$}] (k) [below right= 2.7cm and 0.1 of i] {};

             \draw[reach] (i) -- (m1);
             \draw[pto] (m1) edge node [above right] {$+$} (m2);
             \draw[reach] (j) -- (m2);
             \draw[reach] (m2) to [out=180,in=-75] (k);
             \draw[pto] (k) edge [bend left=30] node [left] {$+$} (m1);
           \end{tikzpicture}
        }
       &
        \scalebox{0.68}{
           \begin{tikzpicture}[baseline]
             \node[dot,label=above:$\avariable_j$] (i) at (0,0) {};
             \node[dot,label=right:{$m_q$($\avariable_j$,$\avariable_i$)}] (m1) [below left = 1.5cm and 0.4cm of i] {};
             \node[dot,label=below:{$m_q$($\avariable_i$,$\avariable_j$)}] (m2) [below left = 1.7cm and 0.4cm of m1] {};
             \node[dot,label=above:$\avariable_i$] (j) [above left=1.5cm and 0.5cm of m2] {};
             \node[dot,label=right:{$\avariableter^i=\avariableter^j$}] (k) [below left= 2.7cm and 0.1 of i] {};

             \draw[reach] (i) -- (m1);
             \draw[pto] (m1) edge node [above left] {$+$} (m2);
             \draw[reach] (j) -- (m2);
             \draw[reach] (m2) to [out=0,in=-105] (k);
             \draw[pto] (k) edge [bend right=30] node [right] {$+$} (m1);
           \end{tikzpicture}
        }
       &
        \scalebox{0.68}{
           \begin{tikzpicture}[baseline]
             \node[dot,label=above:$\avariable_i$] (i) at (0,0) {};
             \node[dot,label=left:{$m_q$($\avariable_i$,$\avariable_j$)}] (m1) [below right = 1.5cm and 0.4cm of i] {};
             \node[dot,label=right:{$m_q$($\avariable_j$,$\avariable_i$)}] (m2) [right = 1.8cm of m1] {};
             \node[dot,label=above:$\avariable_j$] (j) [above right=1.5cm and 0.5cm of m2] {};
             \node[dot,label=below:{$\avariableter^j$}] (k) [below right= 0.5cm and 0.9cm of m1] {};
             \node[dot,label=above:{$\avariableter^i$}] (z) [above right = 0.5 and 0.9cm of m1] {};

             \draw[reach] (i) -- (m1);
             \draw[reach] (j) -- (m2);
             \draw[reach] (m2) to [out=-115,in=0] (k);
             \draw[pto] (k) edge [bend left=30] node [below left] {$+$} (m1);
             \draw[reach] (m1) to [out=65,in=180] (z);
             \draw[pto] (z) edge [bend left=30] node [above right] {$+$} (m2);
           \end{tikzpicture}
        }\\
        \\
       (i)
       &
       (ii)
       &
       (iii)
       &
       (iv)
       &
       (v)
     \end{tabular*}
    \end{center}
    Distinct structures satisfy different test formulae, though they all satisfy the formula
    \begin{center}
      $\firstvar(m_q(\avariable_i,\avariable_j),\avariableter^i) \land \firstvar(m_q(\avariable_j,\avariable_i),\avariableter^j)$.
    \end{center}
    For instance, (i) is the only form not satisfying $\reach^+(\avariableter^i,\avariableter^i)$ (recall that
    this form can be expressed as a Boolean combination of test formulae, as shown in Lemma~\ref{lemma-atomic-formulae-test}), whereas (ii) can be distinguished as the only form satisfying
    both $\reach^+(\avariableter^i,\avariableter^i)$ and $m_q(\avariable_i,\avariable_j) = m_q(\avariable_j,\avariable_i)$.
    Moreover, the last structure is the only one satisfying $\avariableter^i \neq \avariableter^j$ whereas (iii) and (iv) can be distinguished with a formula, similar to $\firstvar(m_q(\avariable_i,\avariable_j),\avariableter^i)$, stating that from the location corresponding to $\avariableter^i$, it is possible to reach the location corresponding to $m_q(\avariable_i,\avariable_j)$ without reaching the location corresponding to $m_q(\avariable_j,\avariable_i)$:
    $$
    \bigvee_{\mathclap{\substack{\aterm_1, \ldots, \aterm_n \in \Terms{q},\ n > 1\\
    {\rm pairwise \ distinct} \ \aterm_1, \ldots, \aterm_{n-1}, \\ \aterm_1 = \avariableter^i, \aterm_n = m_q(\avariable_i,\avariable_j)}}}
    \textstyle\bigwedge_{\delta \in [1,n-1]} \sees_q(\aterm_{\delta},\aterm_{\delta+1}) \geq 1 \land \bigwedge_{1 < m < n} m_q(\avariable_j,\avariable_i) \neq \aterm_m.
    $$
    Differently from (iii), the structure (iv) does not satisfy this formula.
    Since $\pair{\astore}{\aheap} \approx^q_\alpha \pair{\astore'}{\aheap'}$, the heaps
    $\aheap$ and $\aheap'$ agree on the structure of every meet-point.
    The proof that for all $i,j \in \interval{1}{q}$ and $\alocation \in \gverts$,
    $m_q(\avariable_i,\avariable_j) \in \glabels_k(\alocation)$ iff
    $m_q(\avariable_i,\avariable_j) \in \glabels_k'(\amap(\alocation))$, essentially relies on \ref{pathproperty-body}. To show the result we need to proceed by cases, according to the taxonomy.
    \begin{enumerate}[align = left]
    \item[\textbf{Case: $\aheap$ witnesses (i), (ii) or (iv).}]
    Since the heaps
    $\aheap$ and $\aheap'$ agree on the structure of every meet-point,
    $\aheap'$ also witnesses the same form  among (i), (ii) and (iv) as $\aheap$.
    Regarding $\aheap_k$, one of the following holds:
    \begin{itemize}
      \item The path from $\astore(\avariable_i)$ to $\astore(\avariableter^i)$ and the path from $\astore(\avariable_j)$ to $\astore(\avariableter^j)$ are both preserved in $\aheap_k$.
      Then $\aheap_k$ witnesses (i) (ii) or (iv). By~\ref{pathproperty-body} the same holds for $\aheap_k'$.
      In every case ((i), (ii) and (iv)), we conclude that $\sem{m_q(\avariable_i,\avariable_j)}^q_{\astore,\aheap_k} = \sem{m_q(\avariable_i,\avariable_j)}^q_{\astore,\aheap}$
      and
      $\sem{m_q(\avariable_i,\avariable_j)}^q_{\astore',\aheap_k'} = \sem{m_q(\avariable_i,\avariable_j)}^q_{\astore',\aheap'} = \amap(\sem{m_q(\avariable_i,\avariable_j)}^q_{\astore,\aheap})$.
      Then,
      \begin{align*}
      m_q(\avariable_i,\avariable_j) &\in  \glabels_k(\sem{m_q(\avariable_i,\avariable_j)}^q_{\astore,\aheap})\\ m_q(\avariable_i,\avariable_j) &\in  \glabels_k'(\amap(\sem{m_q(\avariable_i,\avariable_j)}^q_{\astore,\aheap})).
      \end{align*}
      \item The path from $\astore(\avariable_i)$ to $\astore(\avariableter^i)$ or the path from $\astore(\avariable_j)$ to $\astore(\avariableter^j)$ are not preserved in $\aheap_k$.
      Then, again by~\ref{pathproperty-body}, the same holds for $\aheap_k'$ with respect to $\pair{\astore'}{\aheap'}$.
     By definition of meet-points, both $\sem{m_q(\avariable_i,\avariable_j)}^q_{\astore,\aheap_k}$ and $\sem{m_q(\avariable_i,\avariable_j)}^q_{\astore',\aheap_k'}$ are therefore not defined.
    \end{itemize}
    \item[\textbf{Case: $\aheap$ witnesses (iii).}]
    If instead $\aheap$ and $\aheap'$ witness (iii) then one of the following holds.
    \begin{itemize}
    \item $\aheap_k$ also witnesses (iii), meaning that the path from $\astore(\avariable_i)$ to $\sem{m_q(\avariable_j,\avariable_i)}^q_{\astore,\aheap}$ and the
    path from $\astore(\avariable_j)$ to $\sem{m_q(\avariable_i,\avariable_j)}^q_{\astore,\aheap}$ are both preserved in $\aheap_k$. Then by \ref{pathproperty-body}
    the heap $\aheap_k'$ also witnesses (iii).
    We have 
    $\sem{m_q(\avariable_i,\avariable_j)}^q_{\astore,\aheap_k} = \sem{m_q(\avariable_i,\avariable_j)}^q_{\astore,\aheap}$ and $\sem{m_q(\avariable_i,\avariable_j)}^q_{\astore',\aheap_k'} = \sem{m_q(\avariable_i,\avariable_j)}^q_{\astore',\aheap'} = \amap(\sem{m_q(\avariable_i,\avariable_j)}^q_{\astore,\aheap})$.
    We conclude that
    \begin{align*}
    m_q(\avariable_i,\avariable_j) &\in  \glabels_k(\sem{m_q(\avariable_i,\avariable_j)}^q_{\astore,\aheap_k})\\
    m_q(\avariable_i,\avariable_j) &\in  \glabels_k'(\amap(\sem{m_q(\avariable_i,\avariable_j)}^q_{\astore,\aheap_k})).
    \end{align*}
    \item The path from $\astore(\avariable_i)$ to $\astore(\avariableter^i)$ and the path from $\astore(\avariable_j)$ to $\astore(\avariableter^i)$ are preserved in $\aheap_k$, whereas the path from $\astore(\avariableter^i)$ to $\sem{m_q(\avariable_i,\avariable_j)}^q_{\astore,\aheap}$ is not preserved in $\aheap_k$ (i.e. at least one of its
    locations is assigned to the other
    heap $\aheap_{3-k}$).
    Then $\aheap_k$ witnesses (i) and by definition of meet-points, it holds that
    \[
      \sem{m_q(\avariable_i,\avariable_j)}^q_{\astore,\aheap_k} = \sem{m_q(\avariable_j,\avariable_i)}^q_{\astore,\aheap_k} = \sem{m_q(\avariable_j,\avariable_i)}^q_{\astore,\aheap}
    \]
    By~\ref{pathproperty-body}, $\aheap_k'$ also witnesses (i). Then,
    \[
      \sem{m_q(\avariable_i,\avariable_j)}^q_{\astore',\aheap_k'} = \sem{m_q(\avariable_j,\avariable_i)}^q_{\astore',\aheap_k'} = \sem{m_q(\avariable_j,\avariable_i)}^q_{\astore',\aheap'} = \amap(\sem{m_q(\avariable_j,\avariable_i)}^q_{\astore,\aheap}).
    \]
    Then, we conclude that
    \begin{align*}
    m_q(\avariable_i,\avariable_j) &\in  \glabels_k(\sem{m_q(\avariable_j,\avariable_i)}^q_{\astore,\aheap})\\ m_q(\avariable_i,\avariable_j) &\in  \glabels_k'(\amap(\sem{m_q(\avariable_j,\avariable_i)}^q_{\astore,\aheap})).
    \end{align*}
    \item The path from $\astore(\avariable_i)$ to $\astore(\avariableter^i)$ or the path from $\astore(\avariable_j)$ to $\astore(\avariableter^i)$ are not preserved in $\aheap_k$.
    Then, by~\ref{pathproperty-body}, the same holds for $\aheap_k'$ with respect to $\pair{\astore'}{\aheap'}$.
    By definition of meet-points, neither $\sem{m_q(\avariable_i,\avariable_j)}^q_{\astore,\aheap_k}$ nor $\sem{m_q(\avariable_i,\avariable_j)}^q_{\astore',\aheap_k'}$ is defined.
    \end{itemize}
    \item[\textbf{Case: $\aheap$ witnesses (v).}]
    Lastly, suppose  that $\aheap$ and $\aheap'$ witness (v). One of the following holds.
    \begin{itemize}
      \item The path from $\astore(\avariable_i)$ to $\astore(\avariableter^i)$ and the path from $\astore(\avariable_j)$ to $\astore(\avariableter^i)$ are both preserved in $\aheap_k$. Then, depending on whether or not the path from $\astore(\avariableter^i)$ to $\sem{m_q(\avariable_j,\avariable^i)}^q_{\astore,\aheap}$ is also preserved, $\aheap_k$ witnesses (i) or (v).
      From~\ref{pathproperty-body}, the same holds for $\aheap_k'$ (where $\aheap_k'$ witnesses (i) iff $\aheap_k$ witnesses (i)).
      In both cases ((i) and (v)), it holds that $\sem{m_q(\avariable_i,\avariable_j)}^q_{\astore,\aheap_k} = \sem{m_q(\avariable_i,\avariable_j)}^q_{\astore,\aheap}$
      and
      $\sem{m_q(\avariable_i,\avariable_j)}^q_{\astore',\aheap_k'} = \sem{m_q(\avariable_i,\avariable_j)}^q_{\astore',\aheap'} = \amap(\sem{m_q(\avariable_i,\avariable_j)}^q_{\astore,\aheap})$.
      Then,
      \begin{align*}
      m_q(\avariable_i,\avariable_j) &\in  \glabels_k(\sem{m_q(\avariable_i,\avariable_j)}^q_{\astore,\aheap})\\ m_q(\avariable_i,\avariable_j) &\in  \glabels_k'(\amap(\sem{m_q(\avariable_i,\avariable_j)}^q_{\astore,\aheap})).
      \end{align*}
      \item The path from $\astore(\avariable_i)$ to $\astore(\avariableter^j)$ and the path from $\astore(\avariable_j)$ to $\astore(\avariableter^j)$ are preserved in $\aheap_k$, whereas the path from $\astore(\avariableter^j)$ to $\sem{m_q(\avariable_i,\avariable_j)}^q_{\astore,\aheap}$ is not preserved in $\aheap_k$.
      Then $\aheap_k$ witnesses (i) and by definition of meet-points, it holds that
      \[
        \sem{m_q(\avariable_i,\avariable_j)}^q_{\astore,\aheap_k} = \sem{m_q(\avariable_j,\avariable_i)}^q_{\astore,\aheap_k} = \sem{m_q(\avariable_j,\avariable_i)}^q_{\astore,\aheap}.
      \]
      By~\ref{pathproperty-body}, the heap $\aheap_k'$ also witnesses (i). Then,
      \[
        \sem{m_q(\avariable_i,\avariable_j)}^q_{\astore',\aheap_k'} = \sem{m_q(\avariable_j,\avariable_i)}^q_{\astore',\aheap_k'} = \sem{m_q(\avariable_j,\avariable_i)}^q_{\astore',\aheap'} = \amap(\sem{m_q(\avariable_j,\avariable_i)}^q_{\astore,\aheap})
      \]
      Then, we conclude that
      \begin{align*}
      m_q(\avariable_i,\avariable_j) &\in  \glabels_k(\sem{m_q(\avariable_j,\avariable_i)}^q_{\astore,\aheap})\\ m_q(\avariable_i,\avariable_j) &\in  \glabels_k'(\amap(\sem{m_q(\avariable_j,\avariable_i)}^q_{\astore,\aheap})).
      \end{align*}
      \item The path from $\astore(\avariable_i)$ to $\astore(\avariableter^j)$ or the path from $\astore(\avariable_j)$ to $\astore(\avariableter^j)$ is
       not preserved in $\aheap_k$. Then, by~\ref{pathproperty-body}, the same holds for $\aheap_k'$ with respect to $\pair{\astore'}{\aheap'}$.
      By definition of meet-points, both $\sem{m_q(\avariable_i,\avariable_j)}^q_{\astore,\aheap_k}$ and $\sem{m_q(\avariable_i,\avariable_j)}^q_{\astore',\aheap_k'}$ are
      undefined.
    \end{itemize}
    \end{enumerate}

      We conclude that for every $\alocation \in \gverts$, $m_q(\avariable_i,\avariable_j) \in \glabels_k(\alocation)$ if and only if $m_q(\avariable_i,\avariable_j) \in \glabels_k'(\amap(\alocation))$.
\cut{
  \item[\textbf{$\amap_k$ satisfies \descLab{(A2)}{mapksatA2}}:] We prove that $\amap_k$ is such that
  \begin{center}
    for every $\alocation \in \gverts_k$, $\alocation \in \galloc_k$ iff $\amap_k(\alocation) \in \galloc_k'$.
  \end{center}
  From (a) in the proof that $\amap_k$ satisfies \ref{mapksatA3-body}, we know that for every $\alocation \in \gverts_k$, $\amap_k(\alocation) = \amap(\alocation) \in \gverts_k'$.
  ($\Rightarrow$) For the left-to-right direction, let $\alocation \in \gverts_k$ such that $\alocation \in \galloc_k$. By definition we have $\alocation \in \domain{\aheap_k}$ and therefore $\alocation \in \domain{\aheap}$ (from $\aheap_k \sqsubseteq \aheap$).
  Since $\gverts_k \subseteq \gverts$ (Lemma~\ref{lemma-labels}) we have $\alocation \in \gverts$ and hence $\alocation \in \galloc$.
  From
  $\pair{\astore}{\aheap} \approx^q_\alpha \pair{\astore'}{\aheap'}$
  we obtain $\amap(\alocation) \in \galloc'$. Hence, from the construction of $\aheap_k'$
  done in \ref{c1-body} we have $\amap(\alocation) \in \aheap_k'$.
  Moreover, since $\amap_k(\alocation) = \amap(\alocation) \in \gverts_k'$,
  we conclude $\amap_k(\alocation) \in \galloc_k'$.

  ($\Leftarrow$) The right-to-left direction follows by using a similar argument.

  \item[\textbf{$\amap_k$ satisfies \descLab{(A1)}{mapksatA1}}:] We prove that $\amap_k$ is a graph isomorphism between $\pair{\gverts_k}{\gedges_k}$ and $\pair{\gverts_k'}{\gedges_k'}$.
  From (a) in the proof that $\amap_k$ satisfies \ref{mapksatA3-body} we already know that $\amap_k$ is a bijection from $\gverts_k$ to $\gverts_k'$. Hence, we only need to show that
    for all $\alocation,\alocation' \in \gverts_k$, $\pair{\alocation}{\alocation'} \in \gedges_k$ $\iff$ $\pair{\amap_k(\alocation)}{\amap_k(\alocation')} \in \gedges_k'$.

  ($\Rightarrow$) By definition of $\supportgraph{\astore,\aheap_k}$, we have that $\pair{\alocation}{\alocation'} \in \gedges_k$ if and only if
  $\alocation,\alocation' \in \gverts_k$ and $\aheap_k$ witnesses a non-empty (minimal) path from $\alocation$ to $\alocation'$ such that
  none of the intermediate locations of this path are in $\gverts_k$.
  From (a) in the proof that $\amap_k$ satisfies \ref{mapksatA3-body} we have
  \begin{center}
    $\alocation \in \gverts_k$ iff $\amap_k(\alocation) \in \gverts_k'$; \qquad $\alocation' \in \gverts_k$ iff $\amap_k(\alocation') \in \gverts_k'$.
  \end{center}
  As moreover $\gverts_k \subseteq \gverts$ and $\gverts_k' \subseteq \gverts'$ (by Lemma~\ref{lemma-labels}), by~\ref{pathproperty-body} it holds that $\aheap_k$ witnesses a non-empty path from $\alocation$ to $\alocation'$ if and only if $\aheap_k'$ witnesses a non-empty path from $\amap_k(\alocation)$ to $\amap_k(\alocation')$.

  Therefore, to conclude the proof, it
  remains to show that there is a non-empty path from $\amap_k(\alocation)$ to $\amap_k(\alocation')$ in $\aheap_k'$ such that no intermediate locations of this path are in $\gverts_k'$.
  If this property is satisfied, then it is satisfied by the minimal non-empty path  from $\amap_k(\alocation)$ to $\amap_k(\alocation')$, which we suppose being of length $L \geq 1$.
  Let us denote this path with $\sigma(\amap_k(\alocation),\amap_k(\alocation'))$.
  If $L = 1$ then the property is trivially satisfied (as there are no intermediate locations between $\amap_k(\alocation)$ and $\amap_k(\alocation')$).
  Otherwise,
  {\em ad absurdum}, suppose there exists $\alocation'' \in \gverts_k'$, different from $\amap_k(\alocation)$ and $\amap_k(\alocation')$,
  such that for some $L_1,L_2 \geq 1$ such that $L = L_1 + L_2$ we have ${\aheap_k'}^{L_1}(\amap_k(\alocation)) = \alocation''$
  and ${\aheap_k'}^{L_2}(\alocation'') = \amap_k(\alocation')$.
  Then, we show that $\amap_k^{-1}(\alocation'')$ (which by~\ref{mapksatA3-body} is in $\gverts_k$) is an intermediate location on the minimal non-empty path from $\alocation$ to $\alocation'$, in $\aheap_k$, in contradiction with $\pair{\alocation}{\alocation'} \in \gedges_k$.
  To do so, we start by considering the set $A = \{\alocation_1,\dots,\alocation_n\}$ of locations in $\gbtw_k'(\amap_k(\alocation),\amap_k(\alocation'))$ that are labelled locations w.r.t.\ $\pair{\astore'}{\aheap'}$, i.e. $A = \gbtw_k'(\amap_k(\alocation),\amap_k(\alocation')) \cap \gverts'$.
  Since by Lemma~\ref{lemma-labels} we have $\alocation'' \in \gverts'$, we conclude $\alocation'' \in A$.
  Without loss of generality, we assume $\alocation_i$ and $\alocation_{i+1}$ ($i \in \interval{1}{n-1}$) to be two consecutive labelled locations in the minimal path from $\amap_k(\alocation)$ to $\amap_k(\alocation')$,
  i.e. none of the intermediate locations in the minimal path from $\alocation_i$ to $\alocation_{i+1}$ is labelled. Then, by minimality of $\sigma(\amap_k(\alocation),\amap_k(\alocation'))$ and from the definition of support graph, $\supportgraph{\astore',\aheap'}$ witness the following linear structure, where the arrow between two locations $\alocation_i$ and $\alocation_{i+1}$ labelled by $\gedges'$ means that $\gedges'(\alocation_i,\alocation_{i+1})$.
 \begin{center}
      \begin{tikzpicture}[baseline, node distance=1.8cm]
        \node[dot,label=above:{$\amap(\alocation)$}] (zero) at (0,0) {};
        \node[dot,label=above:{$\alocation_1$}] (one) [right of=zero] {};
        \node[dot,label=above:{$\alocation_2$}] (two) [right of=one] {};
        \node[dot,label=above:{$\alocation_{n-1}$}] (n1) [right of=two] {};
        \node[dot,label=above:{$\alocation_{n}$}] (n) [right of=n1] {};
        \node[dot,label=above right:{$\amap(\alocation')$}] (np) [right of=n] {};

         \draw[pto] (zero) edge node [above] {$\gedges'$} (one);
         \draw[pto] (one) edge node [above] {$\gedges'$} (two);
         \draw[pto] (n1) edge node [above] {$\gedges'$} (n);
         \draw[pto] (n) edge node [above] {$\gedges'$} (np);
         \draw[pto] (one) edge node [above] {$\gedges'$} (two);

         \path (two) edge [draw=none] node {$\dots$} (n1);

      \end{tikzpicture}
 \end{center}
 Since $\amap$ is a graph isomorphism from $\pair{\gverts}{\gedges}$ to $\pair{\gverts'}{\gedges'}$ (by~\ref{a1}), $\supportgraph{\astore,\aheap}$ witnesses a similar structure, as depicted below:

 \begin{center}
  \scalebox{1}{
      \begin{tikzpicture}[baseline, node distance=1.8cm]
        \node[dot,label=above:{$\amap(\alocation)$}] (zero) at (0,0) {};
        \node[dot,label=above:{$\alocation_1$}] (one) [right of=zero] {};
        \node[dot,label=above:{$\alocation_2$}] (two) [right of=one] {};
        \node[dot,label=above:{$\alocation_{n-1}$}] (n1) [right of=two] {};
        \node[dot,label=above:{$\alocation_{n}$}] (n) [right of=n1] {};
        \node[dot,label=above right:{$\amap(\alocation')$}] (np) [right of=n] {};
        \node[dot,label=below:{$\alocation$}] (Fzero) [below of=zero] {};
        \node[dot,label=below:{$\amap^{-1}(\alocation_1)$}] (Fone) [right of=Fzero] {};
        \node[dot,label=below:{$\amap^{-1}(\alocation_2)$}] (Ftwo) [right of=Fone] {};
        \node[dot,label=below:{$\amap^{-1}(\alocation_{n-1})$}] (Fn1) [right of=Ftwo] {};
        \node[dot,label=below:{$\amap^{-1}(\alocation_{n})$}] (Fn) [right of=Fn1] {};
        \node[dot,label=below:{$\alocation'$}] (Fnp) [right of=Fn] {};

         \draw[pto] (zero) edge node [above] {$\gedges'$} (one);
         \draw[pto] (one) edge node [above] {$\gedges'$} (two);
         \draw[pto] (n1) edge node [above] {$\gedges'$} (n);
         \draw[pto] (n) edge node [above] {$\gedges'$} (np);
         \draw[pto] (one) edge node [above] {$\gedges'$} (two);
         \draw[pto] (Fzero) edge node [below] {$\gedges$} (Fone);
         \draw[pto] (Fone) edge node [below] {$\gedges$} (Ftwo);
         \draw[pto] (Fn1) edge node [below] {$\gedges$} (Fn);
         \draw[pto] (Fn) edge node [below] {$\gedges$} (Fnp);
         \draw[pto] (Fone) edge node [below] {$\gedges$} (Ftwo);

         \path (two) edge [draw=none] node {$\dots$} (n1);
         \path (Ftwo) edge [draw=none] node {$\dots$} (Fn1);

         \path[dotted,->] (zero) edge node [left] {$\amap^{-1}$} (Fzero);
         \path[dotted,->] (one) edge node [left] {$\amap^{-1}$} (Fone);
         \path[dotted,->] (np) edge node [left] {$\amap^{-1}$} (Fnp);
         \path[dotted,->] (n) edge node [left] {$\amap^{-1}$} (Fn);
         \path[dotted,->] (n1) edge node [left] {$\amap^{-1}$} (Fn1);
         \path[dotted,->] (two) edge node [left] {$\amap^{-1}$} (Ftwo);

      \end{tikzpicture}
  }
 \end{center}
 Since $\amap$ is a bijection, for all distinct $i,j \in \interval{1}{n}$, we have $\amap^{-1}(\alocation_1) \neq \amap^{-1}(\alocation_n)$.
 Thus, every $\amap^{-1}(\alocation_i)$ ($i \in \interval{1}{n}$) is in the minimal path from $\alocation$ to $\alocation'$ in $\aheap$, and therefore from $\pair{\alocation}{\alocation'} \in \gedges_k$ we obtain that
 $\{\amap^{-1}(\alocation_1),\dots,\amap^{-1}(\alocation_n)\} \subseteq \gbtw_k(\alocation,\alocation')$.
 This leads to a contradiction. 
 Indeed, from $\alocation'' \in A$ we conclude $\amap^{-1}(\alocation'') \in \gbtw_k(\alocation,\alocation')$. Then, as 
$\amap_k^{-1}(\alocation'') \in \gverts_k$ (by~\ref{mapksatA3-body}), we conclude that $\aheap_k$ witnesses a labelled intermediate 
location in the minimal non-empty path from $\alocation$ to $\alocation'$
 in contradiction with $\pair{\alocation}{\alocation'} \in \gedges_k$.
 Hence, in $\aheap_k'$ there are no labelled locations in the minimal path from $\amap_k(\alocation)$ to $\amap_k(\alocation')$, allowing us to conclude that $\pair{\amap_k(\alocation)}{\amap_k(\alocation')} \in \gedges'$.

 ($\Leftarrow$) The right-to-left direction is analogous (thanks to the fact that $\amap$ and $\amap_k$ are bijections).

It remains to check the satisfaction of the conditions  \descLab{(A4)}{mapksatA4-body} and  \descLab{(A5)}{mapksatA5-body} that
involve arithmetical constraints. 

 \item[\textbf{$\amap_k$ satisfies \descLab{(A4)}{mapksatA4-body}}:]
   We show that for every $\pair{\alocation}{\alocation'} \in \gedges_k$ we have
   \[\min(\alpha_k,\card{\gbtw_k(\alocation,\alocation')}) = \min(\alpha_k,\card{\gbtw_k'(\amap_k(\alocation),\amap_k(\alocation'))}).\]

   First, we show the following intermediate result.

   \begin{enumerate}[label=\textbf{($\amap_k$-A4.\Roman*)}, align = left]
     \item\label{path-prop1} Let $\pair{\alocation}{\alocation'} \in \gedges_k$.
     There is a (possibly empty) set $\{\alocation_1,\dots,\alocation_n\} \subseteq \gverts_k$ such that
     \vspace{2pt}
     \begin{enumerate}[label=(\alph*), align = left]
       \setlength{\itemsep}{3pt}
       \item\label{path-inter-1}
        By defining $\alocation_0 \egdef \alocation$ and $\alocation_{n+1} \egdef \alocation'$, we have
        that for every $i \in \interval{0}{n}$, $\pair{\alocation_i}{\alocation_{i+1}} \in \gedges$;
       \item\label{path-inter-break}
        $\{\alocation_1,\dots,\alocation_n\} = \gbtw_k(\alocation,\alocation') \cap \gverts$ and $\{\amap(\alocation_1),\dots,\amap(\alocation_n)\} = \gbtw_k'(\amap(\alocation),\amap(\alocation'))  \cap \gverts'$;
       \item\label{path-inter-2}
        $\gbtw_k(\alocation,\alocation') = \{\alocation_1,\dots,\alocation_n\} \cup \bigcup_{i \in \interval{0}{n}} \gbtw(\alocation_i,\alocation_{i+1})$;
       \item\label{path-inter-3}
        $\gbtw_k'(\amap(\alocation),\amap(\alocation')) = \{\amap(\alocation_1),\dots,\amap(\alocation_n)\} \cup \bigcup_{i \in \interval{0}{n}} \gbtw'(\amap(\alocation_i),\amap(\alocation_{i+1}))$.
     \end{enumerate}
     \vspace{4pt}
     \item[(Proof of \ref{path-prop1})]
       The proof follows rather closely the arguments used for the proof of \ref{mapksatA1}.
       Suppose $\pair{\alocation}{\alocation'} \in \gedges_k$. Since $\amap_k$ satisfies \ref{mapksatA1} we have $\pair{\amap_k(\alocation)}{\amap_k(\alocation')} \in \gedges_k'$.
       By~Lemma~\ref{lemma-labels}, $\gverts_k \subseteq \gverts$ and $\gverts_k' \subseteq \gverts'$. Informally, this means that a subheap cannot contain new labelled locations w.r.t.\ the original heap.
       It can however be the case that the set $\gbtw_k(\alocation,\alocation')$ contains locations that are
       labelled w.r.t. $\pair{\astore}{\aheap}$ (i.e. locations in $\gverts$) and that are not labelled for $\pair{\astore}{\aheap_k}$ (by definition of $\gbtw_k(\alocation,\alocation')$).
       More precisely, these locations correspond to meet-points of $\aheap$.
       Hence, let us consider the set
       $
       \{\alocation_1,\dots,\alocation_n \} =  \gbtw_k(\alocation,\alocation') \cap \gverts
       $
       (as required by~\ref{path-inter-break}).
       Let us define $\alocation_0 \egdef \alocation$, $\alocation_{n+1} \egdef \alocation'$.
       Since $\gbtw_k(\alocation,\alocation')$ describes the set of locations of the minimal path from $\alocation$ to $\alocation'$ in $\aheap_k$, and moreover $\aheap_k \sqsubseteq \aheap$,
       there is an ordering on the locations $\alocation_1,\dots,\alocation_n$,
       w.l.o.g.\ say $\alocation_1 < \dots < \alocation_n$ such that
      for every $i \in \interval{0}{n}$ $\pair{\alocation_i}{\alocation_{i+1}} \in \gedges$ (property~\ref{path-inter-1} of \ref{path-prop1}).
      So, $\{\alocation_1,\dots,\alocation_n \}$ is a set of locations in $\gbtw_k(\alocation,\alocation')$ that correspond to meet-points of $\pair{\astore}{\aheap}$.
        Now, as done in the proof of~\ref{mapksatA1}, $\amap$ is a graph isomorphism from $\pair{\gverts}{\gedges}$ to $\pair{\gverts'}{\gedges'}$ and therefore these two structures witness the following correspondence

        \begin{center}
         \scalebox{1}{
             \begin{tikzpicture}[baseline, node distance=1.8cm]
               \node[dot,label=above left:{$\alocation={\alocation_0}$}] (zero) at (0,0) {};
               \node[dot,label=above:{$\alocation_1$}] (one) [right of=zero] {};
               \node[dot,label=above:{$\alocation_2$}] (two) [right of=one] {};
               \node[dot,label=above:{$\alocation_{n-1}$}] (n1) [right of=two] {};
               \node[dot,label=above:{$\alocation_{n}$}] (n) [right of=n1] {};
               \node[dot,label=above right:{$\alocation'={\alocation_{n+1}}$}] (np) [right of=n] {};
               \node[dot,label=below:{$\amap(\alocation)$}] (Fzero) [below of=zero] {};
               \node[dot,label=below:{$\amap(\alocation_1)$}] (Fone) [right of=Fzero] {};
               \node[dot,label=below:{$\amap(\alocation_2)$}] (Ftwo) [right of=Fone] {};
               \node[dot,label=below:{$\amap(\alocation_{n-1})$}] (Fn1) [right of=Ftwo] {};
               \node[dot,label=below:{$\amap(\alocation_{n})$}] (Fn) [right of=Fn1] {};
               \node[dot,label=below:{$\amap(\alocation')$}] (Fnp) [right of=Fn] {};

                \draw[pto] (zero) edge node [above] {$\gedges$} (one);
                \draw[pto] (one) edge node [above] {$\gedges$} (two);
                \draw[pto] (n1) edge node [above] {$\gedges$} (n);
                \draw[pto] (n) edge node [above] {$\gedges$} (np);
                \draw[pto] (one) edge node [above] {$\gedges$} (two);
                \draw[pto] (Fzero) edge node [below] {$\gedges'$} (Fone);
                \draw[pto] (Fone) edge node [below] {$\gedges'$} (Ftwo);
                \draw[pto] (Fn1) edge node [below] {$\gedges'$} (Fn);
                \draw[pto] (Fn) edge node [below] {$\gedges'$} (Fnp);
                \draw[pto] (Fone) edge node [below] {$\gedges'$} (Ftwo);

                \path (two) edge [draw=none] node {$\dots$} (n1);
                \path (Ftwo) edge [draw=none] node {$\dots$} (Fn1);

                \path[dotted,->] (zero) edge node [left] {$\amap$} (Fzero);
                \path[dotted,->] (one) edge node [left] {$\amap$} (Fone);
                \path[dotted,->] (np) edge node [left] {$\amap$} (Fnp);
                \path[dotted,->] (n) edge node [left] {$\amap$} (Fn);
                \path[dotted,->] (n1) edge node [left] {$\amap$} (Fn1);
                \path[dotted,->] (two) edge node [left] {$\amap$} (Ftwo);

             \end{tikzpicture}
         }
        \end{center}
        \descLab{($\star$)}{prop-star-prime}: where in particular for every $i \in \interval{0}{n}$,
        $\pair{\amap(\alocation_i)}{\amap(\alocation_{i+1})} \in \gedges'$.
        Notice that,
        since $\amap$ is a bijection, for all two distinct $i,j \in \interval{1}{n}$, we have $\amap(\alocation_i) \neq \amap(\alocation_j)$.
        Most importantly, as $\alocation,\alocation' \not\in \gbtw_k(\alocation,\alocation')$ (this holds by definition of this set), we conclude $\alocation,\alocation' \not\in\{\alocation_1,\dots,\alocation_n\}$ and therefore
        $\amap(\alocation),\amap(\alocation') \not\in\{\amap(\alocation_1),\dots,\amap(\alocation_n)\}$.
        Thanks to this property, $\{\amap(\alocation_1),\dots,\amap(\alocation_n)\}$ is the set of labelled locations in the minimal path from $\amap(\alocation)$ to $\amap(\alocation')$, in $\aheap'$.
        Equivalently, $\{\amap(\alocation_1),\dots,\amap(\alocation_n)\} = \gbtw_k'(\amap(\alocation),\amap(\alocation'))  \cap \gverts'$ (property~\ref{path-inter-break} of~\ref{path-prop1}).
        Lastly, by \ref{path-inter-1} and \ref{path-inter-break},
        from the fact that
        heaps are functional and
        $\aheap_k$ witnesses a path from $\alocation$ to $\alocation'$, it follows that
        for every $i \in \interval{0}{n}$ $\gbtw{(\alocation_i,\alocation_{i+1}}) \subseteq \gbtw(\alocation,\alocation')$.
        Similarly, for every $i \in \interval{0}{n}$ $\gbtw{(\amap(\alocation_i),\amap(\alocation_{i+1}})) \subseteq \gbtw(\amap(\alocation),\amap(\alocation'))$.
        Therefore, we conclude that
        \vspace{3pt}
        \begin{itemize}
        \item
         $\gbtw_k(\alocation,\alocation') = \{\alocation_1,\dots,\alocation_n\} \cup \bigcup_{i \in \interval{0}{n}} \gbtw(\alocation_i,\alocation_{i+1})$;
        \item
         $\gbtw_k'(\amap(\alocation),\amap(\alocation')) = \{\amap(\alocation_1),\dots,\amap(\alocation_n)\} \cup \bigcup_{i \in \interval{0}{n}} \gbtw'(\amap(\alocation_i),\amap(\alocation_{i+1}))$.
       \end{itemize}
       \vspace{4pt}
       ending the proof of~\ref{path-prop1}.
       Notice that the two properties \ref{path-inter-2} and \ref{path-inter-3} (now proved) are well depicted by the figure above, as it shows the structure of the minimal path from $\alocation$ to $\alocation'$ in $\aheap$ (and $\aheap_k$), and the structure of the minimal path from $\amap(\alocation)$ to $\amap(\alocation')$ in $\aheap'$ (and $\aheap_k'$).
   \end{enumerate}

   We are now ready to show that $\amap_k$ satisfies \ref{a4}. Let $\pair{\alocation}{\alocation'} \in \gedges_k$ and let $\{\alocation_1,\dots,\alocation_n\} \subseteq \gverts_k$ satisfying the four properties in \ref{path-prop1}.
   Let us define $\alocation_0 \egdef \alocation$, $\alocation_{n+1} \egdef \alocation'$.
   As already stated in the proof of~\ref{path-prop1},
   $\amap$ is a bijection and therefore for all two distinct $i,j \in \interval{1}{n}$, we have $\amap(\alocation_1) \neq \amap(\alocation_n)$.
   Hence,
   \[
     \card{\{\alocation_1,\dots,\alocation_n\}} = \card{\{\amap(\alocation_1),\dots,\amap(\alocation_n)\}} = n.
    \]
  Moreover, by recalling (from Definition~\ref{definition-support-graph}) that
  \begin{itemize}
    \item $\{\galloc,\grem\} \cup \{ \gbtw(\alocation,\alocation') \mid
    (\alocation,\alocation') \in \gedges \}$ is a partition of $\domain{\aheap}$;
    \item $\{\galloc',\grem'\} \cup \{ \gbtw'(\alocation,\alocation') \mid
    (\alocation,\alocation') \in \gedges' \}$ is a partition of $\domain{\aheap'}$,
  \end{itemize}
   we conclude that for every $i \in \interval{0}{n}$,
   \begin{itemize}
    \item $\{\alocation_1,\dots,\alocation_n\} \cap \gbtw(\alocation_i,\alocation_{i+1}) = \emptyset$ and
    for all $j \in \interval{0}{n} \setminus \{i\}$, $\gbtw(\alocation_i,\alocation_{i+1}) \cap \gbtw(\alocation_j,\alocation_{j+1}) = \emptyset$;
    \item $\{\amap(\alocation_1),\dots,\amap(\alocation_n)\} \cap \gbtw'(\amap(\alocation_i),\amap(\alocation_{i+1})) = \emptyset$ and
    for all $j \in \interval{0}{n} \setminus \{i\}$,
    $\gbtw'(\amap(\alocation_i),\amap(\alocation_{i+1})) \cap \gbtw'(\amap(\alocation_j),\amap(\alocation_{j+1})) = \emptyset$.
   \end{itemize}
   Thus, the cardinalities of $\gbtw_k(\alocation,\alocation')$ and $\gbtw_k'(\amap_k(\alocation),\amap_k(\alocation'))$ can be written respectively as
   \begin{align*}
     \card{\gbtw_k(\alocation,\alocation')} = n + &\sum_{\mathclap{i \in \interval{0}{n}}} \card{\gbtw(\alocation_i,\alocation_{i+1})},\\
     \card{\gbtw_k'(\amap_k(\alocation),\amap_k(\alocation'))} = n + &\sum_{\mathclap{i \in \interval{0}{n}}} \card{\gbtw(\amap(\alocation_i),\amap(\alocation_{i+1}))}.
   \end{align*}
 Now, $\amap$ satisfies~\ref{a4} w.r.t.\ the memory states $\pair{\astore}{\aheap}$ and $\pair{\astore'}{\aheap'}$, and therefore for every $i \in \interval{0}{n}$
  \begin{equation}
    \tag*{\textbf{($\amap$-A4)}}\label{amap-A4}
    \min(\alpha,\card{\gbtw(\alocation_i,\alocation_{i+1})}) =
    \min(\alpha,\card{\gbtw'(\amap(\alocation_i),\amap(\alocation_{i+1}))}).
  \end{equation}
Now, we distinguish two cases. 
 \begin{itemize}
     \item If there exists $i \in \interval{0}{n}$ such that $\card{\gbtw(\alocation_i,\alocation_{i+1})} \geq \alpha$, then
     by \ref{amap-A4} we have $\card{\gbtw'(\amap(\alocation_i),\amap(\alocation_{i+1}))} \geq \alpha$.
     Since $\alpha_k \leq \alpha$, we  conclude
     \[
      \card{\gbtw_k(\alocation,\alocation')} \geq \alpha_k; \qquad
      \card{\gbtw_k'(\amap_k(\alocation),\amap_k(\alocation'))} \geq \alpha_k.
     \]
     \item Otherwise (for all $i \in \interval{0}{n}, \card{\gbtw(\alocation_i,\alocation_{i+1})} < \alpha$)
     from \ref{amap-A4} we conclude that for all $i \in \interval{0}{n}$,
     \[\card{\gbtw(\alocation_i,\alocation_{i+1})} = \card{\gbtw'(\amap(\alocation_i),\amap(\alocation_{i+1}))},\]
     which allows us to conclude that $\card{\gbtw_k(\alocation,\alocation')} = \card{\gbtw_k'(\amap_k(\alocation),\amap_k(\alocation'))}$.
  \end{itemize}
 In both cases we have shown that
 \[\min(\alpha_k,\card{\gbtw_k(\alocation,\alocation')}) = \min(\alpha_k,\card{\gbtw_k'(\amap_k(\alocation),\amap_k(\alocation'))}).\]
 concluding the proof of \ref{a4}.

  \item[\textbf{$\amap_k$ satisfies \descLab{(A5)}{mapksatA5-body}}:]
    Lastly, we show that
    \[
    \min(\alpha_k, \card{\grem_k}) = \min(\alpha_k, \card{\grem_k'}).
    \]
    We take advantage of the following five intermediate results.
    \begin{enumerate}[label=\textbf{(\Roman*)}, align = left]
    \item\label{rem-prop1} For every $\alocation \in \gverts \cap \domain{\aheap_k}$,
      $\alocation \in \grem_k$ if and only if $\amap(\alocation) \in \grem_k'$.

    \item[(Proof of \ref{rem-prop1})]
    ($\Rightarrow$) Suppose $\alocation \in \gverts$.
    By Definition~\ref{definition-support-graph},
    it holds that $\alocation \in \grem_k$ if and only if
    \begin{enumerate}[label=(\roman*), align = left]
      \item\label{rem-prop-l1} $\alocation \not\in \gverts_k$ and $\alocation \in \domain{\aheap_k}$;
      \item\label{rem-prop-l2} there is no $\pair{\alocation'}{\alocation''} \in \gedges_k$ such that $\alocation \in \gbtw_k(\alocation',\alocation'')$.
    \end{enumerate}
    From (a) in the proof that $\amap_k$ satisfies \ref{mapksatA3-body}, we have $\amap(\alocation) = \amap_k(\alocation) \not\in \gverts_k'$. By $\alocation \in \domain{\aheap_k}$ and $\alocation \in \grem_k$, from the construction step~\ref{c1-body} we conclude $\amap(\alocation) \in \domain{\aheap_k'}$.

    Besides, it cannot be that there is
    $\pair{\alocation_1}{\alocation_2} \in \gedges_k'$ such that $\amap(\alocation) \in \gbtw_k'(\alocation_1,\alocation_2)$. Indeed, if that was the case, we can
    apply the proof ad absurdum showed in the left-to-right direction of the proof that
    $\amap_k$ satisfies \ref{mapksatA1}, and derive that 
    $\alocation \in \gbtw_k'(\amap^{-1}_k(\alocation_1),\amap^{-1}_k(\alocation_2))$, in contradiction
    with~\ref{rem-prop-l2}.
    Hence, as we proved that \ref{rem-prop-l1} and \ref{rem-prop-l2} hold for $\amap(\alocation)$, by Definition~\ref{definition-support-graph}, we conclude that $\amap(\alocation) \in \grem_k'$.

    ($\Leftarrow$) The right-to-left direction follows by using similar arguments.

    \item\label{rem-prop-break}  $\card{\grem_k \cap \gverts} = \card{\grem_k' \cap \gverts'}$ and $\card{\grem_k \cap \galloc} = \card{\grem_k' \cap \galloc'}$.
    \item[(Proof of \ref{rem-prop-break})]
    The first equivalence follows directly from \ref{rem-prop1} and the fact that $\amap$ is a bijection.
    Then, since $\galloc \subseteq \gverts$, $\galloc' \subseteq \gverts'$ (by definition) and
    $\alocation \in \galloc\ \Leftrightarrow\ \amap(\alocation) \in \galloc'$ (from the property~\ref{a2} of $\amap$), the second equivalence also holds.

    \item\label{rem-prop2} For all $\pair{\alocation}{\alocation'} \in \gedges$,
    \begin{center}
      $\gbtw(\alocation,\alocation') \cap \domain{\aheap_k} \subseteq \grem_k$ if and only if
     $\gbtw'(\amap(\alocation),\amap(\alocation')) \cap \domain{\aheap_k'} \subseteq \grem_k'$.
    \end{center}
    \item[(Proof of \ref{rem-prop2})]
    We first notice that, directly from \ref{cip2-body},
    \begin{center}
      $\gbtw(\alocation,\alocation') \cap \domain{\aheap_k} = \emptyset$ if and only if
      $\gbtw'(\amap(\alocation),\amap(\alocation')) \cap \domain{\aheap_k'} = \emptyset$,
    \end{center}
    as by definition
    $L_k = \gbtw(\alocation,\alocation') \cap \domain{\aheap_k}$
    and $L_k' = \gbtw'(\amap(\alocation),\amap(\alocation')) \cap \domain{\aheap_k'}$ in the case~\ref{c2-body} of the construction (where we consider $\gbtw(\alocation,\alocation')$).
    Therefore,~\ref{rem-prop2} trivially holds when $\gbtw(\alocation,\alocation') \cap \domain{\aheap_k}$ is empty and in the following, we assume
    $\gbtw(\alocation,\alocation') \cap \domain{\aheap_k}$ and $\gbtw'(\amap(\alocation),\amap(\alocation')) \cap \domain{\aheap_k'}$ both non-empty.

    $(\Rightarrow)$
    Let $\pair{\alocation}{\alocation'} \in \gedges$ and suppose
    $\gbtw(\alocation,\alocation') \cap \domain{\aheap_k} \subseteq \grem_k$.
    Let us consider a location $\tilde{\alocation} \in \gbtw'(\amap(\alocation),\amap(\alocation')) \cap \domain{\aheap_k'}$.
    We show that $\tilde{\alocation} \in \grem_k'$.
    As shown in the proof of \ref{rem-prop1}, $\tilde{\alocation} \in \grem_k'$ holds if and only if
    \begin{enumerate}[label=(\roman*), align = left]
      \item\label{rem-prop-l3} $\tilde{\alocation} \not\in \gverts_k'$ and $\tilde{\alocation} \in \domain{\aheap_k'}$;
      \item\label{rem-prop-l4} there is no $\pair{\tilde{\alocation_1}}{\tilde{\alocation_2}} \in \gedges_k'$ such that $\tilde{\alocation} \in \gbtw_k'(\tilde{\alocation_1},\tilde{\alocation_2})$.
    \end{enumerate}
    The proof of~\ref{rem-prop-l3} is straightforward:
    by $\tilde{\alocation} \in \gbtw'(\amap(\alocation),\amap(\alocation')) \cap \domain{\aheap_k'}$, we directly conclude that $\tilde{\alocation} \in \domain{\aheap_k'}$ (first condition in~\ref{rem-prop-l3}) and
    $\tilde{\alocation} \in \gbtw'(\amap(\alocation),\amap(\alocation'))$.
    From the latter membership and
    by Definition~\ref{definition-support-graph}, we conclude $\tilde{\alocation} \not\in \gverts'$.
    Lastly, by Lemma~\ref{lemma-labels}, $\gverts_k' \subseteq \gverts'$ and therefore $\tilde{\alocation} \not\in \gverts_k'$ (second condition in~\ref{rem-prop-l3}).
    Let us now prove~\ref{rem-prop-l4}.
    \emph{Ad absurdum}, suppose that there is $\pair{\tilde{\alocation_1}}{\tilde{\alocation_2}} \in \gedges_k'$ such that $\tilde{\alocation} \in \gbtw_k'(\tilde{\alocation_1},\tilde{\alocation_2})$.
    Then, as we have shown that $\amap_k$ satisfies \ref{mapksatA4-body}, $\pair{\amap^{-1}_k(\tilde{\alocation_1})}{\amap^{-1}_k(\tilde{\alocation_n})} \in \gedges_k$.
    We apply \ref{path-prop1}: there is a set $\{\alocation_1,\dots,\alocation_n\} \subseteq \gverts_k$ such that
    \vspace{2pt}
    \begin{enumerate}[label=(\alph*), align = left]
      \setlength{\itemsep}{3pt}
      \item\label{rem-path-inter-1}
       By defining $\alocation_0 \egdef \amap^{-1}_k(\tilde{\alocation_1})$ and $\alocation_{n+1} \egdef \amap^{-1}_k(\tilde{\alocation_2})$, for every $i \in \interval{0}{n}$, $\pair{\alocation_i}{\alocation_{i+1}} \in \gedges$;
      \item\label{rem-path-inter-break}
       $\{\alocation_1,\dots,\alocation_n\} = \gbtw_k(\amap^{-1}_k(\tilde{\alocation_1}),\amap^{-1}_k(\tilde{\alocation_2})) \cap \gverts$ and
       $\{\amap(\alocation_1),\dots,\amap(\alocation_n)\} = \gbtw_k'(\tilde{\alocation_1},\tilde{\alocation_2})  \cap \gverts'$;
      \item\label{rem-path-inter-2}
       $\gbtw_k(\amap^{-1}_k(\tilde{\alocation_1}),\amap^{-1}_k(\tilde{\alocation_2})) = \{\alocation_1,\dots,\alocation_n\} \cup \bigcup_{i \in \interval{0}{n}} \gbtw(\alocation_i,\alocation_{i+1})$;
      \item\label{rem-path-inter-3}
       $\gbtw_k'(\tilde{\alocation_1},\tilde{\alocation_2}) = \{\amap(\alocation_1),\dots,\amap(\alocation_n)\} \cup \bigcup_{i \in \interval{0}{n}} \gbtw'(\amap(\alocation_i),\amap(\alocation_{i+1}))$.
    \end{enumerate}
    \vspace{4pt}
    Recalling that
    $\{\galloc',\grem'\} \cup \{ \gbtw'(\alocation,\alocation') \mid
      (\alocation,\alocation') \in \gedges' \}$ is a partition of $\domain{\aheap'}$
    (by Definition~\ref{definition-support-graph}),
    from \ref{rem-path-inter-3} together with $\tilde{\alocation} \in \gbtw'(\amap(\alocation),\amap(\alocation'))$ and $\tilde{\alocation} \in \gbtw_k'(\tilde{\alocation_1},\tilde{\alocation_2})$, we conclude that
    there is $i \in \interval{0}{n}$ such that $\pair{\amap(\alocation)}{\amap(\alocation')} = \pair{\amap(\alocation_i)}{\amap(\alocation_{i+1})}$.
    From \ref{rem-path-inter-2}, we conclude that
    $\gbtw(\alocation,\alocation') \subseteq \gbtw_k(\amap_k^{-1}(\tilde{\alocation_1}),\amap_k^{-1}(\tilde{\alocation_2})).$
    Notice that, by Definition~\ref{definition-support-graph},
    this inclusion also entails $\gbtw(\alocation,\alocation') \cap \domain{\aheap_k} = \gbtw(\alocation,\alocation')$, i.e. every element in $\gbtw(\alocation,\alocation')$ is a memory cell of $\aheap_k$.
    Hence,
    $\gbtw(\alocation,\alocation') \cap \domain{\aheap_k} \subseteq \gbtw_k(\amap_k^{-1}(\tilde{\alocation_1}),\amap_k^{-1}(\tilde{\alocation_2}))$, in contradiction with $\gbtw(\alocation,\alocation') \cap \domain{\aheap_k} \subseteq \grem_k$.
    Indeed, we assumed $\gbtw(\alocation,\alocation') \cap \domain{\aheap_k}$ non-empty, and
    by Definition~\ref{definition-support-graph} $\grem_k \cap \gbtw_k(\amap_k^{-1}(\tilde{\alocation_1}),\amap_k^{-1}(\tilde{\alocation_2})) = \emptyset$.
    Therefore, it cannot be that there is $\pair{\tilde{\alocation_1}}{\tilde{\alocation_2}} \in \gedges_k'$ such that $\tilde{\alocation} \in \gbtw_k'(\tilde{\alocation_1},\tilde{\alocation_2})$.
    As \ref{rem-prop-l3} and \ref{rem-prop-l4} hold, we conclude that
    $\tilde{\alocation} \in \grem_k'$.

    ($\Leftarrow$) The right-to-left direction follows by using similar arguments.

    \item\label{rem-prop-aux2} For every $\pair{\alocation}{\alocation'} \in \gedges$, $\grem_k \cap \gbtw(\alocation,\alocation') = \emptyset$ or $\gbtw(\alocation,\alocation') \cap \domain{\aheap_k} \subseteq \grem_k$.
    Similarly, for every $\pair{\alocation}{\alocation'} \in \gedges'$, $\grem_k' \cap \gbtw'(\alocation,\alocation') = \emptyset$ or $\gbtw'(\alocation,\alocation') \cap \domain{\aheap_k'} \subseteq \grem_k'$.

    \item[(Proof of~\ref{rem-prop-aux2})] Let us prove the first statement (the second one is proven analogously). By Definition~\ref{definition-support-graph}, we have
    \begin{enumerate}[label = (\alph*), align = left]
      \item\label{last-rem-1-aux} $\{\galloc,\grem\} \cup \{ \gbtw(\alocation,\alocation') \mid
      (\alocation,\alocation') \in \gedges \}$ is a partition of $\domain{\aheap}$;
      \item\label{last-rem-2-aux} $\{\galloc_k,\grem_k\} \cup \{ \gbtw_k(\alocation,\alocation') \mid
      (\alocation,\alocation') \in \gedges_k \}$ is a partition of $\domain{\aheap_k}$.
    \end{enumerate}
    Let $\pair{\alocation}{\alocation'} \in \gedges$.
    \emph{Ad absurdum}, suppose that $\grem_k \cap \gbtw(\alocation,\alocation') \neq \emptyset$ and there is $\tilde{\alocation} \in (\gbtw(\alocation,\alocation') \cap \domain{\aheap_k}) \setminus \grem_k$.
    Since $\tilde{\alocation} \in \gbtw(\alocation,\alocation')$, 
    by Definition~\ref{definition-support-graph}, $\tilde{\alocation}$ is not a labelled location w.r.t. $\pair{\astore}{\aheap}$, i.e.\ $\tilde{\alocation} \not \in \gverts$.
    Moreover, from $\tilde{\alocation} \in \domain{\aheap_k}$ and $\tilde{\alocation} \not\in \grem_k$, by
    \ref{last-rem-2-aux} we conclude that $\tilde{\alocation} \in \galloc_k$ or there is $\pair{\alocation_1'}{\alocation_2'} \in \gedges_k$ such that $\tilde{\alocation} \in \gbtw_k(\alocation_1',\alocation_2')$.
    \begin{itemize}
      \item If $\tilde{\alocation} \in \galloc_k$, then $\tilde{\alocation} \in \gverts_k$ (by Definition~\ref{definition-support-graph}). Recall that,
      by Lemma~\ref{lemma-labels}, $\gverts_k \subseteq \gverts$.
      Therefore, we conclude $\tilde{\alocation} \in \gverts$ (contradiction).
      \item Otherwise there is $\pair{\alocation_1'}{\alocation_2'} \in \gedges_k$ such that $\alocation \in \gbtw_k(\alocation_1',\alocation_2')$. By~\ref{path-prop1}, there is a set $\{\alocation_1,\dots,\alocation_n\} \subseteq \gverts_k$ such that
      \vspace{2pt}
      \begin{enumerate}[label=(\alph*), align = left, start = 3]
        \setlength{\itemsep}{3pt}
        \item\label{rem-4-path-1}
         By defining $\alocation_0 \egdef \alocation_1'$ and $\alocation_{n+1} \egdef \alocation_2'$, we have
         that for every $i \in \interval{0}{n}$, $\pair{\alocation_i}{\alocation_{i+1}} \in \gedges$;
        \item\label{rem-4-path-2}
         $\{\alocation_1,\dots,\alocation_n\} = \gbtw_k(\alocation_1',\alocation_2') \cap \gverts$;
        \item\label{rem-4-path-3}
         $\gbtw_k(\alocation_1',\alocation_2') = \{\alocation_1,\dots,\alocation_n\} \cup \bigcup_{i \in \interval{0}{n}} \gbtw(\alocation_i,\alocation_{i+1})$.
      \end{enumerate}
      Since $\tilde{\alocation} \not\in \gverts$, from~\ref{rem-4-path-3} we conclude that there is $i \in \interval{0}{n}$ such that $\tilde{\alocation} \in \gbtw(\alocation_i,\alocation_{i+1})$.
      Then, by \ref{last-rem-1-aux} and $\alocation \in \gbtw(\alocation,\alocation')$, it must be that $\pair{\alocation_i}{\alocation_{i+1}} = \pair{\alocation}{\alocation'}$.
      Hence, from \ref{rem-4-path-3}, $\gbtw(\alocation,\alocation') \subseteq \gbtw_k(\alocation_1',\alocation_2')$.
      By~\ref{last-rem-2-aux} $\grem_k \cap \gbtw_k(\alocation_1',\alocation_2') = \emptyset$ and therefore we conclude that $\grem_k \cap \gbtw(\alocation,\alocation') = \emptyset$, in contradiction with the hypothesis $\grem_k \cap \gbtw(\alocation,\alocation') \neq \emptyset$.
    \end{itemize}
    As in both cases we reached a contradiction, it must be that
    $\grem_k \cap \gbtw(\alocation,\alocation') = \emptyset$ or $\gbtw(\alocation,\alocation') \cap \domain{\aheap_k} \subseteq \grem_k$, concluding the proof.

    \item\label{rem-prop3} It holds that $\grem \cap \domain{\aheap_k} \subseteq \grem_k$ and $\grem' \cap \domain{\aheap_k'} \subseteq \grem_k'$.
    \item[(Proof of \ref{rem-prop3})] The proof is rather immediate. Let us prove that
     $\grem \cap \domain{\aheap_k} \subseteq \grem_k$ (the proof of $\grem' \cap \domain{\aheap_k'} \subseteq \grem_k'$ is analogous).
     By Definition~\ref{definition-support-graph}, we have
     \begin{enumerate}[label = (\alph*), align = left]
       \item\label{last-rem-1} $\{\galloc,\grem\} \cup \{ \gbtw(\alocation,\alocation') \mid
       (\alocation,\alocation') \in \gedges \}$ is a partition of $\domain{\aheap}$;
       \item\label{last-rem-2} $\{\galloc_k,\grem_k\} \cup \{ \gbtw_k(\alocation,\alocation') \mid
       (\alocation,\alocation') \in \gedges_k \}$ is a partition of $\domain{\aheap_k}$.
     \end{enumerate}
     Suppose $\alocation \in \grem \cap \domain{\aheap_k}$.
     \emph{Ad absurdum}, suppose $\alocation \not \in \grem_k$.
     Then, as $\alocation \in \domain{\aheap_k}$, from \ref{last-rem-2} it holds that
     $\alocation \in \galloc_k$ or there is $(\alocation_1,\alocation_2) \in \gedges_k$ such that
     $\alocation \in \gbtw_k(\alocation_1,\alocation_2)$.
     \begin{itemize}
     \item
       If $\alocation \in \galloc_k$ then $\alocation \in \gverts_k$ (by Definition~\ref{definition-support-graph}).
       By Lemma~\ref{lemma-labels} $\gverts_k \subseteq \gverts$ and therefore $\alocation \in \gverts$.
       Together with $\aheap_k \sqsubseteq \aheap$ (hence $\domain{\aheap_k} \subseteq \domain{\aheap}$), the fact that $\alocation \in \gverts$ implies $\alocation \in \galloc$ (again by Definition~\ref{definition-support-graph}).
       From~\ref{last-rem-1} we then obtain $\alocation \not \in \grem$, a contradiction.
     \item
      Otherwise, let $(\alocation_1,\alocation_2) \in \gedges_k$ such that
      $\alocation \in \gbtw_k(\alocation_1,\alocation_2)$.
      Therefore, by Definition~\ref{definition-support-graph}
      there is $L_1,L_2 \geq 1$ such that $\aheap_k^{L_1}(\alocation_1) = \alocation$ and $\aheap_k^{L_2}(\alocation) = \alocation_2$.
      Since $\aheap_k \sqsubseteq \aheap$, we conclude that
      there is $L_1,L_2 \geq 1$ such that $\aheap^{L_1}(\alocation_1) = \alocation$ and $\aheap^{L_2}(\alocation) = \alocation_2$.
      Now, notice that $(\alocation_1,\alocation_2) \in \gedges_k$ implies $\alocation_1,\alocation_2 \in \gverts_k$ (by Definition~\ref{definition-support-graph}) and therefore
      $\alocation_1,\alocation_2 \in \gverts$ (since by Lemma~\ref{lemma-labels} $\gverts_k \subseteq \gverts$).
      We  conclude that $\alocation$ is an intermediate location in the path of $\aheap$ from
      the labelled location $\alocation_1$ to the labelled location $\alocation_2$.
      Hence, by Definition~\ref{definition-support-graph}
      there are $\tilde{\alocation_1}$ and $\tilde{\alocation_2}$ such that $\alocation \in \gbtw(\tilde{\alocation_1},\tilde{\alocation_2})$.
      From~\ref{last-rem-1} we then obtain $\alocation \not \in \grem$, a contradiction.
     \end{itemize}
     In both cases we conclude that $\alocation \not\in \grem_k$ cannot hold, ending the proof of \ref{rem-prop3}.
  \end{enumerate}

    By taking advantagr of \ref{rem-prop1}--\ref{rem-prop3}, we are now 
    ready to show that $\amap_k$ satisfies \ref{a5}, i.e.
    \[
      \min(\alpha_k, \card{\grem_k}) = \min(\alpha_k, \card{\grem_k'}).
    \]
    First of all, we recall that by Definition~\ref{definition-support-graph} we have
    \begin{enumerate}[label = (\alph*), align = left]
      \item\label{a5-rem-L} $\{\galloc,\grem\} \cup \{ \gbtw(\alocation,\alocation') \mid
      (\alocation,\alocation') \in \gedges \}$ is a partition of $\domain{\aheap}$;
      \item\label{a5-rem-R} $\{\galloc',\grem'\} \cup \{ \gbtw'(\alocation,\alocation') \mid
      (\alocation,\alocation') \in \gedges' \}$ is a partition of $\domain{\aheap'}$.
    \end{enumerate}
    Since $\grem_k \subseteq \domain{\aheap_k} \subseteq \domain{\aheap}$ and $\grem_k' \subseteq \domain{\aheap_k'} \subseteq \domain{\aheap'}$, we can use~\ref{a5-rem-L} and~\ref{a5-rem-R} to decompose $\grem_k$ and $\grem_k'$ as follows:
    \begin{itemize}[align = left]
      \item $\grem_k = (\grem_k \cap \galloc) \cup (\grem_k \cap \grem) \cup \bigcup_{{\pair{\alocation}{\alocation'} \in \gedges}} \big( \grem_k \cap \gbtw(\alocation,\alocation')\big)$,
      \item $\grem_k' = (\grem_k' \cap \galloc') \cup (\grem_k' \cap \grem') \cup \bigcup_{{\pair{\alocation}{\alocation'} \in \gedges'}} \big( \grem_k' \cap \gbtw'(\alocation,\alocation')\big)$,
    \end{itemize}
    where all unions are performed on disjoint sets.
    By \ref{rem-prop-aux2},
    $\bigcup_{{\pair{\alocation}{\alocation'} \in \gedges}} ( \grem_k \cap \gbtw(\alocation,\alocation'))$
    and
    $\bigcup_{{\pair{\alocation}{\alocation'} \in \gedges'}} ( \grem_k' \cap \gbtw'(\alocation,\alocation'))$
    are respectively equivalent to
    \[
      \qquad\qquad
      \bigcup_{\mathclap{\substack{\pair{\alocation}{\alocation'} \in \gedges\\ \gbtw(\alocation,\alocation') \cap \domain{\aheap_k} \subseteq \grem_k}}}
      \big(\gbtw(\alocation,\alocation') \cap \domain{\aheap_k}\big)
      \qquad\qquad\qquad\qquad
      \bigcup_{\mathclap{\substack{\pair{\alocation}{\alocation'} \in \gedges'\\ \gbtw'(\alocation,\alocation') \cap \domain{\aheap_k'} \subseteq \grem_k'}}}
      \big(\gbtw'(\alocation,\alocation') \cap \domain{\aheap_k'}\big)
    \]
    Since $\grem_k \subseteq \domain{\aheap_k}$ and $\grem_k' \subseteq \domain{\aheap_k'}$,
    by~\ref{rem-prop3} we have $\grem_k \cap \grem = \grem \cap \domain{\aheap_k}$ and $\grem_k' \cap \grem' = \grem' \cap \domain{\aheap_k'}$. Hence, the previous equivalences can be rewritten as
    \begin{align*}
      \grem_k = \ & (\grem_k \cap \galloc) \cup (\grem \cap \domain{\aheap_k}) \cup \bigcup_{\mathclap{\substack{\pair{\alocation}{\alocation'} \in \gedges\\ \gbtw(\alocation,\alocation') \cap \domain{\aheap_k} \subseteq \grem_k}}}
      \big(\gbtw(\alocation,\alocation') \cap \domain{\aheap_k}\big)\\
      \grem_k' = \ & (\grem_k' \cap \galloc') \cup (\grem' \cap \domain{\aheap_k'}) \cup \bigcup_{\mathclap{\substack{\pair{\alocation}{\alocation'} \in \gedges'\\ \gbtw'(\alocation,\alocation') \cap \domain{\aheap_k'} \subseteq \grem_k'}}}
      \big(\gbtw'(\alocation,\alocation') \cap \domain{\aheap_k'}\big).
    \end{align*}
    Since all unions are performed on disjoint sets (by \ref{a5-rem-L} and \ref{a5-rem-R}),
    we conclude that $\card{\grem_k}$ and $\card{\grem_k'}$ satisfy the following equivalences:
    \begin{align*}
      \qquad \
      \card{\grem_k} = \ & \card{\grem_k \cap \galloc} + \card{\grem \cap \domain{\aheap_k}} +  \!\!\!\sum_{\mathclap{\substack{\pair{\alocation}{\alocation'} \in \gedges\\ \gbtw(\alocation,\alocation') \cap \domain{\aheap_k} \subseteq \grem_k}}}
      \!\card{\gbtw(\alocation,\alocation') \cap \domain{\aheap_k}}\\
      \card{\grem_k'} = \ & \card{\grem_k'\cap \galloc'} + \card{\grem'\!\cap \domain{\aheap_k}} + \!\!\!\sum_{\mathclap{\substack{\pair{\alocation}{\alocation'} \in \gedges'\\ \gbtw'(\alocation,\alocation') \cap \domain{\aheap_k'} \subseteq \grem_k'}}}
      \!\card{\gbtw'(\alocation,\alocation') {\cap} \domain{\aheap_k'}}
    \end{align*}
    Now, we are able to compare the values $\min(\alpha_k, \card{\grem_k})$ and 
    $\min(\alpha_k, \card{\grem_k'})$.
    By using the two equivalences above and recalling that $\min(A,B+C) = \min(A,B + \min(A,C))$,
    the following set of equivalences holds for $\min(\alpha_k, \card{\grem_k})$:
    \begin{align*}
        \min(\alpha_k, \card{\grem_k}) &= \min(\alpha_k, \card{\grem_k \cap \galloc} + R)\\
        R &= \min(\alpha_k, \card{\grem \cap \domain{h_k}} + I)\\
        I &= \min(\alpha_k, \sum_{\mathclap{\substack{\pair{\alocation}{\alocation'} \in \gedges\\ \gbtw(\alocation,\alocation') \cap \domain{\aheap_k} \subseteq \grem_k}}} I_{\alocation,\alocation'})\\
        I_{\alocation,\alocation'} &= \min(\alpha_k,\card{\gbtw(\alocation,\alocation') \cap \domain{\aheap_k}})
    \end{align*}
    The same can be done for $\min(\alpha_k, \card{\grem_k'})$:
    \begin{align*}
        \min(\alpha_k, \card{\grem_k'}) &= \min(\alpha_k, \card{\grem_k' \cap \galloc'} + R')\\
        R' &= \min(\alpha_k, \card{\grem' \cap \domain{h_k'}} + I')\\
        I' &= \min(\alpha_k, \sum_{\mathclap{\substack{\pair{\alocation}{\alocation'} \in \gedges'\\ \gbtw'(\alocation,\alocation') \cap \domain{\aheap_k'} \subseteq \grem_k'}}} I_{\alocation,\alocation'}')\\
        I_{\alocation,\alocation'}' &= \min(\alpha_k,\card{\gbtw'(\alocation,\alocation') \cap \domain{\aheap_k'}}).
    \end{align*}
    By \ref{cip2-body}, for every $\pair{\alocation}{\alocation'} \in \gedges$ we have $I_{\alocation,\alocation'} = I_{\amap(\alocation),\amap(\alocation')}'$.
    From \ref{rem-prop2} together with the fact that $\amap$ witnesses a graph isomorphism between $\pair{\astore}{\aheap}$ and $\pair{\astore'}{\aheap'}$
    we conclude that
    \begin{align*}
      &\{\pair{\alocation}{\alocation'} \in \gedges' \mid \gbtw'(\alocation,\alocation') \cap \domain{\aheap_k'} \subseteq \grem_k' \}\\
      = \ &
      \{\pair{\amap(\alocation)}{\amap(\alocation')} \mid \pair{\alocation}{\alocation'} \in \gedges, \gbtw(\alocation,\alocation') \cap \domain{\aheap_k} \subseteq \grem_k \},
    \end{align*}
    which allows us to conclude that $I = I'$.
    Furthermore, by \ref{crp2-body}, it follows also that $R = R'$.
    Lastly, by \ref{rem-prop-break} we conclude:
    \[\min(\alpha_k, \card{\grem_k}) = \min(\alpha_k, \card{\grem_k'}).\]
}
\end{enumerate}


This concludes the proof: for the support graphs of $h_k$ and $h_k'$, $\amap$ restricted to the domain $\gverts_k$ and the codomain $\gverts_k'$ satisfies all conditions of Lemma~\ref{lemma-test-graph} with respect to $\alpha_k$.
Therefore, it holds that $\pair{\astore}{\aheap_1}  \approx^q_{\alpha_1} \pair{\astore'}{\aheap'_1}$ and
$\pair{\astore}{\aheap_2}  \approx^q_{\alpha_2} \pair{\astore'}{\aheap'_2}$.
\end{proof}

Thanks to Lemma~\ref{lemma-star}, we can now characterise every formula of
$\seplogic{\separate,\reachplus}$ by a Boolean combination of test formulae in $\Test(q,\alpha)$, where $\alpha$ is related to the \defstyle{memory size}
$\msize(\aformula)$ of a formula $\aformula$ in $\seplogic{\separate,\reachplus}$ defined as follows (see also~\cite{Yang01}):
\begin{multicols}{2}
\begin{itemize}
\item $\msize(\aatomicformula) \egdef 1$ for any atomic formula $\aatomicformula$,
\item $\msize(\aformulabis \land \aformulabis') \egdef \max(\msize(\aformulabis),\msize(\aformulabis'))$,
\item $\msize(\lnot \aformulabis) \egdef \msize(\aformulabis)$,
\item $\msize(\aformulabis \sepcnj \aformulabis') \egdef \msize(\aformulabis) + \msize(\aformulabis')$.
\end{itemize}
\end{multicols}
We have $1 \leq \msize(\aformula) \leq \length{\aformula}$, where $\length{\aformula}$ is the size of the syntax tree 
for $\aformula$.
Below, we establish the characterisation of $\seplogic{\separate,\reachplus}$ formulae in terms of test formulae.


\begin{theorem}\label{theo:test_sl_equiv}
Let $\aformula$ be in $\seplogic{\separate,\reachplus}$ built over the variables $\avariable_1,\dots,\avariable_q$.
For all $\alpha \geq \msize(\aformula)$ and all memory states $\pair{\astore}{\aheap}$,
$\pair{\astore'}{\aheap'} $ such that $\pair{\astore}{\aheap} \approx^q_\alpha \pair{\astore'}{\aheap'}$,
we have $\pair{\astore}{\aheap} \models \aformula$ iff $\pair{\astore'}{\aheap'} \models \aformula$.
\end{theorem}

\begin{proof}
Assume that $\aformula$ is a formula with $\msize(\aformula) \leq \alpha$ and $\pair{\astore}{\aheap} \approx^q_\alpha \pair{\astore'}{\aheap'}$.
By structural induction we show that
 $\pair{\astore}{\aheap} \models \aformula$ iff $\pair{\astore'}{\aheap'} \models \aformula$.
It is sufficient to establish one direction of the equivalence thanks to its symmetry.
First, note that $\pair{\astore}{\aheap} \approx^q_\alpha \pair{\astore'}{\aheap'}$ implies that
$\pair{\astore}{\aheap}$ and $\pair{\astore'}{\aheap'}$ agree on the satisfaction of
Boolean combinations built over atomic formulae in $\Test(q, \alpha)$.
Indeed, at the atomic level, this is exactly what $\pair{\astore}{\aheap} \approx^q_\alpha \pair{\astore'}{\aheap'}$ means,
whereas for Boolean connectives $\neg$ and $\wedge$, this is also straighforward by induction.

The basic cases for the atomic formulae $\avariable_i \hpto \avariable_j$ and $\avariable_i = \avariable_j$  are immediate since these are
test formulae. For $\emp$ and $\reachplus(\avariable_i,\avariable_j)$ we use directly Lemma~\ref{lemma-atomic-formulae-test}.
Suppose $\aformulabis = \emp$. Then $\msize(\aformulabis) = 1 \leq \alpha$ and we can express $\emp$ with the Boolean
combination of test formulae
$$
{\lnot \sizeothers_q \geq 1} \land \bigwedge_{\mathclap{i \in \interval{1}{q}}} \lnot \alloc{\avariable_i},
$$
as shown in the proof of Lemma~\ref{lemma-atomic-formulae-test}. Notice that all the test formulae appearing in the formula above are from $\Test(q,1)$. Thus,
from $\pair{\astore}{\aheap} \approx^q_\alpha \pair{\astore'}{\aheap'}$
we conclude that  $\pair{\astore}{\aheap} \models \aformulabis$ iff $\pair{\astore'}{\aheap'} \models \aformulabis$.

Similarly, suppose that $\aformulabis = \reachplus(\avariable_i,\avariable_j)$.
Then $\msize(\aformulabis) = 1 \leq \alpha$ and we can express $\aformulabis$  with the Boolean combination of test formulae:
\[
\bigvee_{\mathclap{\substack{\aterm_1, \ldots, \aterm_n \in \Terms{q}, \\
{\rm pairwise \ distinct} \ \aterm_1, \ldots, \aterm_{n-1}, \\
\avariable_i = \aterm_1, \avariable_j = \aterm_n}}}
\textstyle\bigwedge_{\delta \in [1,n-1]} \sees_q(\aterm_{\delta},\aterm_{\delta+1}) \geq 1,
\]
as shown in the proof of Lemma~\ref{lemma-atomic-formulae-test}. As in the case for the formula $\emp$, notice that all the test formulae appearing in the formula above are from $\Test(q,1)$. Again,
from $\pair{\astore}{\aheap} \approx^q_\alpha \pair{\astore'}{\aheap'}$
we conclude that  $\pair{\astore}{\aheap} \models \aformulabis$ iff $\pair{\astore'}{\aheap'} \models \aformulabis$.

We omit the obvious cases with the Boolean connectives. Let us consider the last case $\aformulabis = \aformulabis_1 \separate \aformulabis_2$.
Suppose that $\pair{\astore}{\aheap} \models \aformulabis_1 \separate \aformulabis_2$ and
$\msize(\aformulabis_1 \separate \aformulabis_2) \leq \alpha$.
There are heaps $\aheap_1$ and $\aheap_2$ such that
$\aheap = \aheap_1 + \aheap_2$,
$\pair{\astore}{\aheap_1} \models \aformulabis_1$
and $\pair{\astore}{\aheap_2} \models \aformulabis_2$.
As $\alpha \geq \msize(\aformulabis_1 \separate \aformulabis_2) = \msize(\aformulabis_1) + \msize(\aformulabis_2)$,
there exist $\alpha_1$ and $\alpha_2$ such that $\alpha = \alpha_1 + \alpha_2$,
$\alpha_1 \geq \msize(\aformulabis_1)$ and $\alpha_2 \geq \msize(\aformulabis_2)$.
By  Lemma~\ref{lemma-star},
there exist heaps $\aheap'_1$ and $\aheap'_2$ such that
      $\aheap'=\aheap'_1 + \aheap'_2$,
      $\pair{\astore}{\aheap_1} \approx^q_{\alpha_1}  \pair{\astore'}{\aheap'_1}$ and
      $\pair{\astore}{\aheap_2} \approx^q_{\alpha_2}  \pair{\astore'}{\aheap'_2}$.
      By the induction hypothesis, we get  $\pair{\astore'}{\aheap_1'} \models \aformulabis_1$
and $\pair{\astore'}{\aheap_2'} \models \aformulabis_2$. Consequently, we obtain
$\pair{\astore'}{\aheap'} \models \aformulabis_1 \separate \aformulabis_2$.
\end{proof}

As an example, we can apply this result to the memory states  from Figure~\ref{figure-fourms}, page~\pageref{figure-fourms}.
We have already shown how we can distinguish $\pair{\astore_1}{\aheap_1}$ from $\pair{\astore_2}{\aheap_2}$
using a formula with only one separating conjunction.
Theorem~\ref{theo:test_sl_equiv} ensures that these two memory states do not satisfy the same set of test formulae for $\alpha \geq 2$.
Indeed, only $\pair{\astore_1}{\aheap_1}$ satisfies $\sees_q(\avariable_i,\avariable_j) \geq 2$. The same argument can be used with
$\pair{\astore_3}{\aheap_3}$ and $\pair{\astore_4}{\aheap_4}$: only $\pair{\astore_3}{\aheap_3}$ satisfies the test formula $m_q(\avariable_i,\avariable_j) \hpto
 m_q(\avariable_j,\avariable_i)$.
As a result, we can relate separation logic with classical logic, as advocated in the
works~\cite{Lozes04,Calcagno&Gardner&Hague05,Demrietal17,Echenim&Iosif&Peltier20} and shown with the following theorem.

\begin{theorem}
\label{theorem-quantifier-elimination}
Let $\aformula$ be a formula in $\seplogic{\separate,\reachplus}$ built over the variables in $\avariable_1$, \ldots, $\avariable_q$.
The formula $\aformula$ is logically equivalent to a Boolean combination of test formulae from $\testformulae{q}{\msize(\aformula)}$.
\end{theorem}

\begin{proof} The proof is rather standard.
Let $\alpha = \msize(\aformula)$.
Given a memory state $\pair{\astore}{\aheap}$, we write $\literals{\astore,\aheap}$ to denote the following
set of literals:
$$
\set{\aformulabis \in \Test(q,\alpha) \ \mid \ \pair{\astore}{\aheap} \models \aformulabis}
\cup
\set{\neg \aformulabis  \mid \ \pair{\astore}{\aheap} \not \models \aformulabis \ {\rm with} \ \aformulabis \in
\Test(q,\alpha)}.
$$
Since  $\Test(q,\alpha)$ is a finite set, $\literals{\astore,\aheap}$ is finite too and let us consider the well-defined
formula $\bigwedge_{\aformulabis \in\literals{\astore,\aheap}} \aformulabis$. We have the following equivalence:
$$
\pair{\astore'}{\aheap'} \models \bigwedge_{\mathclap{\aformulabis \in\literals{\astore,\aheap}}} \aformulabis  \quad \text{ iff } \quad
\pair{\astore}{\aheap} \indist{q}{\alpha} \pair{\astore'}{\aheap'}.
$$
The expression
$
\aformulabis' \egdef
\bigvee_{\pair{\astore}{\aheap}  \models \aformula} (\bigwedge_{\aformulabis \in\literals{\astore,\aheap}} \aformulabis)
$
is equivalent to a Boolean combination $\aformula'$ of formulae from $\Test(q,\alpha)$
because $\literals{\astore,\aheap}$ ranges over the finite
set of elements from $\Test(q,\alpha)$ (just select a finite amount of disjuncts).
Though the number of memory states is infinite, the number of formulae of the form 
$(\bigwedge_{\aformulabis \in\literals{\astore,\aheap}} \aformulabis)$ is finite and therefore, we understand $\aformulabis'$
as a finite disjunction. This is not a constructive way to define $\aformulabis'$ (which can be also done by some other means)
but this is sufficient for the existence. 
By Theorem~\ref{theo:test_sl_equiv}, the formula $\aformula$
is logically equivalent to $\aformulabis'$, which concludes the proof.
Indeed, suppose that $\pair{\astore}{\aheap} \models \aformula$. Obviously, we get $\pair{\astore}{\aheap} \models
\bigwedge_{\aformulabis \in\literals{\astore,\aheap}} \aformulabis$ and therefore $\pair{\astore}{\aheap} \models
 \aformulabis'$. Conversely, suppose that  $\pair{\astore}{\aheap} \models  \aformulabis'$.
This means that there is a memory state  $\pair{\astore'}{\aheap'}$ such that $\pair{\astore'}{\aheap'} \models \aformula$
and $\pair{\astore}{\aheap} \models \bigwedge_{\aformulabis \in\literals{\astore',\aheap'}} \aformulabis$.
Thus $\pair{\astore}{\aheap} \indist{q}{\alpha} \pair{\astore'}{\aheap'}$, $\threshold{\aformula} \leq \alpha$
and since $\pair{\astore'}{\aheap'} \models \aformula$,
by Theorem~\ref{theo:test_sl_equiv}
we get  $\pair{\astore}{\aheap} \models  \aformula$.
\end{proof}

It is now possible to establish a small heap  property of $\seplogic{\separate,\reachplus}$ by inheriting it from the small heap property for
Boolean combinations of test formulae, which is analogous to the small model property for other theories of singly linked lists,
see e.g.~\cite{Ranise&Zarba06,David&Kroening&Lewis15}. Indeed, following Lemma~\ref{lemma-test-graph}, now it is straightforward  to
derive an upper bound on the size of  a small model satisfying a formula in $\seplogic{\separate,\reachplus}$.

Let $\acompmap(q,n)$ be the polynomial $(q^2 + q) \cdot (n + 1) + n$ used in the sequel.

\begin{theorem}
\label{theorem-finite-heap-property}
Let $\aformula$ be a satisfiable $\seplogic{\separate,\reachplus}$ formula built over  $\avariable_1$, \ldots, $\avariable_q$.
There is  $\pair{\astore}{\aheap}$ such that $\pair{\astore}{\aheap} \models \aformula$
and
 $\card{\domain{\aheap}} \leq \acompmap(q,\length{\aformula})$.
\end{theorem}


\begin{proof}
  Let $\aformula$ be a formula built over $\avariable_1,\dots,\avariable_q$ with $\alpha = \msize(\aformula) \leq \length{\aformula}$ and
  let $\pair{\astore}{\aheap}$ be a memory state satisfying $\aformula$.
  Let $(\gverts,\gedges,\galloc,\glabels,\gbtw,\grem)$ be equal to $\supportgraph{\astore,\aheap}$ with respect to $q$
  and $\alocation^*$ be an arbitrary location not in the domain of
  $\aheap$.
  We construct a heap $\aheap'$ such that 
  $\card{\domain{\aheap'}} \leq \acompmap(q,\length{\aformula})$ 
and
  $\pair{\astore}{\aheap'} \models \aformula$.
  \begin{enumerate}
  
    \item Let $R \subseteq \grem$ be a set of $\min(\alpha,\card{\grem})$ locations.
     For all $\alocation \in \grem$, $\alocation \in \domain{\aheap'} \equivdef \alocation \in R$ and
     for all $\alocation \in R$, we have $\aheap'(\alocation) \egdef \alocation^*$.
     As it will be soon clear, these locations in $R$ will happen to be the only ones in $\grem'$ (from $\aheap'$).
     In that way,  $\pair{\astore}{\aheap}$ and $\pair{\astore}{\aheap'}$ shall agree on all test formulae of the form
     $\sizeothers_q \geq \beta$ with $\beta \in \interval{1}{\alpha}$.

    \item For all $\pair{\alocation}{\alocation'} \in \gedges$,
     let
     $L = \{\alocation_1,\dots,\alocation_A\}$ be a set of $A = \min(\alpha,\card{\gbtw(\alocation,\alocation')})$ locations in $\gbtw(\alocation,\alocation')$. 
     Notice that if $\card{\gbtw(\alocation,\alocation')} \geq \alpha$, then $A = \alpha$, else $A = \card{\gbtw(\alocation,\alocation')}$.
     Then, for all $\bar{\alocation} \in \gbtw(\alocation,\alocation')$, $\bar{\alocation} \in \domain{\aheap'} \equivdef \bar{\alocation} \in L$.
     Moreover $\aheap'(\alocation) \egdef \alocation_1$ (and therefore $\alocation \in \domain{\aheap'}$), $\aheap'(\alocation_\alpha) \egdef \alocation'$
     and for all $i \in \interval{1}{\alpha-1}$, $\aheap'(\alocation_i) \egdef \alocation_{i+1}$.
     Since all the locations in $\gbtw(\alocation,\alocation')$ are not labelled locations and we preserve the existence of paths between labelled locations,
     this step guarantees  that for all $\aterm \in \Terms{q}$, we have $\sem{\aterm}^q_{\astore,\aheap} = \sem{\aterm}^q_{\astore,\aheap'}$.
     This implies that $\pair{\astore}{\aheap}$ and $\pair{\astore}{\aheap'}$ satisfy the same set of test formulae of the form $\aterm = \aterm'$,
     where $\aterm,\aterm' \in \Terms{q}$.
     Furthermore, the path from $\alocation$ to $\alocation'$ in $\aheap'$  has length $\min(\alpha,\card{\gbtw(\alocation,\alocation')}) + 1$.
     Consequently, for all $\beta \in \interval{1}{\alpha}$ and for all $\aterm,\aterm'\in \Terms{q}$,
     $\pair{\astore}{\aheap'} \models \sees_q(\aterm,\aterm') \geq \beta + 1$ if and only if $\pair{\astore}{\aheap} \models \sees_q(\aterm,\aterm') \geq \beta + 1$,
     and these two memory states  agree on the test formulae of the form $\aterm \hpto \aterm'$.
     As $\card{\Terms{q}} = q^2 + q$, with this construction, each path between two labelled locations  has at most length $\alpha + 1$.
     Additionally, the heap graphs are functional, and therefore $\pair{\astore}{\aheap} \models  \sees_q(\aterm,\aterm_1) \geq \beta_1 +1 \wedge
     \sees_q(\aterm,\aterm_2) \geq \beta_2 +1$ implies $\aterm_1 = \aterm_2$, which entails that this step adds less than
     $(q^2 + q)(\length{\aformula}+1)$ locations to  $\domain{\aheap'}$.
     \item Lastly, for all $\alocation \in \galloc$, $\aheap'(\alocation) \egdef \aheap(\alocation)$ and therefore $\alocation \in \domain{\aheap'}$.
      So, for all $\aterm \in \Terms{q}$ $\pair{\astore}{\aheap'} \models \alloc{\aterm}$ if and only if $\pair{\astore}{\aheap} \models \alloc{\aterm}$.
      After this step, this implies that  the two memory states  satisfy the same set of test formulae. Note that in the computation of the upper bound,
      there is no need to take into account the location introduced in this step, since the upper bound mentioned in the previous step already includes this case.
  \end{enumerate}
  It follows that $\aheap'$ is a heap such that $\card{\domain{\aheap'}} \leq (q^2 + q) \cdot (\length{\aformula} + 1) + \length{\aformula}$ and
  $\pair{\astore}{\aheap'} \models \aformula$.
\end{proof}

\subsection{Complexity  upper bounds}\label{section:complexity-upper-bounds}
Let us draw some consequences of Theorem~\ref{theorem-finite-heap-property}.
First, for the logic $\seplogic{\separate,\reachplus}$, we get a \pspace upper bound, which matches
the lower bound for $\seplogic{\separate}$~\cite{Calcagno&Yang&OHearn01}.
Indeed, $\seplogic{\separate}$ is a syntactic fragment of $\seplogic{\separate,\reachplus}$
and the satisfiability problem for $\seplogic{\separate}$ is shown \pspace-hard in~\cite{Calcagno&Yang&OHearn01}.

\begin{theorem}
\label{theorem-pspace}
The satisfiability problem for $\seplogic{\separate, \reachplus}$ is \pspace-complete.
\end{theorem}


\begin{proof} Let $\aformula$ be a formula in  $\seplogic{\separate, \reachplus}$ built over $\avariable_1,\dots,\avariable_q$.
By Theorem~\ref{theorem-finite-heap-property},
$\aformula$ is satisfiable if and only if there is a memory state satisfying $\aformula$ with
$\card{\domain{\aheap}} \leq \acompmap(q,\length{\aformula})$. The non-deterministic polynomial-space
algorithm (leading to the \pspace upper bound by Savitch's Theorem~\cite{Savitch70}) works as follows. First, guess a heap $\aheap$ with
$\card{\domain{\aheap}} \leq \acompmap(q,\length{\aformula})$ and
$\domain{\aheap} \cup \range{\aheap} \subseteq  \interval{0}{2 \times \acompmap(q,\length{\aformula})}$
 and a store restricted
to $\avariable_1, \ldots, \avariable_q$ such that $\range{\astore} \subseteq \interval{0}{2 \times \acompmap(q,\length{\aformula}) + q}$
(in the worst case, all the variables have different values, and the memory cells have different values too).

Then, checking whether $\pair{\astore}{\aheap} \models \aformula$ can be done in polynomial-space as for
the standard $\seplogic{\separate}$  by using a recursive algorithm that internalises the semantics (see e.g.~\cite{Calcagno&Yang&OHearn01}): the recursive depth
is linear and at each call, the algorithm uses at most linear space in the size of $\pair{\astore}{\aheap}$ and $\aformula$,
which is polynomial in $\length{\aformula}$. We only need to guarantee that $\pair{\astore}{\aheap} \models \reachplus(\avariable,\avariablebis)$
can be checked in polynomial space, actually this can be done in polynomial time in the size of $\pair{\astore}{\aheap}$, that is
therefore in polynomial space in $\length{\aformula}$.
\end{proof}

Besides, we may consider restricting the usage of Boolean connectives.
Let us briefly define the symbolic heap fragment formulae $\aformula$ as conjunctions
$\apureformula \wedge \aspatialformula$
where $\apureformula$ is a
\defstyle{pure formula} and $\aspatialformula$ is a \defstyle{spatial formula} $\aspatialformula$:
\begin{align*}
\apureformula &::= \perp  \ \mid \ \top  \ \mid \ \avariable_i = \avariable_j \ \mid \  \neg (\avariable_i = \avariable_j) \ \mid \
\apureformula \wedge \apureformula\\
\aspatialformula &::= \emptyconstant \ \mid \ \top \ \mid \ \avariable_i \mapsto \avariable_j \ \mid \ \ls(\avariable_i,\avariable_j) \ \mid \
\aspatialformula \separate \aspatialformula
\end{align*}
As usual,  $\avariable_i \mapsto \avariable_j$ is interpreted as the exact points-to relation.
We  write $\mathsf{Bool(SHF)}$ to denote the set of  Boolean combinations
of  formulae from the symbolic heap fragment~\cite{Berdine&Calcagno&OHearn04}.
A \ptime upper bound
for the entailment/satisfiability problem for the symbolic heap fragment is successfully solved in~\cite{Cooketal11,Haaseetal13},
whereas the satisfiability problem for a slight variant of $\mathsf{Bool(SHF)}$
is shown in \np in~\cite[Theorem 4]{Piskac&Wies&Zufferey13}.
This \np upper bound can be obtained as a by-product of
Theorem~\ref{theorem-finite-heap-property}.

\begin{corollary}
\label{corollary-bool-shf}
The satisfiability problem for $\mathsf{Bool(SHF)}$
is \np-complete.
\end{corollary}
\begin{proof}
The \np-hardness is obtained thanks to the presence of equalities and Boolean connectives.
Indeed, let us show how to simply reduce SAT to the satisfiability problem for $\mathsf{Bool(SHF)}$.
Let $\aformula$ be a formula from the propositional calculus built over the propositional variables
$\avarprop_1$, \ldots, $\avarprop_n$. Let us define the translation $\atranslation$
such that $\atranslation(\avarprop_i) \egdef \avariable_i = \avariablebis_i$
($\avariable_i$ and $\avariablebis_i$ are program variables dedicated to $\avarprop_i$) and
$\atranslation$ is homomorphic for Boolean connectives. It is easy to prove that 
$\aformula$ is satisfiable iff $\atranslation(\aformula)$ is satisfiable.

As far as the complexity upper bound is concerned,
let $\aformula$ be a Boolean combination of pure or spatial formulae built over $\avariable_1,\dots,\avariable_q$.
By Theorem~\ref{theorem-finite-heap-property},
$\aformula$ is satisfiable iff there is $\pair{\astore}{\aheap}$ satisfying $\aformula$ with
$\card{\domain{\aheap}} \leq \acompmap(q,\length{\aformulabis})$, where $\aformulabis$ is the translation of $\aformula$ where
\begin{itemize}
  \item every occurrences of $\avariable_i \pto \avariable_j$ is rewritten as the equivalent formula $\avariable_i \hpto \avariable_j \land \size = 1$, and
  \item every $\ls(\avariable_i,\avariable_j)$ is rewritten as the equivalent formula
  \[
    (\avariable_i = \avariable_j \land \emp) \lor (\avariable_i \neq \avariable_j \land \reachplus(\avariable_i,\avariable_j) \land \lnot (\lnot \emp \sepcnj \reachplus(\avariable_i,\avariable_j))).
  \]
\end{itemize}
This technicality is introduced as Theorem~\ref{theo:test_sl_equiv} requires $\aformula$ to be in $\seplogic{\separate,\reachplus}$ and, as such, the formula $\aformulabis$ is not then used to check for satisfiability.

The non-deterministic polynomial-time
algorithm  works as follows. Similarly to the proof of Theorem~\ref{theorem-pspace}, guess a heap $\aheap$ with
$\card{\domain{\aheap}} \leq \acompmap(q,\length{\aformulabis})$,
$\domain{\aheap} \cup \range{\aheap} \subseteq  \interval{0}{2 \times \acompmap(q,\length{\aformulabis})}$
and a store restricted to $\avariable_1,\dots,\avariable_q$ and such that  $\range{\astore} \subseteq \interval{0}{2 \times \acompmap(q,\length{\aformulabis}) + q}$.
Then, confirming that the valuation is correct can be done in \ptime: checking whether a pure formula is satisfied by
$\pair{\astore}{\aheap}$ can be done in linear time. Similarly, checking whether a spatial formula $\aspatialformula$
is satisfied by $\pair{\astore}{\aheap}$ can be done in \ptime~\cite[Lemma 1]{Berdine&Calcagno&OHearn04}.
So, the satisfiability problem for Boolean combinations of symbolic heap fragment formulae can be solved in \np.
\end{proof}

We have seen that we can take advantage of the small heap property to derive complexity results for fragments of
$\seplogic{\separate,\reachplus}$. However,  it is also possible to push further the \pspace upper bound
by allowing occurrences of $\magicwand$ in a controlled way
(as unrestricted use of the magic wand leads to
undecidability, see Theorem~\ref{theorem-main-undecidability}).
The following result can be shown thanks to Proposition~\ref{proposition-prop-sl} and
Lemma~\ref{lemma-atomic-formulae-test}.
\begin{lemma}
\label{lemma-psl-to-test}
Let $\aformula$ be in $\seplogic{\separate,\magicwand}$ built over  $\avariable_1$, \ldots, $\avariable_q$.
The formula $\aformula$ is logically equivalent to a Boolean combination of test formulae from $\testformulae{q}{2\length{\aformula}}$.
\end{lemma}
\begin{proof}
First, translate $\aformula$ into a Boolean combination of formulae from
\begin{mathpar}
  \avariable_i = \avariable_j

  \alloc{\avariable_i}

  \avariable_i \hpto \avariable_j

  \size \geq \beta,
\end{mathpar}
as stated in Proposition~\ref{proposition-prop-sl} and $\beta \leq 2\length{\aformula}$.
Then, rewrite every occurrence of $\size \geq \beta$ into the equivalent Boolean combination of test formulae shown in Lemma~\ref{lemma-atomic-formulae-test}. The resulting formula is in $\testformulae{q}{2\length{\aformula}}$.
\end{proof}
Let $\seplogic{\separate, \reachplus, \bigcup_{q,\alpha} \Test(q,\alpha)}$ be the extension
of $\seplogic{\separate, \reachplus}$ augmented with the test formulae. The memory size function is extended as follows.
\begin{multicols}{2}
\begin{itemize}
\item $\msize(\alloc{\aterm}) \egdef 1$,
\item $\msize(\sees_q(\aterm,\aterm') \geq \beta + 1) \egdef \beta + 1$,
\item $\msize(\aterm \hpto \aterm') \egdef 1$,
\item $\msize(\sizeothers \geq \beta) \egdef \beta$.
\end{itemize}
\end{multicols}
When formulae are encoded as trees, we have $1 \leq \msize(\aformula) \leq \length{\aformula} \alpha_{\aformula}$ where
$\alpha_{\aformula}$ is the maximal constant in $\aformula$.
Theorem~\ref{theorem-quantifier-elimination} admits a counterpart for
$\seplogic{\separate, \reachplus, \bigcup_{q,\alpha} \Test(q,\alpha)}$
and consequently,
any formula built over  $\avariable_1$, \ldots, $\avariable_q$
can be shown
equivalent to a Boolean combination of test formulae from $\testformulae{q}{ \length{\aformula} \alpha_{\aformula}}$.
By Theorem~\ref{theorem-finite-heap-property}, any satisfiable formula has therefore a model with
$
\card{\domain{\aheap}} \leq (q^2 + q)\cdot (\length{\aformula} \alpha_{\aformula} + 1) + \length{\aformula} \alpha_{\aformula}$.
Hence, the satisfiability problem for $\seplogic{\separate, \reachplus, \bigcup_{q,\alpha} \Test(q,\alpha)}$
is  in \pspace when the constants are encoded in unary.
Then, we conclude by stating the new \pspace upper bound for Boolean combinations of formulae from
$\seplogic{\separate, \magicwand} \cup \seplogic{\separate, \reachplus}$.

\begin{theorem}
\label{theorem-boolean-combination}
The satisfiability problem for  Boolean combinations of formulae from
$\seplogic{\separate, \magicwand} \cup \seplogic{\separate, \reachplus}$
is \pspace-complete.
\end{theorem}
\begin{proof}
Let $\aformula$ be a Boolean combination of formulae from $\seplogic{\separate, \magicwand} \cup \seplogic{\separate, \reachplus}$.
First, replace every maximal subformula $\aformulabis$ in $\seplogic{\separate, \magicwand}$ by an equivalent Boolean
combination of test formulae from $\testformulae{q}{2\length{\aformulabis}}$, as shown by Lemma~\ref{lemma-psl-to-test}.
This replacement may require exponential time in the worst case but this is still fine to establish
the \pspace upper bound as we aim at showing a small heap property.
We obtain a formula $\aformula'$ in $\seplogic{\separate, \reachplus, \bigcup_{q,\alpha} \Test(q,\alpha)}$   with
$\alpha_{\aformula'} \leq 2\length{\aformula}$. So, $\aformula$ is satisfiable iff $\aformula$ has a model $\pair{\astore}{\aheap}$
with
$
\card{\domain{\aheap}} \leq \acompmap(q,2\length{\aformula} \alpha_{\aformula'})$,
which is still polynomial in $\length{\aformula}$.
Again, the non-deterministic polynomial-space
algorithm works as follows. First, guess a small heap $\aheap$ with
$\card{\domain{\aheap}} \leq \acompmap(q,2\length{\aformula} \alpha_{\aformula'})$,
$\domain{\aheap} \cup \range{\aheap} \subseteq \interval{0}{ 2 \times \acompmap(q,2\length{\aformula} \alpha_{\aformula'})}$
and a store restricted to $\avariable_1,\dots,\avariable_q$ and such that  $\range{\astore} \subseteq \interval{0}{2 \times 
\acompmap(q,2\length{\aformula} \alpha_{\aformula'}) + q}$.
Since the respective model-checking problem for $\seplogic{\separate, \magicwand}$ and $\seplogic{\separate, \reachplus}$
are both in \pspace (see~\cite{Calcagno&Yang&OHearn01} and Theorem~\ref{theorem-pspace})
and $\pair{\astore}{\aheap}$ is of polynomial size in $\length{\aformula}$,
checking whether $\pair{\astore}{\aheap} \models \aformula$ can be done in \pspace by performing
several instances of the model-checking problem with maximal subformulae $\aformulabis$ from either
 $\seplogic{\separate, \magicwand}$ or $\seplogic{\separate, \reachplus}$.
\end{proof}

The fragment from Theorem~\ref{theorem-boolean-combination} forbids formulae with
$\ls$ in the scope of the separating implication $\magicwand$.
A fragment with $\ls$ in the scope of $\magicwand$ is considered in~~\cite{Thakur&Breck&Reps14bis}
but decidability is still open. By contrast, a recent work~\cite{Mansutti18} has established a \pspace upper bound
for fragments with $\ls$ in the scope of $\magicwand$ but restrictions apply (full proofs soon available
in~\cite{Mansutti20}).

%
%
\section{Conclusion}
We studied the effects of adding  $\ls$ to $\seplogic{\separate, \magicwand}$, giving us the opportunity to consider
several variants.
$\seplogic{\separate, \magicwand,\ls}$ is shown undecidable (Theorem~\ref{theorem-main-undecidability}) and non-finitely
axiomatisable, which remains
quite unexpected since there are no first-order quantifications. This result is strengthened to even weaker extensions
of $\seplogic{\separate, \magicwand}$ such as the one
 augmented with  $n(\avariable) = n(\avariablebis)$, $n(\avariable) \hpto n(\avariablebis)$ and
$\allocback{\avariable}$, or the one
  augmented with
  $\reach(\avariable,\avariablebis) = 2$ and $\reach(\avariable,\avariablebis) = 3$.
If the magic wand is discarded, we have established that the satisfiability problem for
$\seplogic{\separate,\ls}$ is \pspace-complete by introducing a class of test formulae that captures
the expressive power of $\seplogic{\separate,\ls}$ and that leads to a small heap property. Such a logic contains the Boolean combinations of symbolic heaps
and our proof technique allows us to get an \np upper bound for such formulae. Moreover, we have shown
that the satisfiability problem for $\seplogic{\separate, \magicwand, \reachplus}$
restricted to Boolean combination of formulae from $\seplogic{\separate,\magicwand}$ and $\seplogic{\separate,\reachplus}$ is also \pspace-complete.
So,  we have provided proof techniques to establish undecidability when $\separate$, $\magicwand$ and $\ls$ are present
and to establish decidability based on test formulae. This paves the way  to investigate the decidability status
of $\seplogic{\magicwand, \ls}$ as well as of the positive fragment of  $\seplogic{\separate, \septraction, \ls}$
from~\cite{Thakur14,Thakur&Breck&Reps14bis}.

{\bf Acknowledgments.} We would like to thank the anonymous reviewers for their numerous suggestions and
remarks that help us to improve the quality of the document. 
\bibliographystyle{ACM-Reference-Format}
\bibliography{biblio-sl}

\newpage
\appendix
\section{Electronic Appendix}
\label{appendix-star}

In this appendix, we present the full proof of Lemma~\ref{lemma-star}. 
The material about the proof included in the body of the paper
is  a subset of the material below.

\begin{proof}
  Let $q$, $\alpha$, $\alpha_1$, $\alpha_2$, $\pair{\astore}{\aheap}$, $\pair{\astore'}{\aheap'}$, $\aheap_1$ and $\aheap_2$ be
  defined as in the statement. Let
  \begin{itemize}
  \item $\supportgraph{\astore,\aheap} = (\gverts,\gedges,\galloc,\glabels,\gbtw,\grem)$ and
  \item $\supportgraph{\astore',\aheap'} = (\gverts',\gedges',\galloc',\glabels',\gbtw',\grem')$
  \end{itemize}
  be the support graphs
of $\pair{\astore}{\aheap}$ and $\pair{\astore'}{\aheap'}$ respectively, with respect to $q$.
  As $\pair{\astore}{\aheap} \approx^q_\alpha
  \pair{\astore'}{\aheap'}$, let $\amap: \gverts \to \gverts'$ be a map satisfying \ref{a1}--\ref{a5} from
   Lemma~\ref{lemma-test-graph}.
  Below, for the sake of conciseness, let $k \in \{1,2\}$ and 
  $\supportgraph{\astore,\aheap_k} = (\gverts_k,\gedges_k,\galloc_k,\glabels_k,\gbtw_k,\grem_k)$
  be the support graph of $\pair{\astore}{\aheap_k}$.

  The proof is rather long and can be summed up with the following steps.
  \begin{enumerate}
  \item First, we define a strategy to split $\aheap'$ into $\aheap_1'$ and $\aheap_2'$ by closely following the way that $\aheap$ is split into $\aheap_1$ and $\aheap_2$. To do this, we look at the support graphs. For instance, suppose that $\grem$ is split into two sets $R_1 \subseteq \domain{\aheap_1}$ and $R_2 \subseteq \domain{\aheap_2}$.
  By definition, it is quite easy to see that the sets $R_1$ and $R_2$ must be subsets of $\grem_1$ and $\grem_2$, respectively. 
Then, following Lemma~\ref{lemma-test-graph}, to obtain
  $\pair{\astore}{\aheap_k} \approx^q_{\alpha_k} \pair{\astore'}{\aheap_k'}$,
  we are required to split $\grem'$ into
  $R_1' \subseteq \domain{\aheap_1'}$ and
  $R_2' \subseteq \domain{\aheap_2'}$ so that $\min(\alpha_k, \card{R_k}) = \min(\alpha_k, \card{R_k'})$. Indeed, otherwise the equisatisfaction of
  the test formulae of the form $\sizeothers_q \geq \beta$ is not ensured (details on this are formalised later).
  \item After defining $\aheap_1'$ and $\aheap_2'$, we show that $\pair{\astore}{\aheap_k} \approx^q_{\alpha_k} \pair{\astore'}{\aheap_k'}$.
  To do so, again, we  follow Lemma~\ref{lemma-test-graph} and we show that we can find suitable bijections from labelled locations of $\aheap_k$ to the ones of $\aheap_k'$ satisfying \ref{a1}--\ref{a5}.
  \end{enumerate}
  According to the summary above, let us first define explicitly $\aheap_1'$ and $\aheap_2'$ via an iterative process that consists in adding locations
  to $\domain{\aheap'_1}$ or to $\domain{\aheap'_2}$. Whenever we enforce that $\alocation \in \domain{\aheap_k'}$,
  implicitly we have $\aheap'_k(\alocation) \egdef \aheap'(\alocation)$ as $\aheap'_k$ is intended to be a subheap of $\aheap'$.
  \begin{description}
    \item[\descLab{(CA)}{c1}] For all $\alocation \in \galloc'$, $\alocation \in \domain{\aheap_k'} \equivdef \amap^{-1}(\alocation) \in \domain{\aheap_k}$.
    This step of the construction, as well as its usefulness, should be self-explanatory.
    For example, if $\pair{\astore}{\aheap_k} \models \alloc{\avariable_i}$ then, by relying on \ref{a3}, this step allows us to conclude that $\pair{\astore'}{\aheap_k'} \models \alloc{\avariable_i}$ (independently on how the definition of $\aheap_k'$ will be completed in the next steps of the construction).
    \item[\descLab{(CR)}{c3}] The heaps $\aheap_1'$ and $\aheap_2'$ are further populated depending on $\grem$. Let $R_k = \grem \cap \domain{\aheap_k}$.
    By definition, we have $R_1 \uplus R_2 = \grem$.
    Below, we partition $\grem'$ into two sets $R_1'$ and $R_2'$ so that by definition $R_k' \subseteq \domain{\aheap_k'}$.
    The strategy for defining the partition is split into three cases:
    \begin{center}
    \begin{enumerate}[align=left]
    \item[\descLab{(CR.C1)}{c31}] If $\card{R_1} < \alpha_1$ then $R_1'$ is a set of $\card{R_1}$ locations from $\grem'$ and $R_2' \egdef \grem' \setminus R_1'$.
    \item[\descLab{(CR.C2)}{c32}] Otherwise, if $\card{R_2} < \alpha_2$ then $R_2'$ is a set of 
    $\card{R_2}$ locations of $\grem'$ and $R_1' \egdef \grem' \setminus R_2'$.
    \item[\descLab{(CR.C3)}{c33}] Otherwise we have  $\card{R_1} \geq \alpha_1$ and $\card{R_2} \geq \alpha_2$. Then, $R_1'$ is a set of $\alpha_1$ locations
    from $\grem'$ and
    $R_2' \egdef \grem' \setminus R_1'$.
    \end{enumerate}
    \end{center}
    It is easy to show that the construction satisfies the following property (where $k \in \{1,2\}$):
    \begin{equation}
      \tag*{\textbf{(CR.P1)}}
      \label{crp2} \min(\alpha_k,\card{R_k}) = \min(\alpha_k,\card{R_k'})
    \end{equation}

    The property \ref{crp2} directly follows from the
     property~\ref{a5} satisfied by $\amap$. The proof (that can be applied also for the next step of the construction, see~\ref{cip2}), works as follows.

  First, suppose that the sets of remaining locations in the heap domain are small, i.e.
     \[
     \min(\alpha,\card{\grem}) = \min(\alpha,\card{\grem'}) < \alpha_1 + \alpha_2.
     \]
     So, $\card{\grem} = \card{\grem'}$ and therefore $\card{R_1} + \card{R_2} = \card{R_1'} + \card{R_2'}$.
    By definition, $\card{R_1} = \card{R_1'}$ and $\card{R_2} = \card{R_2'}$ trivially hold for the cases~\ref{c31}
    and~\ref{c32}, whereas the case~\ref{c33}
     ($\card{R_1} \geq \alpha_1$ and $\card{R_2} \geq \alpha_2$) can never be applied since
    $\card{R_1} + \card{R_2} < \alpha_1 + \alpha_2$. We conclude that $\min(\alpha_k,\card{R_k}) = \min(\alpha_k,\card{R_k'})$.

  Second, suppose instead
  \[\min(\alpha,\card{\grem}) = \min(\alpha,\card{\grem'}) = \alpha_1 + \alpha_2.\]
  If the first case~\ref{c31} applies, i.e.\ $\card{R_1} < \alpha_1$, then $\card{R_2} \geq \alpha_2$ and by definition $\card{R_1'} = \card{R_1}$.
  Then, $\card{R_2'} \geq \alpha_2$ trivially follows from $\card{R_1'} + \card{R_2'} \geq \alpha_1 + \alpha_2$.
  Symmetrically, $\min(\alpha_k,\card{R_k}) = \min(\alpha_k,\card{R_k'})$ holds when the second case~\ref{c32}  applies ($\card{R_2} < \alpha_2$).
     Lastly, suppose $\card{R_1} \geq \alpha_1$ and $\card{R_2} \geq \alpha_2$. Then, the third case~\ref{c33} applies and by definition $\card{R_1'} = \alpha_1$. Again, we conclude that $\min(\alpha_k,\card{R_k}) = \min(\alpha_k,\card{R_k'})$ since
    $\card{R_2'} \geq \alpha_2$ trivially follows from $\card{R_1'} + \card{R_2'} \geq \alpha_1 + \alpha_2$.

    \item[\descLab{(CI)}{c2}] Lastly, the heaps $\aheap_1'$ and $\aheap_2'$ are further populated with respect to the memory cells in $\gbtw(\alocation,\alocation')$.
    For all $\pair{\alocation}{\alocation'} \in \gedges$, let
     $L_k \egdef \gbtw(\alocation,\alocation') \cap \domain{\aheap_k}$. We have
    $L_1 \uplus L_2 = \gbtw(\alocation,\alocation')$.
    Below, we partition $\gbtw'(\amap(\alocation),\amap(\alocation'))$ into $L_1'$ and $L_2'$ so that by definition
    $L_k' \subseteq \domain{\aheap_k'}$:
    \begin{description}
    \item[\descLab{(CI.C1)}{c21}] If $\card{L_1} < \alpha_1$ then $L_1'$ is a set of $\card{L_1}$ locations from $\gbtw'(\amap(\alocation),\amap(\alocation'))$, whereas $L_2' \egdef
     \gbtw'(\amap(\alocation),\amap(\alocation')) \setminus L_1'$.
    \item[\descLab{(CI.C2)}{c22}] Else, if $\card{L_2} < \alpha_2$ then $L_2'$ is a set of $\card{L_2}$ locations from $\gbtw'(\amap(\alocation),\amap(\alocation'))$,
    whereas  $L_1' \egdef \gbtw'(\amap(\alocation),\amap(\alocation')) \setminus L_2'$.
    \item[\descLab{(CI.C3)}{c23}] Otherwise, we have  $\card{L_1} \geq \alpha_1$ and $\card{L_2} \geq \alpha_2$. Then $L_1'$ is a set of $\alpha_1$ locations
    from $\gbtw'(\amap(\alocation),\amap(\alocation'))$ and
    $L_2' \egdef \gbtw'(\amap(\alocation),\amap(\alocation')) \setminus L_1'$.
    \end{description}
    It is easy to show that the construction satisfies the following properties (where $k \in \{1,2\}$):
    \begin{equation}
      \tag*{\textbf{(CI.P1)}}\label{cip3}
    \text{if } L_k' = \emptyset \text{ then } \gbtw'(\amap(\alocation),\amap(\alocation')) \subseteq \domain{\aheap_{3-k}'}
    \end{equation}
    \begin{equation}
      \tag*{\textbf{(CI.P2)}}\label{cip2}
    \min(\alpha_k,\card{L_k}) = \min(\alpha_k,\card{L_k'})
    \end{equation}
    (Proof of~\ref{cip3}) The first property trivially holds from the cases~\ref{c21} and~\ref{c22} of the construction.
    Notice that given $k \in \{1,2\}$, $3-k$ corresponds to the index in $\{1,2\}$ that is different from $k$.

    (Proof of~\ref{cip2}) The second property directly follows from~\ref{a4} (which is satisfied by $\amap$) and is proved as done for~\ref{crp2}.


  \end{description}

This ends the construction of $\aheap'_1$ and $\aheap'_2$ as
any location in $\domain{\aheap'}$ has been assigned to one of the two heaps.
Indeed, $\{\galloc',\grem'\} \cup \{ \gbtw'(\alocation,\alocation') \mid
(\alocation,\alocation') \in \gedges' \}$ is a partition of $\domain{\aheap'}$.
As  $\pair{\astore}{\aheap} \approx^q_\alpha
\pair{\astore'}{\aheap'}$,
the support graphs  $\supportgraph{\astore,\aheap}$ and  $\supportgraph{\astore',\aheap'}$ witness
the existence of a map $\amap:  \gverts \to \gverts'$ satisfying~\ref{a1}--\ref{a5} and therefore
there is an underlying isomorphism between these structures satisfying quantitative properties up to
the value $\alpha$. The construction above can be understood as a way to split $\aheap'$
into $\aheap_1'$ and $\aheap_2'$ mimicking the splitting of $\aheap$ into $\aheap_1$ and $\aheap_2$.
It remains to show below that this is done in a way that guarantees that
 $\pair{\astore}{\aheap_k} \approx^q_{\alpha_{k}}
\pair{\astore'}{\aheap'_k}$ ($k \in \set{1,2}$).

In the following, we denote the support graphs of $\pair{\astore}{\aheap_k}$ and $\pair{\astore'}{\aheap_k'}$ respectively as
\begin{itemize}
\item $\supportgraph{\astore,\aheap_k} = (\gverts_k,\gedges_k,\galloc_k,\glabels_k,\gbtw_k,\grem_k)$
and,
\item  $\supportgraph{\astore',\aheap_k'} = (\gverts_k',\gedges_k',\galloc_k',\glabels_k',\gbtw_k',\grem_k')$.
\end{itemize}

  First, let us formalise an essential property of the construction of $\aheap_1'$ and $\aheap_2'$.
  \begin{enumerate}[align=left]
  \item[\descLab{(Paths)}{pathproperty}] Let $k \in \{1,2\}$ and let $\alocation,\alocation' \in \gverts$ be two labelled locations w.r.t. $\pair{\astore}{\aheap}$.
  $\aheap_k$ witnesses a non-empty path from $\alocation$ to $\alocation'$ if and only if $\aheap_k'$ witnesses a non-empty path from $\amap(\alocation)$ to $\amap(\alocation')$.
  \vspace{2pt}
  \item[(Proof of \ref{pathproperty})]
  The proof mainly relies on the properties \ref{cip3} and \ref{cip2} of the construction.
  Recall that $\amap : \gverts \to \gverts'$ is a bijection satisfying~\ref{a1}--\ref{a5} w.r.t.\ $\pair{\astore}{\aheap}$ and $\pair{\astore'}{\aheap'}$.

  ($\Rightarrow$) Let $\alocation,\alocation' \in \gverts$ be such that
  $\aheap_k$ witnesses a non-empty path from $\alocation$ to $\alocation'$.
  Since $\aheap_k \sqsubseteq \aheap$, then $\aheap$ also witnesses a non-empty path from $\alocation$ to $\alocation'$.
  In particular, by Definition~\ref{definition-support-graph}, this path corresponds to a path
  in the support graph $\supportgraph{\astore,\aheap}$:
  \begin{center}
       \begin{tikzpicture}[baseline, node distance=1.8cm]
         \node[dot,label=above:{$\alocation$}] (zero) at (0,0) {};
         \node[dot,label=above:{$\alocation_1$}] (one) [right of=zero] {};
         \node[dot,label=above:{$\alocation_2$}] (two) [right of=one] {};
         \node[dot,label=above:{$\alocation_{n-1}$}] (n1) [right of=two] {};
         \node[dot,label=above:{$\alocation_{n}$}] (n) [right of=n1] {};
         \node[dot,label=above right:{$\alocation'$}] (np) [right of=n] {};

          \draw[pto] (zero) edge node [above] {$\gedges$} (one);
          \draw[pto] (one) edge node [above] {$\gedges$} (two);
          \draw[pto] (n1) edge node [above] {$\gedges$} (n);
          \draw[pto] (n) edge node [above] {$\gedges$} (np);
          \draw[pto] (one) edge node [above] {$\gedges$} (two);

          \path (two) edge [draw=none] node {$\dots$} (n1);

       \end{tikzpicture}
  \end{center}
  Let us define $\alocation_0 \egdef \alocation$  and $\alocation_{n+1} \egdef \alocation'$.
  In particular, following the picture above, the support graph $\supportgraph{\astore,\aheap}$
  witnesses a path $\{\pair{\alocation_0}{\alocation_1},\dots,\pair{\alocation_n}{\alocation_{n+1}}\} \subseteq \gedges$ from $\alocation_0$ to $\alocation_{n+1}$.
  Since $\amap$ is a graph isomorphism from $\pair{\gverts}{\gedges}$ to $\pair{\gverts'}{\gedges'}$ (by~\ref{a1}), $\supportgraph{\astore',\aheap'}$ witnesses a similar structure, as depicted below:
  \begin{center}
   \scalebox{1}{
       \begin{tikzpicture}[baseline, node distance=1.8cm]
         \node[dot,label=above left:{$\alocation$=$\alocation_0$}] (zero) at (0,0) {};
         \node[dot,label=above:{$\alocation_1$}] (one) [right of=zero] {};
         \node[dot,label=above:{$\alocation_2$}] (two) [right of=one] {};
         \node[dot,label=above:{$\alocation_{n-1}$}] (n1) [right of=two] {};
         \node[dot,label=above:{$\alocation_{n}$}] (n) [right of=n1] {};
         \node[dot,label=above right:{$\alocation'$=$\alocation_{n+1}$}] (np) [right of=n] {};
         \node[dot,label=below:{$\amap(\alocation)$}] (Fzero) [below of=zero] {};
         \node[dot,label=below:{$\amap(\alocation_1)$}] (Fone) [right of=Fzero] {};
         \node[dot,label=below:{$\amap(\alocation_2)$}] (Ftwo) [right of=Fone] {};
         \node[dot,label=below:{$\amap(\alocation_{n-1})$}] (Fn1) [right of=Ftwo] {};
         \node[dot,label=below:{$\amap(\alocation_{n})$}] (Fn) [right of=Fn1] {};
         \node[dot,label=below:{$\amap(\alocation')$}] (Fnp) [right of=Fn] {};

          \draw[pto] (zero) edge node [above] {$\gedges$} (one);
          \draw[pto] (one) edge node [above] {$\gedges$} (two);
          \draw[pto] (n1) edge node [above] {$\gedges$} (n);
          \draw[pto] (n) edge node [above] {$\gedges$} (np);
          \draw[pto] (one) edge node [above] {$\gedges$} (two);
          \draw[pto] (Fzero) edge node [below] {$\gedges'$} (Fone);
          \draw[pto] (Fone) edge node [below] {$\gedges'$} (Ftwo);
          \draw[pto] (Fn1) edge node [below] {$\gedges'$} (Fn);
          \draw[pto] (Fn) edge node [below] {$\gedges'$} (Fnp);
          \draw[pto] (Fone) edge node [below] {$\gedges'$} (Ftwo);

          \path (two) edge [draw=none] node {$\dots$} (n1);
          \path (Ftwo) edge [draw=none] node {$\dots$} (Fn1);

          \path[dotted,->] (zero) edge node [left] {$\amap$} (Fzero);
          \path[dotted,->] (one) edge node [left] {$\amap$} (Fone);
          \path[dotted,->] (np) edge node [left] {$\amap$} (Fnp);
          \path[dotted,->] (n) edge node [left] {$\amap$} (Fn);
          \path[dotted,->] (n1) edge node [left] {$\amap$} (Fn1);
          \path[dotted,->] (two) edge node [left] {$\amap$} (Ftwo);
       \end{tikzpicture}
   }
  \end{center}
  Let us consider $i \in \interval{0}{n}$. Since the path belongs to $\aheap_k$, it must hold that
  $\alocation_i \in \domain{\aheap_k}$ and $\gbtw(\alocation_i,\alocation_{i+1}) \subseteq \domain{\aheap_k}$. We show that then $\amap(\alocation_i) \in \domain{\aheap_k'}$ and $\gbtw'(\amap(\alocation_i),\amap(\alocation_{i+1})) \subseteq \domain{\aheap_k'}$,
  which entails that $\aheap_k'$ witnesses a path from $\amap(\alocation)$ to $\amap(\alocation')$, concluding the proof.
  \begin{itemize}[align = left]
    \item From $\pair{\alocation_i}{\alocation_{i+1}} \in \gedges$, by Definition~\ref{definition-support-graph} we conclude that $\alocation_i \in \gverts$.
    Since $\alocation_i \in \domain{\aheap}$, again by Definition~\ref{definition-support-graph},
    we have $\alocation_i \in \galloc$ and therefore by~\ref{a2}, $\amap(\alocation_i) \in \galloc'$.
    By~\ref{c1} together with the fact that $\alocation_i \in \domain{\aheap_k}$, we then conclude that $\amap(\alocation_i) \in \domain{\aheap_k'}$.
    \item If $\gbtw(\alocation_i,\alocation_{i+1})$ is empty then by~\ref{a4} $\gbtw'(\amap(\alocation_i),\amap(\alocation_{i+1}))$ is also empty and the inclusion w.r.t. $\domain{\aheap_k'}$ trivially holds.
    Suppose now $\gbtw(\alocation_i,\alocation_{i+1})$ is non-empty.
    From $\gbtw(\alocation_i,\alocation_{i+1}) \subseteq \domain{\aheap_k}$ and the fact that $\aheap_k$ and $\aheap_{3-k}$ are disjoint we conclude that $\gbtw(\alocation_i,\alocation_{i+1}) \cap \domain{\aheap_{3-k}} = \emptyset$. Hence, by \ref{cip2} we conclude that
    $\gbtw'(\amap(\alocation_i),\amap(\alocation_{i+1})) \cap \domain{\aheap_{3-k}'} = \emptyset$
    (notice that in \ref{cip2}, $L_{3-k}$ corresponds to $\gbtw(\alocation_i,\alocation_{i+1}) \cap \domain{\aheap_{3-k}}$ whereas $L_{3-k}'$ corresponds to $\gbtw'(\amap(\alocation_i),\amap(\alocation_{i+1})) \cap \domain{\aheap_{3-k}'}$).
    Since $\gbtw'(\amap(\alocation_i),\amap(\alocation_{i+1})) \cap \domain{\aheap_{3-k}'} = \emptyset$,
    by  \ref{cip3} we then conclude that
    $\gbtw'(\amap(\alocation_i),\amap(\alocation_{i+1})) \subseteq \domain{\aheap_{k}'}$.
  \end{itemize}

  ($\Leftarrow$) The right-to-left direction is analogous (thanks to the fact that $\amap^{-1}$ is a graph isomorphism from $\pair{\gverts'}{\gedges'}$ to $\pair{\gverts}{\gedges}$). \\ 
  \end{enumerate}

  Here is the last step of the proof.
  Given $k \in \{1,2\}$, let $\amap_k$ be the restriction of $\amap$ to $\gverts_k$ and
  $\gverts_k'$.
  We prove that $\amap_k$ satisfies \ref{a1}--\ref{a5} w.r.t. the memory states $\pair{\astore}{\aheap_k}$ and $\pair{\astore'}{\aheap_k'}$.
  Thanks to Lemma~\ref{lemma-test-graph}, this implies $\pair{\astore}{\aheap_1} \approx^q_{\alpha_1} \pair{\astore'}{\aheap_1'}$ and
  $\pair{\astore}{\aheap_2} \approx^q_{\alpha_2} \pair{\astore'}{\aheap_2'}$, ending the proof.
  For convenience, we prove the five properties in the following order: 
  \ref{a3}, \ref{a2}, \ref{a1}, \ref{a4} and \ref{a5} (in the body of the paper only the proof of~\ref{a3} is provided).\\

  \begin{enumerate}[label=\textbf{(P\arabic*)}, align = left]
  \item[\textbf{$\amap_k$ satisfies \descLab{(A3)}{mapksatA3}}:]
    We prove that for every $\alocation \in \gverts$, the set of terms corresponding to $\alocation$ in $\pair{\astore}{\aheap_k}$ is equivalent to the set of terms corresponding to $\amap(\alocation)$ in $\pair{\astore'}{\aheap_k'}$. Formally:
    \begin{center}
      \begin{minipage}{0.9\linewidth}
        for every $\alocation \in \gverts$,
        \begin{enumerate}[label=(\alph*)]
            \item $\alocation \in \gverts_k$ iff $\amap(\alocation) \in \gverts_k'$;
            \item if $\alocation \in \gverts_k$, $\glabels_k(\alocation) = \glabels_k'(\amap(\alocation))$.
        \end{enumerate}
      \end{minipage}
    \end{center}
    Notice that (a) implies that $\amap_k$ (which we recall being the restriction of $\amap$ to $\gverts_k$ and $\gverts_k'$) is well-defined and it is a bijection from the labelled locations of $\pair{\astore}{\aheap_k}$ (i.e. $\gverts_k$) to the labelled locations of $\pair{\astore'}{\aheap_k'}$ (i.e. $\gverts_k'$).
This is due to the fact that
   $\amap$ is a bijection from $\gverts$ to $\gverts'$ and by Lemma~\ref{lemma-labels}, 
we have $\gverts_k \subseteq \gverts$ and $\gverts_k' \subseteq \gverts'$.
    Again by $\gverts_k \subseteq \gverts$ (Lemma~\ref{lemma-labels}), (b) is then equivalent
    to \ref{a3}.

    We prove (a) and (b) together, by showing that
    for every $\alocation \in \gverts$, the set of terms corresponding to $\alocation$ in $\pair{\astore}{\aheap_k}$ is equivalent to the set of terms corresponding to $\amap(\alocation)$
    in $\pair{\astore'}{\aheap_k'}$.
    We first show the result for program variables, and then for meet-points.

    (Program variables) Let $\alocation \in \gverts$ and  $i \in \interval{1}{q}$.
    It holds that $\avariable_i \in \glabels_k(\alocation)$
     if and only if $\astore(\avariable_i) = \alocation$, or equivalently
     $\avariable_i \in \glabels(\alocation)$ which, by \ref{a3} in Lemma~\ref{lemma-test-graph}, holds whenever
     $\avariable_i \in \glabels'(\amap(\alocation))$. The latter is equivalent to $\astore'(\avariable_i) =
      \amap(\alocation)$, or equivalently $\avariable_i \in \glabels_k'(\amap(\alocation))$.

    (Meet-points)
    In order to conclude the proof, we show that for all $i,j \in \interval{1}{q}$ and $\alocation \in \gverts$,
    \begin{center}
     $m_q(\avariable_i,\avariable_j) \in \glabels_k(\alocation)$
    if and only if $m_q(\avariable_i,\avariable_j) \in \glabels_k'(\amap(\alocation))$.
    \end{center}
    If $\sem{m_q(\avariable_i,\avariable_j)}^q_{\astore,\aheap}$ is undefined, then so is $\sem{m_q(\avariable_j,\avariable_i)}^q_{\astore,\aheap}$ (by def.)
    and by $\pair{\astore}{\aheap} \approx^q_\alpha \pair{\astore'}{\aheap'}$, so are $\sem{m_q(\avariable_i,\avariable_j)}^q_{\astore',\aheap'}$ and $\sem{m_q(\avariable_j,\avariable_i)}^q_{\astore',\aheap'}$.
    By Lemma~\ref{lemma-labels}, if $\sem{m_q(\avariable_i,\avariable_j)}^q_{\astore,\aheap}$ is undefined, then so are
    $\sem{m_q(\avariable_i,\avariable_j)}^q_{\astore,\aheap_k}$, $\sem{m_q(\avariable_j,\avariable_i)}^q_{\astore,\aheap_k}$,
    $\sem{m_q(\avariable_i,\avariable_j)}^q_{\astore',\aheap'_k}$ and $\sem{m_q(\avariable_j,\avariable_i)}^q_{\astore',\aheap'_k}$.
    Otherwise, if $\sem{m_q(\avariable_i,\avariable_j)}^q_{\astore,\aheap} = \alocation$ then by definition $\sem{m_q(\avariable_j,\avariable_i)}^q_{\astore,\aheap} = \alocation'$ for some $\alocation' \in \gverts$ and
    moreover $\sem{m_q(\avariable_i,\avariable_j)}^q_{\astore',\aheap'} = \amap(\alocation)$ and
    $\sem{m_q(\avariable_j,\avariable_i)}^q_{\astore',\aheap'} = \amap(\alocation')$ by $\pair{\astore}{\aheap} \approx^q_\alpha \pair{\astore'}{\aheap'}$.
    Let $\avariableter^i$ (resp. $\avariableter^j$) be the program variable in $\{\avariable_1,\dots,\avariable_q\}$ such that $\astore(\avariableter^i)$ is the first location corresponding to a program variable that is reachable from $\sem{m_q(\avariable_i,\avariable_j)}^q_{\astore,\aheap}$ (resp. $\sem{m_q(\avariable_j,\avariable_i)}^q_{\astore,\aheap}$), itself included.
    The characterisation of such a program variable $\avariableter^i$
    can be captured by the formula
    $\firstvar(m_q(\avariable_i,\avariable_j),\avariableter^i)$ defined as the following Boolean combination of test formulae:
    $$
    m_q(\avariable_i,\avariable_j) = \avariableter^i \vee \ \ \ \
    \bigvee_{\mathclap{\substack{\aterm_1, \ldots, \aterm_n \in \Terms{q}, n>1 \\
    {\rm pairwise \ distinct} \ \aterm_1, \ldots, \aterm_{n-1}, \\ \aterm_1 = m_q(\avariable_i,\avariable_j), \aterm_n = \avariableter^i}}}
    \textstyle\bigwedge_{\delta \in [1,n-1]} \sees_q(\aterm_{\delta},\aterm_{\delta+1}) \geq 1 \land \bigwedge_{\substack{m < n\hphantom{m.}\\ k \in [1,q]}} \avariable_k \neq \aterm_m$$
    Indeed, this formula is satisfied only by memory states where there is a (possibly empty) path from the location corresponding to $m_q(\avariable_i,\avariable_j)$ to the location $\overline{\alocation}$ corresponding to $\avariableter^i$ so that each labelled location in the path, apart from $\overline{\alocation}$, does not correspond to the interpretation of  such program variables.
    Now,  we  recall the taxonomy of meet-points, where the rightmost case from Figure~\ref{figure-taxonomy} is split into three cases, highlighting additional cases depending on $\avariableter^i$ and $\avariableter^j$.
    \begin{center}
       \begin{tabular*}{\linewidth}{c @{\extracolsep{\fill}}  cccc}
         \scalebox{0.68}{
           \begin{tikzpicture}[baseline]
             \node[dot,label=above:$\avariable_i$] (i) at (0,0) {};
             \node[dot,label=left:{$m_q$($\avariable_i$,$\avariable_j$)\\$m_q$($\avariable_j$,$\avariable_i$)}] (m) [below right = 1.5cm and 0.5cm of i] {};
             \node[dot,label=above:$\avariable_j$] (j) [above right=1.5cm and 0.5cm of m] {};
             \node[dot,label=below:{$\avariableter^i=\avariableter^j$\\$\avariableter^i$ not inside a loop}] (k) [below of=m] {};

             \draw[reach] (i) -- (m);
             \draw[reach] (j) -- (m);
             \draw[reach] (m) -- (k);
           \end{tikzpicture}
         }
       &
        \scalebox{0.68}{
           \begin{tikzpicture}[baseline]
             \node[dot,label=above:$\avariable_i$] (i) at (0,0) {};
             \node[dot,label=left:{$m_q$($\avariable_i$,$\avariable_j$)\\$m_q$($\avariable_j$,$\avariable_i$)}] (m) [below right = 1.5cm and 0.5cm of i] {};
             \node[dot,label=above:$\avariable_j$] (j) [above right=1.5cm and 0.5cm of m] {};
             \node[dot] (mid) [below=0.8cm of m] {};
             \node[dot,label=below:{$\avariableter^i=\avariableter^j$}] (k) [below=1.65cm of m] {};

             \draw[reach] (i) -- (m);
             \draw[reach] (j) -- (m);
             \draw[reach] (m) -- (mid);
             \draw[reach] (mid) to [out=-35,in=35] (k);
             \draw[pto] (k) edge [out=155,in=-155] node [left] {$+$} (mid);
           \end{tikzpicture}
        }
       &
        \scalebox{0.68}{
           \begin{tikzpicture}[baseline]
             \node[dot,label=above:$\avariable_i$] (i) at (0,0) {};
             \node[dot,label=left:{$m_q$($\avariable_i$,$\avariable_j$)}] (m1) [below right = 1.5cm and 0.4cm of i] {};
             \node[dot,label=below:{$m_q$($\avariable_j$,$\avariable_i$)}] (m2) [below right = 1.7cm and 0.4cm of m1] {};
             \node[dot,label=above:$\avariable_j$] (j) [above right=1.5cm and 0.5cm of m2] {};
             \node[dot,label=left:{$\avariableter^i=\avariableter^j$}] (k) [below right= 2.7cm and 0.1 of i] {};

             \draw[reach] (i) -- (m1);
             \draw[pto] (m1) edge node [above right] {$+$} (m2);
             \draw[reach] (j) -- (m2);
             \draw[reach] (m2) to [out=180,in=-75] (k);
             \draw[pto] (k) edge [bend left=30] node [left] {$+$} (m1);
           \end{tikzpicture}
        }
       &
        \scalebox{0.68}{
           \begin{tikzpicture}[baseline]
             \node[dot,label=above:$\avariable_j$] (i) at (0,0) {};
             \node[dot,label=right:{$m_q$($\avariable_j$,$\avariable_i$)}] (m1) [below left = 1.5cm and 0.4cm of i] {};
             \node[dot,label=below:{$m_q$($\avariable_i$,$\avariable_j$)}] (m2) [below left = 1.7cm and 0.4cm of m1] {};
             \node[dot,label=above:$\avariable_i$] (j) [above left=1.5cm and 0.5cm of m2] {};
             \node[dot,label=right:{$\avariableter^i=\avariableter^j$}] (k) [below left= 2.7cm and 0.1 of i] {};

             \draw[reach] (i) -- (m1);
             \draw[pto] (m1) edge node [above left] {$+$} (m2);
             \draw[reach] (j) -- (m2);
             \draw[reach] (m2) to [out=0,in=-105] (k);
             \draw[pto] (k) edge [bend right=30] node [right] {$+$} (m1);
           \end{tikzpicture}
        }
       &
        \scalebox{0.68}{
           \begin{tikzpicture}[baseline]
             \node[dot,label=above:$\avariable_i$] (i) at (0,0) {};
             \node[dot,label=left:{$m_q$($\avariable_i$,$\avariable_j$)}] (m1) [below right = 1.5cm and 0.4cm of i] {};
             \node[dot,label=right:{$m_q$($\avariable_j$,$\avariable_i$)}] (m2) [right = 1.8cm of m1] {};
             \node[dot,label=above:$\avariable_j$] (j) [above right=1.5cm and 0.5cm of m2] {};
             \node[dot,label=below:{$\avariableter^j$}] (k) [below right= 0.5cm and 0.9cm of m1] {};
             \node[dot,label=above:{$\avariableter^i$}] (z) [above right = 0.5 and 0.9cm of m1] {};

             \draw[reach] (i) -- (m1);
             \draw[reach] (j) -- (m2);
             \draw[reach] (m2) to [out=-115,in=0] (k);
             \draw[pto] (k) edge [bend left=30] node [below left] {$+$} (m1);
             \draw[reach] (m1) to [out=65,in=180] (z);
             \draw[pto] (z) edge [bend left=30] node [above right] {$+$} (m2);
           \end{tikzpicture}
        }\\
        \\
       (i)
       &
       (ii)
       &
       (iii)
       &
       (iv)
       &
       (v)
     \end{tabular*}
    \end{center}
    Distinct structures satisfy different test formulae, though they all satisfy the formula
    \begin{center}
      $\firstvar(m_q(\avariable_i,\avariable_j),\avariableter^i) \land \firstvar(m_q(\avariable_j,\avariable_i),\avariableter^j)$.
    \end{center}
    For instance, (i) is the only form not satisfying $\reach^+(\avariableter^i,\avariableter^i)$ (recall that
    this form can be expressed as a Boolean combination of test formulae, as shown in Lemma~\ref{lemma-atomic-formulae-test}), whereas (ii) can be distinguished as the only form satisfying
    both $\reach^+(\avariableter^i,\avariableter^i)$ and $m_q(\avariable_i,\avariable_j) = m_q(\avariable_j,\avariable_i)$.
    Moreover, the last structure is the only one satisfying $\avariableter^i \neq \avariableter^j$ whereas (iii) and (iv) can be distinguished with a formula, similar to $\firstvar(m_q(\avariable_i,\avariable_j),\avariableter^i)$, stating that from the location corresponding to $\avariableter^i$, it is possible to reach the location corresponding to $m_q(\avariable_i,\avariable_j)$ without reaching the location corresponding to $m_q(\avariable_j,\avariable_i)$:
    $$
    \bigvee_{\mathclap{\substack{\aterm_1, \ldots, \aterm_n \in \Terms{q},\ n > 1\\
    {\rm pairwise \ distinct} \ \aterm_1, \ldots, \aterm_{n-1}, \\ \aterm_1 = \avariableter^i, \aterm_n = m_q(\avariable_i,\avariable_j)}}}
    \textstyle\bigwedge_{\delta \in [1,n-1]} \sees_q(\aterm_{\delta},\aterm_{\delta+1}) \geq 1 \land \bigwedge_{1 < m < n} m_q(\avariable_j,\avariable_i) \neq \aterm_m.
    $$
    Differently from (iii), the structure (iv) does not satisfy this formula.
    Since $\pair{\astore}{\aheap} \approx^q_\alpha \pair{\astore'}{\aheap'}$, the heaps
    $\aheap$ and $\aheap'$ agree on the structure of every meet-point.
    The proof that for all $i,j \in \interval{1}{q}$ and $\alocation \in \gverts$,
    $m_q(\avariable_i,\avariable_j) \in \glabels_k(\alocation)$ iff
    $m_q(\avariable_i,\avariable_j) \in \glabels_k'(\amap(\alocation))$, essentially relies on \ref{pathproperty}. To show the result we need to proceed by cases, according to the taxonomy.
    \begin{enumerate}[align = left]
    \item[\textbf{Case: $\aheap$ witnesses (i), (ii) or (iv).}]
    Since the heaps
    $\aheap$ and $\aheap'$ agree on the structure of every meet-point,
    $\aheap'$ also witnesses the same form  among (i), (ii) and (iv) as $\aheap$.
    Regarding $\aheap_k$, one of the following holds:
    \begin{itemize}
      \item The path from $\astore(\avariable_i)$ to $\astore(\avariableter^i)$ and the path from $\astore(\avariable_j)$ to $\astore(\avariableter^j)$ are both preserved in $\aheap_k$.
      Then $\aheap_k$ witnesses (i) (ii) or (iv). By~\ref{pathproperty} the same holds for $\aheap_k'$.
      In every case ((i), (ii) and (iv)), we conclude that $\sem{m_q(\avariable_i,\avariable_j)}^q_{\astore,\aheap_k} = \sem{m_q(\avariable_i,\avariable_j)}^q_{\astore,\aheap}$
      and
      $\sem{m_q(\avariable_i,\avariable_j)}^q_{\astore',\aheap_k'} = \sem{m_q(\avariable_i,\avariable_j)}^q_{\astore',\aheap'} = \amap(\sem{m_q(\avariable_i,\avariable_j)}^q_{\astore,\aheap})$.
      Then,
      \begin{align*}
      m_q(\avariable_i,\avariable_j) &\in  \glabels_k(\sem{m_q(\avariable_i,\avariable_j)}^q_{\astore,\aheap})\\ m_q(\avariable_i,\avariable_j) &\in  \glabels_k'(\amap(\sem{m_q(\avariable_i,\avariable_j)}^q_{\astore,\aheap})).
      \end{align*}
      \item The path from $\astore(\avariable_i)$ to $\astore(\avariableter^i)$ or the path from $\astore(\avariable_j)$ to $\astore(\avariableter^j)$ are not preserved in $\aheap_k$.
      Then, again by~\ref{pathproperty}, the same holds for $\aheap_k'$ with respect to $\pair{\astore'}{\aheap'}$.
     By definition of meet-points, both $\sem{m_q(\avariable_i,\avariable_j)}^q_{\astore,\aheap_k}$ and $\sem{m_q(\avariable_i,\avariable_j)}^q_{\astore',\aheap_k'}$ are therefore not defined.
    \end{itemize}
    \item[\textbf{Case: $\aheap$ witnesses (iii).}]
    If instead $\aheap$ and $\aheap'$ witness (iii) then one of the following holds.
    \begin{itemize}
    \item $\aheap_k$ also witnesses (iii), meaning that the path from $\astore(\avariable_i)$ to $\sem{m_q(\avariable_j,\avariable_i)}^q_{\astore,\aheap}$ and the
    path from $\astore(\avariable_j)$ to $\sem{m_q(\avariable_i,\avariable_j)}^q_{\astore,\aheap}$ are both preserved in $\aheap_k$. Then by \ref{pathproperty}
    the heap $\aheap_k'$ also witnesses (iii).
    We have 
    $\sem{m_q(\avariable_i,\avariable_j)}^q_{\astore,\aheap_k} = \sem{m_q(\avariable_i,\avariable_j)}^q_{\astore,\aheap}$ and $\sem{m_q(\avariable_i,\avariable_j)}^q_{\astore',\aheap_k'} = \sem{m_q(\avariable_i,\avariable_j)}^q_{\astore',\aheap'} = \amap(\sem{m_q(\avariable_i,\avariable_j)}^q_{\astore,\aheap})$.
    We conclude that
    \begin{align*}
    m_q(\avariable_i,\avariable_j) &\in  \glabels_k(\sem{m_q(\avariable_i,\avariable_j)}^q_{\astore,\aheap_k})\\
    m_q(\avariable_i,\avariable_j) &\in  \glabels_k'(\amap(\sem{m_q(\avariable_i,\avariable_j)}^q_{\astore,\aheap_k})).
    \end{align*}
    \item The path from $\astore(\avariable_i)$ to $\astore(\avariableter^i)$ and the path from $\astore(\avariable_j)$ to $\astore(\avariableter^i)$ are preserved in $\aheap_k$, whereas the path from $\astore(\avariableter^i)$ to $\sem{m_q(\avariable_i,\avariable_j)}^q_{\astore,\aheap}$ is not preserved in $\aheap_k$ (i.e. at least one of its
    locations is assigned to the other
    heap $\aheap_{3-k}$).
    Then $\aheap_k$ witnesses (i) and by definition of meet-points, it holds that
    \[
      \sem{m_q(\avariable_i,\avariable_j)}^q_{\astore,\aheap_k} = \sem{m_q(\avariable_j,\avariable_i)}^q_{\astore,\aheap_k} = \sem{m_q(\avariable_j,\avariable_i)}^q_{\astore,\aheap}
    \]
    By~\ref{pathproperty}, $\aheap_k'$ also witnesses (i). Then,
    \[
      \sem{m_q(\avariable_i,\avariable_j)}^q_{\astore',\aheap_k'} = \sem{m_q(\avariable_j,\avariable_i)}^q_{\astore',\aheap_k'} = \sem{m_q(\avariable_j,\avariable_i)}^q_{\astore',\aheap'} = \amap(\sem{m_q(\avariable_j,\avariable_i)}^q_{\astore,\aheap}).
    \]
    Then, we conclude that
    \begin{align*}
    m_q(\avariable_i,\avariable_j) &\in  \glabels_k(\sem{m_q(\avariable_j,\avariable_i)}^q_{\astore,\aheap})\\ m_q(\avariable_i,\avariable_j) &\in  \glabels_k'(\amap(\sem{m_q(\avariable_j,\avariable_i)}^q_{\astore,\aheap})).
    \end{align*}
    \item The path from $\astore(\avariable_i)$ to $\astore(\avariableter^i)$ or the path from $\astore(\avariable_j)$ to $\astore(\avariableter^i)$ are not preserved in $\aheap_k$.
    Then, by~\ref{pathproperty}, the same holds for $\aheap_k'$ with respect to $\pair{\astore'}{\aheap'}$.
    By definition of meet-points, both $\sem{m_q(\avariable_i,\avariable_j)}^q_{\astore,\aheap_k}$ and $\sem{m_q(\avariable_i,\avariable_j)}^q_{\astore',\aheap_k'}$ are not defined.
    \end{itemize}
    \item[\textbf{Case: $\aheap$ witnesses (v).}]
    Lastly, suppose  that $\aheap$ and $\aheap'$ witness (v). One of the following holds.
    \begin{itemize}
      \item The path from $\astore(\avariable_i)$ to $\astore(\avariableter^i)$ and the path from $\astore(\avariable_j)$ to $\astore(\avariableter^i)$ are both preserved in $\aheap_k$. Then, depending on whether or not the path from $\astore(\avariableter^i)$ to $\sem{m_q(\avariable_j,\avariable^i)}^q_{\astore,\aheap}$ is also preserved, $\aheap_k'$ witnesses (i) or (v).
      From~\ref{pathproperty}, the same holds for $\aheap_k'$ (where $\aheap_k'$ witnesses (i) iff $\aheap_k$ witnesses (i)).
      In both cases ((i) and (v)), it holds that $\sem{m_q(\avariable_i,\avariable_j)}^q_{\astore,\aheap_k} = \sem{m_q(\avariable_i,\avariable_j)}^q_{\astore,\aheap}$
      and
      $\sem{m_q(\avariable_i,\avariable_j)}^q_{\astore',\aheap_k'} = \sem{m_q(\avariable_i,\avariable_j)}^q_{\astore',\aheap'} = \amap(\sem{m_q(\avariable_i,\avariable_j)}^q_{\astore,\aheap})$.
      Then,
      \begin{align*}
      m_q(\avariable_i,\avariable_j) &\in  \glabels_k(\sem{m_q(\avariable_i,\avariable_j)}^q_{\astore,\aheap})\\ m_q(\avariable_i,\avariable_j) &\in  \glabels_k'(\amap(\sem{m_q(\avariable_i,\avariable_j)}^q_{\astore,\aheap})).
      \end{align*}
      \item The path from $\astore(\avariable_i)$ to $\astore(\avariableter^j)$ and the path from $\astore(\avariable_j)$ to $\astore(\avariableter^j)$ are preserved in $\aheap_k$, whereas the path from $\astore(\avariableter^j)$ to $\sem{m_q(\avariable_i,\avariable_j)}^q_{\astore,\aheap}$ is not preserved in $\aheap_k$.
      Then $\aheap_k$ witnesses (i) and by definition of meet-points, it holds that
      \[
        \sem{m_q(\avariable_i,\avariable_j)}^q_{\astore,\aheap_k} = \sem{m_q(\avariable_j,\avariable_i)}^q_{\astore,\aheap_k} = \sem{m_q(\avariable_j,\avariable_i)}^q_{\astore,\aheap}.
      \]
      By~\ref{pathproperty}, the heap $\aheap_k'$ also witnesses (i). Then,
      \[
        \sem{m_q(\avariable_i,\avariable_j)}^q_{\astore',\aheap_k'} = \sem{m_q(\avariable_j,\avariable_i)}^q_{\astore',\aheap_k'} = \sem{m_q(\avariable_j,\avariable_i)}^q_{\astore',\aheap'} = \amap(\sem{m_q(\avariable_j,\avariable_i)}^q_{\astore,\aheap})
      \]
      Then, we conclude that
      \begin{align*}
      m_q(\avariable_i,\avariable_j) &\in  \glabels_k(\sem{m_q(\avariable_j,\avariable_i)}^q_{\astore,\aheap})\\ m_q(\avariable_i,\avariable_j) &\in  \glabels_k'(\amap(\sem{m_q(\avariable_j,\avariable_i)}^q_{\astore,\aheap})).
      \end{align*}
      \item The path from $\astore(\avariable_i)$ to $\astore(\avariableter^j)$ and the path from $\astore(\avariable_j)$ to $\astore(\avariableter^j)$ are
       not preserved in $\aheap_k$. Then, by~\ref{pathproperty}, the same holds for $\aheap_k'$ with respect to $\pair{\astore'}{\aheap'}$.
      By definition of meet-points, both $\sem{m_q(\avariable_i,\avariable_j)}^q_{\astore,\aheap_k}$ and $\sem{m_q(\avariable_i,\avariable_j)}^q_{\astore',\aheap_k'}$ are
      undefined.
    \end{itemize}
    \end{enumerate}

      We conclude that for every $\alocation \in \gverts$, $m_q(\avariable_i,\avariable_j) \in \glabels_k(\alocation)$ if and only if $m_q(\avariable_i,\avariable_j) \in \glabels_k'(\amap(\alocation))$.

  \item[\textbf{$\amap_k$ satisfies \descLab{(A2)}{mapksatA2}}:] We prove that $\amap_k$ is such that
  \begin{center}
    for every $\alocation \in \gverts_k$, $\alocation \in \galloc_k$ iff $\amap_k(\alocation) \in \galloc_k'$.
  \end{center}
  From (a) in the proof that $\amap_k$ satisfies \ref{mapksatA3}, we know that for every $\alocation \in \gverts_k$, $\amap_k(\alocation) = \amap(\alocation) \in \gverts_k'$.
  ($\Rightarrow$) For the left-to-right direction, let $\alocation \in \gverts_k$ such that $\alocation \in \galloc_k$. By definition we have $\alocation \in \domain{\aheap_k}$ and therefore $\alocation \in \domain{\aheap}$ (from $\aheap_k \sqsubseteq \aheap$).
  Since $\gverts_k \subseteq \gverts$ (Lemma~\ref{lemma-labels}) we have $\alocation \in \gverts$ and hence $\alocation \in \galloc$.
  From
  $\pair{\astore}{\aheap} \approx^q_\alpha \pair{\astore'}{\aheap'}$
  we obtain $\amap(\alocation) \in \galloc'$. Hence, from the construction of $\aheap_k'$
  done in \ref{c1} we have $\amap(\alocation) \in \domain{\aheap_k'}$.
  Moreover, since $\amap_k(\alocation) = \amap(\alocation) \in \gverts_k'$,
  we conclude $\amap_k(\alocation) \in \galloc_k'$.

  ($\Leftarrow$) The right-to-left direction follows by using a similar argument.

  \item[\textbf{$\amap_k$ satisfies \descLab{(A1)}{mapksatA1}}:] We prove that $\amap_k$ is a graph isomorphism between $\pair{\gverts_k}{\gedges_k}$ and $\pair{\gverts_k'}{\gedges_k'}$.
  From (a) in the proof that $\amap_k$ satisfies \ref{mapksatA3} we already know that $\amap_k$ is a bijection from $\gverts_k$ to $\gverts_k'$. Hence, we only need to show that
    for all $\alocation,\alocation' \in \gverts_k$, $\pair{\alocation}{\alocation'} \in \gedges_k$ $\iff$ $\pair{\amap_k(\alocation)}{\amap_k(\alocation')} \in \gedges_k'$.

  ($\Rightarrow$) By definition of $\supportgraph{\astore,\aheap_k}$, we have that $\pair{\alocation}{\alocation'} \in \gedges_k$ if and only if
  $\alocation,\alocation' \in \gverts_k$ and $\aheap_k$ witnesses a non-empty (minimal) path from $\alocation$ to $\alocation'$ such that
  none of the intermediate locations of this path are in $\gverts_k$.
  From (a) in the proof that $\amap_k$ satisfies \ref{mapksatA3} we have
  \begin{center}
    $\alocation \in \gverts_k$ iff $\amap_k(\alocation) \in \gverts_k'$; \qquad $\alocation' \in \gverts_k$ iff $\amap_k(\alocation') \in \gverts_k'$.
  \end{center}
  As moreover $\gverts_k \subseteq \gverts$ and $\gverts_k' \subseteq \gverts'$ (by Lemma~\ref{lemma-labels}), by~\ref{pathproperty} it holds that $\aheap_k$ witnesses a non-empty path from $\alocation$ to $\alocation'$ if and only if $\aheap_k'$ witnesses a non-empty path from $\amap_k(\alocation)$ to $\amap_k(\alocation')$.

  Therefore, to conclude the proof, it
  remains to show that there is a non-empty path from $\amap_k(\alocation)$ to $\amap_k(\alocation')$ in $\aheap_k'$ such that no intermediate locations of this path are in $\gverts_k'$.
  If this property is satisfied, then it is satisfied by the minimal non-empty path  from $\amap_k(\alocation)$ to $\amap_k(\alocation')$, which we suppose being of length $L \geq 1$.
  Let us denote this path with $\sigma(\amap_k(\alocation),\amap_k(\alocation'))$.
  If $L = 1$ then the property is trivially satisfied (as there are no intermediate locations between $\amap_k(\alocation)$ and $\amap_k(\alocation')$).
  Otherwise,
  {\em ad absurdum}, suppose there exists $\alocation'' \in \gverts_k'$, different from $\amap_k(\alocation)$ and $\amap_k(\alocation')$,
  such that for some $L_1,L_2 \geq 1$ such that $L = L_1 + L_2$ we have ${\aheap_k'}^{L_1}(\amap_k(\alocation)) = \alocation''$
  and ${\aheap_k'}^{L_2}(\alocation'') = \amap_k(\alocation')$.
  Then, we show that $\amap_k^{-1}(\alocation'')$ (which by~\ref{mapksatA3} is in $\gverts_k$) is an intermediate location on the minimal non-empty path from $\alocation$ to $\alocation'$, in $\aheap_k$, in contradiction with $\pair{\alocation}{\alocation'} \in \gedges_k$.
  To do so, we start by considering the set $A$ of locations of $\gverts'$ that belongs to the minimal path from $\amap_k(\alocation)$ to $\amap_k(\alocation')$ (excluding these two locations). Formally, 
  \begin{gather*}
    A \egdef \left\{
      \tilde{\alocation} \in \gverts' \
      \middle|\
    \begin{aligned}
      \text{there are } L_1',L_2' \geq 1, \ \text{such that } L = L_1' + L_2',\\ {\aheap_k'}^{L_1'}(\amap_k(\alocation)) = \tilde{\alocation} \text{ and } {\aheap_k'}^{L_2'}(\tilde{\alocation}) = \amap_k(\alocation')
    \end{aligned}
    \right\}
  \end{gather*} 
  For simplicity, let $A =  \{\alocation_1,\dots,\alocation_n\}$.
  Since by Lemma~\ref{lemma-labels} we have $\alocation'' \in \gverts'$, we conclude $\alocation'' \in A$.
  Without loss of generality, we assume $\alocation_i$ and $\alocation_{i+1}$ ($i \in \interval{1}{n-1}$) to be two consecutive labelled locations in the minimal path from $\amap_k(\alocation)$ to $\amap_k(\alocation')$,
  i.e. none of the intermediate locations in the minimal path from $\alocation_i$ to $\alocation_{i+1}$ is labelled. Then, by minimality of $\sigma(\amap_k(\alocation),\amap_k(\alocation'))$ and from the definition of support graph, $\supportgraph{\astore',\aheap'}$ witnesses the following linear structure, where the arrow between two locations $\alocation_i$ and $\alocation_{i+1}$ labelled by $\gedges'$ means that $\gedges'(\alocation_i,\alocation_{i+1})$.
 \begin{center}
      \begin{tikzpicture}[baseline, node distance=1.8cm]
        \node[dot,label=above:{$\amap(\alocation)$}] (zero) at (0,0) {};
        \node[dot,label=above:{$\alocation_1$}] (one) [right of=zero] {};
        \node[dot,label=above:{$\alocation_2$}] (two) [right of=one] {};
        \node[dot,label=above:{$\alocation_{n-1}$}] (n1) [right of=two] {};
        \node[dot,label=above:{$\alocation_{n}$}] (n) [right of=n1] {};
        \node[dot,label=above right:{$\amap(\alocation')$}] (np) [right of=n] {};

         \draw[pto] (zero) edge node [above] {$\gedges'$} (one);
         \draw[pto] (one) edge node [above] {$\gedges'$} (two);
         \draw[pto] (n1) edge node [above] {$\gedges'$} (n);
         \draw[pto] (n) edge node [above] {$\gedges'$} (np);
         \draw[pto] (one) edge node [above] {$\gedges'$} (two);

         \path (two) edge [draw=none] node {$\dots$} (n1);

      \end{tikzpicture}
 \end{center}
 Since $\amap$ is a graph isomorphism from $\pair{\gverts}{\gedges}$ to $\pair{\gverts'}{\gedges'}$ (by~\ref{a1}), $\supportgraph{\astore,\aheap}$ witnesses a similar structure, as depicted below:

 \begin{center}
  \scalebox{1}{
      \begin{tikzpicture}[baseline, node distance=1.8cm]
        \node[dot,label=above:{$\amap(\alocation)$}] (zero) at (0,0) {};
        \node[dot,label=above:{$\alocation_1$}] (one) [right of=zero] {};
        \node[dot,label=above:{$\alocation_2$}] (two) [right of=one] {};
        \node[dot,label=above:{$\alocation_{n-1}$}] (n1) [right of=two] {};
        \node[dot,label=above:{$\alocation_{n}$}] (n) [right of=n1] {};
        \node[dot,label=above right:{$\amap(\alocation')$}] (np) [right of=n] {};
        \node[dot,label=below:{$\alocation$}] (Fzero) [below of=zero] {};
        \node[dot,label=below:{$\amap^{-1}(\alocation_1)$}] (Fone) [right of=Fzero] {};
        \node[dot,label=below:{$\amap^{-1}(\alocation_2)$}] (Ftwo) [right of=Fone] {};
        \node[dot,label=below:{$\amap^{-1}(\alocation_{n-1})$}] (Fn1) [right of=Ftwo] {};
        \node[dot,label=below:{$\amap^{-1}(\alocation_{n})$}] (Fn) [right of=Fn1] {};
        \node[dot,label=below:{$\alocation'$}] (Fnp) [right of=Fn] {};

         \draw[pto] (zero) edge node [above] {$\gedges'$} (one);
         \draw[pto] (one) edge node [above] {$\gedges'$} (two);
         \draw[pto] (n1) edge node [above] {$\gedges'$} (n);
         \draw[pto] (n) edge node [above] {$\gedges'$} (np);
         \draw[pto] (one) edge node [above] {$\gedges'$} (two);
         \draw[pto] (Fzero) edge node [below] {$\gedges$} (Fone);
         \draw[pto] (Fone) edge node [below] {$\gedges$} (Ftwo);
         \draw[pto] (Fn1) edge node [below] {$\gedges$} (Fn);
         \draw[pto] (Fn) edge node [below] {$\gedges$} (Fnp);
         \draw[pto] (Fone) edge node [below] {$\gedges$} (Ftwo);

         \path (two) edge [draw=none] node {$\dots$} (n1);
         \path (Ftwo) edge [draw=none] node {$\dots$} (Fn1);

         \path[dotted,->] (zero) edge node [left] {$\amap^{-1}$} (Fzero);
         \path[dotted,->] (one) edge node [left] {$\amap^{-1}$} (Fone);
         \path[dotted,->] (np) edge node [left] {$\amap^{-1}$} (Fnp);
         \path[dotted,->] (n) edge node [left] {$\amap^{-1}$} (Fn);
         \path[dotted,->] (n1) edge node [left] {$\amap^{-1}$} (Fn1);
         \path[dotted,->] (two) edge node [left] {$\amap^{-1}$} (Ftwo);

      \end{tikzpicture}
  }
 \end{center}
 Since $\amap$ is a bijection, for all distinct $i,j \in \interval{1}{n}$, we have $\amap^{-1}(\alocation_i) \neq \amap^{-1}(\alocation_j)$.
 Thus, every $\amap^{-1}(\alocation_i)$ ($i \in \interval{1}{n}$) is in the minimal path from $\alocation$ to $\alocation'$ in $\aheap$, and therefore from $\pair{\alocation}{\alocation'} \in \gedges_k$ we obtain that
 $\{\amap^{-1}(\alocation_1),\dots,\amap^{-1}(\alocation_n)\} \subseteq \gbtw_k(\alocation,\alocation')$.
 This leads to a contradiction. 
 Indeed, from $\alocation'' \in A$ we conclude $\amap^{-1}(\alocation'') \in \gbtw_k(\alocation,\alocation')$. Then, as 
$\amap_k^{-1}(\alocation'') \in \gverts_k$ (by~\ref{mapksatA3}), we conclude that $\aheap_k$ witnesses a labelled intermediate 
location in the minimal non-empty path from $\alocation$ to $\alocation'$
 in contradiction with $\pair{\alocation}{\alocation'} \in \gedges_k$.
 Hence, in $\aheap_k'$ there are no labelled locations in the minimal path from $\amap_k(\alocation)$ to $\amap_k(\alocation')$, allowing us to conclude that $\pair{\amap_k(\alocation)}{\amap_k(\alocation')} \in \gedges'$.

 ($\Leftarrow$) The right-to-left direction is analogous (thanks to the fact that $\amap$ and $\amap_k$ are bijections).

It remains to check the satisfaction of the conditions  \ref{mapksatA4} and  \ref{mapksatA5} that
involve arithmetical constraints. 

 \item[\textbf{$\amap_k$ satisfies \descLab{(A4)}{mapksatA4}}:]
   We show that for every $\pair{\alocation}{\alocation'} \in \gedges_k$ we have
   \[\min(\alpha_k,\card{\gbtw_k(\alocation,\alocation')}) = \min(\alpha_k,\card{\gbtw_k'(\amap_k(\alocation),\amap_k(\alocation'))}).\]

   First, we show the following intermediate result.

   \begin{enumerate}[label=\textbf{($\amap_k$-A4.\Roman*)}, align = left]
     \item\label{path-prop1} Let $\pair{\alocation}{\alocation'} \in \gedges_k$.
     There is a (possibly empty) set $\{\alocation_1,\dots,\alocation_n\} \subseteq \gverts$ such that
     \vspace{2pt}
     \begin{enumerate}[label=(\alph*), align = left]
       \setlength{\itemsep}{3pt}
       \item\label{path-inter-1}
        By defining $\alocation_0 \egdef \alocation$ and $\alocation_{n+1} \egdef \alocation'$, we have
        that for every $i \in \interval{0}{n}$, $\pair{\alocation_i}{\alocation_{i+1}} \in \gedges$;
       \item\label{path-inter-break}
        $\{\alocation_1,\dots,\alocation_n\} = \gbtw_k(\alocation,\alocation') \cap \gverts$ and $\{\amap(\alocation_1),\dots,\amap(\alocation_n)\} = \gbtw_k'(\amap(\alocation),\amap(\alocation'))  \cap \gverts'$;
       \item\label{path-inter-2}
        $\gbtw_k(\alocation,\alocation') = \{\alocation_1,\dots,\alocation_n\} \cup \bigcup_{i \in \interval{0}{n}} \gbtw(\alocation_i,\alocation_{i+1})$;
       \item\label{path-inter-3}
        $\gbtw_k'(\amap(\alocation),\amap(\alocation')) = \{\amap(\alocation_1),\dots,\amap(\alocation_n)\} \cup \bigcup_{i \in \interval{0}{n}} \gbtw'(\amap(\alocation_i),\amap(\alocation_{i+1}))$.
     \end{enumerate}
     \vspace{4pt}
     \item[(Proof of \ref{path-prop1})]
       The proof follows rather closely the arguments used for the proof of \ref{mapksatA1}.
       Suppose $\pair{\alocation}{\alocation'} \in \gedges_k$. Since $\amap_k$ satisfies \ref{mapksatA1} we have $\pair{\amap_k(\alocation)}{\amap_k(\alocation')} \in \gedges_k'$.
       By~Lemma~\ref{lemma-labels}, $\gverts_k \subseteq \gverts$ and $\gverts_k' \subseteq \gverts'$. Informally, this means that a subheap cannot contain new labelled locations w.r.t.\ the original heap.
       It can however be the case that the set $\gbtw_k(\alocation,\alocation')$ contains locations that are
       labelled w.r.t. $\pair{\astore}{\aheap}$ (i.e. locations in $\gverts$) and that are not labelled for $\pair{\astore}{\aheap_k}$ (by definition of $\gbtw_k(\alocation,\alocation')$).
       More precisely, these locations correspond to meet-points of $\aheap$.
       Hence, let us consider the set
       $
       \{\alocation_1,\dots,\alocation_n \} =  \gbtw_k(\alocation,\alocation') \cap \gverts
       $
       (as required by~\ref{path-inter-break}).
       Let us define $\alocation_0 \egdef \alocation$, $\alocation_{n+1} \egdef \alocation'$.
       Since $\gbtw_k(\alocation,\alocation')$ describes the set of locations of the minimal path from $\alocation$ to $\alocation'$ in $\aheap_k$, and moreover $\aheap_k \sqsubseteq \aheap$,
       there is an ordering on the locations $\alocation_1,\dots,\alocation_n$,
       w.l.o.g.\ say $\alocation_1 < \dots < \alocation_n$ such that
      for every $i \in \interval{0}{n}$ $\pair{\alocation_i}{\alocation_{i+1}} \in \gedges$ (property~\ref{path-inter-1} of \ref{path-prop1}).
      So, $\{\alocation_1,\dots,\alocation_n \}$ is a set of locations in $\gbtw_k(\alocation,\alocation')$ that correspond to meet-points of $\pair{\astore}{\aheap}$.
        Now, as done in the proof of~\ref{mapksatA1}, $\amap$ is a graph isomorphism from $\pair{\gverts}{\gedges}$ to $\pair{\gverts'}{\gedges'}$ and therefore these two structures witness the following correspondence

        \begin{center}
         \scalebox{1}{
             \begin{tikzpicture}[baseline, node distance=1.8cm]
               \node[dot,label=above left:{$\alocation={\alocation_0}$}] (zero) at (0,0) {};
               \node[dot,label=above:{$\alocation_1$}] (one) [right of=zero] {};
               \node[dot,label=above:{$\alocation_2$}] (two) [right of=one] {};
               \node[dot,label=above:{$\alocation_{n-1}$}] (n1) [right of=two] {};
               \node[dot,label=above:{$\alocation_{n}$}] (n) [right of=n1] {};
               \node[dot,label=above right:{$\alocation'={\alocation_{n+1}}$}] (np) [right of=n] {};
               \node[dot,label=below:{$\amap(\alocation)$}] (Fzero) [below of=zero] {};
               \node[dot,label=below:{$\amap(\alocation_1)$}] (Fone) [right of=Fzero] {};
               \node[dot,label=below:{$\amap(\alocation_2)$}] (Ftwo) [right of=Fone] {};
               \node[dot,label=below:{$\amap(\alocation_{n-1})$}] (Fn1) [right of=Ftwo] {};
               \node[dot,label=below:{$\amap(\alocation_{n})$}] (Fn) [right of=Fn1] {};
               \node[dot,label=below:{$\amap(\alocation')$}] (Fnp) [right of=Fn] {};

                \draw[pto] (zero) edge node [above] {$\gedges$} (one);
                \draw[pto] (one) edge node [above] {$\gedges$} (two);
                \draw[pto] (n1) edge node [above] {$\gedges$} (n);
                \draw[pto] (n) edge node [above] {$\gedges$} (np);
                \draw[pto] (one) edge node [above] {$\gedges$} (two);
                \draw[pto] (Fzero) edge node [below] {$\gedges'$} (Fone);
                \draw[pto] (Fone) edge node [below] {$\gedges'$} (Ftwo);
                \draw[pto] (Fn1) edge node [below] {$\gedges'$} (Fn);
                \draw[pto] (Fn) edge node [below] {$\gedges'$} (Fnp);
                \draw[pto] (Fone) edge node [below] {$\gedges'$} (Ftwo);

                \path (two) edge [draw=none] node {$\dots$} (n1);
                \path (Ftwo) edge [draw=none] node {$\dots$} (Fn1);

                \path[dotted,->] (zero) edge node [left] {$\amap$} (Fzero);
                \path[dotted,->] (one) edge node [left] {$\amap$} (Fone);
                \path[dotted,->] (np) edge node [left] {$\amap$} (Fnp);
                \path[dotted,->] (n) edge node [left] {$\amap$} (Fn);
                \path[dotted,->] (n1) edge node [left] {$\amap$} (Fn1);
                \path[dotted,->] (two) edge node [left] {$\amap$} (Ftwo);

             \end{tikzpicture}
         }
        \end{center}
        \descLab{($\star$)}{prop-star-prime}: where in particular for every $i \in \interval{0}{n}$,
        $\pair{\amap(\alocation_i)}{\amap(\alocation_{i+1})} \in \gedges'$.
        Notice that,
        since $\amap$ is a bijection, for all two distinct $i,j \in \interval{1}{n}$, we have $\amap(\alocation_i) \neq \amap(\alocation_j)$.
        Most importantly, as $\alocation,\alocation' \not\in \gbtw_k(\alocation,\alocation')$ (this holds by definition of this set), we conclude $\alocation,\alocation' \not\in\{\alocation_1,\dots,\alocation_n\}$ and therefore
        $\amap(\alocation),\amap(\alocation') \not\in\{\amap(\alocation_1),\dots,\amap(\alocation_n)\}$.
        Thanks to this property, $\{\amap(\alocation_1),\dots,\amap(\alocation_n)\}$ is the set of labelled locations in the minimal path from $\amap(\alocation)$ to $\amap(\alocation')$ in $\aheap'$.
        Equivalently, $\{\amap(\alocation_1),\dots,\amap(\alocation_n)\} = \gbtw_k'(\amap(\alocation),\amap(\alocation'))  \cap \gverts'$ (property~\ref{path-inter-break} of~\ref{path-prop1}).
        Lastly, by \ref{path-inter-1} and \ref{path-inter-break},
        from the fact that
        heaps are functional and
        $\aheap_k$ witnesses a path from $\alocation$ to $\alocation'$, it follows that
        for every $i \in \interval{0}{n}$ $\gbtw{(\alocation_i,\alocation_{i+1}}) \subseteq \gbtw_k(\alocation,\alocation')$.
        Similarly, for every $i \in \interval{0}{n}$ $\gbtw'{(\amap(\alocation_i),\amap(\alocation_{i+1}})) \subseteq \gbtw_k'(\amap(\alocation),\amap(\alocation'))$.
        Therefore, we conclude that
        \vspace{3pt}
        \begin{itemize}
        \item
         $\gbtw_k(\alocation,\alocation') = \{\alocation_1,\dots,\alocation_n\} \cup \bigcup_{i \in \interval{0}{n}} \gbtw(\alocation_i,\alocation_{i+1})$,
        \item
         $\gbtw_k'(\amap(\alocation),\amap(\alocation')) = \{\amap(\alocation_1),\dots,\amap(\alocation_n)\} \cup \bigcup_{i \in \interval{0}{n}} \gbtw'(\amap(\alocation_i),\amap(\alocation_{i+1}))$,
       \end{itemize}
       \vspace{4pt}
       which ends the proof of~\ref{path-prop1}.
       Notice that the two properties \ref{path-inter-2} and \ref{path-inter-3} (now proved) are well depicted by the figure above, as it shows the structure of the minimal path from $\alocation$ to $\alocation'$ in $\aheap$ (and $\aheap_k$), and the structure of the minimal path from $\amap(\alocation)$ to $\amap(\alocation')$ in $\aheap'$ (and $\aheap_k'$).
   \end{enumerate}

   We are now ready to show that $\amap_k$ satisfies \ref{a4}. Let $\pair{\alocation}{\alocation'} \in \gedges_k$ and let $\{\alocation_1,\dots,\alocation_n\} \subseteq \gverts_k$ satisfying the four properties in \ref{path-prop1}.
   Let us define $\alocation_0 \egdef \alocation$, $\alocation_{n+1} \egdef \alocation'$.
   As already stated in the proof of~\ref{path-prop1},
   $\amap$ is a bijection and therefore for all two distinct $i,j \in \interval{1}{n}$, we have $\amap(\alocation_i) \neq \amap(\alocation_j)$.
   Hence,
   \[
     \card{\{\alocation_1,\dots,\alocation_n\}} = \card{\{\amap(\alocation_1),\dots,\amap(\alocation_n)\}} = n.
    \]
  Moreover, by recalling (from Definition~\ref{definition-support-graph}) that
  \begin{itemize}
    \item $\{\galloc,\grem\} \cup \{ \gbtw(\alocation,\alocation') \mid
    (\alocation,\alocation') \in \gedges \}$ is a partition of $\domain{\aheap}$;
    \item $\{\galloc',\grem'\} \cup \{ \gbtw'(\alocation,\alocation') \mid
    (\alocation,\alocation') \in \gedges' \}$ is a partition of $\domain{\aheap'}$,
  \end{itemize}
   we conclude that for every $i \in \interval{0}{n}$,
   \begin{itemize}
    \item $\{\alocation_1,\dots,\alocation_n\} \cap \gbtw(\alocation_i,\alocation_{i+1}) = \emptyset$ and
    for all $j \in \interval{0}{n} \setminus \{i\}$, $\gbtw(\alocation_i,\alocation_{i+1}) \cap \gbtw(\alocation_j,\alocation_{j+1}) = \emptyset$;
    \item $\{\amap(\alocation_1),\dots,\amap(\alocation_n)\} \cap \gbtw'(\amap(\alocation_i),\amap(\alocation_{i+1})) = \emptyset$ and
    for all $j \in \interval{0}{n} \setminus \{i\}$,
    $\gbtw'(\amap(\alocation_i),\amap(\alocation_{i+1})) \cap \gbtw'(\amap(\alocation_j),\amap(\alocation_{j+1})) = \emptyset$.
   \end{itemize}
   Thus, the cardinalities of $\gbtw_k(\alocation,\alocation')$ and $\gbtw_k'(\amap_k(\alocation),\amap_k(\alocation'))$ can be written respectively as
   \begin{align*}
     \card{\gbtw_k(\alocation,\alocation')} = n + &\sum_{\mathclap{i \in \interval{0}{n}}} \card{\gbtw(\alocation_i,\alocation_{i+1})},\\
     \card{\gbtw_k'(\amap_k(\alocation),\amap_k(\alocation'))} = n + &\sum_{\mathclap{i \in \interval{0}{n}}} \card{\gbtw'(\amap(\alocation_i),\amap(\alocation_{i+1}))}.
   \end{align*}
 Now, $\amap$ satisfies~\ref{a4} w.r.t.\ the memory states $\pair{\astore}{\aheap}$ and $\pair{\astore'}{\aheap'}$, and therefore for every $i \in \interval{0}{n}$
  \begin{equation}
    \tag*{\textbf{($\amap$-A4)}}\label{amap-A4}
    \min(\alpha,\card{\gbtw(\alocation_i,\alocation_{i+1})}) =
    \min(\alpha,\card{\gbtw'(\amap(\alocation_i),\amap(\alocation_{i+1}))}).
  \end{equation}
We distinguish two cases. 
 \begin{itemize}
     \item If there exists $i \in \interval{0}{n}$ such that $\card{\gbtw(\alocation_i,\alocation_{i+1})} \geq \alpha$, then
     by \ref{amap-A4} we have $\card{\gbtw'(\amap(\alocation_i),\amap(\alocation_{i+1}))} \geq \alpha$.
     Since $\alpha_k \leq \alpha$, we  conclude
     \[
      \card{\gbtw_k(\alocation,\alocation')} \geq \alpha_k; \qquad
      \card{\gbtw_k'(\amap_k(\alocation),\amap_k(\alocation'))} \geq \alpha_k.
     \]
     \item Otherwise (for all $i \in \interval{0}{n}, \card{\gbtw(\alocation_i,\alocation_{i+1})} < \alpha$)
     from \ref{amap-A4} we conclude that for all $i \in \interval{0}{n}$,
     \[\card{\gbtw(\alocation_i,\alocation_{i+1})} = \card{\gbtw'(\amap(\alocation_i),\amap(\alocation_{i+1}))},\]
     which allows us to conclude that $\card{\gbtw_k(\alocation,\alocation')} = \card{\gbtw_k'(\amap_k(\alocation),\amap_k(\alocation'))}$.
  \end{itemize}
 In both cases we have shown that
 \[\min(\alpha_k,\card{\gbtw_k(\alocation,\alocation')}) = \min(\alpha_k,\card{\gbtw_k'(\amap_k(\alocation),\amap_k(\alocation'))}),\]
 concluding the proof of \ref{a4}.

  \item[\textbf{$\amap_k$ satisfies \descLab{(A5)}{mapksatA5}}:]
    Lastly, we show that
    \[
    \min(\alpha_k, \card{\grem_k}) = \min(\alpha_k, \card{\grem_k'}).
    \]
    We take advantage of the following five intermediate results.
    \begin{enumerate}[label=\textbf{(\Roman*)}, align = left]
    \item\label{rem-prop1} For every $\alocation \in \gverts \cap \domain{\aheap_k}$,
      $\alocation \in \grem_k$ if and only if $\amap(\alocation) \in \grem_k'$.

    \item[(Proof of \ref{rem-prop1})]
    ($\Rightarrow$) Suppose $\alocation \in \gverts$.
    By Definition~\ref{definition-support-graph},
    it holds that $\alocation \in \grem_k$ if and only if
    \begin{enumerate}[label=(\roman*), align = left]
      \item\label{rem-prop-l1} $\alocation \not\in \gverts_k$ and $\alocation \in \domain{\aheap_k}$;
      \item\label{rem-prop-l2} there is no $\pair{\alocation'}{\alocation''} \in \gedges_k$ such that $\alocation \in \gbtw_k(\alocation',\alocation'')$.
    \end{enumerate}
    From (a) in the proof that $\amap_k$ satisfies \ref{mapksatA3}, we have $\amap(\alocation) = \amap_k(\alocation) \not\in \gverts_k'$. By $\alocation \in \domain{\aheap_k}$ and $\alocation \in \grem_k$, from the construction step~\ref{c1} we conclude $\amap(\alocation) \in \domain{\aheap_k'}$.

    Besides, it cannot be that there is
    $\pair{\alocation_1}{\alocation_2} \in \gedges_k'$ such that $\amap(\alocation) \in \gbtw_k'(\alocation_1,\alocation_2)$. Indeed, if that was the case, we can
    apply the proof ad absurdum showed in the left-to-right direction of the proof that
    $\amap_k$ satisfies \ref{mapksatA1}, and derive that 
    $\alocation \in \gbtw_k'(\amap^{-1}_k(\alocation_1),\amap^{-1}_k(\alocation_2))$, in contradiction
    with~\ref{rem-prop-l2}.
    Hence, as we proved that \ref{rem-prop-l1} and \ref{rem-prop-l2} hold for $\amap(\alocation)$, by Definition~\ref{definition-support-graph}, we conclude that $\amap(\alocation) \in \grem_k'$.

    ($\Leftarrow$) The right-to-left direction follows by using similar arguments.

    \item\label{rem-prop-break}  $\card{\grem_k \cap \gverts} = \card{\grem_k' \cap \gverts'}$ and $\card{\grem_k \cap \galloc} = \card{\grem_k' \cap \galloc'}$.
    \item[(Proof of \ref{rem-prop-break})]
    The first equality follows directly from \ref{rem-prop1} and the fact that $\amap$ is a bijection.
    Then, since $\galloc \subseteq \gverts$, $\galloc' \subseteq \gverts'$ (by definition) and
    $\alocation \in \galloc\ \Leftrightarrow\ \amap(\alocation) \in \galloc'$ (from the property~\ref{a2} of $\amap$), the second equality also holds.

    \item\label{rem-prop2} For all $\pair{\alocation}{\alocation'} \in \gedges$,
    \begin{center}
      $\gbtw(\alocation,\alocation') \cap \domain{\aheap_k} \subseteq \grem_k$ if and only if
     $\gbtw'(\amap(\alocation),\amap(\alocation')) \cap \domain{\aheap_k'} \subseteq \grem_k'$.
    \end{center}
    \item[(Proof of \ref{rem-prop2})]
    Recall that $\pair{\alocation}{\alocation'} \in \gedges$ implies $\pair{\amap(\alocation)}{\amap(\alocation')} \in \gedges'$.
    From \ref{cip2}, we notice that
    \begin{center}
      $\gbtw(\alocation,\alocation') \cap \domain{\aheap_k} = \emptyset$ if and only if
      $\gbtw'(\amap(\alocation),\amap(\alocation')) \cap \domain{\aheap_k'} = \emptyset$,
    \end{center}
    as by definition
    $L_k = \gbtw(\alocation,\alocation') \cap \domain{\aheap_k}$
    and $L_k' = \gbtw'(\amap(\alocation),\amap(\alocation')) \cap \domain{\aheap_k'}$ in the case~\ref{c2} of the construction (where we consider $\gbtw(\alocation,\alocation')$).
    Therefore,~\ref{rem-prop2} trivially holds when $\gbtw(\alocation,\alocation') \cap \domain{\aheap_k}$ is empty. In the following, we assume
    $\gbtw(\alocation,\alocation') \cap \domain{\aheap_k}$ and $\gbtw'(\amap(\alocation),\amap(\alocation')) \cap \domain{\aheap_k'}$ to be non-empty.

    $(\Rightarrow)$
    Let $\pair{\alocation}{\alocation'} \in \gedges$ and suppose
    $\gbtw(\alocation,\alocation') \cap \domain{\aheap_k} \subseteq \grem_k$.
    Let us consider a location $\tilde{\alocation} \in \gbtw'(\amap(\alocation),\amap(\alocation')) \cap \domain{\aheap_k'}$.
    We show that $\tilde{\alocation} \in \grem_k'$.
    As shown in the proof of \ref{rem-prop1}, $\tilde{\alocation} \in \grem_k'$ holds if and only if
    \begin{enumerate}[label=(\roman*), align = left]
      \item\label{rem-prop-l3} $\tilde{\alocation} \not\in \gverts_k'$ and $\tilde{\alocation} \in \domain{\aheap_k'}$;
      \item\label{rem-prop-l4} there is no $\pair{\tilde{\alocation_1}}{\tilde{\alocation_2}} \in \gedges_k'$ such that $\tilde{\alocation} \in \gbtw_k'(\tilde{\alocation_1},\tilde{\alocation_2})$.
    \end{enumerate}
    The proof of~\ref{rem-prop-l3} is straightforward:
    by $\tilde{\alocation} \in \gbtw'(\amap(\alocation),\amap(\alocation')) \cap \domain{\aheap_k'}$, we directly conclude that $\tilde{\alocation} \in \domain{\aheap_k'}$ (first condition in~\ref{rem-prop-l3}) and
    $\tilde{\alocation} \in \gbtw'(\amap(\alocation),\amap(\alocation'))$.
    From the latter membership and
    by Definition~\ref{definition-support-graph}, we conclude $\tilde{\alocation} \not\in \gverts'$.
    Lastly, by Lemma~\ref{lemma-labels}, $\gverts_k' \subseteq \gverts'$ and therefore $\tilde{\alocation} \not\in \gverts_k'$ (second condition in~\ref{rem-prop-l3}).
    Let us now prove~\ref{rem-prop-l4}.
    \emph{Ad absurdum}, suppose that there is $\pair{\tilde{\alocation_1}}{\tilde{\alocation_2}} \in \gedges_k'$ such that $\tilde{\alocation} \in \gbtw_k'(\tilde{\alocation_1},\tilde{\alocation_2})$.
    Then, as we have shown that $\amap_k$ satisfies \ref{mapksatA1}, $\pair{\amap^{-1}_k(\tilde{\alocation_1})}{\amap^{-1}_k(\tilde{\alocation_n})} \in \gedges_k$.
    We apply \ref{path-prop1}: there is a set $\{\alocation_1,\dots,\alocation_n\} \subseteq \gverts$ such that
    \vspace{2pt}
    \begin{enumerate}[label=(\alph*), align = left]
      \setlength{\itemsep}{3pt}
      \item\label{rem-path-inter-1}
       by defining $\alocation_0 \egdef \amap^{-1}_k(\tilde{\alocation_1})$ and $\alocation_{n+1} \egdef \amap^{-1}_k(\tilde{\alocation_2})$, for every $i \in \interval{0}{n}$, $\pair{\alocation_i}{\alocation_{i+1}} \in \gedges$;
      \item\label{rem-path-inter-break}
       $\{\alocation_1,\dots,\alocation_n\} = \gbtw_k(\amap^{-1}_k(\tilde{\alocation_1}),\amap^{-1}_k(\tilde{\alocation_2})) \cap \gverts$ and
       $\{\amap(\alocation_1),\dots,\amap(\alocation_n)\} = \gbtw_k'(\tilde{\alocation_1},\tilde{\alocation_2})  \cap \gverts'$;
      \item\label{rem-path-inter-2}
       $\gbtw_k(\amap^{-1}_k(\tilde{\alocation_1}),\amap^{-1}_k(\tilde{\alocation_2})) = \{\alocation_1,\dots,\alocation_n\} \cup \bigcup_{i \in \interval{0}{n}} \gbtw(\alocation_i,\alocation_{i+1})$;
      \item\label{rem-path-inter-3}
       $\gbtw_k'(\tilde{\alocation_1},\tilde{\alocation_2}) = \{\amap(\alocation_1),\dots,\amap(\alocation_n)\} \cup \bigcup_{i \in \interval{0}{n}} \gbtw'(\amap(\alocation_i),\amap(\alocation_{i+1}))$.
    \end{enumerate}
    \vspace{4pt}
    Recalling that
    $\{\galloc',\grem'\} \cup \{ \gbtw'(\alocation,\alocation') \mid
      (\alocation,\alocation') \in \gedges' \}$ is a partition of $\domain{\aheap'}$
    (by Definition~\ref{definition-support-graph}),
    from \ref{rem-path-inter-3} together with $\tilde{\alocation} \in \gbtw'(\amap(\alocation),\amap(\alocation'))$ and $\tilde{\alocation} \in \gbtw_k'(\tilde{\alocation_1},\tilde{\alocation_2})$, we conclude that
    there is $i \in \interval{0}{n}$ such that $\pair{\amap(\alocation)}{\amap(\alocation')} = \pair{\amap(\alocation_i)}{\amap(\alocation_{i+1})}$.
    From \ref{rem-path-inter-2}, we conclude that
    $\gbtw(\alocation,\alocation') \subseteq \gbtw_k(\amap_k^{-1}(\tilde{\alocation_1}),\amap_k^{-1}(\tilde{\alocation_2})).$
    Notice that, by Definition~\ref{definition-support-graph},
    this inclusion also entails $\gbtw(\alocation,\alocation') \cap \domain{\aheap_k} = \gbtw(\alocation,\alocation')$, i.e. every element in $\gbtw(\alocation,\alocation')$ is a memory cell of $\aheap_k$.
    Hence,
    $\gbtw(\alocation,\alocation') \cap \domain{\aheap_k} \subseteq \gbtw_k(\amap_k^{-1}(\tilde{\alocation_1}),\amap_k^{-1}(\tilde{\alocation_2}))$, in contradiction with $\gbtw(\alocation,\alocation') \cap \domain{\aheap_k} \subseteq \grem_k$.
    Indeed, we assumed $\gbtw(\alocation,\alocation') \cap \domain{\aheap_k}$ to be non-empty, and
    by Definition~\ref{definition-support-graph}, $\grem_k \cap \gbtw_k(\amap_k^{-1}(\tilde{\alocation_1}),\amap_k^{-1}(\tilde{\alocation_2})) = \emptyset$.
    Therefore, it cannot be that there is $\pair{\tilde{\alocation_1}}{\tilde{\alocation_2}} \in \gedges_k'$ such that $\tilde{\alocation} \in \gbtw_k'(\tilde{\alocation_1},\tilde{\alocation_2})$.
    As \ref{rem-prop-l3} and \ref{rem-prop-l4} hold, we conclude that
    $\tilde{\alocation} \in \grem_k'$.

    ($\Leftarrow$) The right-to-left direction follows by using similar arguments.

    \item\label{rem-prop-aux2} For every $\pair{\alocation}{\alocation'} \in \gedges$, $\grem_k \cap \gbtw(\alocation,\alocation') = \emptyset$ or $\gbtw(\alocation,\alocation') \cap \domain{\aheap_k} \subseteq \grem_k$.
    Similarly, for every $\pair{\alocation}{\alocation'} \in \gedges'$, $\grem_k' \cap \gbtw'(\alocation,\alocation') = \emptyset$ or $\gbtw'(\alocation,\alocation') \cap \domain{\aheap_k'} \subseteq \grem_k'$.

    \item[(Proof of~\ref{rem-prop-aux2})] Let us prove the first statement (the second one is proven analogously). By Definition~\ref{definition-support-graph}, we have
    \begin{enumerate}[label = (\alph*), align = left]
      \item\label{last-rem-1-aux} $\{\galloc,\grem\} \cup \{ \gbtw(\alocation,\alocation') \mid
      (\alocation,\alocation') \in \gedges \}$ is a partition of $\domain{\aheap}$;
      \item\label{last-rem-2-aux} $\{\galloc_k,\grem_k\} \cup \{ \gbtw_k(\alocation,\alocation') \mid
      (\alocation,\alocation') \in \gedges_k \}$ is a partition of $\domain{\aheap_k}$.
    \end{enumerate}
    Let $\pair{\alocation}{\alocation'} \in \gedges$.
    \emph{Ad absurdum}, suppose that $\grem_k \cap \gbtw(\alocation,\alocation') \neq \emptyset$ and there is $\tilde{\alocation} \in (\gbtw(\alocation,\alocation') \cap \domain{\aheap_k}) \setminus \grem_k$.
    Since $\tilde{\alocation} \in \gbtw(\alocation,\alocation')$, 
    by Definition~\ref{definition-support-graph}, $\tilde{\alocation}$ is not a labelled location w.r.t. $\pair{\astore}{\aheap}$, i.e.\ $\tilde{\alocation} \not \in \gverts$.
    Moreover, from $\tilde{\alocation} \in \domain{\aheap_k}$ and $\tilde{\alocation} \not\in \grem_k$, by
    \ref{last-rem-2-aux} we conclude that $\tilde{\alocation} \in \galloc_k$ or there is $\pair{\alocation_1'}{\alocation_2'} \in \gedges_k$ such that $\tilde{\alocation} \in \gbtw_k(\alocation_1',\alocation_2')$.
    \begin{itemize}
      \item If $\tilde{\alocation} \in \galloc_k$, then $\tilde{\alocation} \in \gverts_k$ (by Definition~\ref{definition-support-graph}). Recall that,
      by Lemma~\ref{lemma-labels}, $\gverts_k \subseteq \gverts$.
      Therefore, we conclude $\tilde{\alocation} \in \gverts$ (contradiction).
      \item Otherwise, there is $\pair{\alocation_1'}{\alocation_2'} \in \gedges_k$ such that $\tilde{\alocation} \in \gbtw_k(\alocation_1',\alocation_2')$. By~\ref{path-prop1}, there is a set $\{\alocation_1,\dots,\alocation_n\} \subseteq \gverts$ such that
      \vspace{2pt}
      \begin{enumerate}[label=(\alph*), align = left, start = 3]
        \setlength{\itemsep}{3pt}
        \item\label{rem-4-path-1}
         by defining $\alocation_0 \egdef \alocation_1'$ and $\alocation_{n+1} \egdef \alocation_2'$, we have
         that for every $i \in \interval{0}{n}$, $\pair{\alocation_i}{\alocation_{i+1}} \in \gedges$;
        \item\label{rem-4-path-2}
         $\{\alocation_1,\dots,\alocation_n\} = \gbtw_k(\alocation_1',\alocation_2') \cap \gverts$;
        \item\label{rem-4-path-3}
         $\gbtw_k(\alocation_1',\alocation_2') = \{\alocation_1,\dots,\alocation_n\} \cup \bigcup_{i \in \interval{0}{n}} \gbtw(\alocation_i,\alocation_{i+1})$.
      \end{enumerate}
      Since $\tilde{\alocation} \not\in \gverts$, from~\ref{rem-4-path-3} we conclude that there is $i \in \interval{0}{n}$ such that $\tilde{\alocation} \in \gbtw(\alocation_i,\alocation_{i+1})$.
      Then, by \ref{last-rem-1-aux} and $\tilde{\alocation} \in \gbtw(\alocation,\alocation')$, it must be that $\pair{\alocation_i}{\alocation_{i+1}} = \pair{\alocation}{\alocation'}$.
      Hence, from \ref{rem-4-path-3}, $\gbtw(\alocation,\alocation') \subseteq \gbtw_k(\alocation_1',\alocation_2')$.
      By~\ref{last-rem-2-aux} $\grem_k \cap \gbtw_k(\alocation_1',\alocation_2') = \emptyset$ and therefore we conclude that $\grem_k \cap \gbtw(\alocation,\alocation') = \emptyset$, in contradiction with the hypothesis $\grem_k \cap \gbtw(\alocation,\alocation') \neq \emptyset$.
    \end{itemize}
    As in both cases we reached a contradiction, it must be that
    $\grem_k \cap \gbtw(\alocation,\alocation') = \emptyset$ or $\gbtw(\alocation,\alocation') \cap \domain{\aheap_k} \subseteq \grem_k$, concluding the proof.

    \item\label{rem-prop3} It holds that $\grem \cap \domain{\aheap_k} \subseteq \grem_k$ and $\grem' \cap \domain{\aheap_k'} \subseteq \grem_k'$.
    \item[(Proof of \ref{rem-prop3})] The proof is rather immediate. Let us prove that
     $\grem \cap \domain{\aheap_k} \subseteq \grem_k$ (the proof of $\grem' \cap \domain{\aheap_k'} \subseteq \grem_k'$ is analogous).
     By Definition~\ref{definition-support-graph}, we have
     \begin{enumerate}[label = (\alph*), align = left]
       \item\label{last-rem-1} $\{\galloc,\grem\} \cup \{ \gbtw(\alocation,\alocation') \mid
       (\alocation,\alocation') \in \gedges \}$ is a partition of $\domain{\aheap}$;
       \item\label{last-rem-2} $\{\galloc_k,\grem_k\} \cup \{ \gbtw_k(\alocation,\alocation') \mid
       (\alocation,\alocation') \in \gedges_k \}$ is a partition of $\domain{\aheap_k}$.
     \end{enumerate}
     Suppose $\alocation \in \grem \cap \domain{\aheap_k}$.
     \emph{Ad absurdum}, suppose $\alocation \not \in \grem_k$.
     Then, as $\alocation \in \domain{\aheap_k}$, from \ref{last-rem-2} it holds that
     $\alocation \in \galloc_k$ or there is $(\alocation_1,\alocation_2) \in \gedges_k$ such that
     $\alocation \in \gbtw_k(\alocation_1,\alocation_2)$.
     \begin{itemize}
     \item
       If $\alocation \in \galloc_k$ then $\alocation \in \gverts_k$ (by Definition~\ref{definition-support-graph}).
       By Lemma~\ref{lemma-labels}, $\gverts_k \subseteq \gverts$ and therefore $\alocation \in \gverts$.
       Together with $\aheap_k \sqsubseteq \aheap$ (hence $\domain{\aheap_k} \subseteq \domain{\aheap}$), the fact that $\alocation \in \gverts$ implies $\alocation \in \galloc$ (again by Definition~\ref{definition-support-graph}).
       From~\ref{last-rem-1} we then obtain $\alocation \not \in \grem$, a contradiction.
     \item
      Otherwise, let $(\alocation_1,\alocation_2) \in \gedges_k$ such that
      $\alocation \in \gbtw_k(\alocation_1,\alocation_2)$.
      Therefore, by Definition~\ref{definition-support-graph}
      there are $L_1,L_2 \geq 1$ such that $\aheap_k^{L_1}(\alocation_1) = \alocation$ and $\aheap_k^{L_2}(\alocation) = \alocation_2$.
      Since $\aheap_k \sqsubseteq \aheap$, we conclude that
      there are $L_1,L_2 \geq 1$ such that $\aheap^{L_1}(\alocation_1) = \alocation$ and $\aheap^{L_2}(\alocation) = \alocation_2$.
      Now, notice that $(\alocation_1,\alocation_2) \in \gedges_k$ implies $\alocation_1,\alocation_2 \in \gverts_k$ (by Definition~\ref{definition-support-graph}) and therefore
      $\alocation_1,\alocation_2 \in \gverts$ (since by Lemma~\ref{lemma-labels} $\gverts_k \subseteq \gverts$).
      We  conclude that $\alocation$ is an intermediate location in the path of $\aheap$ from
      the labelled location $\alocation_1$ to the labelled location $\alocation_2$.
      Hence, 
      either $\alocation \in \galloc$ or, by Definition~\ref{definition-support-graph}, there is $\pair{\tilde{\alocation_1}}{\tilde{\alocation_2}} \in \gedges$ such that $\alocation \in \gbtw(\tilde{\alocation_1},\tilde{\alocation_2})$.
      In both cases,
      by~\ref{last-rem-1} we have $\alocation \not \in \grem$, a contradiction.
     \end{itemize}
     In both cases we conclude that $\alocation \not\in \grem_k$ cannot hold, ending the proof of \ref{rem-prop3}.
  \end{enumerate}

    By taking advantage of \ref{rem-prop1}--\ref{rem-prop3}, we are now 
    ready to show that $\amap_k$ satisfies \ref{a5}, i.e.
    \[
      \min(\alpha_k, \card{\grem_k}) = \min(\alpha_k, \card{\grem_k'}).
    \]
    First of all, we recall that by Definition~\ref{definition-support-graph} we have
    \begin{enumerate}[label = (\alph*), align = left]
      \item\label{a5-rem-L} $\{\galloc,\grem\} \cup \{ \gbtw(\alocation,\alocation') \mid
      (\alocation,\alocation') \in \gedges \}$ is a partition of $\domain{\aheap}$;
      \item\label{a5-rem-R} $\{\galloc',\grem'\} \cup \{ \gbtw'(\alocation,\alocation') \mid
      (\alocation,\alocation') \in \gedges' \}$ is a partition of $\domain{\aheap'}$.
    \end{enumerate}
    Since $\grem_k \subseteq \domain{\aheap_k} \subseteq \domain{\aheap}$ and $\grem_k' \subseteq \domain{\aheap_k'} \subseteq \domain{\aheap'}$, we can use~\ref{a5-rem-L} and~\ref{a5-rem-R} to decompose $\grem_k$ and $\grem_k'$ as follows:
    \begin{itemize}[align = left]
      \item $\grem_k = (\grem_k \cap \galloc) \cup (\grem_k \cap \grem) \cup \bigcup_{{\pair{\alocation}{\alocation'} \in \gedges}} \big( \grem_k \cap \gbtw(\alocation,\alocation')\big)$,
      \item $\grem_k' = (\grem_k' \cap \galloc') \cup (\grem_k' \cap \grem') \cup \bigcup_{{\pair{\alocation}{\alocation'} \in \gedges'}} \big( \grem_k' \cap \gbtw'(\alocation,\alocation')\big)$,
    \end{itemize}
    where all unions are performed on disjoint sets.
    By \ref{rem-prop-aux2},
    $\bigcup_{{\pair{\alocation}{\alocation'} \in \gedges}} ( \grem_k \cap \gbtw(\alocation,\alocation'))$
    and
    $\bigcup_{{\pair{\alocation}{\alocation'} \in \gedges'}} ( \grem_k' \cap \gbtw'(\alocation,\alocation'))$
    are respectively equal to
    \[
      \qquad\qquad
      \bigcup_{\mathclap{\substack{\pair{\alocation}{\alocation'} \in \gedges\\ \gbtw(\alocation,\alocation') \cap \domain{\aheap_k} \subseteq \grem_k}}}
      \big(\gbtw(\alocation,\alocation') \cap \domain{\aheap_k}\big)
      \qquad\qquad\qquad\qquad
      \bigcup_{\mathclap{\substack{\pair{\alocation}{\alocation'} \in \gedges'\\ \gbtw'(\alocation,\alocation') \cap \domain{\aheap_k'} \subseteq \grem_k'}}}
      \big(\gbtw'(\alocation,\alocation') \cap \domain{\aheap_k'}\big)
    \]
    Since $\grem_k \subseteq \domain{\aheap_k}$ and $\grem_k' \subseteq \domain{\aheap_k'}$,
    by~\ref{rem-prop3} we have $\grem_k \cap \grem = \grem \cap \domain{\aheap_k}$ and $\grem_k' \cap \grem' = \grem' \cap \domain{\aheap_k'}$. Hence, the previous equalities can be rewritten as
    \begin{align*}
      \grem_k = \ & (\grem_k \cap \galloc) \cup (\grem \cap \domain{\aheap_k}) \cup \bigcup_{\mathclap{\substack{\pair{\alocation}{\alocation'} \in \gedges\\ \gbtw(\alocation,\alocation') \cap \domain{\aheap_k} \subseteq \grem_k}}}
      \big(\gbtw(\alocation,\alocation') \cap \domain{\aheap_k}\big)\\
      \grem_k' = \ & (\grem_k' \cap \galloc') \cup (\grem' \cap \domain{\aheap_k'}) \cup \bigcup_{\mathclap{\substack{\pair{\alocation}{\alocation'} \in \gedges'\\ \gbtw'(\alocation,\alocation') \cap \domain{\aheap_k'} \subseteq \grem_k'}}}
      \big(\gbtw'(\alocation,\alocation') \cap \domain{\aheap_k'}\big).
    \end{align*}
    Since all unions are performed on disjoint sets (by \ref{a5-rem-L} and \ref{a5-rem-R}),
    we conclude that $\card{\grem_k}$ and $\card{\grem_k'}$ satisfy the following equalities:
    \begin{align*}
      \qquad \
      \card{\grem_k} = \ & \card{\grem_k \cap \galloc} + \card{\grem \cap \domain{\aheap_k}} +  \!\!\!\sum_{\mathclap{\substack{\pair{\alocation}{\alocation'} \in \gedges\\ \gbtw(\alocation,\alocation') \cap \domain{\aheap_k} \subseteq \grem_k}}}
      \!\card{\gbtw(\alocation,\alocation') \cap \domain{\aheap_k}}\\
      \card{\grem_k'} = \ & \card{\grem_k'\cap \galloc'} + \card{\grem'\!\cap \domain{\aheap_k}} + \!\!\!\sum_{\mathclap{\substack{\pair{\alocation}{\alocation'} \in \gedges'\\ \gbtw'(\alocation,\alocation') \cap \domain{\aheap_k'} \subseteq \grem_k'}}}
      \!\card{\gbtw'(\alocation,\alocation') {\cap} \domain{\aheap_k'}}
    \end{align*}
    Now, we are able to compare the values $\min(\alpha_k, \card{\grem_k})$ and 
    $\min(\alpha_k, \card{\grem_k'})$.
    By using the two equalities above and recalling that $\min(A,B+C) = \min(A,B + \min(A,C))$,
    the following set of equalities holds for $\min(\alpha_k, \card{\grem_k})$:
    \begin{align*}
        \min(\alpha_k, \card{\grem_k}) &= \min(\alpha_k, \card{\grem_k \cap \galloc} + R)\\
        R &= \min(\alpha_k, \card{\grem \cap \domain{h_k}} + I)\\
        I &= \min(\alpha_k, \sum_{\mathclap{\substack{\pair{\alocation}{\alocation'} \in \gedges\\ \gbtw(\alocation,\alocation') \cap \domain{\aheap_k} \subseteq \grem_k}}} I_{\alocation,\alocation'})\\
        I_{\alocation,\alocation'} &= \min(\alpha_k,\card{\gbtw(\alocation,\alocation') \cap \domain{\aheap_k}}) &\text{where}~\pair{\alocation}{\alocation'} \in \gedges.
    \end{align*}
    The same can be done for $\min(\alpha_k, \card{\grem_k'})$:
    \begin{align*}
        \min(\alpha_k, \card{\grem_k'}) &= \min(\alpha_k, \card{\grem_k' \cap \galloc'} + R')\\
        R' &= \min(\alpha_k, \card{\grem' \cap \domain{h_k'}} + I')\\
        I' &= \min(\alpha_k, \sum_{\mathclap{\substack{\pair{\alocation}{\alocation'} \in \gedges'\\ \gbtw'(\alocation,\alocation') \cap \domain{\aheap_k'} \subseteq \grem_k'}}} I_{\alocation,\alocation'}')\\
        I_{\alocation,\alocation'}' &= \min(\alpha_k,\card{\gbtw'(\alocation,\alocation') \cap \domain{\aheap_k'}})
        &\text{where}~\pair{\alocation}{\alocation'} \in \gedges'.
    \end{align*}
    By \ref{cip2}, for every $\pair{\alocation}{\alocation'} \in \gedges$ we have $I_{\alocation,\alocation'} = I_{\amap(\alocation),\amap(\alocation')}'$.
    From \ref{rem-prop2} together with the fact that $\amap$ witnesses a graph isomorphism between $\pair{\astore}{\aheap}$ and $\pair{\astore'}{\aheap'}$
    we conclude that
    \begin{align*}
      &\{\pair{\alocation}{\alocation'} \in \gedges' \mid \gbtw'(\alocation,\alocation') \cap \domain{\aheap_k'} \subseteq \grem_k' \}\\
      = \ &
      \{\pair{\amap(\alocation)}{\amap(\alocation')} \mid \pair{\alocation}{\alocation'} \in \gedges, \gbtw(\alocation,\alocation') \cap \domain{\aheap_k} \subseteq \grem_k \},
    \end{align*}
    which allows us to conclude that $I = I'$.
    Furthermore, by \ref{crp2}, it follows also that $R = R'$.
    Lastly, by \ref{rem-prop-break} we conclude:
    \[\min(\alpha_k, \card{\grem_k}) = \min(\alpha_k, \card{\grem_k'}).\]

\end{enumerate}


This concludes the proof: for the support graphs of $h_k$ and $h_k'$, $\amap$ restricted to the domain $\gverts_k$ and the codomain $\gverts_k'$ satisfies all conditions of Lemma~\ref{lemma-test-graph} with respect to $\alpha_k$.
Therefore, it holds that $\pair{\astore}{\aheap_1}  \approx^q_{\alpha_1} \pair{\astore'}{\aheap'_1}$ and
$\pair{\astore}{\aheap_2}  \approx^q_{\alpha_2} \pair{\astore'}{\aheap'_2}$.
\end{proof}

\end{document}